\newcommand{\typeof}{1} %

\newcommand{\longv}[1]{\ifthenelse{\equal{\typeof}{0}}{}{#1}}
\newcommand{\shortv}[1]{\ifthenelse{\equal{\typeof}{0}}{#1}{}}
\newcommand{\longshortv}[2]{\ifthenelse{\equal{\typeof}{0}}{#2}{#1}}
\newcommand{\drop}[1]{\ifthenelse{\equal{\typeof}{0}}{}{}}

\documentclass[10pt,a4paper]{article}

\usepackage{color}

\newenvironment{varitemize}
{
\begin{list}{\labelitemii}
{\setlength{\itemsep}{0pt}
 \setlength{\topsep}{0pt}
 \setlength{\parsep}{0pt}
 \setlength{\partopsep}{0pt}
 \setlength{\leftmargin}{15pt}
 \setlength{\rightmargin}{0pt}
 \setlength{\itemindent}{0pt}
 \setlength{\labelsep}{5pt}
 \setlength{\labelwidth}{10pt}
}}
{
 \end{list}
}

 \usepackage[utf8]{inputenc}
 \usepackage{mathrsfs}
 \usepackage{amsmath}
 \usepackage{amssymb}
 \usepackage{bussproofs}
 \usepackage{graphicx}
\usepackage{stmaryrd}
\usepackage{chngpage}
\usepackage{tikz-qtree}
\usepackage{wrapfig}
\usepackage{url}
\usepackage{a4wide}

\newcounter{numberone}

\renewcommand{\epsilon}{\varepsilon}
\renewcommand{\phi}{\varphi}
\renewcommand{\vec}{\overrightarrow}
\newcommand{\intterm}[1]{\underset{\bar{}}{#1}}

\newcommand{\bnf}{\;::=\;}
\newcommand{\sep}{\ \ \big|\ \ }
 
\newcommand{\smallsep}{\ \ \vert\ \ }
\newcommand{\proves}{\vdash}
\newcommand{\typsep}{\,:\,}

\newcommand{\supremumnat}[1]{\operatorname{sup}_{n \in \NN} \left( #1 \right)}
\newcommand{\semantics}[1]{[\![\,#1\,]\!]}

\newcommand{\valdec}{\,\overset{\mathit{VD}}{=}\,}

\newcommand{\NN}{\mathbb{N}}

\newcommand{\RR}{\mathbb{R}}

\newcommand{\indexone}{i}
\newcommand{\indextwo}{j}
\newcommand{\indexthree}{k}
\newcommand{\indexfour}{l}
\newcommand{\indexfive}{g}
\newcommand{\indexsix}{h}

\newcommand{\indexsetone}{\mathcal{I}}
\newcommand{\indexsettwo}{\mathcal{J}}
\newcommand{\indexsetthree}{\mathcal{K}}
\newcommand{\indexsetfour}{\mathcal{L}}
\newcommand{\indexsetfive}{\mathcal{G}}
\newcommand{\indexsetsix}{\mathcal{H}}

\newcommand{\probaconv}[2]{\mathit{Pr}_{#1,#2}}

\newcommand{\values}{\termsproba^{V}}

\newcommand{\dom}[1]{\mathit{dom}(#1)}

\newcommand{\abstr}[2]{\lambda #1.#2}

\newcommand{\casenme}{\mathsf{case}}
\newcommand{\caseof}[2]{\casenme \ #1\ \mathsf{of} \ \left\{\, #2 \,\right\}}
\newcommand{\casenat}[3]{\casenme \ #1\ \mathsf{of} \ \left\{\, \natsucc \rightarrow #2 \sep \natzero \rightarrow #3 \,\right\}}
\newcommand{\letrec}[2]{\letrecname\ #1\ = \ #2}
\newcommand{\letrecname}{\mathsf{letrec}}
\newcommand{\letin}[3]{\mathsf{let}\ #1\ = \ #2 \ \mathsf{in}\ #3}
\newcommand{\choice}{\oplus}

\newcommand{\subst}[3]{#1[#3/#2]}

\newcommand{\languageproba}{\lambda_{\choice}}

\newcommand{\PCF}{\ensuremath{\mathsf{PCF}}}

\newcommand{\termsproba}{\Lambda_{\choice}}

\newcommand{\setredtms}[2]{\mathsf{TRed}^{#1}_{#2}}
\newcommand{\setredtmsfin}[2]{\mathsf{TRed}^{#1}_{#2}}
\newcommand{\setredvals}[2]{\mathsf{VRed}^{#1}_{#2}}
\newcommand{\setreddists}[2]{\mathsf{DRed}^{#1}_{#2}}
\newcommand{\settypedclosedterms}[1]{\termsproba \left(#1\right)}
\newcommand{\settypedclosedvalues}[1]{\termsproba^V \left(#1\right)}
\newcommand{\settypedterms}[2]{\termsproba \left(#2,#1\right)}
\newcommand{\settypedvalues}[2]{\termsproba^V \left(#2,#1\right)}
\newcommand{\setdistrtypedclosedterms}[1]{\termsproba^{\sizeone} \left(#1\right)}
\newcommand{\setdistrtypedclosedvalues}[1]{\termsproba^{\sizeone,V} \left(#1\right)}
\newcommand{\setdistrtypedterms}[2]{\termsproba^{\sizeone} \left(#2,#1\right)}
\newcommand{\setdistrtypedvalues}[2]{\termsproba^{\sizeone,V} \left(#2,#1\right)}

\newcommand{\vars}{\mathcal{X}}

\newcommand{\termvector}[1]{\vec{#1}}

\newcommand{\termone}{M}
\newcommand{\termtwo}{N}
\newcommand{\termthree}{L}
\newcommand{\termfour}{P}
\newcommand{\termfive}{R}
\newcommand{\termsix}{T}

\newcommand{\termnonaff}{\termone_{\mathit{NAFF}}}

\newcommand{\varone}{x}
\newcommand{\vartwo}{y}
\newcommand{\varthree}{z}

\newcommand{\funcone}{f}
\newcommand{\functwo}{g}

\newcommand{\valone}{V}
\newcommand{\valtwo}{W}
\newcommand{\valthree}{Z}
\newcommand{\valfour}{X}
\newcommand{\valfive}{Y}

\newcommand{\unfoldings}[1]{\mathit{Unfold}\left(#1\right)}

\newcommand{\dataconstrone}{c}

\newcommand{\rcbv}{\rightarrow_{v}}
\newcommand{\redval}{\Rrightarrow_{v}}

\newcommand{\setoneton}{\left\{1,\,\ldots,\,n\right\}}
\newcommand{\setcard}[1]{\#\left(#1\right)}

\newcommand{\typone}{\sigma}
\newcommand{\typtwo}{\tau}
\newcommand{\typthree}{\theta}

\newcommand{\expectype}[1]{\mathbb{E}\left(#1\right)}

\newcommand{\distrtypone}{\mu}
\newcommand{\distrtyptwo}{\nu}
\newcommand{\distrtypthree}{\xi}
\newcommand{\distrtypfour}{\rho}

\newcommand{\simpletypone}{\kappa}
\newcommand{\simpletyptwo}{\kappa'}
\newcommand{\simpletypthree}{\kappa''}

\newcommand{\refines}{\,::\,}

\newcommand{\typarrow}{\rightarrow}

\newcommand{\spine}[1]{\operatorname{spine}\left(#1\right)}

\newcommand{\contextone}{\Gamma}
\newcommand{\contexttwo}{\Delta}
\newcommand{\contextsep}{\,\vert\,}
\newcommand{\contextsizedone}{\Gamma}
\newcommand{\contextsizedtwo}{\Delta}
\newcommand{\contextsizedthree}{\Xi}
\newcommand{\contextdistrone}{\Theta}
\newcommand{\contextdistrtwo}{\Psi}
\newcommand{\contextdistrthree}{\Phi}
\newcommand{\contextsumdisj}{\uplus}

\newcommand{\underlying}[1]{\langle #1\rangle }
\newcommand{\positive}[2]{#1\ \mathsf{pos}\ #2}
\newcommand{\negative}[2]{#1\ \mathsf{neg}\ #2}

\newcommand{\sizeleq}{\preccurlyeq}
\newcommand{\sizevars}{\mathcal{S}}
\newcommand{\sizevarone}{\mathfrak{i}}
\newcommand{\sizevartwo}{\mathfrak{j}}
\newcommand{\sizevarthree}{\mathfrak{l}}

\newcommand{\sizeone}{\mathfrak{s}}
\newcommand{\sizetwo}{\mathfrak{r}}
\newcommand{\sizethree}{\mathfrak{t}}
\newcommand{\sizefour}{\mathfrak{u}}

\newcommand{\sizeinf}{\infty}

\newcommand{\sizesucc}[1]{\widehat{#1}}
\newcommand{\sizesuccit}[2]{\widehat{#1}^{^{#2}}}

\newcommand{\subtypeleq}{\sqsubseteq}

\newcommand{\seone}{\rho}

\newcommand{\sesem}[2]{\llbracket #1\rrbracket_{#2}}

\newcommand{\setredopentms}[2]{\mathsf{OTRed}^{#1}_{#2}}
\newcommand{\setredopenvals}[2]{\mathsf{OVRed}^{#1}_{#2}}

\newcommand{\states}{Q}
\newcommand{\mainstate}{q_{\alpha}}
\newcommand{\zerostate}{q_{\mathit{zero}}}
\newcommand{\transitionzero}{\delta^{=0}}
\newcommand{\transitionsupzero}{\delta^{>0}}
\newcommand{\probatransitionzero}{P^{=0}}
\newcommand{\probatransitionsupzero}{P^{>0}}

\newcommand{\markovvertices}{V}
\newcommand{\markovverticeone}{v}
\newcommand{\markovtransition}{\mapsto}
\newcommand{\markovproba}{\mathit{Pr}}

\newcommand{\distribs}{\mathcal{P}}
\newcommand{\distrsum}[1]{\sum\,#1}
\newcommand{\distrleq}{\preccurlyeq}

\newcommand{\distrone}{\mathscr{D}}
\newcommand{\distroneval}{\mathscr{D}_{|V}}
\newcommand{\distroneterm}{\mathscr{D}_{|T}}
\newcommand{\distrtwo}{\mathscr{E}}
\newcommand{\distrthree}{\mathscr{F}}

\newcommand{\closedtypeddistrone}{\mathscr{T}}
\newcommand{\closedtypeddistrtwo}{\mathscr{U}}
\newcommand{\closedtypeddistrthree}{\mathscr{V}}

\newcommand{\distrelts}[1]{\left\{\,#1 \,\right\}}
\newcommand{\pseudorep}[1]{\left[\,#1 \,\right]}

\newcommand{\supp}[1]{\mathcal{S}(#1)}

\newcommand{\termbiased}{\termone_{\mathit{BIAS}}}
\newcommand{\termunbiased}{\termone_{\mathit{UNB}}}
\newcommand{\termnat}{\termone_{\mathit{EXP}}}

\newcommand{\Nat}{\mathsf{Nat}}

\newcommand{\natsucc}{\mathsf{S}}
\newcommand{\natzero}{\mathsf{0}}

\newcommand{\setone}{X}

\newtheorem{theorem}{Theorem}
\newtheorem{proposition}{Proposition}
\newtheorem{lemma}{Lemma}
\newtheorem{definition}{Definition}
\newtheorem{example}{Example}
\newtheorem{corollary}{Corollary}

\newenvironment{proof}{\begin{trivlist}
       \item[\hskip \labelsep {\bfseries Proof.}]}{\hfill $\Box$ \end{trivlist}}

\usepackage{ifthen}

  \author{Ugo Dal Lago\and Charles Grellois}
\title{Probabilistic Termination\\ by Monadic Affine Sized Typing\\(Long Version)}

\begin{document}
\maketitle

\begin{abstract}
  We introduce a system of monadic affine sized types, which
  substantially generalise usual sized types, and allows this way to
  capture probabilistic higher-order programs which terminate almost
  surely. Going beyond plain, strong normalisation without losing
  soundness turns out to be a hard task, which cannot be accomplished
  without a richer, quantitative notion of types, but also without
  imposing some affinity constraints. The proposed type system is
  powerful enough to type classic examples of probabilistically
  terminating programs such as random walks. The way typable programs
  are proved to be almost surely terminating is based on reducibility,
  but requires a substantial adaptation of the technique.
\end{abstract}
\section{Introduction}

Probabilistic models are more and more pervasive in computer 
science~\cite{manning-schutze:foundations-statistical-language-processing,pearl:machine-learning,thrun:robotic-mapping}.
Moreover, the concept of algorithm, originally assuming
determinism, has been relaxed so as to allow probabilistic evolution
since the very early days of theoretical computer 
science~\cite{shannon-schapiro:computability-proba-machines}. All
this has given impetus to research on probabilistic programming
languages, which however have been studied at a large scale only in
the last twenty years, following advances in randomized 
computation~\cite{motwani-raghavan:randomized-algorithms},
cryptographic protocol verification~\cite{barthe-gregoire-beguelin:certicrypt,barthe-gregoire-heraud-beguelin:easycrypt},
and machine learning~\cite{goodman-et-al:church-language-generative-models}.
Probabilistic programs can be seen as ordinary
programs in which specific instructions are provided to make the
program evolve probabilistically rather than deterministically.
The typical example are instructions for sampling from a given
distribution toolset, or for performing probabilistic choice.

One of the most crucial properties a program should satisfy is
\emph{termination}: the execution process should be guaranteed to
end. In (non)deterministic computation, this is easy to formalize,
since any possible computation path is only considered qualitatively,
and termination is a boolean predicate on programs: any
non-deterministic program either terminates -- in must or may sense -- or
it does not. In probabilistic programs, on the other hand, any
terminating computation path is attributed a probability, and thus
termination becomes a \emph{quantitative} property. It is therefore natural
to consider a program terminating when its terminating paths form a set
of measure one or, equivalently, when it terminates with maximal
probability.  This is dubbed ``almost sure termination'' (AST for
short) in the literature~\cite{bournez-kirchner:proba-rewrite-strategies},
and many techniques for automatically and
semi-automatically checking programs for AST have been introduced in
the last years~\cite{esparza-gaiser-kiefer:probabilistic-termination-using-patterns,fioriti-hermans:probabilistic-termination,chatterjee-et-al:analysis-proba-termination,chatterjee-et-al:termination-analysis-proba-positivstellensatz}.
All of them, however, focus on imperative programs;
while probabilistic functional programming languages are nowadays
among the most successful ones in the realm of probabilistic
programming~\cite{goodman-et-al:church-language-generative-models}.
It is not clear at all whether the existing techniques
for imperative languages could be easily applied to functional ones,
especially when higher-order functions are involved.

In this paper, we introduce a system of monadic affine sized types for
a simple probabilistic $\lambda$-calculus with recursion
and show that it guarantees
the AST property for all typable programs. The type system, described
in Section \ref{sect:sizedtypes}, can be seen as a non-trivial
variation on Hughes et al.'s sized
types~\cite{hughes-pareto-sabry:sized-types}, whose main novelties are
the following:
\begin{varitemize}
\item
  Types are generalised so as to be \emph{monadic},
  this way encapsulating the kind of information we need to type
  non-trivial examples. This information, in particular, is taken
  advantage of when typing recursive programs.
\item
  Typing rules are \emph{affine}:
  higher-order variables cannot be freely duplicated. This is quite
  similar to what happens when characterising polynomial time
  functions by restricting higher-order languages akin to the
  $\lambda$-calculus~\cite{hofmann:mixed-modal-linear-lambda-calc}.
  Without affinity, the type system is bound to be unsound for AST.
\end{varitemize}
The necessity of both these variations is discussed in Section
\ref{sect:necessity} below.  The main result of this paper is that
typability in monadic affine sized types entails AST, a property which
is proved using an adaptation of the Girard-Tait reducibility
technique~\cite{girard-taylor-lafont:proofs-and-types}.  This
adaptation is technically involved, as it needs substantial
modifications allowing to deal with possibly infinite and
probabilistic computations. In particular, every reducibility set
must be parametrized by a quantitative parameter $p$
guaranteeing that terms belonging to this set terminate with
probability at least $p$.  The idea of parametrizing such sets already
appears in work by the first author and
Hofmann~\cite{dal-lago-hofmann:realizability-models-icc}, in which a
notion of realizability parametrized by resource monoids is
considered. These realizability models are however studied in relation
to linear logic and to the complexity of normalisation, and do not fit
as such to our setting, even if they provided some crucial
inspiration. In our approach, the fact that recursively-defined terms
are AST comes from a continuity argument on this parameter: we can
prove, by unfolding such terms, that they terminate with probability
$p$ for every $p<1$, and continuity then allows to take the limit and
deduce that they are AST. This soundness result is technically
speaking the main contribution of this paper, and is described in Section
\ref{sect:reducibility}.

\shortv{An extended version with more details and proofs is available
  \cite{dal-lago-grellois:monadic-affine-sized-types-full}.}
\subsection{Related Works}
Sized types have been originally introduced by Hughes, Pareto, and
Sabry~\cite{hughes-pareto-sabry:sized-types} in the context of
reactive programming. A series of papers by Barthe and
colleagues~\cite{barthe-et-al:type-based-termination,barthe-gregoire-riba:type-based-termination-tutorial,barthe-gregoire-riba:type-based-termination-sized-products}
presents sized types in a way similar to the one we will adopt here,
although still for a deterministic functional language. Contrary to the other works on sized types,
their type system is proved to admit a decidable type inference, see the unpublished tutorial~\cite{barthe-gregoire-riba:type-based-termination-tutorial}.
Abel developed independently of Barthe and colleagues a similar type system
featuring size informations~\cite{abel:termination-checking-with-types}.
These three lines of work allow polymorphism, arbitrary inductive data constructors, and ordinal sizes, so
that data such as infinite trees can be manipulated. These three features will be absent of our system, in order to focus the challenge on the
treatment of probabilistic recursive programs.
Another interesting approach is the one of 
Xi's Dependent ML~\cite{xi:dependent-types-program-termination},
in which a system of lightweight dependent types allows a more liberal treatment of the notion of size,
over which arithmetic or conditional operations may in particular be applied,
see~\cite{abel:termination-checking-with-types} for a detailed comparison. This type system
is well-adapted for practical termination checking, but does not handle ordinal sizes either.
Some works along these lines are able to deal with coinductive data, as well as
inductive ones~\cite{hughes-pareto-sabry:sized-types,barthe-et-al:type-based-termination,abel:termination-checking-with-types}.
They are related to Amadio and Coupet-Grimal's work on guarded types ensuring
productivity of infinite structures such as streams~\cite{amadio-coupet-grimal:guard-condition-type-theory}.
None of these works deal with probabilistic computation, and in particular
with almost sure termination.

There has been a lot of interest, recently, about probabilistic
termination as a verification problem in the context of imperative
programming~\cite{esparza-gaiser-kiefer:probabilistic-termination-using-patterns,fioriti-hermans:probabilistic-termination,chatterjee-et-al:analysis-proba-termination,chatterjee-et-al:termination-analysis-proba-positivstellensatz}.
All of them deal, invariably, with some form of while-style language without 
higher-order functions. A possible approach is to reduce AST for probabilistic
programs to termination of non-deterministic programs~\cite{esparza-gaiser-kiefer:probabilistic-termination-using-patterns}.
Another one is to extend the concept of ranking function to the probabilistic case.
Bournez and Garnier obtained in this way the notion of Lyapunov ranking function~\cite{bournez-garnier:proving-past},
but such functions capture a notion more restrictive than AST: \emph{positive} almost
sure termination, meaning that the program is AST and terminates in expected finite time.
To capture AST, the notion of ranking supermartingale~\cite{chakarov-sankaranarayanan:proba-programs-martingales}
has been used. Note that the use of ranking supermartingales allows to deal with programs which are both probabilistic
and non-deterministic~\cite{fioriti-hermans:probabilistic-termination,chatterjee-et-al:termination-analysis-proba-positivstellensatz}
and even to reason about programs with real-valued variables~\cite{chatterjee-et-al:analysis-proba-termination}.

Some recent work by Cappai, the first author, and Parisen
Toldin~\cite{cappai-dal-lago:equivalences-metrics-ptime,dal-lago-parisen-tolding:ho-carac-probaptime}
introduce type systems ensuring that all typable programs can be
evaluated in probabilistic polynomial time. This is too restrictive
for our purposes. On the one hand, we aim at termination, and
restricting to polynomial time algorithms would be an overkill. On the
other hand, the above-mentioned type systems guarantee that the length
of \emph{all} probabilistic branches are uniformly bounded 
(by the same polynomial). In our setting, this would restrict the focus
to terms in which infinite computations are forbidden, while we
simply want the set of such computations to have probability 0.

\section{Why is Monadic Affine Typing Necessary?}\label{sect:necessity}
In this section, we justify the design choices that guided
us in the development of our type system. As we will see, the nature
of AST requires a significant and non-trivial extension
of the system of sized types originally introduced to ensure
termination in the deterministic case~\cite{hughes-pareto-sabry:sized-types}.
\paragraph*{Sized Types for Deterministic Programs.}
The simply-typed $\lambda$-calculus endowed with a typed recursion
operator $\letrecname$ and appropriate constructs for the natural
numbers, sometimes called \PCF, is already Turing-complete, so that there is
no hope to prove it strongly normalizing.
Sized types~\cite{hughes-pareto-sabry:sized-types} refine the simple
type system by enriching base types with annotations, so as to ensure
the termination of any recursive definition. Let us explain the idea
of sizes in the simple, yet informative case in which the base type is
$\Nat$. Sizes are defined by the grammar
$$
\sizeone \ \ \bnf\ \ \sizevarone \sep \sizeinf \sep \sizesucc{\sizeone}
$$
where $\sizevarone$ is a size variable and $\sizesucc{\sizeone}$ is
the successor of the size $\sizeone$ --- with
$\sizesucc{\sizeinf}\,=\,\sizeinf$.  These sizes permit to consider
decorations $\Nat^\sizeone$ of the base type $\Nat$, whose elements
are natural numbers of size at most $\sizeone$.  The type system
ensures that the only constant value of type
$\Nat^{\sizesucc{\sizevarone}}$ is $\natzero$, that the only constant
values of type $\Nat^{\sizesucc{\sizesucc{\sizevarone}}}$ are
$\natzero$ or $\intterm{1}\,=\,\natsucc\ \natzero$, and so on. The
type $\Nat^{\sizeinf}$ is the one of all natural numbers, and is therefore
often denoted as $\Nat$.

The crucial rule of the sized type system, which we present here
following Barthe et al.~\cite{barthe-et-al:type-based-termination},
allows one to type recursive definitions as follows:
\begin{equation}
\label{eq:original-fix-rule}
\AxiomC{$\contextone,\funcone \typsep \Nat^\sizevarone \typarrow \typone \proves \termone \typsep \Nat^{\sizesucc{\sizevarone}} \typarrow 
\subst{\typone}{\sizevarone}{\sizesucc{\sizevarone}}$}
\AxiomC{$\sizevarone \text{ pos } \typone$}
\BinaryInfC{$\contextone \proves \letrec{\funcone}{\termone} \typsep \Nat^{\sizeone} \typarrow 
\subst{\typone}{\sizevarone}{\sizeone}$}
\DisplayProof
\end{equation}
This typing rule ensures that, to recursively define the function
$\funcone\,=\,\termone$, the term $\termone$ taking an input of size
$\sizesucc{\sizevarone}$ calls $\funcone$ on inputs of \emph{strictly lesser}
size $\sizevarone$. This is for instance the case when typing the
program 
$$
\termone_{\mathit{DBL}}\ \ =\ \ \letrec{\funcone}{\abstr{\varone}{\casenat{\varone}{\abstr{\vartwo}{\natsucc\ \natsucc \ (\funcone\ \vartwo)}}{\natzero}}}
$$
computing recursively the double of an input integer, as the hypothesis of the fixpoint rule 
in a typing derivation of $\termone_{\mathit{DBL}}$ is
$$
\funcone \typsep \Nat^{\sizevarone} \typarrow \Nat 
\proves \abstr{\varone}{\casenat{\varone}{\abstr{\vartwo}{\natsucc\ \natsucc \ (\funcone\ \vartwo)}}{\natzero}}
\typsep \Nat^{\sizesucc{\sizevarone}} \typarrow \Nat 
$$
The fact that $\funcone$ is called on an input $\vartwo$ of strictly
lesser size $\sizevarone$ is ensured by the rule typing the $\casenme$
construction:
$$
\AxiomC{$\contextone \proves \varone \typsep \Nat^{\sizesucc{\sizevarone}}$}
\AxiomC{$\contextone \proves \abstr{\vartwo}{\natsucc\ \natsucc \ (\funcone\ \vartwo)} \typsep \Nat^\sizevarone \typarrow \Nat$}
\AxiomC{$\contextone \proves \natzero \typsep \Nat$}
\TrinaryInfC{$\contextone \proves \casenat{\varone}{\abstr{\vartwo}{\natsucc\ \natsucc \ (\funcone\ \vartwo)}}{\natzero} \typsep \Nat$}
\DisplayProof
$$
where $\contextone\ =\ \funcone \typsep \Nat^{\sizevarone} \typarrow
\Nat ,\ \varone\typsep \Nat^{\sizesucc{\sizevarone}}$.  The soundness
of sized types for strong normalization allows to conclude that
$\termone_{\mathit{DBL}}$ is indeed SN.

\paragraph*{A Naïve Generalization to Probabilistic Terms.}
The aim of this paper is to obtain a probabilistic,
\emph{quantitative} counterpart to this soundness result for sized
types.  Note that unlike the result for sized types, which was
focusing on \emph{all} reduction strategies of terms, we only consider
a \emph{call-by-value} calculus\footnote{Please notice that choosing a reduction
strategy is crucial in a probabilistic setting, otherwise one risks
getting nasty forms of non-confluence~\cite{dal-lago-zorzi:probabilistic-operational-semantics-lambda-calculus}.}.  Terms can now
contain a probabilistic choice operator $\choice_p$, such that
$\termone\, \choice_p\, \termtwo$ reduces to the term $\termone$ with
probability $p\in\RR_{[0,1]}$, and to $\termtwo$ with probability $1-p$.  The
language and its operational semantics will be defined more extensively in
Section~\ref{sect:proba-language}.  Suppose for the moment that we
type the choice operator in a naïve way:
$$
 \AxiomC{$\contextsizedone \ \proves\ \termone \typsep \typone$}
 \AxiomC{$\contextsizedone \ \proves\ \termtwo \typsep \typone$}
 \LeftLabel{Choice $\quad$}
 \BinaryInfC{$\contextsizedone\  \proves\ \termone \choice_p \termtwo \typsep \typone$}
 \DisplayProof
$$
On the one hand, the original system of sized types features
subtyping, which allows some flexibility to ``unify'' the types of
$\termone$ and $\termtwo$ to $\typone$. On the other hand,
it is easy to realise
that \emph{all} probabilistic branches would have to be terminating,
without any hope of capturing interesting AST programs: nothing has
been done to capture the \emph{quantitative} nature of probabilistic
termination. An instance of a term which is not strongly normalizing
but which is almost-surely terminating --- meaning that it normalizes with
probability 1 --- is
\begin{equation}
 \label{eq:termbiased}
 \resizebox{0.9\textwidth}{!}{
 $\termbiased \ =\ \left(\letrec{\funcone}{\abstr{\varone}{\casenat{\varone}{\abstr{\vartwo}{\funcone(\vartwo) \choice_{\frac{2}{3}} \left( 
\funcone(\natsucc\,\natsucc\,\vartwo))\right)}}{\natzero}}}\right)\ \intterm{n}$}
\end{equation}
simulating a \emph{biased random walk} which, on $\varone = m+1$, goes
to $m$ with probability $\frac{2}{3}$ and to $m+2$ with probability
$\frac{1}{3}$.  The naïve generalization of the sized type system only
allows us to type the body of the recursive definition as follows:
\begin{equation}
\label{eq:naive-type1}
\funcone \typsep \ \Nat^{\sizesucc{\sizesucc{\sizevarone}}} \rightarrow \Nat^{\sizeinf} \ \proves\ 
 \abstr{\vartwo}{\funcone(\vartwo) \choice_{\frac{2}{3}} \left( 
\funcone(\natsucc\,\natsucc\,\vartwo))\right)} \ \typsep\ \Nat^{\sizesucc{\sizevarone}} \rightarrow \Nat^{\sizeinf}
\end{equation}
and thus does not allow us to deduce any relevant information on the
\emph{quantitative} termination of this term: nothing tells us
that the recursive call $\funcone(\natsucc\,\natsucc\,\vartwo)$
is performed with a relatively low probability.
\paragraph*{A Monadic Type System.}
Along the evaluation of $\termbiased$, there is \emph{indeed} a
quantity which decreases during each recursive call to the function
$\funcone$: the \emph{average} size of the input on which the call is
performed. Indeed, on an input of size $\sizesucc{\sizevarone}$,
$\funcone$ calls itself on an input of smaller size
$\sizevarone$ with probability $\frac{2}{3}$, and on an input of
greater size $\sizesucc{\sizesucc{\sizevarone}}$ with probability only
$\frac{1}{3}$. To capture such a relevant \emph{quantitative}
information on the recursive calls of $\funcone$, and with the aim to
capture almost sure termination, we introduce a \emph{monadic} type
system, in which \emph{distributions of types} can be used to type in a
finer way the functions to be used recursively.  Contexts
$\contextsizedone \contextsep \contextdistrone$ will be generated by a
context $\contextsizedone$ attributing sized types to any number of
variables, while $\contextdistrone$ will attribute a
\emph{distribution} of sized types to at most \emph{one} variable ---
typically the one we want to use to recursively define a function.  In
such a context, terms will be typed by a distribution type, formed by
combining the Dirac distributions of types introduced in the Axiom
rules using the following rule for probabilistic choice:
$$
   \AxiomC{$\contextsizedone \contextsep \contextdistrone\ \proves\ \termone \typsep \distrtypone$}
 \AxiomC{$\contextsizedone \contextsep \contextdistrtwo\ \proves\ \termtwo \typsep \distrtyptwo$}
 \AxiomC{$\underlying{\distrtypone}\,=\,\underlying{\distrtyptwo}$}
 \LeftLabel{Choice $\quad$}
 \TrinaryInfC{$\contextsizedone \contextsep \contextdistrone \choice_p \contextdistrtwo\ 
 \proves\ \termone \choice_p \termtwo \typsep \distrtypone \choice_p \distrtyptwo$}
 \DisplayProof
 $$
The guard condition
$\underlying{\distrtypone}\,=\,\underlying{\distrtyptwo}$ ensures that
$\distrtypone$ and $\distrtyptwo$ are distributions of types
decorating of the \emph{same} simple type.  Without this condition,
there is no hope to aim for a decidable type inference algorithm.
\paragraph*{The Fixpoint Rule.}
Using these monadic types, instead of the insufficiently informative typing (\ref{eq:naive-type1}),
we can derive the sequent
  \begin{equation}
  \label{eq:monadic-typing-ex1}
  \resizebox{0.9\textwidth}{!}{
 $\funcone \typsep \ \left\{
 \left(\Nat^{\sizevarone} \rightarrow \Nat^{\sizeinf}\right)^{\frac{2}{3}},\ 
 \left(\Nat^{\sizesucc{\sizesucc{\sizevarone}}} \rightarrow \Nat^{\sizeinf}\right)^{\frac{1}{3}}\right\}
\ \proves\
 \abstr{\vartwo}{\funcone(\vartwo) \choice_{\frac{2}{3}} \left( 
\funcone(\natsucc\,\natsucc\,\vartwo))\right)} \ \typsep\ \Nat^{\sizesucc{\sizevarone}} \rightarrow \Nat^{\sizeinf}$}
  \end{equation}
in which the type of $\funcone$ contains finer information on the sizes of arguments over which it
is called recursively, and with which probability. This information enables us to perform a first
switch from a qualitative to a quantitative notion of termination: we will adapt the hypothesis
\begin{equation}
 \label{eq:hyp-fix-original-sized-types}
 \contextone,\funcone \typsep \Nat^\sizevarone \typarrow \typone \proves \termone \typsep \Nat^{\sizesucc{\sizevarone}} \typarrow 
\subst{\typone}{\sizevarone}{\sizesucc{\sizevarone}}
\end{equation}
of the original fix rule (\ref{eq:original-fix-rule}) of sized types, expressing that $\funcone$ is called on an argument of size one
less than the one on which $\termone$ is called, to a condition meaning that there is probability $1$
to call $\funcone$ on arguments of a lesser size \emph{after enough iterations of recursive calls}.
We therefore define a random walk associated to the distribution type $\distrtypone$ of $\funcone$,
the \emph{sized walk} associated to $\distrtypone$, and which is as follows for the typing (\ref{eq:monadic-typing-ex1}):
\begin{varitemize}
 \item the random walk starts on $1$, corresponding to the size $\sizesucc{\sizevarone}$,
 \item on an integer $n+1$, the random walk jumps to $n$ with probability $\frac{2}{3}$ and to $n+2$ with probability $\frac{1}{3}$,
 \item $0$ is stationary: on it, the random walk loops.
\end{varitemize}
This random walk -- as all sized walks will be -- is an instance of 
\emph{one-counter Markov decision problem}~\cite{bradzdil-et-al:one-counter-markov-decision-processes},
so that it is decidable in polynomial time whether the walk reaches $0$ with probability $1$.
We will therefore replace the hypothesis
(\ref{eq:hyp-fix-original-sized-types}) of the $\letrecname$ rule by the quantitative counterpart we just sketched,
obtaining
$$
\AxiomC{$\distrelts{\left(\Nat^{\sizeone_{\indextwo}} \typarrow \subst{\distrtyptwo}{\sizevarone}{\sizeone_{\indextwo}}
\right)^{p_{\indextwo}} \sep \indextwo \in \indexsettwo} \text{induces an AST sized walk}$}
\noLine
\UnaryInfC{$\contextsizedone \contextsep
\funcone \typsep \distrelts{\left(\Nat^{\sizeone_{\indextwo}} \typarrow \subst{\distrtyptwo}{\sizevarone}{\sizeone_{\indextwo}}
\right)^{p_{\indextwo}} \sep \indextwo \in \indexsettwo} \proves \valone \typsep 
\Nat^{ \sizesucc{\sizevarone}} \typarrow \subst{\distrtyptwo}{\sizevarone}{\sizesucc{\sizevarone}}$}
\LeftLabel{$\letrecname$ $\quad$}
\UnaryInfC{$\contextsizedone,\,\contextsizedtwo \contextsep \contextdistrone \proves \letrec{\funcone}{\valone} \typsep  
\Nat^{\sizetwo} \typarrow \subst{\distrtyptwo}{\sizevarone}{\sizetwo}$}
 \DisplayProof
 $$
where we omit two additional technical conditions to be found in
Section~\ref{sect:monadic-types} and which justify the weakening 
on contexts incorporated to this rule. The resulting type system allows
to type a varieties of examples, among which the following program
computing the geometric distribution over the natural numbers:
\begin{equation}
 \label{eq:termnat}
 \termnat \ =\ \left(\letrec{\funcone}{\abstr{\varone}{\varone \choice_{\frac{1}{2}} \natsucc\ \left(\funcone\ \varone\right)}}\right)\ \natzero
\end{equation}
and for which the decreasing quantity is the size of the set of
probabilistic branches of the term making recursive calls to
$\funcone$.  Another example is the unbiased random walk
\begin{equation}
 \label{eq:termunbiased}
 \resizebox{0.9\textwidth}{!}{$
 \termunbiased \ =\ \left(\letrec{\funcone}{\abstr{\varone}{\casenat{\varone}{\abstr{\vartwo}{\funcone(\vartwo) \choice_{\frac{1}{2}} \left( 
\funcone(\natsucc\,\natsucc\,\vartwo))\right)}}{\natzero}}}\right)\ \intterm{n}$}
\end{equation}
\begin{wrapfigure}[18]{R}{.5\textwidth}
\vspace{-0.4cm}
\fbox{
 \begin{tikzpicture}[baseline=(current  bounding  box.center)]
 \tikzset{level distance=25pt}
 \Tree [.$[1]$ $[0]$ [.$[2\ 2]$ [.$[2\ 1]$ [.$[2]$ [.$[1]$ $[0]$ [.$[2\ 2]$ $\vdots$ ] ] [.$[3\ 3]$ $\vdots$ ] ] 
 [.$[2\ 2\ 2]$ $\vdots$ ] ] [.$[2\ 3\ 3]$ $\vdots$ ] ] ]
 \end{tikzpicture}}
 \caption{A Tree of Recursive Calls.}
 \label{fig:tree-calls-termnonaff}
\end{wrapfigure}
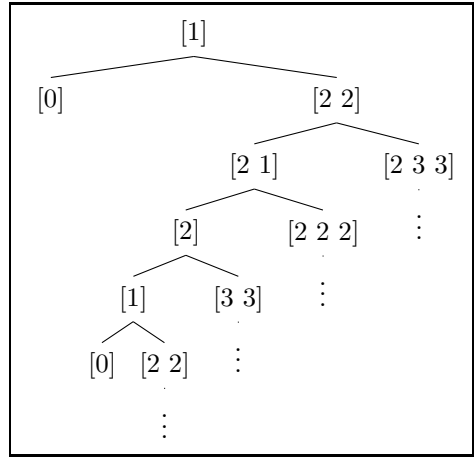
for which there is no clear notion of decreasing measure during
recursive calls, but which yet terminates almost surely, as
witnessed by the sized walk associated to an appropriate derivation in the sized
type system.  We therefore
claim that the use of this external guard condition on associated sized
walks, allowing us to give a general condition of termination, is satisfying
as it both captures an interesting class of examples, and is computable
in polynomial time.

In Section~\ref{sect:reducibility}, we prove that this shift from a
qualitative to a quantitative hypothesis in the type system results in
a shift from the soundness for strong normalization of the original
sized type system to a soundness for its quantitative counterpart:
\emph{almost-sure termination}.
\paragraph{Why Affinity?}
To ensure the soundness of the $\letrecname$ rule, we need one more
structural restriction on the type system.  For the sized walk
argument to be adequate, we must ensure that the recursive calls of
$\funcone$ are indeed precisely modelled by the sized walk, and this is
not the case when considering for instance the following term:
\begin{equation}
 \label{eq:termnonaff}
\resizebox{0.9\textwidth}{!}{$
\termnonaff\ \ =\ \ \left(\letrec{\funcone}{\abstr{\varone}{\casenat{\varone}{\abstr{\vartwo}{\funcone(\vartwo) \choice_{\frac{2}{3}} \left( 
\funcone(\natsucc\,\natsucc\,\vartwo)\,;\,\funcone(\natsucc\,\natsucc\,\vartwo)\right)}}{\natzero}}}\right)\ \intterm{n}$}
\end{equation}
where the sequential composition $;$ is defined in this call-by-value
calculus as
\longv{$$}\shortv{$}
\termone \,;\,\termtwo \ \ =\ \ \left(\abstr{\varone}{\abstr{\vartwo}{\natzero}}\right)\ \termone\ \termtwo
\shortv{$.}\longv{$$}
Note that $\termnonaff$ calls recursively $\funcone$ \emph{twice} in
the right branch of its probabilistic choice, and is not therefore
modelled appropriately by the sized walk associated to its type. In
fact, we would need a generalized notion of random walk to model the
recursive calls of this process; it would be a random walk on
\emph{stacks} of integers. In the case where $n=1$, the
recursive calls to $f$ can indeed be represented by a tree of stacks
as depicted
in Figure~\ref{fig:tree-calls-termnonaff}, where leftmost edges have
probability $\frac{2}{3}$ and rightmost ones $\frac{1}{3}$.  The root
indicates that the first call on $\funcone$ was on the integer
$1$. From it, there is either a call of $\funcone$ on $0$ which
terminates, or \emph{two} calls on $2$ which are put into a stack of
calls, and so on. We could prove that, without the \emph{affine}
restriction we are about to formulate, the term $\termnonaff$ is
typable with monadic sized types and the fixpoint rule we just
designed. However, this term is not almost-surely terminating. Notice,
indeed, that the \emph{sum} of the integers appearing in a stack
labelling a node of the tree in Figure~\ref{fig:tree-calls-termnonaff}
decreases by $1$ when the left edge of probability $\frac{2}{3}$ is
taken, and increases by \emph{at least} $3$ when the right edge of
probability $\frac{1}{3}$ is taken.  It follows that the expected
increase of the sum of the elements of the stack during one step is at
least $-1 \times \frac{2}{3} + 3 \times \frac{1}{3}\,=\,\frac{1}{3} >
0$.
This implies that the probability that $\funcone$ is called on an
input of size $0$ after enough iterations is strictly less than $1$,
so that the term $\termnonaff$ cannot be almost surely terminating.

Such general random processes have stacks as states and are rather
complex to analyse. To the best of the authors' knowledge, they do not seem
to have been considered in the literature.  We also believe that the
complexity of determining whether $0$ can be reached almost surely in
such a process, if decidable, would be very high. This leads us to
the design of an \emph{affine} type system, in which the management of
contexts ensures that a given probabilistic branch of a term may only use at
most once a given higher-order symbol. We do not however formulate
restrictions on variables of simple type $\Nat$, as affinity is only
used on the $\letrecname$ rule and thus on higher-order symbols. 
Remark that this is in the spirit of certain systems from implicit computational
complexity~\cite{hofmann:mixed-modal-linear-lambda-calc,dal-lago:geometry-linear-ho-recursion}.

\section{A Simple Probabilistic Functional Programming Language}\label{sect:proba-language}

We consider the language $\languageproba$, which is an extension of
the $\lambda$-calculus with recursion, constructors for the natural
numbers, and a choice operator. In this section, we introduce this
language and its operational semantics, and use them to define the
crucial notion of \emph{almost-sure termination}.

\paragraph{Terms and Values.}  Given a set of variables $\vars$, terms
and values of the language $\languageproba$ are defined by mutual
induction as follows:
$$
\begin{array}{lrcl}
\text{Terms:} \qquad \quad & \termone, \,\termtwo,\,\ldots & \ \ \bnf\ \ & \valone \sep \valone \ \valone 
\sep \letin{\varone}{\termone}{\termtwo}
\sep \termone \choice_p \termtwo\\
& & &\!\!\!\!\!\!\!\! \sep \caseof{\valone}{\natsucc \rightarrow \valtwo \smallsep \natzero \rightarrow \valthree}\\[0.4cm]
\text{Values:} \qquad \quad &\valone,\,\valtwo,\,\valthree,\,\ldots & \ \ ::=\ \ & \varone \sep \natzero \sep \natsucc\ \valone
\sep \abstr{\varone}{\termone} \sep \letrec{\funcone}{\valone}\\
\end{array}
$$
where $\varone,\,\funcone \in \vars,\,p \in ]0,1[$. When
    $p\,=\,\frac{1}{2}$, we often write $\choice$ as a shorthand for
    $\choice_{\frac{1}{2}}$.  The set of terms is denoted
    $\termsproba$ and the set of values is denoted $\values$. Terms of
    the calculus are assumed to be in A-normal form~\cite{sabry-felleisen:reasoning-about-programs-in-cps}.
    This allows to
    formulate crucial definitions in a simpler way, concentrating in the
    Let construct the study of the probabilistic behaviour of terms.  We
    claim that all traditional constructions can be encoded in this
    formalism.  For instance, the usual application
    $\termone\ \termtwo$ of two terms can be harmlessly recovered via
    the encoding
    $\letin{\varone}{\termone}{\left(\letin{\vartwo}{\termtwo}{\varone\ \vartwo}\right)}$.
    In the sequel, we write $\dataconstrone \ \vec{\valone}$ when a
    value may be either $\natzero$ or of the shape $\natsucc
    \ \valone$.
\paragraph{Term Distributions.}
The introduction of a probabilistic choice operator in the syntax
leads to a \emph{probabilistic} reduction relation.  It is therefore
meaningful to consider the (operational) semantics of a term as a
\emph{distribution} of values modelling the outcome of \emph{all} the
finite probabilistic reduction paths of the term.  For instance, the
term $\termnat$ defined in (\ref{eq:termnat}) evaluates to the term
distribution assigning probability $\frac{1}{2^{n+1}}$ to the value
$\intterm{n}$.  Let us define this notion more formally:

\begin{definition}[Distribution]
A \emph{distribution} on $\setone$ is a function
$\distrone\,:\,\setone \rightarrow [0,1]$ satisfying the constraint
$\distrsum{\distrone}\ =\ \sum_{\varone \in \setone}\ \distrone (\varone)\ \leq\ 1$,
where $\distrsum{\distrone}$ is called the \emph{sum} of the distribution $\distrone$.
We say that $\distrone$ is \emph{proper} precisely when $\distrsum{\distrone}\,=\,1$.
We denote by $\distribs$ the set of all distributions, would they be proper or not.
We define the \emph{support} $\supp{\distrone}$ of a distribution $\distrone$ as:
$\supp{\distrone} =  \left\{ \varone \in \setone \sep \distrone(\varone) > 0 \right\}$.
When $\supp{\distrone}$ consists only of closed terms,
we say that $\distrone$ is a \emph{closed} distribution.
When it is finite, we say that $\distrone$ is a \emph{finite} distribution.
We call \emph{Dirac} a proper distribution $\distrone$ such that $\supp{\distrone}$ is a singleton.
We denote by $0$ the null distribution, mapping every term to the probability $0$.
\end{definition}

When $\setone\,=\,\termsproba$, we say that $\distrone$ is a
\emph{term distribution}.  In the sequel, we will use a more
practical notion of \emph{representation} of distributions, which
enumerates the terms with their probabilities as a family of
assignments. For technical reasons, notably related to the subject
reduction property, we will also need \emph{pseudo-representations},
which are essentially multiset-like decompositions of the
representation of a distribution.

\begin{definition}[Representations and Pseudo-Representations]  
Let $\distrone \in \distribs$ be of support
$\left\{\varone_{\indexone} \sep \indexone \in \indexsetone\right\}$,
where $\varone_{\indexone}\,=\,\varone_{\indextwo}$ implies
$\indexone=\indextwo$ for every $\indexone,\indextwo\in\indexsetone$.
The \emph{representation} of $\distrone$ is the set
$\distrone\,=\,\distrelts{\varone_{\indexone}^{\distrone(\varone_{\indexone})}
  \sep \indexone \in \indexsetone}$ where
$\varone_{\indexone}^{\distrone(\varone_{\indexone})}$ is just an intuitive
way to write the pair $(\varone_{\indexone},\distrone(\varone_{\indexone}))$.    
A \emph{pseudo-representation} of $\distrone$ is any multiset
$\pseudorep{\vartwo_{\indextwo}^{p_{\indextwo}} \sep \indextwo \in
  \indexsettwo}$ such that
$$
\forall \indextwo \in \indexsettwo,\ \ \vartwo_{\indextwo} \in \supp{\distrone}
\qquad\qquad
\forall \indexone \in \indexsetone, \ \ 
 \distrone(\varone_{\indexone})=\sum_{\vartwo_\indextwo = \varone_\indexone}\ p_{\indextwo}.
$$
By abuse of notation, we will simply write
$\distrone\,=\,\pseudorep{\vartwo_{\indextwo}^{p_{\indextwo}} \sep
  \indextwo \in \indexsettwo}$ to mean that $\distrone$ admits
$\pseudorep{\vartwo_{\indextwo}^{p_{\indextwo}} \sep \indextwo \in
  \indexsettwo}$ as pseudo-representation.  
Any distribution has a canonical pseudo-representation obtained by
simply replacing the set-theoretic notation with the
multiset-theoretic one.
\end{definition}
\longv{
\begin{definition}[$\omega$-CPO of distributions]
We define the pointwise order on distributions over $\setone$ as
$$
\distrone \distrleq \distrtwo \qquad \text{if and only if} \qquad \forall \varone \in \setone,\quad \distrone(\varone) \leq \distrtwo(\varone).
$$
This turns $(\distribs,\,\distrleq)$ into a partial order. This partial order is an $\omega$-CPO,
but not a lattice as the join of two distributions
does not necessarily exist. The bottom element of this $\omega$-CPO is the null distribution $0$.
\end{definition}

\begin{definition}[Operations on distributions]
Given a distribution $\distrone$ and a real number $\alpha \leq 1$, we
define the distribution $\alpha \cdot \distrone$ as $\varone \mapsto
\alpha \cdot \distrone (\varone)$. We similarly define the sum
$\distrone + \distrtwo$ of two distributions over a same set $\setone$
as the function $\varone \mapsto \distrone (\varone) +
\distrtwo(\varone)$.  Note that this is a total operation on functions
$\setone \rightarrow \RR$, but a partial operation on distributions:
it is defined if and only if $\distrsum{\distrone} +
\distrsum{\distrtwo} \leq 1$.  When $\distrone \distrleq \distrtwo$,
we define the partial operation of difference of distributions
$\distrtwo - \distrone$ as the function $\valone \mapsto \distrtwo
(\valone) - \distrone(\valone)$.  We naturally extend these operations
to representations and pseudo-representations of distributions.\\
\end{definition}
}
\shortv{
Distributions support operations like affine combinations and sums --
the latter being only a partial operation. We extend these operations
to (pseudo)-representations, in a natural way. Distributions, endowed
with the pointwise partial-order $\distrleq$, form an $\omega$-CPO, but not
a lattice, since the join of two distributions is not guaranteed
to exist.
}
\begin{definition}[Value Decomposition of a Term Distribution]
 Let $\distrone$ be a term distribution. We write its
 \emph{value decomposition} as $\distrone\ \valdec \ \distroneval +
 \distroneterm$, where $\distroneval$ is the maximal subdistribution
 of $\distrone$ whose support consists of values, and
 $\distroneterm\,=\,\distrone - \distroneval$ is the subdistribution
 of ``non-values'' contained in $\distrone$.
\end{definition}
\paragraph{Operational Semantics.}
The semantics of a term will be the value distribution to which it
reduces via the probabilistic reduction relation, iterated up to the
limit.  \shortv{As a first
  step, we define the call-by-value reduction relation $\rcbv
  \subseteq \distribs \times \distribs$ on
  Figure~\ref{fig:def-weak-cbv-distribs}.}  \longv{As a first step, we
  define the call-by-value reduction relation $\rcbv \subseteq
  \distribs \times \RR^{\termsproba}$ on
  Figure~\ref{fig:def-weak-cbv-distribs}.  The relation $\rcbv$ is in
  fact a relation on distributions:

\begin{lemma}
  Let $\distrone$ be a distribution such that $\distrone \rcbv
  \distrtwo$. Then $\distrtwo$ is a distribution.
\end{lemma}}
Note that we write Dirac distributions simply as terms on the left
side of $\rcbv$, to improve readability.  As usual, we denote by
$\rcbv^n$ the $n$-th iterate of the relation $\rcbv$, with
$\rcbv^0$ being the identity relation.  We then define the relation $\redval^n$ as
follows. Let $\distrone \rcbv^n \distrtwo \valdec \distrtwo_{|V} +
\distrtwo_{|T}$.  Then $\distrone \redval^n \distrtwo_{|V}$. Note
that, for every $n \in \NN$ and $\distrone \in \distribs$, there is a
unique distribution $\distrtwo$ such that $\distrone \rcbv^n
\distrtwo$. Moreover, $\distrtwo_{|V}$ is the only distribution such
that $\distrone \redval^n \distrtwo_{|V}$.

\begin{figure}[t]
\centering
 \fbox{
 \begin{minipage}{0.95\textwidth}
 \begin{center}
 \vspace{0.3cm}
\begin{tabular}{ccc}
 \AxiomC{}
\UnaryInfC{$\letin{\varone}{\valone}{\termone} \ \ \rcbv\ \ \distrelts{\left(\subst{\termone}{\varone}{\valone}\right)^1}$}
 \DisplayProof
 &
 $\qquad$
 &
 \AxiomC{}
\UnaryInfC{$\left(\abstr{\varone}{\termone}\right)\ \valone \ \ \rcbv\ \ \distrelts{\left(\subst{\termone}{\varone}{\valone}\right)^1}$}
 \DisplayProof
\\[0.5cm]
 
\end{tabular}

 \AxiomC{}
 \UnaryInfC{$\termone\ \choice_p\ \termtwo \ \ \rcbv\ \ \distrelts{\termone^p,\,\termtwo^{1-p}}$}
 \DisplayProof

 \vspace{0.5cm}

  \AxiomC{$\termone \ \ \rcbv\ \ \distrelts{\termthree_{\indexone}^{p_{\indexone}} \sep \indexone \in \indexsetone}$}
 \UnaryInfC{$\letin{\varone}{\termone}{\termtwo} \ \ \rcbv\ \ 
 \distrelts{\left(\letin{\varone}{\termthree_{\indexone}}{\termtwo}\right)^{p_{\indexone}} \sep \indexone \in \indexsetone}$}
 \DisplayProof

 \vspace{0.5cm}


%

\AxiomC{}
\UnaryInfC{$ \caseof{\natsucc\ \valone}{\natsucc \rightarrow \valtwo \smallsep
 \natzero \rightarrow \valthree}  \ \ \rcbv\ \ \distrelts{\left(\valtwo\ \valone\right)^1}$}
\DisplayProof

\vspace{0.5cm}

\AxiomC{}
\UnaryInfC{$ \caseof{\natzero}{\natsucc \rightarrow \valtwo \smallsep
 \natzero \rightarrow \valthree}  \ \ \rcbv\ \ \distrelts{\left(\valthree\right)^1}$}
\DisplayProof

\vspace{0.5cm}

\AxiomC{}
\UnaryInfC{$ \left(\letrec{\funcone}{\valone}\right)\ \left(\dataconstrone\ \termvector{\valtwo}\right)\ \ \rcbv\ \ 
 \distrelts{\left(\subst{\valone}{\funcone}{\left(\letrec{\funcone}{\valone}\right)}\ \left(\dataconstrone\ \termvector{\valtwo}\right)\right)^1}$}
\DisplayProof

\vspace{0.5cm}

\AxiomC{$\distrone\ \valdec\ \distrelts{\termone_{\indextwo}^{p_{\indextwo}} \sep \indextwo \in \indexsettwo} + \distroneval$}
\AxiomC{$\forall \indextwo \in \indexsettwo,\ \ \termone_{\indextwo} \ \ \rcbv\ \ \distrtwo_{\indextwo}$}
\BinaryInfC{$\distrone \ \ \rcbv\ \ \left(\sum_{\indextwo \in \indexsettwo} p_j \cdot \distrtwo_j\right) + \distroneval$}
\DisplayProof

\vspace{0.2cm}

 \end{center}
 \end{minipage}
}

\vspace{0.35cm}

 \caption{Call-by-value reduction relation $\rcbv$ on distributions.}
 \label{fig:def-weak-cbv-distribs} 
\end{figure}

\longv{
\begin{lemma}
\label{lemma/semantics-is-computed-monotonically3}
 Let $n,m \in \NN$ with $n < m$. Let $\distrone_n$ (resp $\distrone_m$) be the distribution such that $\termone \rcbv^n \distrone_n$
 (resp $\termone \rcbv^m \distrone_m$).
 Then $\distrone_n \distrleq \distrone_m$.
\end{lemma}
}

\begin{lemma}
\label{lemma/semantics-is-computed-monotonically2}
 Let $n,m \in \NN$ with $n < m$. Let $\distrone_n$ (resp $\distrone_m$) be the distribution such that $\termone \redval^n \distrone_n$
 (resp $\termone \redval^m \distrone_m$).
 Then $\distrone_n \distrleq \distrone_m$.
\end{lemma}


\begin{definition}[Semantics of a Term, of a Distribution]
 The semantics of a distribution $\distrone$ is the distribution
 $\semantics{\distrone}\ \ =\ \ \supremumnat{\left\{\distrone_n\sep
   \distrone \redval^n \distrone_n\right\}}$. This supremum exists
 thanks to Lemma~\ref{lemma/semantics-is-computed-monotonically2},
 combined with the fact that $(\distribs,\,\distrleq)$
 is an $\omega$-CPO.
 We define the semantics of a term $\termone$ as $\semantics{\termone}\ =\ \semantics{\distrelts{\termone^1}}$.
\end{definition}
\longv{
\begin{corollary}
\label{corollary:semantics-is-computed-monotonically}
 Let $n \in \NN$ and $\distrone_n$ be such that $\termone \redval^n \distrone_n$.
 Then $\distrone_n \distrleq \semantics{\termone}$.
\end{corollary}}
We now have all the ingredients required to define the central
concept of this paper, the one of almost-surely terminating
term:
\begin{definition}[Almost-Sure Termination]
  We say that a term $\termone$ is \emph{almost-surely terminating}
  precisely when $\distrsum{\semantics{\termone}}\,=\,1$.
\end{definition}

\longv{Before we terminate this section, let us formulate the following
lemma on the operational semantics of the $\mathsf{let}$ construction, which will be
used in the proof of typing soundness for monadic affine sized types:
\begin{lemma}
\label{lemma:operational-decomposition-let}
 Suppose that $\termone \redval^n \pseudorep{\valone^{p_\indexone} \sep \indexone \in \indexsetone}$
 and that, for every $\indexone \in \indexsetone$, $\subst{\termtwo}{\varone}{\valone_\indexone} \redval^m \distrtwo_\indexone$.
 Then $\letin{\varone}{\termone}{\termtwo} \redval^{n+m+1} \sum_{\indexone \in \indexsetone} p_\indexone\cdot\distrtwo_\indexone$.
\end{lemma}
}

\longv{
\begin{proof}
 Easy from the definition of $\redval$ and of $\rcbv$ in the case of $\mathsf{let}$.
\end{proof}
}
\section{Monadic Affine Sized Typing}\label{sect:monadic-types}\label{sect:sizedtypes}
Following the discussion from Section~\ref{sect:necessity}, we
introduce in this section a non-trivial lifting of sized types to our
probabilistic setting.  As a first step, we design an \emph{affine}
simple type system for $\languageproba$.  This means that no
higher-order variable may be used more than once in the same
probabilistic branch.  However, variables of base type $\Nat$ may be
used freely. 
In spite of this restriction, the resulting system
allows to type terms corresponding to any probabilistic Turing machine.
In Section~\ref{subsect:monadic-types}, we introduce a more sophisticated type
system, which will be \emph{monadic} and affine, and which will be
sound for almost-sure termination as we prove in
Section~\ref{sect:reducibility}.
\subsection{Affine Simple Types for $\languageproba$}

The terms of the language $\languageproba$ can be typed using a
variant of the simple types of the $\lambda$-calculus, extended to
type $\letrecname$ and $\choice_p$, but also restricted to an \emph{affine}
management of contexts. Recall that the constraint of affinity ensures that a
given higher-order symbol is used at most \emph{once} in a
probabilistic branch.  We define simple types over the base type
$\Nat$ in the usual way:
 $
 \simpletypone,\,\simpletyptwo,\,\ldots \bnf \Nat \sep \simpletypone \typarrow \simpletyptwo
 $
where, by convention, the arrow associates to the right.  Contexts
$\contextone,\,\contexttwo,\,\ldots$ are sequences of simply-typed
variables $\varone \refines \simpletypone$.  We write sequents as
$\contextone \proves \termone \refines \simpletypone$ to distinguish
these sequents from the ones using distribution types appearing later
in this section.  Before giving the rules of the type system, we need
to define two policies for contracting contexts: an affine and a general one.
\paragraph*{Context Contraction.} 
\longv{
The contraction $\contextone \cup \contexttwo$ of two contexts is a
non-affine operation, and is partially defined as follows:
\begin{itemize}
 \item $\varone \refines\simpletypone \in \contextone \setminus \contexttwo \ \implies\ x\refines\simpletypone \in \contextone \cup \contexttwo$,
 \item $\varone \refines\simpletypone \in \contexttwo \setminus \contextone \ \implies\ x\refines\simpletypone \in \contextone \cup \contexttwo$ ,
 \item if $\varone \refines\simpletypone \in \contextone$ and $\varone \refines\simpletyptwo \in \contexttwo$,
 \begin{itemize}
  \item if $\simpletypone\,=\,\simpletyptwo$, $x\refines\simpletypone \in \contextone \cup \contexttwo$,
  \item else the operation is undefined.
 \end{itemize}
\end{itemize}
This operation will be used to contract contexts in the rule typing the choice operation $\choice_p$:
indeed, we allow a same higher-order variable $\funcone$ to occur both in $\termone$ and in $\termtwo$ when
forming $\termone \choice_p \termtwo$, as both terms correspond to different probabilistic branches.

\bigbreak

\noindent
\textbf{Affine contraction of contexts.} The affine contraction $\contextone \contextsumdisj \contexttwo$
will be used in all rules but the one for $\choice_p$. It is partially defined as follows:
\begin{itemize}
 \item $\varone \refines\simpletypone \in \contextone \setminus \contexttwo \ \implies\ x\refines\simpletypone \in \contextone \contextsumdisj \contexttwo$,
 \item $\varone \refines\simpletypone \in \contexttwo \setminus \contextone \ \implies\ x\refines\simpletypone \in \contextone \contextsumdisj \contexttwo$ ,
 \item if $\varone \refines\simpletypone \in \contextone$ and $\varone \refines\simpletyptwo \in \contexttwo$,
 \begin{itemize}
  \item if $\simpletypone\,=\,\simpletyptwo\,=\,\Nat$, $x\refines\simpletypone \in \contextone \contextsumdisj \contexttwo$,
  \item in any other case, the operation is undefined.
 \end{itemize}
\end{itemize}
As we explained earlier, only variables of base type $\Nat$ may be contracted.
}
\shortv{
Contexts can be combined in two different ways. On the one hand, one
can form the \emph{non-affine contraction} $\contextone \cup
\contexttwo$ of two contexts, for which $\contextone$ and $\contexttwo$
are allowed to share some variables, but these variables must be
attributed the \emph{same type} in both contexts. On the other hand, one
can form the \emph{affine contraction $\contextone \contextsumdisj
  \contexttwo$}, in which variables in common between
$\contextone$ and $\contexttwo$ must be attributed the type $\Nat$.
}
\paragraph*{The Affine Type System.}
The affine simple type system is then defined in
Figure~\ref{fig:affine-simple-type-system}. All the rules are quite
standard. Higher-order variables can occur at most once in any
probabilistic branch because all binary typing rules -- except
probabilistic choice -- treat contexts affinely.
We set $\settypedvalues{\simpletypone}{\contextone}\ =\ \left\{\valone
\in \values \sep\right.$ $\left. \contextone \proves \valone \refines
\simpletypone\right\}$ and
$\settypedterms{\simpletypone}{\contextone}\ =\ \left\{\termone \in
\termsproba \sep \contextone \proves \termone\refines
\simpletypone\right\}$.  We simply write
$\settypedclosedvalues{\simpletypone}\,=\,\settypedvalues{\simpletypone}{\emptyset}$
and
$\settypedclosedterms{\simpletypone}\,=\,\settypedterms{\simpletypone}{\emptyset}$
when the terms or values are closed.
%
These closed, typable terms enjoy subject reduction and the progress property.

\begin{figure}[t]
\centering
\fbox{
\begin{minipage}[c]{0.95\textwidth}
\vspace{0.3cm}
\begin{center}
\begin{tabular}{ccccc}
 \AxiomC{}
 \LeftLabel{Var $\quad$}
 \UnaryInfC{$\contextone,\,\varone\refines \simpletypone\ \proves\ \varone\refines \simpletypone$}
 \DisplayProof
 &
 $\qquad$
 &
   \AxiomC{$\contextone\ \proves\ \valone \refines \Nat$}
 \UnaryInfC{$\contextone\ \proves\ \natsucc\ \valone \refines \Nat$}
 \DisplayProof
 &
 $\qquad$
 &
 \AxiomC{}
 \UnaryInfC{$\contextone\ \proves\ \natzero \refines \Nat$}
 \DisplayProof
 \\[0.65cm]
 
 \end{tabular} 
 \begin{tabular}{ccc}
  \AxiomC{$\contextone,\,\varone\refines\simpletypone \ \proves\ \termone \refines\simpletyptwo$}
 \LeftLabel{$\lambda$ $\quad$}
 \UnaryInfC{$\contextone\ \proves\ \abstr{\varone}{\termone} \refines \simpletypone \typarrow\simpletyptwo$}
 \DisplayProof
 &
 $\qquad$
 &
  \AxiomC{$\contextone \ \proves\ \valone\refines \simpletypone \typarrow\simpletyptwo$}
  \AxiomC{$\contexttwo \ \proves\ \valtwo \refines \simpletypone$}
  \RightLabel{$\ \ $ App}
 \BinaryInfC{$\contextone \contextsumdisj \contexttwo \ \proves\ \valone\ \valtwo \refines \simpletyptwo$}
 \DisplayProof
  \\[0.65cm]
 \end{tabular} 
\end{center}

 \vspace{-0.4cm}
 
$$
  \AxiomC{$\contextone\ \proves\ \termone \refines \simpletypone$}
 \AxiomC{$\contexttwo\ \proves\ \termtwo \refines \simpletypone$}
 \LeftLabel{Choice $\quad$}
 \BinaryInfC{$\contextone \cup \contexttwo\ \proves\ \termone \choice_p \termtwo \refines \simpletypone$}
 \DisplayProof
$$  

 \vspace{0.2cm}
 
 $$
 \AxiomC{$\contextone \proves \termone \refines \simpletypone$}
 \AxiomC{$\contexttwo,\,\varone\refines \simpletypone \proves \termtwo \refines \simpletyptwo$}
 \LeftLabel{Let $\quad$}
 \BinaryInfC{$ \contextone \contextsumdisj \contexttwo \proves
 \letin{\varone}{\termone}{\termtwo} \refines \simpletyptwo$}
 \DisplayProof
 $$

  \vspace{0.2cm}
  
 $$
 \AxiomC{$\contextone \proves \valone \refines \Nat$}
 \AxiomC{$\contexttwo \proves \valtwo \refines \Nat \typarrow \simpletypone$}
 \AxiomC{$\contexttwo \proves \valthree \refines \simpletypone$}
 \LeftLabel{Case $\quad$}
 \TrinaryInfC{$\contextone \contextsumdisj \contexttwo \proves \caseof{\valone}{\natsucc \rightarrow \valtwo \smallsep
 \natzero \rightarrow \valthree} \refines \simpletypone$}
 \DisplayProof
 $$
 \vspace{0.4cm}
$$
\AxiomC{$\contextone,\,\funcone \refines \Nat \typarrow \simpletypone \proves \valone \refines \Nat \typarrow \simpletypone$}
 \AxiomC{$\forall x \in \Gamma,\ \ x \refines \Nat$}
\LeftLabel{$\letrecname$ $\quad$}
\BinaryInfC{$\contextone \proves \letrec{\funcone}{\valone} \refines  
\Nat \typarrow \simpletypone$}
 \DisplayProof
 $$
  \vspace{0.2cm}
 \end{minipage}
} 
 
 \caption{Affine simple types for $\languageproba$.}
 \label{fig:affine-simple-type-system}
\end{figure}
\subsection{Monadic Affine Sized Types}
\label{subsect:monadic-types}
This section is devoted to giving the basic definitions and results
about monadic affine sized types (MASTs, for short), which can be seen
as decorations of the affine simple types with some \emph{size
  information}. 
\paragraph{Sized Types.}
We consider a set $\sizevars$ of \emph{size variables}, denoted
$\sizevarone,\,\sizevartwo,\,\ldots$ and define \emph{sizes} (called
\emph{stages} in~\cite{barthe-et-al:type-based-termination}) as:
$$
\sizeone,\,\sizetwo  \ \ \bnf\ \ \sizevarone \sep \sizeinf \sep \sizesucc{\sizeone}
$$
where $\sizesucc{\cdot}$ denotes the \emph{successor} operation. We
denote the iterations of $\sizesucc{\cdot}$ as follows:
$\sizesucc{\sizesucc{\sizeone}}$ is denoted
$\sizesuccit{\sizeone}{2}$,
$\sizesucc{\sizesucc{\sizesucc{\sizeone}}}$ is denoted
$\sizesuccit{\sizeone}{3}$,and so on.  By definition, at most one
variable $\sizevarone \in \sizevars$ appears in a given size
$\sizeone$. We call it its \emph{spine variable}, denoted as
$\spine{\sizeone}$. We write $\spine{\sizeone}\,=\,\emptyset$ when
there is no variable in $\sizeone$.  An order $\sizeleq$ on
sizes can be defined as follows:
 $$
 \begin{array}{ccccccc}
  \AxiomC{}
  \UnaryInfC{$\sizeone \sizeleq \sizeone$}
  \DisplayProof
  &
  \quad
  &
  \AxiomC{$\sizeone \sizeleq \sizetwo$}
  \AxiomC{$\sizetwo \sizeleq \sizethree$}
  \BinaryInfC{$\sizeone \sizeleq \sizethree$}
  \DisplayProof
  &
  \quad
  &
  \AxiomC{}
  \UnaryInfC{$\sizeone \sizeleq \sizesucc{\sizeone}$}
  \DisplayProof
  &
  \quad
  &
  \AxiomC{}
  \UnaryInfC{$\sizeone \sizeleq \sizeinf$}
  \DisplayProof
 \end{array}
 $$
Notice that these rules imply notably that $\sizesucc{\sizeinf}$ is
equivalent to $\sizeinf$, i.e., $\sizesucc{\sizeinf}\sizeleq\sizeinf$
and $\sizeinf\sizeleq\sizesucc{\sizeinf}$. We consider sizes modulo
this equivalence.  We can now define sized types and distribution
types by mutual induction, calling distributions of (sized) types the
distributions over the set of sized types:

\begin{definition}[Sized Types, Distribution Types]
Sized types and distribution types are defined by mutual induction,
contextually with the function $\underlying{\cdot}$ which maps any
sized or distribution type to its \emph{underlying} affine type.
$$
\begin{array}{lrcl}
\text{Sized types:} \qquad & \typone,\,\typtwo &  \bnf & \typone \typarrow \distrtypone \sep \Nat^{\sizeone}\\
\text{Distribution types:} \qquad &  
  \distrtypone,\,\distrtyptwo & \bnf & \distrelts{\typone_{\indexone}^{p_{\indexone}} \sep \indexone \in \indexsetone},\\
\text{Underlying map:} \qquad & \underlying{\typone \typarrow \distrtypone}&=&\underlying{\typone} \typarrow \underlying{\distrtypone}\\
& \underlying{\Nat^{\sizeone}}&=&\Nat\\
& \underlying{\distrelts{\typone_{\indexone}^{p_{\indexone}} \sep \indexone \in \indexsetone}}&=&
   \underlying{\typone_\indextwo}
\end{array}
$$
For distribution types we require additionally that
$\sum_{\indexone \in \indexsetone}\ p_{\indexone}\ \leq\ 1$, that
$\indexsetone$ is a finite non-empty set, and that
$\underlying{\typone_\indexone}=\underlying{\typone_\indextwo}$ for
every $\indexone,\indextwo\in\indexsetone$. In the last
equation, $\indextwo$ is any element of $\indexsetone$. 
\end{definition}

The definition of sized types is \emph{monadic} in that a
higher-order sized type is of the shape $\typone \typarrow
\distrtypone$ where $\typone$ is again a sized type, and
$\distrtypone$ is a \emph{distribution} of sized types.  

\longv{
The definition of the fixpoint will refer to the notion
of \emph{positivity} of a size variable in a sized or
distribution type. We define positive and negative occurrences
of a size variable in such a type in Figure~\ref{fig:positivity}.
}

%

\longv{ 
\begin{figure}
$$
\begin{array}{ccccc}
 \AxiomC{}
 \UnaryInfC{$\positive{\sizevarone}{\Nat^{\sizeone}}$}
 \DisplayProof
 & \qquad \qquad &
 \AxiomC{$\negative{\sizevarone}{\typone}$}
 \AxiomC{$\positive{\sizevarone}{\distrtypone}$}
 \BinaryInfC{$\positive{\sizevarone}{\typone \typarrow \distrtypone}$}
 \DisplayProof
 & \qquad \qquad &
 \AxiomC{$\forall \indexone \in \indexsetone,\ \ \positive{\sizevarone}{\typone_\indexone}$}
 \UnaryInfC{$\positive{\sizevarone}{\distrelts{\typone^{p_\indexone}_\indexone \sep \indexone \in \indexsetone}}$}
 \DisplayProof
 \\[0.7cm]
 \AxiomC{$\sizevarone \notin \sizeone$}
 \UnaryInfC{$\negative{\sizevarone}{\Nat^{\sizeone}}$}
 \DisplayProof
 & \qquad \qquad &
 \AxiomC{$\positive{\sizevarone}{\typone}$}
 \AxiomC{$\negative{\sizevarone}{\distrtypone}$}
 \BinaryInfC{$\negative{\sizevarone}{\typone \typarrow \distrtypone}$}
 \DisplayProof
 & \qquad \qquad &
 \AxiomC{$\forall \indexone \in \indexsetone,\ \ \negative{\sizevarone}{\typone_\indexone}$}
 \UnaryInfC{$\negative{\sizevarone}{\distrelts{\typone^{p_\indexone}_\indexone \sep \indexone \in \indexsetone}}$}
 \DisplayProof
 \\
\end{array}
$$
 \caption{Positive and negative occurrences of a size
 variable in a size type.}
 \label{fig:positivity}
\end{figure}
}


%
%
%
%
\paragraph{Contexts and Operations on Them.} 
Contexts are sequences of variables together with a sized
type, and at most one distinguished variable with a distribution type:
\begin{definition}[Contexts]
 Contexts are of the shape $\contextsizedone \contextsep \contextdistrone$, with
 $$
 \begin{array}{llrclrr}
  \text{\emph{Sized contexts:}} & \qquad \ \  & \contextsizedone,\,\contextsizedtwo,\,\ldots &\ \bnf \ & \emptyset \sep \varone \typsep \typone,\,
  \contextsizedone & \quad& (\varone \notin \dom{\contextsizedone})\\
  \text{\emph{Distribution contexts:}} & \qquad & \contextdistrone,\,\contextdistrtwo,\,\ldots & \ \bnf\  & \emptyset \sep \varone \typsep
  \distrtypone  & & \\
 \end{array}
 $$
 As usual, we define the \emph{domain} $\dom{\contextsizedone}$ of a sized context $\contextsizedone$ by induction:
 $\dom{\emptyset}\, =\,\emptyset$ and 
 $\dom{\varone \typsep \typone,\,\contextsizedone}\,=\, \left\{\varone\right\} \uplus \dom{\contextsizedone}$.
 We proceed similarly for the domain $\dom{\contextdistrone}$ of a distribution context $\contextdistrone$.
 When a sized context $\contextsizedone\,=\,\varone_1\typsep\typone_1,\,\ldots,\,\varone_n\typsep\typone_n$ ($n \geq 1$)
 is such that there is a simple type $\simpletypone$ with $\forall \indexone \in \setoneton,\ \ \underlying{\typone_i}\,=\,\simpletypone$,
 we say that $\contextsizedone$ is \emph{uniform} of simple type $\simpletypone$.
 We write this as $\underlying{\contextsizedone}\,=\,\simpletypone$.
 
 We write $\contextsizedone,\,\contextsizedtwo$ for the \emph{disjoint union} of these sized contexts: it is defined whenever
 $\dom{\contextsizedone} \cap \dom{\contextsizedtwo}\,=\,\emptyset$. We proceed similarly for
 $\contextdistrone,\,\contextdistrtwo$, but note that due to the restriction on the cardinality of such contexts, there is the additional requirement
 that $\contextdistrone\,=\,\emptyset$ or $\contextdistrtwo\,=\,\emptyset$.
 
 We finally define  \emph{contexts} as pairs $\contextsizedone \contextsep \contextdistrone$ of a sized context and of a distribution context,
 with the constraint that $\dom{\contextsizedone} \cap \dom{\contextdistrone}\,=\,\emptyset$.
\end{definition}

\begin{definition}[Probabilistic Sum]
Let $\distrtypone$ and $\distrtyptwo$ be two distribution types. We define their probabilistic sum
$\distrtypone \choice_p \distrtyptwo$ as the distribution type $p \cdot \distrtypone + (1-p) \cdot \distrtyptwo$.
We extend this operation to a \emph{partial} operation on distribution contexts:
\begin{varitemize}
 \item For two distribution types $\distrtypone$ and $\distrtyptwo$ such that
 $\underlying{\distrtypone}\,=\,\underlying{\distrtyptwo}$, we define $\left(\varone\typsep \distrtypone\right)\, \choice_p\,
 \left(\varone\typsep \distrtyptwo\right) \ \ =\ \ \varone\typsep \distrtypone\,\choice_p\, \distrtyptwo$,
 \item $\left(\varone\typsep \distrtypone\right)\, \choice_p \,\emptyset
 \ \ =\ \ \varone\typsep p \cdot \distrtypone$,
 \item $\emptyset \,\choice_p\, \left(\varone\typsep \distrtypone\right)
 \ \ =\ \ \varone\typsep (1-p) \cdot \distrtypone$,
 \item In any other case, the operation is undefined.
\end{varitemize}
\end{definition}

\begin{definition}[Weighted Sum of Distribution Contexts]
 Let $\left(\contextdistrone_\indexone\right)_{\indexone \in \indexsetone}$ be a non-empty family of distribution contexts and
 $\left(p_\indexone\right)_{\indexone \in \indexsetone}$ be a family of reals of $[0,1]$. 
 We define the weighted sum 
 $\sum_{\indexone\in\indexsetone}\ p_\indexone \cdot \contextdistrone_\indexone$
 as the distribution context
 $x\,:\,\sum_{\indexone\in\indexsetone}\ p_\indexone \cdot \distrtypone_\indexone$ when the following conditions are met:
 \begin{enumerate}
  \item $\exists \varone,\ \ \forall \indexone \in \indexsetone,\ \contextdistrone_\indexone \,=\,\varone \typsep\distrtypone_\indexone$,
  \item $\forall (\indexone,\indextwo) \in \indexsetone^2,\ \ \underlying{\contextdistrone_\indexone}\,=\,
  \underlying{\contextdistrone_\indextwo}$,
  \item and $\sum_{\indexone \in \indexsetone}\ p_\indexone \leq 1$,
 \end{enumerate}
 In any other case, the operation is undefined.
\end{definition}

\shortv{
\noindent
We define the substitution $\subst{}{\sizevarone}{\sizetwo}$ of a size variable in a size
or in a sized or distribution type in the expected way; see the long version \cite{dal-lago-grellois:monadic-affine-sized-types-full}
for details. A subtyping relation allows to lift the order $\sizeleq$ on sizes to monadic sized types:}
\longv{
\begin{definition}[Substitution of a Size Variable]
 We define the substitution $\subst{\sizeone}{\sizevarone}{\sizetwo}$ of a size variable in a size as follows:
 %
$$
\subst{\sizevarone}{\sizevarone}{\sizetwo}\ =\ \sizetwo \qquad\quad
 \subst{\sizevartwo}{\sizevarone}{\sizetwo}\ =\ \sizevartwo\qquad\quad
 \subst{\sizeinf}{\sizevarone}{\sizetwo}\ =\ \sizeinf\qquad\quad
  \subst{\sizesucc{\sizeone}}{\sizevarone}{\sizetwo}\ =\ \sizesucc{\subst{\sizeone}{\sizevarone}{\sizetwo}}
$$
 where $\sizevarone \neq \sizevartwo$.
 We then define the substitution $\subst{\typone}{\sizevarone}{\sizeone}$ (resp. $\subst{\distrtypone}{\sizevarone}{\sizeone}$)
 of a size variable $\sizevarone$ by a size $\sizeone$ in a sized or distribution type as:
 $$
 \subst{\left(\typone \typarrow \distrtypone\right)}{\sizevarone}{\sizeone}\ \ =\ \ 
 \subst{\typone}{\sizevarone}{\sizeone}  \typarrow \subst{\distrtypone}{\sizevarone}{\sizeone}
 \qquad\qquad\qquad
 \subst{\left(\Nat^{\sizeone}\right)}{\sizevarone}{\sizetwo}\ \ =
 \ \ \Nat^{\subst{\sizeone}{\sizevarone}{\sizetwo}}
 $$
 $$
 \subst{\left(\distrelts{\typone_{\indexone}^{p_{\indexone}} \sep \indexone \in \indexsetone}\right)}{\sizevarone}{\sizeone}
 \ =\ \distrelts{\left(\subst{\typone_{\indexone}}{\sizevarone}{\sizeone}\right)^{p_{\indexone}} \sep \indexone \in \indexsetone}
 $$
 We define the substitution of a size variable in a sized or distribution context in the obvious way:
 $$
 \subst{\emptyset}{\sizevarone}{\sizeone}\ =\ \emptyset \qquad \qquad
  \subst{\left(\varone \typsep \typone,\,\contextsizedone\right)}{\sizevarone}{\sizeone}\ =\ 
  \varone \typsep \subst{\typone}{\sizevarone}{\sizeone},\,\subst{\contextsizedone}{\sizevarone}{\sizeone}
 $$
 $$
 \subst{\left(\varone \typsep \distrtypone\right)}{\sizevarone}{\sizeone}\ =\ \varone \typsep \subst{\distrtypone}{\sizevarone}{\sizeone}
 $$  
\end{definition}
}

\longv{
\begin{lemma}~
 \label{lemma:commuting-size-substitutions-with-operations}
 \begin{enumerate}
 \item $\subst{\left(\distrtypone \choice_p \distrtyptwo\right)}{\sizevarone}{\sizeone}\ \ =\ \ 
 \subst{\distrtypone}{\sizevarone}{\sizeone}  \choice_p \subst{\distrtyptwo}{\sizevarone}{\sizeone}$
   \item For distribution contexts, $\subst{\left(\contextdistrone \choice_p \contextdistrtwo\right)}{\sizevarone}{\sizeone}\ \ =\ \ 
 \subst{\contextdistrone}{\sizevarone}{\sizeone}  \choice_p \subst{\contextdistrtwo}{\sizevarone}{\sizeone}$
 
  \item For distribution contexts, $\subst{\left(\sum_{\indexone \in \indexsetone}\ p_{\indexone} \cdot \contextone_{\indexone}\right)}{\sizevarone}{\sizeone}\ \ =\ \
 \sum_{\indexone \in \indexsetone}\ p_{\indexone} \cdot \subst{\contextone_{\indexone}}{\sizevarone}{\sizeone}$
 \end{enumerate}
 \end{lemma}
 }
 
\longv{ 
\begin{proof}
 \begin{enumerate}
  \item Let $\distrtypone\,=\,\distrelts{\typone_\indexone^{p'_\indexone} \sep \indexone \in \indexsetone}$
  and $\distrtyptwo\,=\,\distrelts{\typtwo_\indextwo^{p''_\indextwo} \sep \indextwo \in \indexsettwo}$.
  Then
  $$
  \begin{array}{rl}
   & \subst{\distrtypone}{\sizevarone}{\sizeone}  \choice_p \subst{\distrtyptwo}{\sizevarone}{\sizeone}\\
   = \ \ &
   \subst{\distrelts{\typone_\indexone^{p'_\indexone} \sep \indexone \in \indexsetone}}{\sizevarone}{\sizeone}  \choice_p \subst{\distrelts{\typtwo_\indextwo^{p''_\indextwo} \sep \indextwo \in \indexsettwo}}{\sizevarone}{\sizeone}\\
   = \ \ &
   \distrelts{\left(\subst{\typone_\indexone}{\sizevarone}{\sizeone}\right)^{p'_\indexone} \sep \indexone \in \indexsetone} \choice_p \distrelts{\left(\subst{\typtwo_\indextwo}{\sizevarone}{\sizeone}\right)^{p''_\indextwo} \sep \indextwo \in \indexsettwo}\\
     = \ \ &
   \pseudorep{\left(\subst{\typone_\indexone}{\sizevarone}{\sizeone}\right)^{pp'_\indexone} \sep \indexone \in \indexsetone} + \pseudorep{\left(\subst{\typtwo_\indextwo}{\sizevarone}{\sizeone}\right)^{(1-p)p''_\indextwo} \sep \indextwo \in \indexsettwo}\\
   = \ \ &
   \subst{\left(\pseudorep{\left(\typone_\indexone\right)^{pp'_\indexone} \sep \indexone \in \indexsetone} + \pseudorep{\left(\typtwo_\indextwo\right)^{(1-p)p''_\indextwo} \sep \indextwo \in \indexsettwo}\right)}{\sizevarone}{\sizeone}\\
   = \ \ &
   \subst{\left(\distrtypone \choice_p \distrtyptwo\right)}{\sizevarone}{\sizeone}\\
  \end{array}
  $$

  \item Suppose that $\contextdistrone\,=\,\varone \typsep \distrtypone$
  and that $\contextdistrtwo\,=\,\varone \typsep \distrtyptwo$.
  Then $\contextdistrone\choice_p \contextdistrtwo\ =\ \varone \typsep \distrtypone \choice_p \distrtyptwo$. It follows from (1) that 
  $\subst{\contextdistrone}{\sizevarone}{\sizeone} \choice_p \subst{\contextdistrtwo}{\sizevarone}{\sizeone} 
  \ =\ \varone \typsep \subst{\distrtypone}{\sizevarone}{\sizeone} \choice_p
  \subst{\distrtyptwo}{\sizevarone}{\sizeone}
  \ =\ \varone \typsep \subst{\left(\distrtypone \choice_p \distrtyptwo\right)}{\sizevarone}{\sizeone}
  \ =\ \subst{\left(\contextdistrone \choice_p \contextdistrtwo\right)}{\sizevarone}{\sizeone}$

  \item The proof is similar to the previous cases.
 \end{enumerate}

\end{proof}
}

\longv{
\noindent
A subtyping relation allows to lift the order $\sizeleq$ on sizes to monadic sized types:
}

\begin{definition}[Subtyping]
 We define the subtyping relation $\subtypeleq$ on sized types and distribution types as follows:
 $$
 \begin{array}{ccccc}
  \AxiomC{}
  \UnaryInfC{$\typone \subtypeleq \typone$}
  \DisplayProof
  &
  \quad
  &
  \AxiomC{$\sizeone \sizeleq \sizetwo$}
  \UnaryInfC{$\Nat^{\sizeone} \subtypeleq\ \Nat^{\sizetwo}$}
  \DisplayProof
  &
  \quad
  &
  \AxiomC{$\typtwo \subtypeleq \typone$}
  \AxiomC{$\distrtypone \subtypeleq \distrtyptwo$}
  \BinaryInfC{$\typone \typarrow \distrtypone \ \subtypeleq\ \typtwo \typarrow \distrtyptwo$}
  \DisplayProof
 \end{array}
 $$
 $$
 \AxiomC{$\exists \funcone \,:\, \indexsetone \to \indexsettwo,\ \ \left(\forall \indexone \in \indexsetone,\ \ \typone_{\indexone}
 \,\subtypeleq\,\typtwo_{\funcone(\indexone)} \right) \text{ and }
 \left(\forall \indextwo \in \indexsettwo,\ \ \sum_{\indexone \in \funcone^{-1}(\indextwo)}\ p_{\indexone} \leq p'_{\indextwo} \right)$}
 \UnaryInfC{$\distrelts{\sigma_{\indexone}^{p_{\indexone}} \sep \indexone \in \indexsetone}\ \subtypeleq\ 
 \distrelts{\typtwo_{\indextwo}^{p'_{\indextwo}} \sep \indextwo \in \indexsettwo}$}
 \DisplayProof
 $$
\end{definition}
\paragraph{Sized Walks and Distribution Types.}
As we explained in Section~\ref{sect:necessity}, the rule typing
$\letrecname$ in the monadic, affine type system relies on an external
decision procedure, computable in polynomial time.  This procedure
ensures that the \emph{sized walk} --- a particular instance of
\emph{one-counter Markov decision process}
(OC-MDP, see~\cite{bradzdil-et-al:one-counter-markov-decision-processes}),
but which does not make use of non-determinism --- associated to the
type of the recursive function of interest indeed ensures almost
sure termination. Let us now define the sized walk associated to a
distribution type $\distrtypone$.\shortv{ For the precise connection
  with OC-MDPs, see the long
  version~\cite{dal-lago-grellois:monadic-affine-sized-types-full}.}
\longv{We then make precise the connection with OC-MDPs, from which
  the computability in polynomial time of the almost-sure termination
  of the random walk follows.}

\begin{definition}[Sized Walk]
Let $\indexsetone\ \subseteq_{\mathit{fin}} \NN$ be a
finite set of integers. Let
$\left\{p_\indexone\right\}_{\indexone\in\indexsetone}$
be such that $\sum_{\indexone \in
  \indexsetone}\ p_{\indexone}\ \leq\ 1$. These parameters define a
Markov chain whose set of states is $\NN$ and
whose transition relation is defined as follows:
\begin{varitemize}
\item 
  the state $0\in\NN$ is stationary (i.e. one goes from $0$ to $0$
  with probability $1$),
\item
  from the state $s+1\in\NN$ one moves:
  \begin{varitemize}
  \item 
    to the state $s + \indexone$ with probability $p_\indexone$, 
    for every $\indexone \in \indexsetone$;
  \item 
    to $0$ with probability $1 -
    \left(\sum_{\indexone \in \indexsetone}\ p_{\indexone}\right)$.
  \end{varitemize}
\end{varitemize}
We call this Markov chain the \emph{sized walk} on $\NN$
associated to $\left(\indexsetone,\left(p_\indexone\right)_{\indexone
  \in \indexsetone}\right)$.
 A sized walk is \emph{almost surely terminating} when it
 reaches $0$ with probability 1 from any initial state.
\end{definition}
Notably, checking whether a sized walk is terminating
is relatively easy:
\begin{proposition}[Decidability of AST for Sized Walks]\label{proposition/decidability-ast}
 It is decidable in polynomial time whether a sized walk is AST.
\end{proposition}
\shortv{\begin{proof}
 By encoding sized walks into OC-MDPs, which enjoy this property~\cite{bradzdil-et-al:one-counter-markov-decision-processes}.
 See the long version~\cite{dal-lago-grellois:monadic-affine-sized-types-full}.
\end{proof}}
\longv{ 
\begin{proof}
 See Section~\ref{section:proof-oc-mdp}.
\end{proof}
}

\begin{definition}[From Types to Sized Walks]
Consider a distribution type $\distrtypone\ =\ \distrelts{\left(\Nat^{\sizeone_{\indextwo}}
\typarrow \distrtyptwo_\indextwo\right)^{p_{\indextwo}} \sep \indextwo \in \indexsettwo}$ 
such that $\forall \indextwo \in
\indexsettwo,\ \spine{\sizeone_\indextwo}\,=\,\sizevarone$.  Then
$\distrtypone$ induces a sized walk, defined as follows. First, by
definition, $\sizeone_\indextwo$ must be of the shape
$\sizesuccit{\sizevarone}{k_\indextwo}$ with $k_{\indextwo} \geq 0$
for every $\indextwo \in \indexsettwo$.  We set
$\indexsetone\ =\ \left\{k_\indextwo \sep \indextwo \in
\indexsettwo\right\}$ and $q_{k_\indextwo}\,=\,p_\indextwo$ for
every $\indextwo \in \indexsettwo$.  The sized walk induced by the
distribution type $\distrtypone$ is then the sized walk associated to
$\left(\indexsetone,\left(q_\indexone)_{\indexone \in
  \indexsetone}\right)\right)$.
\end{definition}

\begin{example}
 Let $\distrtypone\ =\ \distrelts{\left(\Nat^{\sizevarone} \typarrow \Nat^{\sizeinf}\right)^{\frac{1}{2}},\ 
 \left(\Nat^{\sizesuccit{\sizevarone}{2}} \typarrow \Nat^{\sizeinf}\right)^{\frac{1}{3}}}$.
 Then the induced sized walk is the one associated to $\left(\left\{0,2\right\},\left(p_{0}=\frac{1}{2},\,
 p_{2}=\frac{1}{3}\right)\right)$. In other words, it is the random walk on $\NN$ which is stationary on $0$, and which on non-null integers
 $\indexone +1$ moves to $\indexone$ with probability $\frac{1}{2}$, to $\indexone +2$ with probability $\frac{1}{3}$,
 and jumps to $0$ with probability $\frac{1}{6}$. Note that the type $\distrtypone$, and therefore the associated sized walk,
 models a recursive function which calls itself on a size lesser by one unit with probability $\frac{1}{2}$,
 on a size greater by one unit with probability $\frac{1}{3}$, and which does not call itself with probability $\frac{1}{6}$.
\end{example}

\paragraph{Typing Rules.}
Judgements are of the shape $\contextsizedone \contextsep
\contextdistrone \proves \termone \typsep \distrtypone$.  When a
distribution $\distrtypone\,=\,\distrelts{\typone^1}$ is Dirac, we
simply write it $\typone$.  The type system is defined in
Figure~\ref{figure/affine-distribution-type-system}.  As earlier, we
define sets of typable terms, and set
$\setdistrtypedvalues{\typone}{\contextsizedone \contextsep
  \contextdistrone} \ =\ \left\{\valone \sep \contextsizedone
\contextsep \contextdistrone \proves \valone \typsep \typone\right\}$,
and $\setdistrtypedterms{\distrtypone}{\contextsizedone \contextsep
  \contextdistrone} \ =\ \left\{\termone \sep \contextsizedone
\contextsep \contextdistrone \proves \termone \typsep
\distrtypone\right\}$.  We abbreviate
$\setdistrtypedvalues{\typone}{\emptyset \contextsep \emptyset}$ as
$\setdistrtypedclosedvalues{\typone}$ and
$\setdistrtypedterms{\typone}{\emptyset \contextsep \emptyset}$ as
$\setdistrtypedclosedterms{\typone}$.

This sized type system is a refinement of the affine simple type system
for $\languageproba$: if $\varone_1 \typsep \typone_1,\,\ldots,\,\varone_n
\typsep \typone_n \contextsep \funcone \typsep \distrtypone \proves
\termone \typsep \distrtyptwo$,
then it is easily checked that
$\varone_1 \refines \underlying{\typone_1},\,\ldots,\,\varone_n
\refines \underlying{\typone_n},\,  \funcone \refines \underlying{\distrtypone}
\proves \termone \refines \underlying{\distrtyptwo}$.


\begin{figure}[t!]
\centering
\fbox{
\resizebox{0.95\textwidth}{!}{
\begin{minipage}[c]{0.97\linewidth}
\vspace{0.3cm}
\begin{center}
\begin{tabular}{ccc}
 \AxiomC{}
 \LeftLabel{Var $\quad$}
 \UnaryInfC{$\contextsizedone,\,\varone \typsep \typone \contextsep \contextdistrone \ \proves\ \varone\typsep \typone$}
 \DisplayProof
 &
 $\qquad$
 &
 \AxiomC{}
 \RightLabel{$\quad$ Var'}
 \UnaryInfC{$\contextsizedone \contextsep \varone \typsep \typone \ \proves\ \varone\typsep \typone$}
 \DisplayProof
 \\[0.65cm]
   \AxiomC{$\contextsizedone \contextsep \contextdistrone\ \proves\ \valone \typsep \Nat^{\sizeone}$}
  \LeftLabel{Succ $\quad$}
 \UnaryInfC{$\contextsizedone \contextsep \contextdistrone\ \proves\ \natsucc\ \valone \typsep \Nat^{\sizesucc{\sizeone}}$}
 \DisplayProof
 &
 $\qquad$
 &
 \AxiomC{}
  \RightLabel{$\quad$ Zero}
 \UnaryInfC{$\contextsizedone \contextsep \contextdistrone\ \proves\ \natzero \typsep \Nat^{\sizesucc{\sizeone}}$}
 \DisplayProof
 \\[0.65cm]
 
 \end{tabular} 
 \begin{tabular}{ccc}
  \AxiomC{$\contextsizedone,\,\varone\typsep\typone \contextsep \contextdistrone \ \proves\ \termone \typsep\distrtypone$}
 \LeftLabel{$\lambda$ $\quad$}
 \UnaryInfC{$\contextsizedone \contextsep \contextdistrone\ \proves\ \abstr{\varone}{\termone} \typsep \typone \typarrow\distrtypone$}
 \DisplayProof
 &
 $\qquad$
 &
\AxiomC{$\contextsizedone \contextsep \contextdistrone\ \proves\ \termone\typsep \distrtypone$}
 \AxiomC{$\distrtypone\ \subtypeleq\ \distrtyptwo$}
 \RightLabel{$\quad \text{Sub} $}
 \BinaryInfC{$\contextsizedone \contextsep \contextdistrone\ \proves\ \termone\typsep\distrtyptwo$}
 \DisplayProof
  \\[0.65cm]
 \end{tabular} 
\end{center}

$$
  \AxiomC{$\contextsizedone,\,\contextsizedtwo \contextsep \contextdistrone \ \proves\ \valone\typsep \typone \typarrow\distrtypone$}
  \AxiomC{$\contextsizedone,\,\contextsizedthree \contextsep \contextdistrtwo \ \proves\ \valtwo \typsep \typone$}
  \AxiomC{$\underlying{\contextsizedone}\,=\,\Nat$}
  \LeftLabel{App $\quad$}
 \TrinaryInfC{$\contextsizedone ,\, \contextsizedtwo ,\, \contextsizedthree \contextsep
 \contextdistrone ,\, \contextdistrtwo\ \proves\ \valone\ \valtwo \typsep \distrtypone$}
 \DisplayProof
 $$

 \vspace{0.2cm}
 
$$
  \AxiomC{$\contextsizedone \contextsep \contextdistrone\ \proves\ \termone \typsep \distrtypone$}
 \AxiomC{$\contextsizedone \contextsep \contextdistrtwo\ \proves\ \termtwo \typsep \distrtyptwo$}
 \AxiomC{$\underlying{\distrtypone}\,=\,\underlying{\distrtyptwo}$}
 \LeftLabel{Choice $\quad$}
 \TrinaryInfC{$\contextsizedone \contextsep \contextdistrone \choice_p \contextdistrtwo\ 
 \proves\ \termone \choice_p \termtwo \typsep \distrtypone \choice_p \distrtyptwo$}
 \DisplayProof
$$  

 \vspace{0.2cm}
 
 $$
 \AxiomC{$\contextsizedone,\,\contextsizedtwo \contextsep \contextdistrone
 \proves \termone \typsep \distrelts{\typone_{\indexone}^{p_{\indexone}} \sep \indexone \in \indexsetone}$}
 \AxiomC{$\underlying{\contextsizedone}\,=\,\Nat$}
 \noLine
 \BinaryInfC{$\contextsizedone,\,\contextsizedthree,\,\varone\typsep \typone_{\indexone} \contextsep \contextdistrtwo_{\indexone}
 \proves \termtwo \typsep \distrtypone_{\indexone}\quad (\forall \indexone \in \indexsetone)$}
 \LeftLabel{Let $\quad$}
 \UnaryInfC{$\contextsizedone,\,\contextsizedtwo,\,\contextsizedthree \contextsep \contextdistrone ,\,
 \left(\sum_{\indexone \in \indexsetone}\ p_{\indexone} \cdot \contextdistrtwo_{\indexone}\right)
 \proves \letin{\varone}{\termone}{\termtwo} \typsep \sum_{\indexone \in \indexsetone}\ p_{\indexone} \cdot \distrtypone_{\indexone}$}
 \DisplayProof
 $$

  \vspace{0.2cm}
  
 $$
 \AxiomC{$\contextsizedone \contextsep \emptyset \proves \valone \typsep \Nat^{\sizesucc{\sizeone}}$}
 \AxiomC{$\contextsizedtwo \contextsep \contextdistrone \proves \valtwo \typsep \Nat^{\sizeone} \typarrow \distrtypone$}
 \AxiomC{$\contextsizedtwo \contextsep \contextdistrone \proves \valthree \typsep \distrtypone$}
 \LeftLabel{Case $\quad$}
 \TrinaryInfC{$\contextsizedone,\,\contextsizedtwo \contextsep \contextdistrone \proves \caseof{\valone}{\natsucc \rightarrow \valtwo \smallsep
 \natzero \rightarrow \valthree} \typsep \distrtypone$}
 \DisplayProof
 $$
 \vspace{0.4cm}
$$
 \AxiomC{$\underlying{\contextsizedone}\,=\,\Nat$}
 \noLine
 \UnaryInfC{$\sizevarone \notin \contextone \text{ and } \sizevarone \text{ positive in } \nu
 \text{ and } \forall \indextwo \in \indexsettwo,\ \spine{\sizeone_\indextwo}\,=\,\sizevarone$}
 \noLine
\UnaryInfC{$\distrelts{\left(\Nat^{\sizeone_{\indextwo}} \typarrow \subst{\distrtyptwo}{\sizevarone}{\sizeone_{\indextwo}}
\right)^{p_{\indextwo}} \sep \indextwo \in \indexsettwo} \text{induces an AST sized walk}$}
\noLine
\UnaryInfC{$\contextsizedone \contextsep
\funcone \typsep \distrelts{\left(\Nat^{\sizeone_{\indextwo}} \typarrow \subst{\distrtyptwo}{\sizevarone}{\sizeone_{\indextwo}}
\right)^{p_{\indextwo}} \sep \indextwo \in \indexsettwo} \proves \valone \typsep 
\Nat^{ \sizesucc{\sizevarone}} \typarrow \subst{\distrtyptwo}{\sizevarone}{\sizesucc{\sizevarone}}$}
\LeftLabel{$\letrecname$ $\quad$}
\UnaryInfC{$\contextsizedone,\,\contextsizedtwo \contextsep \contextdistrone \proves \letrec{\funcone}{\valone} \typsep  
\Nat^{\sizetwo} \typarrow \subst{\distrtyptwo}{\sizevarone}{\sizetwo}$}
 \DisplayProof
 $$
  \vspace{0.2cm}
 \end{minipage}}
} 
 
 \caption{Affine distribution types for $\languageproba$.}
 \label{figure/affine-distribution-type-system}
\end{figure}

\begin{lemma}[Properties of Distribution Types]~
 \label{lemma/dirac-types-values}
 
\begin{varitemize}
\item $\contextsizedone \contextsep \contextdistrone \proves \valone \typsep \distrtypone \ \ \implies\ \ \distrtypone \text{ is Dirac.}$
\item  $\contextsizedone \contextsep \contextdistrone \proves \termone \typsep \distrtypone \ \ \implies\ \ \distrtypone$ is proper.
\end{varitemize}
\end{lemma}


\longv{
\begin{proof}
 Immediate inspection of the rules.
\end{proof}
}
\longv{
\subsection{Proof of Proposition~\ref{proposition/decidability-ast}}
\label{section:proof-oc-mdp}
We now prove Proposition~\ref{proposition/decidability-ast} by reducing sized walks to
deterministic one-counter Markov decision processes (DOC-MDPs), and using then a result 
of~\cite{bradzdil-et-al:one-counter-markov-decision-processes} to conclude. 
Please note that in~\cite{bradzdil-et-al:one-counter-markov-decision-processes}
the Markov decision processes are more general, as they allow non-determinism.
They are called one-counter Markov decision processes (OC-MDPs), and contain 
in particular all the DOC-MDPs.
We omit this feature in our presentation.

\begin{definition}[Markov Decision Process]
 A Markov decision process (MDP) is a tuple $\left(\markovvertices,\markovtransition,\markovproba\right)$ such that $\markovvertices$ is a 
 finite-or-countable set of vertices, $\markovtransition \,\subseteq\, \markovvertices \times \markovvertices$ is a total transition relation, 
 and $\markovproba$ is a probability assignment mapping each $\markovverticeone \in \markovvertices$ to a probability distribution associating
 a rational and non-null probability to each edge outgoing of $\markovverticeone$. These distributions are moreover required to sum to 1.
 %
\end{definition}

\begin{definition}[Deterministic One-Counter Markov Decision Process]
 A \emph{deterministic one-counter Markov decision process} (DOC-MDP) is a tuple 
 $\left(\states,\transitionzero,\transitionsupzero,\probatransitionzero,\probatransitionsupzero\right)$
 such that:
 \begin{itemize}
  \item $\states$ is a finite set of states,
  \item $\transitionzero \subseteq \states \times \{0,1\} \times \states$ and 
  $\transitionsupzero \subseteq \states \times \{-1,0,1\} \times \states$ are sets of \emph{zero} and \emph{positive} transitions,
  satisfying that every $q \in \states$ has at least a zero and a positive outgoing transition,
  \item $\probatransitionzero$ (resp. $\probatransitionsupzero$) is a probability assignment mapping every
  $q \in \states$ to a probability distribution over the outgoing transitions of $\transitionzero$
  (resp. $\transitionsupzero$) from $q$. These distributions are required to attribute a non-null, rational probability to every outgoing
  transition, and to sum to 1.
 \end{itemize}

\end{definition}

\begin{definition}[Induced Markov Decision Process]
 A DOC-MDP $\left(\states,\transitionzero,\transitionsupzero,\probatransitionzero,\probatransitionsupzero\right)$
 induces a MDP $\left(\states \times \NN, \markovtransition, \markovproba\right)$ such that, for 
 $q \in \states$ and $n \in \NN$:
 \begin{itemize}
  \item for every state $q'$ such that $(q,m,q') \in \transitionzero$, $(q,0) \markovtransition (q',m)$, and the probability of this transition
  is the one attributed by $\probatransitionzero(q)$ to the transition $(q,m,q')$,
  \item and for every state $q'$ such that $(q,m,q') \in \transitionsupzero$, $(q,n) \markovtransition (q',n+m)$, and the probability of this transition
  is the one attributed by $\probatransitionsupzero(q)$ to the transition $(q,m,q')$,
 \end{itemize}
 This MDP is said to \emph{terminate} when it reaches the value counter $0$ in \emph{any} state $q \in \states$.
\end{definition}

Recall that, by definition, $|m| \leq 1$. This is the only restriction to overcome (using intermediate states)
to encode sized walks in DOC-MDPs, so that the MDP they induce coincide with the original sized walk.
We will then obtain the result of polynomial time decidability of termination with probability $1$
using the following proposition:

\begin{proposition}[\cite{bradzdil-et-al:one-counter-markov-decision-processes}, Theorem 4.1]
 It is decidable in polynomial time whether the MDP induced by an OC-MDP --- and thus by a DOC-MDP --- terminates with probability $1$.
\end{proposition}

We now encode sized walks as DOC-MDPs:

\begin{definition}[DOC-MDP Corresponding to a Sized Walk]
 Consider the sized walk on $\NN$ associated to $\left(\indexsetone,\left(p_\indexone\right)_{\indexone \in \indexsetone}\right)$.
 We define the corresponding DOC-MDP
 $\left(\states,\transitionzero,\transitionsupzero,\probatransitionzero,\probatransitionsupzero\right)$
 as follows. Let us first consider the following set of states:
 $$
 \states = \{\mainstate,\,\zerostate\}\ \cup\ \left\{q_1,\,\ldots,q_{\indextwo-2} \sep \indextwo = \max\{\indexone \in \indexsetone
 \sep \indexone \geq 2\} \right\}
 $$
 where $\mainstate$ is the ``main'' state of the DOC-MDP and the other ones will be used for encoding purposes.
 We define the transitions of $\transitionsupzero$ as follows:
 \begin{itemize}
  \item we add the transition $(\zerostate,-1,\zerostate)$ with probability $1$,
  \item for every $\indextwo \in \left\{2,\ldots,\max\{\indexone \in \indexsetone
 \sep \indexone \geq 2\} - 2\right\}$, we add the transition $(q_{\indextwo},1,q_{\indextwo -1})$ with probability $1$,
  \item we add the transition $(q_1,1,\mainstate)$ with probability $1$,
  \item for $\indexone \in \indexsetone \cap \{0,1,2\}$, we add the transition $(\mainstate,\indexone -1,\mainstate)$ 
  and attribute it probability $p_{\indexone}$,
  \item for $\indexone \in \indexsetone \setminus \{0,1,2\}$, we add the transition $(\mainstate,1,q_{\indexone -2})$ 
  and attribute it probability $p_{\indexone}$,
  \item if $1 - \left(\sum_{\indexone \in \indexsetone}\ p_{\indexone}\right) > 0$,
  we add the transition $(\mainstate,-1,\zerostate)$ with probability $1 -
    \left(\sum_{\indexone \in \indexsetone}\ p_{\indexone}\right)$.
 \end{itemize}
 Finally, we define $\transitionzero$ as follows: for every state $q \in \states$, we add the transition $(q,0,q)$ and
  attribute it probability $1$.

\end{definition}

It is easily checked that, by construction, these DOC-MDP induce the same Markov decision processes as sized walks:

\begin{proposition}
 The sized walk on $\NN$ associated to $\left(\indexsetone,\left(p_\indexone\right)_{\indexone \in \indexsetone}\right)$
 coincides with the induced MDP of the corresponding DOC-MDP.
\end{proposition}

This allows us to deduce from the result of \cite{bradzdil-et-al:one-counter-markov-decision-processes}
the polynomial time decidability of AST for sized walks:

\begin{corollary}[Proposition~\ref{proposition/decidability-ast}]
  It is decidable in polynomial time whether a sized walk is almost-surely terminating.
\end{corollary}

}
\longv{\section{Subject Reduction for Monadic Affine Sized Types.}}
\shortv{\bigbreak

\noindent
\paragraph{Subject Reduction for Monadic Affine Sized Types.}}
The type system enjoys a form of subject reduction which can be understood from the following example.
Remark that the type system allows to derive the sequent
\begin{equation}
\label{eq:example-subject-reduction1}
\emptyset \contextsep \emptyset \proves \natzero \choice \natzero \typsep \distrelts{\left(\Nat^{\sizesucc{\sizeone}}\right)^{\frac{1}{2}},
\left(\Nat^{\sizesucc{\sizesucc{\sizetwo}}}\right)^{\frac{1}{2}}}
\end{equation}
The distribution type typing $\natzero \choice \natzero$ contains information about the types of the two probabilistic
branches of $\natzero \choice \natzero$, which will be separated into two different terms during the reduction, but these
two different terms will not be distinguished by the operational semantics: 
$\semantics{\natzero \choice \natzero} \,=\,\distrelts{\natzero^1}$.
The subject reduction procedure needs to keep track of more information, namely that 
$\natzero \choice \natzero$ reduced to $\natzero$ with type $\Nat^{\sizesucc{\sizeone}}$
in a copy, and again to $\natzero$ but with type $\Nat^{\sizesucc{\sizesucc{\sizetwo}}}$ in the other copy.
To formalize this distinction, we require a few preliminary definitions: the \emph{typed} term $\natzero \choice \natzero$
will reduce to the following \emph{closed distribution of typed terms}:
\begin{equation}
\label{eq:example-subject-reduction2}
\distrelts{\left(\natzero \typsep \Nat^{\sizesucc{\sizeone}}\right)^{\frac{1}{2}},
\left(\natzero \typsep \Nat^{\sizesucc{\sizesucc{\sizetwo}}}\right)^{\frac{1}{2}}}
\end{equation}
which types the \emph{pseudo-representation} $\pseudorep{\natzero^{\frac{1}{2}},\natzero^{\frac{1}{2}}}$
of $\semantics{\natzero \choice \natzero}$.
The quantity which will be preserved during the reduction is the \emph{average} type of (\ref{eq:example-subject-reduction2}):
$$
\frac{1}{2} \cdot \distrelts{\left(\Nat^{\sizesucc{\sizeone}}\right)^1} \ +\ 
\frac{1}{2} \cdot \distrelts{\left(\Nat^{\sizesucc{\sizesucc{\sizetwo}}}\right)^1}
\ \ =\ \ 
\distrelts{\left(\Nat^{\sizesucc{\sizeone}}\right)^{\frac{1}{2}},
\left(\Nat^{\sizesucc{\sizesucc{\sizetwo}}}\right)^{\frac{1}{2}}}
$$
which we call the \emph{expectation type} of (\ref{eq:example-subject-reduction2}), and which coincides with the type of
the initial term (\ref{eq:example-subject-reduction1}).

\begin{definition}[Distributions of Distribution Types, of Typed Terms]
\begin{varitemize}
 \item 
   A \emph{distribution of distribution types} is a distribution
   $\distrone$ over the set of distribution types, and such that
   $\distrtypone,\,\distrtyptwo \in \supp{\distrone}
   \ \Rightarrow\ \underlying{\distrtypone}\,=\,\underlying{\distrtyptwo}$.
 \item 
   A \emph{distribution of typed terms}, or \emph{typed distribution},
   is a distribution of typing sequents which are derivable in the
   monadic, affine sized type system. The representation of such
   a distribution has thus the following form:
   $
   \distrelts{\left(\contextsizedone_{\indexone} \contextsep \contextdistrone_{\indexone} \proves \termone_{\indexone}
     \typsep \distrtypone_{\indexone}\right)^{p_{\indexone}}  \sep \indexone \in \indexsetone}.
   $
   In the sequel, we restrict to the \emph{uniform} case in which all
   the terms appearing in the sequents are typed with distribution
   types of the same fixed underlying type.  \longv{We denote this unique simple type
   $\simpletypone$ as $\underlying{\vec{\distrtypone}}$.}
 \item 
   A \emph{distribution of closed typed terms}, or \emph{closed typed distribution}, is a typed distribution
   in which all contexts are $\emptyset \contextsep \emptyset$.
   In this case, we simply write the representation of the distribution as
   $\distrelts{\left(\termone_{\indexone} \typsep \distrtypone_{\indexone}\right)^{p_{\indexone}} \sep \indexone \in \indexsetone}$,
   or even as  $\left(\termone_{\indexone} \typsep \distrtypone_{\indexone}\right)^{p_{\indexone}}$
   when the indexing is clear from context.
   We write pseudo-representations in a similar way.
  \longv{
  \item The \emph{underlying term distribution} of a closed typed distribution 
  $\distrelts{\left(\termone_{\indexone} \typsep \distrtypone_{\indexone}\right)^{p_{\indexone}}
  \sep \indexone \in \indexsetone}$
  is the distribution $\distrelts{\left(\termone_{\indexone}\right)^{p_{\indexone}}\sep \indexone \in \indexsetone}$.
  }
\end{varitemize}
\end{definition}

\begin{definition}[Expectation Types]
 Let $\left(\termone_{\indexone} \typsep \distrtypone_{\indexone}\right)^{p_{\indexone}}$ be a closed typed distribution.
 We define its \emph{expectation type} as the distribution type 
 $
 \expectype{\left(\termone_{\indexone} \typsep \distrtypone_{\indexone}\right)^{p_{\indexone}}}
 = \sum_{\indexone \in \indexsetone}\ p_{\indexone} \distrtypone_{\indexone}
 $.
\end{definition}

\longv{
\begin{lemma}
 Expectation is linear:
 \begin{itemize}
  \item $\expectype{\left(\termone_{\indexone} \typsep \distrtypone_{\indexone}\right)^{p_{\indexone}} + 
  \left(\termtwo_{\indextwo} \typsep \distrtyptwo_{\indextwo}\right)^{p'_{\indextwo}}}\ =\ 
  \expectype{\left(\termone_{\indexone} \typsep \distrtypone_{\indexone}\right)^{p_{\indexone}}}
  + \expectype{\left(\termtwo_{\indextwo} \typsep \distrtyptwo_{\indextwo}\right)^{p'_{\indextwo}}}$, 
  \item $\expectype{\left(\termone_{\indexone} \typsep \distrtypone_{\indexone}\right)^{pp'_{\indexone}}}\ =\ 
  p \cdot \expectype{\left(\termone_{\indexone} \typsep \distrtypone_{\indexone}\right)^{p'_{\indexone}}}$.
 \end{itemize}

\end{lemma}
}

\longv{

\subsection{Subtyping Probabilistic Sums}

}

%
%
%

%
%
%


\longv{
\begin{lemma}[Subtyping Probabilistic Sums]
\label{lemma/subtyping-prob-sum}
 Suppose that $\distrsum{(\distrtyptwo \choice_p \distrtypthree)}\,=\,1$ and that
 $\distrtyptwo \choice_p \distrtypthree \subtypeleq \distrtypone$. Then there exists $\distrtyptwo'$ and $\distrtypthree'$
 such that $\distrtypone\,=\,\distrtyptwo' \choice_p \distrtypthree'$, $\distrtyptwo \subtypeleq \distrtyptwo'$,
 and that $\distrtypthree \subtypeleq \distrtypthree'$. Note that this implies that
 $\supp{\distrtyptwo'} \cup \supp{\distrtypthree'}\,=\,\supp{\distrtypone}$.
\end{lemma}

}

\longv{
\begin{proof}
Let $\distrtyptwo\ =\ \distrelts{\typone_{\indexone}^{p'_{\indexone}}\sep \indexone \in \indexsetone}$
and $\distrtypthree\ =\ \distrelts{\typtwo_{\indextwo}^{p''_{\indextwo}} \sep \indextwo \in \indexsettwo}$.
We assume, without loss of generality, that $\indexsetone$ and $\indexsettwo$ are chosen in such a way that, setting
$\indexsetthree\,=\,\indexsetone \cap \indexsettwo$, 
$$
\exists \left(\indexone,\indextwo\right) \in \indexsetone \times \indexsettwo,\ \ \typone_{\indexone}\,=\,\typtwo_{\indextwo} \ \ \Longleftrightarrow \ \ 
\indexone\,=\,\indextwo \in \indexsetthree.
$$
It follows that 
$$
\distrtyptwo \choice_p \distrtypthree \ \ =\ \ \distrelts{\typone_{\indexone}^{pp'_{\indexone}}\sep \indexone \in \indexsetone\setminus\indexsetthree}
\ + \ \distrelts{\typtwo_{\indextwo}^{(1-p)p''_{\indextwo}} \sep \indextwo \in \indexsettwo\setminus\indexsetthree}
\ +\ \distrelts{\typone_{\indexone}^{pp'_{\indexone}+(1-p)p''_{\indexone}}\sep \indexone \in \indexsetthree}
$$
Set $\distrtypone\ =\ \distrelts{\typthree_{\indexfour}^{p'''_{\indexfour}}\sep \indexfour \in \indexsetfour}$.
Since $\distrtyptwo \choice_p \distrtypthree \subtypeleq \distrtypone$ and
$\distrsum{(\distrtyptwo \choice_p \distrtypthree)}\,=\,1$, there exists
a decomposition
$$
\distrtypone\ \ =\ \ \pseudorep{\typthree_{\indexone}^{pp'_{\indexone}}\sep \indexone \in \indexsetone\setminus\indexsetthree}
\ + \ \pseudorep{\typthree_{\indextwo}^{(1-p)p''_{\indextwo}}\sep \indextwo \in \indexsettwo\setminus\indexsetthree}
\ + \ \pseudorep{\typthree_{\indexthree}^{pp'_{\indexone}+(1-p)p''_{\indexone}}\sep \indexthree \in \indexsetthree}
$$
(note that the supports of these distributions may have a non-empty intersection), and this decomposition
is such that $\forall \indexone \in \indexsetone,\ \ \typone_{\indexone} \subtypeleq \typthree_{\indexone}$
and $\forall \indextwo \in \indexsettwo,\ \ \typtwo_{\indextwo} \subtypeleq \typthree_{\indextwo}$.
We define
$\distrtyptwo'\ =\ \distrelts{\typthree_{\indexone}^{p'_{\indexone}}\sep \indexone \in \indexsetone}$
and $\distrtypthree'\ =\ \distrelts{\typthree_{\indextwo}^{p''_{\indextwo}} \sep \indextwo \in \indexsettwo}$
which satisfy $\distrtyptwo \subtypeleq \distrtyptwo'$ and $\distrtypthree \subtypeleq \distrtypthree'$
but also, by construction, $\distrtypone\,=\,\distrtyptwo' \choice_p \distrtypthree'$. 
\end{proof}

}

\longv{
\begin{corollary}
 \label{corollary/subtyping-probabilistic-sums}
 Suppose that $\distrtypone\,=\,\sum_{\indexone \in \indexsetone}\ p_{\indexone} \cdot \distrtypone_{\indexone}$ is a distribution such that
 $\distrtypone \subtypeleq \distrtyptwo$
 and that $\distrsum{\distrtypone} \,=\,1$.
 Then there exists a family $\left(\distrtyptwo_{\indexone}\right)_{\indexone \in \indexsetone}$ of distributions such that
 $\distrtyptwo\,=\,\sum_{\indexone \in \indexsetone} p_{\indexone} \cdot \distrtyptwo_{\indexone}$ and that, for all $\indexone \in \indexsetone$,
 $\distrtypone_{\indexone} \subtypeleq \distrtyptwo_{\indexone}$.
\end{corollary}

}

\longv{
Note that the requirement that $\distrsum{\distrtypone} \,=\,1$ is not necessary to obtain this result, although it simplifies the reasoning.

}

\longv{

\subsection{Generation Lemma for Typing}

}
\longv{
\begin{lemma}[Generation Lemma for Typing]~
\label{lemma/generation-typing}
 
 \begin{enumerate}
  \item $\emptyset \contextsep \emptyset \proves \letin{\varone}{\valone}{\termtwo} \typsep \distrtypone \ \ \implies\ \ \exists 
  \left(\distrtyptwo,\typone\right),\ \ 
  \emptyset \contextsep \emptyset \proves \valone \typsep \typone$ and 
  $\varone \typsep \typone \contextsep \emptyset \proves \termtwo \typsep \distrtyptwo$ and $\distrtyptwo \subtypeleq \distrtypone$.
  \item $\emptyset \contextsep \emptyset \ \proves\ \valone\ \valtwo \typsep \distrtypone\ \ \implies\ \ \exists \left(\distrtyptwo,\typone\right),\ \ 
  \emptyset \contextsep \emptyset \ \proves\ \valone\typsep \typone \typarrow\distrtyptwo$ 
  and $\emptyset \contextsep \emptyset \ \proves\ \valtwo \typsep \typone$
  and $\distrtyptwo \subtypeleq \distrtypone$.
  \item $\emptyset \contextsep \emptyset\ \proves\ \abstr{\varone}{\termone} \typsep \typone \typarrow\distrtypone\ \ 
  \implies\ \ \exists \left(\distrtyptwo,\typtwo\right),\ \
  \varone\typsep\typtwo \contextsep \emptyset \ \proves\ 
  \termone \typsep\distrtyptwo$ and $\typone \subtypeleq \typtwo$ and $\distrtyptwo \subtypeleq \distrtypone$.
  \item $\emptyset \contextsep \emptyset \proves \termone\ \choice_p\ \termtwo\typsep \distrtypone\ \ \implies\ \ 
  \exists \left(\distrtyptwo,\distrtypthree\right),\ \ \emptyset \contextsep \emptyset \proves \termone\typsep \distrtyptwo$
   and $\emptyset \contextsep \emptyset \proves \termtwo\typsep \distrtypthree$ with
   $\distrsum{(\distrtyptwo \choice_p \distrtypthree)}\,=\,1$ and
   $\distrtyptwo \choice_p \distrtypthree \subtypeleq \distrtypone$
   and $\underlying{\distrtypone}\,=\,\underlying{\distrtyptwo}\,=\,\underlying{\distrtypthree}$.
   \item $\emptyset \contextsep \emptyset \proves \letin{\varone}{\termone}{\termtwo} \typsep \distrtyptwo
   \ \ \implies\ \ \exists \left(\indexsetone,\left(\typone_{\indexone}\right)_{\indexone \in \indexsetone},
   \left(p_{\indexone}\right)_{\indexone \in \indexsetone}, \left(\distrtypone_{\indexone}\right)_{\indexone \in \indexsetone}\right)$ such that
   \begin{itemize}
    \item $\sum_{\indexone \in \indexsetone}\ p_{\indexone} \cdot \distrtypone_{\indexone} \subtypeleq \distrtyptwo$,
    \item $\distrsum{\left(\sum_{\indexone \in \indexsetone}\ p_{\indexone} \cdot \distrtypone_{\indexone}\right)}\ =\ 1$,
    \item $\emptyset \contextsep \emptyset \proves \termone \typsep \distrelts{\typone_{\indexone}^{p_{\indexone}} \sep \indexone \in \indexsetone}$,
    \item $\forall \indexone \in \indexsetone,\ \ \varone\typsep \typone_{\indexone} \contextsep \emptyset \proves \termtwo \typsep \distrtypone_{\indexone}$.
   \end{itemize}

   \item $\emptyset \contextsep \emptyset \proves \caseof{\valone}{\natsucc \rightarrow \valtwo \smallsep \natzero \rightarrow \valthree} \typsep
   \distrtypone \ \ \implies\ \ \exists \left(\sizeone,\distrtyptwo\right)$ such that
   $\emptyset \contextsep \emptyset \proves \valone \typsep \Nat^{\sizesucc{\sizeone}}$ and 
   $\emptyset \contextsep \emptyset \proves \valtwo \typsep \Nat^{\sizeone} \typarrow \distrtyptwo$ and
   $\emptyset \contextsep \emptyset \proves \valthree \typsep \distrtyptwo$
   with $\distrtyptwo \subtypeleq \distrtypone$.

   \item $\emptyset \contextsep \emptyset \proves \letrec{\funcone}{\valone} \typsep \distrtypone \ \ \implies\ \ \exists\,
   \left(\left(p_{\indextwo}\right)_{\indextwo \in \indexsettwo},\,\left(\sizeone_{\indextwo}\right)_{\indextwo \in \indexsettwo},\,\sizevarone \right)$
   such that 
   \begin{itemize}
    \item $\Nat^{\sizetwo} \typarrow \subst{\distrtyptwo}{\sizevarone}{\sizetwo} \subtypeleq \distrtypone$,
    \item $\forall \indextwo \in \indexsettwo,\ \ \spine{\sizeone_{\indextwo}}\,=\,\sizevarone$,
    \item $\sizevarone \notin \contextone \text{ and } \sizevarone \text{ positive in } \nu$,
    \item $\distrelts{\left(\Nat^{\sizeone_{\indextwo}} \typarrow \subst{\distrtyptwo}{\sizevarone}{\sizeone_{\indextwo}}
\right)^{p_{\indextwo}} \sep \indextwo \in \indexsettwo} \text{induces an AST sized walk}$,
    \item $\emptyset \contextsep \funcone \typsep \distrelts{\left(\Nat^{\sizeone_{\indextwo}} \typarrow 
    \subst{\distrtyptwo}{\sizevarone}{\sizeone_{\indextwo}}
    \right)^{p_{\indextwo}} \sep \indextwo \in \indexsettwo} \proves \valone \typsep 
    \Nat^{ \sizesucc{\sizevarone}} \typarrow \subst{\distrtyptwo}{\sizevarone}{\sizesucc{\sizevarone}}$
   \end{itemize}
 \end{enumerate}
\end{lemma}

}

\longv{
\begin{proof}
 By inspection of the rules, the key point being that the subtyping rule is the only one which is not syntax-directed, and that
 by transitivity of $\subtypeleq$ we can compose several successive subtyping rules.
 In case (5), we have $\distrsum{\left(\sum_{\indexone \in \indexsetone}\ p_{\indexone} \cdot \distrtypone_{\indexone}\right)}\ =\ 1$
 since it appears that $\emptyset \contextsep \emptyset \proves \letin{\varone}{\termone}{\termtwo} \typsep \sum_{\indexone \in \indexsetone}\ p_{\indexone} \cdot \distrtypone_{\indexone}$. Lemma~\ref{lemma/dirac-types-values} allows then to conclude that this distribution of types has sum 1.
\end{proof}

}

\longv{\subsection{Value Substitutions}}

\longv{
\begin{definition}[Context Extending Another]
 We say that a context $\contextsizedtwo \contextsep \contextdistrtwo$
 extends a context $\contextsizedone \contextsep \contextdistrone$
 when:
 \begin{itemize}
  \item for every $\varone \typsep \typone \in \contextsizedone$, we have $\varone \typsep \typone \in \contextsizedtwo$.
  \item and either $\contextdistrone\,=\,\emptyset$ or $\contextdistrone\,=\,\contextdistrtwo$.
 \end{itemize}
 In other words, $\contextsizedtwo \contextsep \contextdistrtwo$
 extends $\contextsizedone \contextsep \contextdistrone$ when there exists
 $\contextsizedthree$ and $\contextdistrthree$ such that
 $\contextsizedtwo\,=\,\contextsizedone,\,\contextsizedthree$
 and $\contextdistrtwo\,=\,\contextdistrone,\,\contextdistrthree$.
\end{definition}
}

\longv{
\begin{lemma}
\label{lemma:weakening-contexts}
 Let $\termone$ be a closed term such that $\contextsizedone \contextsep \contextdistrone \proves \termone \typsep \distrtypone$.
 Then for every context $\contextsizedtwo \contextsep \contextdistrtwo$ extending
 $\contextsizedone \contextsep \contextdistrone$, we have
 $\contextsizedtwo \contextsep \contextdistrtwo \proves \termone \typsep \distrtypone$.
\end{lemma}
}

\longv{
\begin{proof}
 We proceed by induction on the structure of $\termone$. We set
 $\contextsizedtwo\,=\,\contextsizedone,\,\contextsizedthree$
 and $\contextdistrtwo\,=\,\contextdistrone,\,\contextdistrthree$.
 \begin{itemize}
  \item If $\termone\,=\,\varone$ is a variable, the result is immediate.
  \item If $\termone\,=\,0$, the result is immediate.
  \item If $\termone\,=\,\natsucc\ \valone$, we have by typing rules that $\typone = \Nat^{\sizesucc{\sizeone}}$
  and that $\contextsizedone \contextsep \contextdistrone \proves \valone \typsep \Nat^{\sizeone}$.
  By induction hypothesis $\contextsizedtwo \contextsep \contextdistrtwo \proves \valone \typsep \Nat^{\sizeone}$
  from which we conclude using the typing rule for $\natsucc$.
  \item If $\termone\,=\,\abstr{\varone}{\termtwo}$, we have $\typone = \typtwo \typarrow \distrtypone$ and 
  $\contextsizedone, \varone \typsep \typtwo\ \contextsep \contextdistrone \proves \termtwo \typsep \distrtypone$.
  By definition, $\contextsizedtwo, \varone \typsep \typtwo\ \contextsep \contextdistrtwo$
  extends $\contextsizedone, \varone \typsep \typtwo\ \contextsep \contextdistrone$
  so that we have $\contextsizedtwo, \varone \typsep \typtwo\ \contextsep \contextdistrtwo \proves \termtwo \typsep \distrtypone$.
  The result follows using the Lambda rule.
  \item If $\termone\,=\,\letrec{\funcone}{\valone}$, the typing rule is of the shape
  $$
 \AxiomC{$\underlying{\contextsizedone_1}\,=\,\Nat$}
 \noLine
 \UnaryInfC{$\sizevarone \notin \contextone_1 \text{ and } \sizevarone \text{ positive in } \nu
 \text{ and } \forall \indextwo \in \indexsettwo,\ \spine{\sizeone_\indextwo}\,=\,\sizevarone$}
 \noLine
 \UnaryInfC{$\distrelts{\left(\Nat^{\sizeone_{\indextwo}} \typarrow \subst{\distrtyptwo}{\sizevarone}{\sizeone_{\indextwo}}
 \right)^{p_{\indextwo}} \sep \indextwo \in \indexsettwo} \text{induces an AST sized walk}$}
 \noLine
 \UnaryInfC{$\contextsizedone_1 \contextsep
 \funcone \typsep \distrelts{\left(\Nat^{\sizeone_{\indextwo}} \typarrow \subst{\distrtyptwo}{\sizevarone}{\sizeone_{\indextwo}}
 \right)^{p_{\indextwo}} \sep \indextwo \in \indexsettwo} \proves \valone \typsep 
 \Nat^{ \sizesucc{\sizevarone}} \typarrow \subst{\distrtyptwo}{\sizevarone}{\sizesucc{\sizevarone}}$}
 \LeftLabel{$\letrecname$ $\quad$}
 \UnaryInfC{$\contextsizedone_1,\,\contextsizedone_2 \contextsep \contextdistrone \proves \letrec{\funcone}{\valone} \typsep  
 \Nat^{\sizetwo} \typarrow \subst{\distrtyptwo}{\sizevarone}{\sizetwo}$}
 \DisplayProof
 $$
 Let $\contextsizedtwo\,=\,\contextsizedtwo_1,\,\contextsizedtwo_2$ with
 $\contextsizedtwo_1$ the maximal subcontext consisting only of variables of affine type $\Nat$.
 Then
 $$
 \contextsizedtwo_1 \contextsep
 \funcone \typsep \distrelts{\left(\Nat^{\sizeone_{\indextwo}} \typarrow \subst{\distrtyptwo}{\sizevarone}{\sizeone_{\indextwo}}
 \right)^{p_{\indextwo}} \sep \indextwo \in \indexsettwo}
 $$
 extends
 $$
 \contextsizedone_1 \contextsep
 \funcone \typsep \distrelts{\left(\Nat^{\sizeone_{\indextwo}} \typarrow \subst{\distrtyptwo}{\sizevarone}{\sizeone_{\indextwo}}
 \right)^{p_{\indextwo}} \sep \indextwo \in \indexsettwo}
 $$
 so that by induction hypothesis $\contextsizedtwo_1 \contextsep
 \funcone \typsep \distrelts{\left(\Nat^{\sizeone_{\indextwo}} \typarrow \subst{\distrtyptwo}{\sizevarone}{\sizeone_{\indextwo}}
 \right)^{p_{\indextwo}} \sep \indextwo \in \indexsettwo} \proves \valone \typsep 
 \Nat^{ \sizesucc{\sizevarone}} \typarrow \subst{\distrtyptwo}{\sizevarone}{\sizesucc{\sizevarone}}$
 so that we can conclude using the $\letrecname$ rule again that
 $$
 \contextsizedtwo_1,\,\contextsizedtwo_2 \contextsep \contextdistrtwo \proves \letrec{\funcone}{\valone} \typsep  
 \Nat^{\sizetwo} \typarrow \subst{\distrtyptwo}{\sizevarone}{\sizetwo}.
 $$
 \item If $\termone\,=\,\valone \ \valtwo$, the typing derivation provides contexts such that
 $\contextsizedone\,=\,\contextsizedone_1,\,\contextsizedone_2,\,\contextsizedone_3$
 and that $\contextdistrone\,=\,\contextdistrone_1,\,\contextdistrone_2$
 with $\contextsizedone_1,\,\contextsizedone_2 \contextsep \contextdistrone_1 
 \proves \valone \typsep \typone \typarrow \distrtypone$
 and $\contextsizedone_1,\,\contextsizedone_3 \contextsep \contextdistrone_2 
 \proves \valtwo \typsep \typone$.
 By induction hypothesis, $\contextsizedone_1,\,\contextsizedone_3,\,\contextsizedthree \contextsep \contextdistrone_2,
 \,\contextdistrthree
 \proves \valtwo \typsep \typone$ from which we conclude using the App rule.
 \item If $\termone\,=\,\letin{\varone}{\termtwo}{\termthree}$, the typing derivation provides contexts such that
 $\contextsizedone\,=\,\contextsizedone_1,\,\contextsizedone_2,\,\contextsizedone_3$
 and that $\contextdistrone\,=\,\contextdistrone_1,\,\sum_{\indexone \in \indexsetone}\ p_\indexone \cdot \contextdistrone_{2,\indexone}$ with 
 $\contextsizedone_1,\,\contextsizedone_2 \contextsep \contextdistrone_1
 \proves \termone \typsep \distrelts{\typone_{\indexone}^{p_{\indexone}} \sep \indexone \in \indexsetone}$
 and
 $\contextsizedone_1,\,\contextsizedone_3,\,\varone\typsep \typone_{\indexone} \contextsep \contextdistrone_{2,\indexone} \proves \termtwo \typsep \distrtypone_{\indexone}$.
 By induction hypothesis, $\contextsizedone_1,\,\contextsizedone_2,\,\contextsizedthree \contextsep \contextdistrone_1,\,\contextdistrthree \proves \termone \typsep \distrelts{\typone_{\indexone}^{p_{\indexone}} \sep \indexone \in \indexsetone}$
 from which we conclude using the Let rule.
 \item If $\termone\,=\termtwo  \choice_p  \termthree$,
 then $\contextdistrone\,=\,\contextdistrone_1 \choice_p \contextdistrone_2$
 with 
 $\contextsizedone \contextsep \contextdistrone_1\ \proves\ \termone \typsep \distrtypone$
 and
 $\contextsizedone \contextsep \contextdistrone_2\ \proves\ \termtwo \typsep \distrtyptwo$.
 By applying induction hypothesis twice, we obtain 
 $\contextsizedone,\,\contextsizedthree \contextsep \contextdistrone_1,\,\contextdistrthree\ \proves\ \termone \typsep \distrtypone$
 and
 $\contextsizedone,\,\contextsizedthree \contextsep \contextdistrone_2,\,\contextdistrthree\ \proves\ \termtwo \typsep \distrtyptwo$.
 We apply the Choice rule; it remains to prove that 
 $\left(\contextdistrone_1,\,\contextdistrthree\right) \choice_p \left(\contextdistrone_2,\,\contextdistrthree\right)
 \ =\ \contextdistrone_1 \choice_p \contextdistrone_2,\,\contextdistrthree$
 which is easily done by definition of $\choice_p$.
 \item If $\termone\,=\,\caseof{\valone}{\natsucc \rightarrow \valtwo \smallsep \natzero \rightarrow \valthree}$,
 the typing derivation provides contexts such that
 $\contextsizedone\,=\,\contextsizedone_1,\,\contextsizedone_2$
 with $\contextsizedone_1 \contextsep \emptyset \proves \valone \typsep \Nat^{\sizesucc{\sizeone}}$
 and 
 $\contextsizedone_2 \contextsep \contextdistrone \proves \valtwo \typsep \Nat^{\sizeone} \typarrow \distrtypone$
 and
 $\contextsizedone_2 \contextsep \contextdistrone \proves \valthree \typsep \distrtypone$.
 By induction hypothesis, 
 $\contextsizedone_2,\,\contextsizedthree \contextsep \contextdistrone,\,\contextdistrthree \proves \valtwo \typsep \Nat^{\sizeone} \typarrow \distrtypone$
 and 
 $\contextsizedone_2,\,\contextsizedthree \contextsep \contextdistrone,\,\contextdistrthree \proves \valthree \typsep \distrtypone$
 from which we conclude using the Case rule.
 \end{itemize}

\end{proof}
}

\longv{
\begin{lemma}[Closed Value Substitution]
 \label{lemma/value-substitution}
 Suppose that $\contextsizedone,\,\varone \typsep \typone \contextsep \contextdistrone \proves \termone \typsep \distrtypone$
 and that $\emptyset \contextsep \emptyset \proves \valone \typsep \typone$.
 Then $\contextsizedone\contextsep \contextdistrone \proves \subst{\termone}{\varone}{\valone} \typsep \distrtypone$.
\end{lemma}

}

\longv{
\begin{proof}
 As usual, the proof is by induction on the structure of the typing derivation. We proceed by case analysis
 on the last rule:
 \begin{varitemize}
  \item If it is Var, we have two cases.
  \begin{itemize}
   \item If the conclusion is
   $\contextsizedone,\,\varone \typsep \typone \contextsep \contextdistrone \ \proves\ \varone\typsep \typone$
   then $\subst{\varone}{\varone}{\valone}\,=\,\valone$. By Lemma~\ref{lemma:weakening-contexts},
   we obtain that $\contextsizedone\contextsep \contextdistrone \proves \valone \typsep \typone$.
   \item If the conclusion is
   $\contextsizedone,\,\varone \typsep \typone,\,\vartwo \typsep \typtwo \contextsep \contextdistrone \ \proves\ \vartwo\typsep \typtwo$
   then $\subst{\vartwo}{\varone}{\valone}\,=\,\vartwo$ and we obtain 
   $\contextsizedone,\,\vartwo \typsep \typtwo \contextsep \contextdistrone \ \proves\ \vartwo\typsep \typtwo$ using the Var rule.
  \end{itemize}
  \item If it is Var', the situation is similar to the latter case of the previous one. The conclusion is
  $\contextsizedone,\,\varone \typsep \typone \contextsep \vartwo \typsep \typtwo \ \proves\ \vartwo\typsep \typtwo$ and $\subst{\vartwo}{\varone}{\valone}\,=\,\vartwo$
  so that we obtain 
  $\contextsizedone\contextsep \vartwo \typsep \typtwo \ \proves\ \vartwo\typsep \typtwo$
  using the Var' rule.
  \item If it is Succ, then $\termone\,=\,\natsucc\ \valtwo$ and $\distrtypone\,=\,\Nat^{\sizesucc{\sizeone}}$.
  We obtain by induction hypothesis that 
  $\contextsizedone\contextsep \contextdistrone \proves \subst{\valtwo}{\varone}{\valone} \typsep \Nat^{\sizeone}$
  and we conclude using the Succ rule that 
  $\contextsizedone\contextsep \contextdistrone \proves \subst{(\natsucc\ \valtwo)}{\varone}{\valone} \typsep \Nat^{\sizesucc{\sizeone}}$.
  \item If it is Zero, we obtain immediately the result.
  \item If it is $\lambda$, suppose that 
  $\contextsizedone,\,\varone \typsep \typone \contextsep \contextdistrone \proves \abstr{\vartwo}{\termone}
  \typsep \typtwo \typarrow \distrtypone$. 
  This comes from $\contextsizedone,\,\varone \typsep \typone,\,\vartwo \typsep \typtwo \contextsep \contextdistrone \proves \termone \typsep \distrtypone$ to which we apply the induction hypothesis, obtaining that
  $\contextsizedone,\,\vartwo \typsep \typtwo \contextsep \contextdistrone \proves \subst{\termone}{\varone}{\valone} \typsep \distrtypone$. Then applying the $\lambda$ rule gives the expected result.
  \item For all the remaining cases, as for the $\lambda$ rule, the result is obtained in a straightforward way
  from the induction hypothesis.
 \end{varitemize}

\end{proof}

}

\longv{
\begin{lemma}[Substitution for distributions]
 \label{lemma/distrib-substitution}
 Suppose that $\contextsizedone \contextsep \varone \typsep \distrelts{\typone_{\indexone}^{p_{\indexone}} \sep \indexone \in \indexsetone}
 \proves \termone \typsep
 \distrtypone$ and that, for every $\indexone \in \indexsetone$, we have $\emptyset \contextsep \emptyset \proves \valone \typsep \typone_{\indexone}$.
 Then $\contextsizedone \contextsep \emptyset \proves \subst{\termone}{\varone}{\valone} \typsep \distrtypone$.
\end{lemma}

}

\longv{
\begin{proof}
 The proof is by induction on the structure of the typing derivation. We proceed by case analysis
 on the last rule:
 \begin{varitemize}
  \item If it is Var, we have $\termone\,=\,\vartwo \neq \varone$ and $\vartwo \in \contextsizedone$. It follows 
  that $\subst{\vartwo}{\varone}{\valone}\,=\,\vartwo$ and we obtain
  $\contextsizedone \contextsep \emptyset \proves \subst{\termone}{\varone}{\valone} \typsep \distrtypone$
  simply by the Var rule.
  \item If it is Var', we have $\termone\,=\,\varone$ so that 
  $\subst{\termone}{\varone}{\valone}\,=\,\valone$. Moreover, the distribution 
  $\distrelts{\typone_{\indexone}^{p_{\indexone}} \sep \indexone \in \indexsetone}$
  must be Dirac; we denote by $\typone$ the unique element of its support. Note that we also
  obtain $\typone = \distrtypone$. As we supposed that
  $\emptyset \contextsep \emptyset \proves \valone \typsep \typone$, Lemma~\ref{lemma:weakening-contexts}
  gives $\contextsizedone \contextsep \emptyset \proves \valone \typsep \typone$ from which we conclude.
  \item If it is LetRec, then $\varone$ does not occur free in $\termone$. It
  follows that $\subst{\termone}{\varone}{\valone}\,=\,\termone$, and we can derive 
  $\contextsizedone \contextsep \emptyset \proves \subst{\termone}{\varone}{\valone} \typsep \distrtypone$
  using a LetRec rule with the same hypothesis.

  \item All others cases are treated straightforwardly using the induction hypothesis.
 \end{varitemize}
\end{proof}

}

\longv{
\begin{lemma}~
\label{lemma:size-constructors-properties}

\begin{enumerate}
 \item $\contextsizedone \contextsep \contextdistrone \proves \natsucc\ \valone \typsep \Nat^{\sizesucc{\sizeone}} \ \ \implies \ \ 
 \contextsizedone \contextsep \contextdistrone \proves \valone \typsep \Nat^{\sizeone}$,
 \item $\contextsizedone \contextsep \contextdistrone \proves \natzero \typsep \Nat^{\sizeone} \ \ \implies\ \ \exists \sizetwo,\ \ \sizeone\,=\,\sizesucc{\sizetwo}$.
 \item $\contextsizedone \contextsep \contextdistrone \proves \natsucc\ \valone \typsep \Nat^{\sizeone} \ \ \implies\ \ \exists \sizetwo,\ \ \sizeone\,=\,\sizesucc{\sizetwo}$.
\end{enumerate}
\end{lemma}

}

\longv{
\begin{proof}
 All points are immediate due to the typing rules introducing $\natzero$ and $\natsucc$.
 Recall that by the subtyping rules $\sizesucc{\sizeinf}\,=\,\sizeinf$.
\end{proof}

}

\longv{\subsection{Size Substitutions}}

\longv{ 
\begin{lemma}[Successor and Size Order]
 \label{lemma:successor-size-order}
 Suppose that $\sizeone \sizeleq \sizetwo$. Then $\sizesucc{\sizeone} \sizeleq \sizesucc{\sizetwo}$.
\end{lemma}
}

\longv{ 
\begin{proof}
 By definition of $\sizeleq$, if $\sizeone \sizeleq \sizetwo$, 
 there are two cases: either $\sizetwo = \sizeinf$, or $\spine{\sizeone}=\spine{\sizetwo}=\sizevarone$
 with $\sizeone = \sizesuccit{\sizevarone}{k}$, $\sizetwo = \sizesuccit{\sizevarone}{k'}$
 and $k \leq k'$. In both cases the conclusion is immediate.
\end{proof}
}

\longv{
\begin{lemma}[Size Substitutions are Monotonic]~
\label{lemma:size-substitution-monotonic}
\begin{enumerate}
 \item Suppose that $\sizeone \sizeleq \sizetwo$,
  then for any size $\sizethree$ and size variable $\sizevarone$
  we have $\subst{\sizeone}{\sizevarone}{\sizethree} \sizeleq\subst{\sizetwo}{\sizevarone}{\sizethree}$.
 \item Suppose that $\sizeone \sizeleq \sizetwo$,
  then for any size $\sizethree$ and size variable $\sizevarone$
  we have $\subst{\sizethree}{\sizevarone}{\sizeone} \sizeleq\subst{\sizethree}{\sizevarone}{\sizetwo}$.
\end{enumerate}
\end{lemma}
}

\longv{ 
\begin{proof} 
\begin{enumerate}
 \item
 We proceed by induction on the derivation proving that $\sizeone \sizeleq \sizetwo$, by case analysis
 on the last rule.
 \begin{varitemize}
  \item If it is $\sizeone \sizeleq \sizeone$, then $\sizeone\,=\,\sizetwo$ and the result is immediate.
  \item If it is 
  $$
  \AxiomC{$\sizeone \sizeleq \sizefour$}
  \AxiomC{$\sizefour \sizeleq \sizetwo$}
  \BinaryInfC{$\sizeone \sizeleq \sizetwo$}
  \DisplayProof
  $$
  then by induction hypothesis
  $\subst{\sizeone}{\sizevarone}{\sizethree} \sizeleq\subst{\sizefour}{\sizevarone}{\sizethree}$
  and
  $\subst{\sizefour}{\sizevarone}{\sizethree} \sizeleq\subst{\sizetwo}{\sizevarone}{\sizethree}$
  so that we conclude using this same deduction rule.
  \item If it is $\sizeone \sizeleq \sizesucc{\sizeone}$,
  we have $\sizetwo\,=\,\sizesucc{\sizeone}$ and using the definition of size substitution we obtain
  $\subst{\sizetwo}{\sizevarone}{\sizethree}\,=\,\subst{\sizesucc{\sizeone}}{\sizevarone}{\sizethree}
  \,=\,\sizesucc{\subst{\sizeone}{\sizevarone}{\sizethree}}$. We conclude using the same deduction rule.
  \item If it is $\sizeone \sizeleq \sizeinf$, we have 
  $\subst{\sizeinf}{\sizevarone}{\sizethree}\,=\,\sizeinf$ and we obtain immediately
  $\subst{\sizeone}{\sizevarone}{\sizethree} \sizeleq \sizeinf$.
 \end{varitemize}
\item We proceed by case analysis on $\sizethree$. There are four cases:
 \begin{varitemize}
  \item If $\sizethree\,=\,\sizevarone$, 
  then $\subst{\sizethree}{\sizevarone}{\sizeone} = \sizeone \sizeleq
  \sizetwo = \subst{\sizethree}{\sizevarone}{\sizetwo}$.
  \item If $\sizethree\,=\,\sizevartwo \neq \sizevarone$, 
  then $\subst{\sizethree}{\sizevarone}{\sizeone} = \sizevartwo \sizeleq
  \sizevartwo = \subst{\sizethree}{\sizevarone}{\sizetwo}$.
  \item If $\sizethree\,=\,\sizesucc{\sizefour}$, we have by induction hypothesis
  that $\subst{\sizefour}{\sizevarone}{\sizeone} \sizeleq\subst{\sizefour}{\sizevarone}{\sizetwo}$.
  We conclude using Lemma~\ref{lemma:successor-size-order}.
  \item If $\sizethree\,=\,\sizeinf$, 
  $\subst{\sizethree}{\sizevarone}{\sizeone} = \sizeinf \sizeleq \sizeinf = \subst{\sizethree}{\sizevarone}{\sizetwo}$.
 \end{varitemize}

\end{enumerate}
\end{proof}
}

\longv{
\begin{lemma}[Size Substitutions and Subtyping]~
\label{lemma:size-substitution-subtyping}

\begin{enumerate}
 \item If $\typone \subtypeleq \typtwo$, then for any size $\sizeone$ and size variable $\sizevarone$,
 we have $\subst{\typone}{\sizevarone}{\sizeone} \subtypeleq \subst{\typtwo}{\sizevarone}{\sizeone}$.

 If $\distrtypone \subtypeleq \distrtyptwo$, then for any size $\sizeone$ and size variable $\sizevarone$,
 we have $\subst{\distrtypone}{\sizevarone}{\sizeone} \subtypeleq \subst{\distrtyptwo}{\sizevarone}{\sizeone}$.

 \item If $\positive{\sizevarone}{\typone}$ and $\sizeone \sizeleq \sizetwo$, we have
 $\subst{\typone}{\sizevarone}{\sizeone} \subtypeleq \subst{\typone}{\sizevarone}{\sizetwo}$.
 
  If $\positive{\sizevarone}{\distrtypone}$ and $\sizeone \sizeleq \sizetwo$, we have
 $\subst{\distrtypone}{\sizevarone}{\sizeone} \subtypeleq \subst{\distrtypone}{\sizevarone}{\sizetwo}$.

 \item If $\negative{\sizevarone}{\typone}$ and $\sizeone \sizeleq \sizetwo$, we have
 $\subst{\typone}{\sizevarone}{\sizetwo} \subtypeleq \subst{\typone}{\sizevarone}{\sizeone}$.

 If $\negative{\sizevarone}{\distrtypone}$ and $\sizeone \sizeleq \sizetwo$, we have
 $\subst{\distrtypone}{\sizevarone}{\sizetwo} \subtypeleq \subst{\distrtypone}{\sizevarone}{\sizeone}$.
\end{enumerate}
\end{lemma}
}

\longv{
\begin{proof}~

 \begin{enumerate}
  \item We prove both statements at the same time
  by induction on the derivation proving that $\distrtypone \subtypeleq \distrtyptwo$
  (or $\typone \subtypeleq \typtwo$).
  \begin{varitemize}
   \item If the last rule is $\typone \subtypeleq \typone$, then $\distrtypone\,=\,\distrtyptwo\,=\,\typone$
   and the result is immediate.
   \item If the last rule is 
   $$
   \AxiomC{$\sizethree \sizeleq \sizetwo$}
  \UnaryInfC{$\Nat^{\sizethree} \subtypeleq\ \Nat^{\sizetwo}$}
  \DisplayProof
  $$
  then by Lemma~\ref{lemma:size-substitution-monotonic} we have 
  $\subst{\sizethree}{\sizevarone}{\sizeone} \sizeleq \subst{\sizetwo}{\sizevarone}{\sizeone}$
  so that $\subst{\left(\Nat^{\sizethree}\right)}{\sizevarone}{\sizeone}
  \,=\,\Nat^{\subst{\sizethree}{\sizevarone}{\sizeone}}\ \subtypeleq\ 
  \Nat^{\subst{\sizetwo}{\sizevarone}{\sizeone}} \,=\,
  \subst{\left(\Nat^{\sizetwo}\right)}{\sizevarone}{\sizeone}$.
  \item If the last rule is
   $$
   \AxiomC{$\typtwo \subtypeleq \typone$}
   \AxiomC{$\distrtypone \subtypeleq \distrtyptwo$}
   \BinaryInfC{$\typone \typarrow \distrtypone \ \subtypeleq\ \typtwo \typarrow \distrtyptwo$}
   \DisplayProof
   $$
   then by induction hypothesis 
   $\subst{\typtwo}{\sizevarone}{\sizeone} \subtypeleq \subst{\typone}{\sizevarone}{\sizeone}$
   and
   $\subst{\distrtypone}{\sizevarone}{\sizeone} \subtypeleq \subst{\distrtyptwo}{\sizevarone}{\sizeone}$
   from which we conclude using the same rule.

  \item If the last rule is
   $$
   \AxiomC{$\exists \funcone \,:\, \indexsetone \to \indexsettwo,\ \ \left(\forall \indexone \in \indexsetone,\ \   \typone_{\indexone}
   \,\subtypeleq\,\typtwo_{\funcone(\indexone)} \right) \text{ and }
   \left(\forall \indextwo \in \indexsettwo,\ \ \sum_{\indexone \in \funcone^{-1}(\indextwo)}\ p_{\indexone} \leq   p'_{\indextwo} \right)$}
   \UnaryInfC{$\distrelts{\sigma_{\indexone}^{p_{\indexone}} \sep \indexone \in \indexsetone}\ \subtypeleq\ 
   \distrelts{\typtwo_{\indextwo}^{p'_{\indextwo}} \sep \indextwo \in \indexsettwo}$}
   \DisplayProof
   $$
   we obtain by induction hypothesis that for every $\indexone \in \indexsetone$
   $\subst{\typone_{\indexone}}{\sizevarone}{\sizeone}
   \,\subtypeleq\,\subst{\typtwo_{\funcone(\indexone)}}{\sizevarone}{\sizeone}$
   from which we conclude using the same rule.
  \end{varitemize}
  \item We prove (2) and (3) by mutual induction on $\distrtypone$ (or $\typone$).
  Let $\sizeone \sizeleq \sizetwo$.
  \begin{varitemize}
   \item If $\typone\,=\,\Nat^{\sizethree}$,
   \begin{varitemize}
    \item Suppose that $\positive{\sizevarone}{\Nat^{\sizethree}}$.
    Note that this does not assume anything on $\sizethree$.
    Since $\sizeone \sizeleq \sizetwo$,    
    we have $\subst{\left(\Nat^{\sizethree}\right)}{\sizevarone}{\sizeone}
    \ =\ \Nat^{\subst{\sizethree}{\sizevarone}{\sizeone}}
    \ \subtypeleq\  \Nat^{\subst{\sizethree}{\sizevarone}{\sizetwo}}
    \ =\ \subst{\left(\Nat^{\sizethree}\right)}{\sizevarone}{\sizetwo}
    $ where we used the monotonicity of size substitution 
    (Lemma~\ref{lemma:size-substitution-monotonic}).

    \item Suppose that $\negative{\sizevarone}{\Nat^{\sizethree}}$.
    Then $\sizevarone \notin \sizethree$ and 
    $\subst{\left(\Nat^{\sizethree}\right)}{\sizevarone}{\sizeone}
    \ =\ \subst{\left(\Nat^{\sizethree}\right)}{\sizevarone}{\sizetwo}$
    so that we can conclude.
   \end{varitemize}

   \item If $\typone\,=\,\typtwo\typarrow \distrtypone$,
   \begin{varitemize}
    \item Suppose that $\positive{\sizevarone}{\typone}$.
    Then $\negative{\sizevarone}{\typtwo}$ and $\positive{\sizevarone}{\distrtypone}$.
    By induction hypothesis, $\subst{\typtwo}{\sizevarone}{\sizetwo} \subtypeleq \subst{\typtwo}{\sizevarone}{\sizeone}$ and
    $\subst{\distrtypone}{\sizevarone}{\sizeone} \subtypeleq \subst{\distrtypone}{\sizevarone}{\sizetwo}$.
    By the subtyping rules, $\subst{\typone}{\sizevarone}{\sizeone} =
    \subst{\typtwo}{\sizevarone}{\sizeone} \typarrow \subst{\distrtypone}{\sizevarone}{\sizeone}
    \subtypeleq
    \subst{\typtwo}{\sizevarone}{\sizetwo} \typarrow \subst{\distrtypone}{\sizevarone}{\sizetwo}
    = \subst{\typone}{\sizevarone}{\sizetwo}$.
    \item Suppose that $\negative{\sizevarone}{\typone}$. The reasoning is symmetrical.
   \end{varitemize}
   \item If $\distrtypone\,=\,\distrelts{\typone_\indexone^{p_\indexone} \sep \indexone \in \indexsetone}$,
   \begin{varitemize}
    \item Suppose that $\positive{\sizevarone}{\distrtypone}$.
    Then for every $\indexone \in \indexsetone$ we have $\positive{\sizevarone}{\typone_\indexone}$
    and by induction hypothesis
    $\subst{\typone_\indexone}{\sizevarone}{\sizeone} \subtypeleq \subst{\typone_\indexone}{\sizevarone}{\sizetwo}$.
    We obtain that 
    $\subst{\distrtypone}{\sizevarone}{\sizeone} \subtypeleq \subst{\distrtypone}{\sizevarone}{\sizetwo}$
    using the identity as reindexing function.
    \item Suppose that $\negative{\sizevarone}{\distrtypone}$. The reasoning is symmetrical.
   \end{varitemize}   
  \end{varitemize}

 \end{enumerate}
\end{proof}
}

\longv{
\begin{lemma}[Size substitution]
\label{lemma/size-substitution}
If $\contextsizedone \contextsep \contextdistrone \proves \termone \typsep \distrtypone$, then for any size variable $\sizevarone$ and any size $\sizeone$
we have that $\subst{\contextsizedone}{\sizevarone}{\sizeone} \contextsep \subst{\contextdistrone}{\sizevarone}{\sizeone}
\proves \termone \typsep \subst{\distrtypone}{\sizevarone}{\sizeone}$.
\end{lemma}

}

\longv{

\begin{proof}
 We assume that $\sizevarone \notin \sizeone$, without loss of generality: else we introduce a fresh
 size variable $\sizevartwo$, substitute it with $\sizeone$, and then substitute $\sizevarone$ with
 $\sizevartwo$.
 The proof is by induction on the typing derivation. We proceed by case analysis on the last rule.
 \begin{varitemize}
  \item If it is Var: we have 
  $\contextsizedone,\,\varone \typsep \typone \contextsep \contextdistrone \ \proves\ \varone\typsep \typone$
  and deduce immediately using Var rule again that
  $\subst{\contextsizedone}{\sizevarone}{\sizeone},\,\varone \typsep \subst{\typone}{\sizevarone}{\sizeone} \contextsep \subst{\contextdistrone}{\sizevarone}{\sizeone} \ \proves\ \varone\typsep \subst{\typone}{\sizevarone}{\sizeone}$.\\[0.2cm]
  \item If it is Var': we have 
  $\contextsizedone \contextsep \varone \typsep \typone \ \proves\ \varone\typsep \typone$
  and deduce immediately using Var' rule again that
  $\subst{\contextsizedone}{\sizevarone}{\sizeone} \contextsep \varone \typsep \subst{\typone}{\sizevarone}{\sizeone}
  \ \proves\ \varone\typsep \subst{\typone}{\sizevarone}{\sizeone}$.\\[0.2cm]
  \item If it is Succ: then $\termone\,=\,\natsucc\ \valone$
  and $\distrtypone\,=\,\Nat^{\sizesucc{\sizetwo}}$. By induction hypothesis,
  $\subst{\contextsizedone}{\sizevarone}{\sizeone} \contextsep \subst{\contextdistrone}{\sizevarone}{\sizeone}
  \proves \valone \typsep \subst{\left(\Nat^{\sizetwo}\right)}{\sizevarone}{\sizeone}$.
  But $\subst{\left(\Nat^{\sizetwo}\right)}{\sizevarone}{\sizeone}\,=\,\Nat^{\subst{\sizetwo}{\sizevarone}{\sizeone}}$
  so that by the Succ rule 
  $\subst{\contextsizedone}{\sizevarone}{\sizeone} \contextsep \subst{\contextdistrone}{\sizevarone}{\sizeone}
  \proves \natsucc\ \valone \typsep \Nat^{\sizesucc{\subst{\sizetwo}{\sizevarone}{\sizeone}}}$.
  We use the equality $\Nat^{\sizesucc{\subst{\sizetwo}{\sizevarone}{\sizeone}}}
  \,=\,\subst{\left(\Nat^{\sizesucc{\sizetwo}}\right)}{\sizevarone}{\sizeone}$
  to conclude.\\[0.2cm]
  \item If it is Zero: the result is immediate.\\[0.2cm]
  \item If it is $\lambda$: we have $\termone\,=\,\abstr{\varone}{\termtwo}$ and
  $\distrtypone\,=\,\typone \typarrow \distrtyptwo$. By induction hypothesis,
  $\subst{\contextsizedone}{\sizevarone}{\sizeone},\,\varone\typsep\subst{\typone}{\sizevarone}{\sizeone} \contextsep \subst{\contextdistrone}{\sizevarone}{\sizeone} \ \proves\ \termtwo \typsep\subst{\distrtyptwo}{\sizevarone}{\sizeone}$. By application of the $\lambda$ rule,
  $\subst{\contextsizedone}{\sizevarone}{\sizeone} \contextsep \subst{\contextdistrone}{\sizevarone}{\sizeone} \ \proves\ \abstr{\varone}{\termtwo} \typsep
  \subst{\typone}{\sizevarone}{\sizeone} \typarrow \subst{\distrtyptwo}{\sizevarone}{\sizeone}$.
  We conclude using $\subst{\typone}{\sizevarone}{\sizeone} \typarrow \subst{\distrtyptwo}{\sizevarone}{\sizeone}
  \ =\ \subst{\left(\typone \typarrow \distrtyptwo\right)}{\sizevarone}{\sizeone}$.\\[0.2cm]
  \item If it is Sub: the hypothesis of the rule is 
  $\contextsizedone \contextsep \contextdistrone \proves \termone \typsep \distrtyptwo$
  for $\distrtyptwo \subtypeleq \distrtypone$.
  By induction hypothesis, 
  $\subst{\contextsizedone}{\sizevarone}{\sizeone} \contextsep \subst{\contextdistrone}{\sizevarone}{\sizeone}
  \proves \termone \typsep \subst{\distrtyptwo}{\sizevarone}{\sizeone}$.
  But by Lemma~\ref{lemma:size-substitution-subtyping} we have 
  $\subst{\distrtyptwo}{\sizevarone}{\sizeone} \subtypeleq
  \subst{\distrtypone}{\sizevarone}{\sizeone}$.
  We conclude using the Sub rule.\\[0.2cm]
  \item If it is App, we have $\termone\,=\,\valone\ \valtwo$ and
  $\contextsizedone\,=\,\contextsizedone_1,\,\contextsizedone_2,\,\contextsizedone_3$
  and $\contextdistrone\,=\,\contextdistrone_1,\,\contextdistrone_2$
  with $\underlying{\contextsizedone_1} = \Nat$,
  $\contextsizedone_1,\,\contextsizedone_2 \contextsep \contextdistrone_1 \proves 
  \valone \typsep \typone \typarrow \distrtypone$
  and $\contextsizedone_1,\,\contextsizedone_3 \contextsep \contextdistrone_2 \proves 
  \valtwo \typsep \typone$.
  Applying the induction hypothesis twice gives 
  $\subst{\contextsizedone_1}{\sizevarone}{\sizeone},\,\subst{\contextsizedone_2}{\sizevarone}{\sizeone} \contextsep 
  \subst{\contextdistrone_1}{\sizevarone}{\sizeone} \proves 
  \valone \typsep \subst{\left(\typone\typarrow \distrtypone\right)}{\sizevarone}{\sizeone}$
  and $\subst{\contextsizedone_1}{\sizevarone}{\sizeone},\,\subst{\contextsizedone_3}{\sizevarone}{\sizeone}
  \contextsep \subst{\contextdistrone_2}{\sizevarone}{\sizeone} \proves \valtwo \typsep 
  \subst{\typone}{\sizevarone}{\sizeone}$.
  Since  $\subst{\typone}{\sizevarone}{\sizeone} \typarrow \subst{\distrtypone}{\sizevarone}{\sizeone}
  \ =\ \subst{\left(\typone\typarrow \distrtypone\right)}{\sizevarone}{\sizeone}$, we can use the Application rule
  to conclude.\\[0.2cm]
  \item If it is Choice, then $\termone\,=\,\termtwo \choice_p \termthree$ and
  $\distrtypone\,=\,\distrtyptwo \choice_p \distrtypthree$ and
  $\contextdistrone\,=\,\contextdistrone_1 \choice_p \contextdistrone_2$
  with $\contextsizedone \contextsep \contextdistrone_1 \proves \termtwo \typsep \distrtyptwo$
  and $\contextsizedone \contextsep \contextdistrone_2 \proves \termthree \typsep \distrtypthree$
  and $\underlying{\distrtyptwo}=\underlying{\distrtypthree}$.
  The induction hypothesis, applied twice, gives 
  $\subst{\contextsizedone}{\sizevarone}{\sizeone} \contextsep \subst{\contextdistrone_1}{\sizevarone}{\sizeone}
  \proves \termtwo \typsep \subst{\distrtyptwo}{\sizevarone}{\sizeone}$
  and $\subst{\contextsizedone}{\sizevarone}{\sizeone} \contextsep \subst{\contextdistrone_2}{\sizevarone}{\sizeone}
  \proves \termthree \typsep \subst{\distrtypthree}{\sizevarone}{\sizeone}$
  from which we conclude using the Choice rule again and the equality
  $\subst{\distrtyptwo}{\sizevarone}{\sizeone} \choice_p 
  \subst{\distrtypthree}{\sizevarone}{\sizeone}\,=\,
  \subst{\left(\distrtyptwo \choice_p\distrtypthree\right)}{\sizevarone}{\sizeone}$
  from Lemma~\ref{lemma:commuting-size-substitutions-with-operations}.\\[0.2cm]
  \item If it is Let, then $\termone\,=\,\left(\letin{\varone}{\termtwo}{\termthree}\right)$
  and $\distrtypone\,=\,\sum_{\indexone\in\indexsetone}\ p_\indexone \cdot \distrtyptwo_\indexone$
  and $\contextsizedone\,=\,\contextsizedone_1,\,\contextsizedone_2,\,\contextsizedone_3$
  and $\contextdistrone\,=\,\contextdistrone_1,\,\sum_{\indexone \in \indexsetone}\ 
  \contextdistrone_{2,\indexone}$
  with $\contextsizedone_1,\,\contextsizedone_2 \contextsep \contextdistrone_1
  \proves \termtwo \typsep \distrelts{\typone_\indexone^{p_\indexone}\sep \indexone \in \indexsetone}$
  and, for every $\indexone \in \indexsetone$,
  $\contextsizedone_1,\,\contextsizedone_3 \contextsep \contextdistrone_{2,\indexone}
  \proves \termthree \typsep \distrtyptwo_\indexone$
  and $\underlying{\contextsizedone_1} = \Nat$.
  By repeated applications of the induction hypothesis,
  $\subst{\contextsizedone_1}{\sizevarone}{\sizeone},\,\subst{\contextsizedone_2}{\sizevarone}{\sizeone}
  \contextsep \subst{\contextdistrone_1}{\sizevarone}{\sizeone}
  \proves \termtwo \typsep 
  \subst{\distrelts{\typone_\indexone^{p_\indexone}\sep \indexone \in \indexsetone}}{\sizevarone}{\sizeone}$
  and, for every $\indexone \in \indexsetone$,
  $\subst{\contextsizedone_1}{\sizevarone}{\sizeone},\,\subst{\contextsizedone_3}{\sizevarone}{\sizeone}
  \contextsep \subst{\contextdistrone_{2,\indexone}}{\sizevarone}{\sizeone}
  \proves \termthree \typsep \subst{\distrtyptwo_\indexone}{\sizevarone}{\sizeone}$.
  We use in a first time the equality
  $\subst{\distrelts{\typone_\indexone^{p_\indexone}\sep \indexone \in \indexsetone}}{\sizevarone}{\sizeone}
  \ =\ \distrelts{\left(\subst{\typone_\indexone}{\sizevarone}{\sizeone}\right)^{p_\indexone}\sep \indexone \in \indexsetone}$
  coming from the definition of size substitutions.
  We conclude using the Let rule again and the equality
  $\subst{\left(\sum_{\indexone\in\indexsetone}\ p_\indexone \cdot \distrtyptwo_\indexone\right)}{\sizevarone}{\sizeone}\ =\ \sum_{\indexone\in\indexsetone}\ p_\indexone \cdot \subst{\distrtyptwo_\indexone}{\sizevarone}{\sizeone}$  
  from Lemma~\ref{lemma:commuting-size-substitutions-with-operations}.\\[0.2cm]
  \item If it is Case, then 
  $\termone\,=\,\casenat{\valone}{\valtwo}{\valthree}$
  and $\contextsizedone\,=\,\contextsizedone_1,\,\contextsizedone_2$
  with $\contextsizedone_1 \contextsep\emptyset \proves \valone \typsep \Nat^{\sizesucc{\sizetwo}}$
  and $\contextsizedone_2 \contextsep \contextdistrone \proves \valtwo \typsep 
  \Nat^{\sizetwo} \typarrow \distrtypone$
  and $\contextsizedone_2 \contextsep \contextdistrone \proves \valthree \typsep \distrtypone$.
  We apply induction hypothesis three times, and obtain
  $\subst{\contextsizedone_1}{\sizevarone}{\sizeone} \contextsep\emptyset 
  \proves \valone \typsep \subst{\left(\Nat^{\sizesucc{\sizetwo}}\right)}{\sizevarone}{\sizeone}$
  and $\subst{\contextsizedone_2}{\sizevarone}{\sizeone} \contextsep 
  \subst{\contextdistrone}{\sizevarone}{\sizeone} \proves \valtwo \typsep 
  \subst{\left(\Nat^{\sizetwo} \typarrow \distrtypone\right)}{\sizevarone}{\sizeone}$
  and $\subst{\contextsizedone_2}{\sizevarone}{\sizeone}
  \contextsep \subst{\contextdistrone}{\sizevarone}{\sizeone} \proves \valthree \typsep 
  \subst{\distrtypone}{\sizevarone}{\sizeone}$.
  We use the equalities 
  $\subst{\left(\Nat^{\sizesucc{\sizetwo}}\right)}{\sizevarone}{\sizeone}\ =\ 
  \Nat^{\sizesucc{\subst{\sizetwo}{\sizevarone}{\sizeone}}}$
  and
  $\subst{\left(\Nat^{\sizetwo} \typarrow \distrtypone\right)}{\sizevarone}{\sizeone}\ =\ 
  \Nat^{\subst{\sizetwo}{\sizevarone}{\sizeone}} \typarrow \subst{\distrtypone}{\sizevarone}{\sizeone}$
  and then the Case rule to conclude.\\[0.2cm]

  \item If it is $\letrecname$, we carefully adapt the proof scheme 
  of~\cite[Lemma 3.8]{barthe-et-al:type-based-termination}.
  We have $\termone\,=\,\letrec{\funcone}{\valone}$ and
  $\distrtypone\,=\,\Nat^{\sizetwo} \typarrow \subst{\distrtyptwo}{\sizevartwo}{\sizetwo}$
  and $\contextsizedone\,=\,\contextsizedone_1,\,\contextsizedone_2$
  with 
  \begin{varitemize}
   \item $\underlying{\contextsizedone_1}\,=\,\Nat$,
   \item $\sizevartwo \notin \contextsizedone_1 \text{ and } \sizevartwo \text{ positive in } \nu
   \text{ and } \forall \indextwo \in \indexsettwo,\ \spine{\sizetwo_\indextwo}\,=\,\sizevartwo$,
   \item $\distrelts{\left(\Nat^{\sizetwo_{\indextwo}} \typarrow 
   \subst{\distrtyptwo}{\sizevartwo}{\sizetwo_{\indextwo}}
   \right)^{p_{\indextwo}} \sep \indextwo \in \indexsettwo}$ induces an AST sized walk,
   \item and 
   \begin{equation}
   \label{eq:size-substitution-letrec1}
   \contextsizedone_1 \contextsep \funcone \typsep \distrelts{\left(\Nat^{\sizetwo_{\indextwo}} \typarrow
   \subst{\distrtyptwo}{\sizevartwo}{\sizetwo_{\indextwo}}
   \right)^{p_{\indextwo}} \sep \indextwo \in \indexsettwo} \proves \valone \typsep 
   \Nat^{\sizesucc{\sizevartwo}} \typarrow \subst{\distrtyptwo}{\sizevartwo}{\sizesucc{\sizevartwo}}
   \end{equation}
  \end{varitemize}
  We suppose, without loss of generality as this can be easily obtained by renaming
  $\sizevartwo$ to a fresh variable,
  that $\sizevarone \neq \sizevartwo$ and that $\sizevartwo \notin \sizeone$.
  Let $\sizevarthree$ be a fresh size variable; it follows in particular
  that $\sizevarthree \notin \contextsizedone_1,\,\contextsizedone_2,\,\distrtyptwo,\,\sizeone$.
  We apply the induction hypothesis to (\ref{eq:size-substitution-letrec1}) and obtain
  $$
  \subst{\contextsizedone_1}{\sizevartwo}{\sizevarthree} \contextsep \funcone \typsep \subst{\left(\distrelts{\left(\Nat^{\sizetwo_{\indextwo}} \typarrow
   \subst{\distrtyptwo}{\sizevartwo}{\sizetwo_{\indextwo}}
   \right)^{p_{\indextwo}} \sep \indextwo \in \indexsettwo}\right)}{\sizevartwo}{\sizevarthree}
   \proves \valone \typsep 
   \subst{\left(\Nat^{\sizesucc{\sizevartwo}} \typarrow 
   \subst{\distrtyptwo}{\sizevartwo}{\sizesucc{\sizevartwo}}\right)}{\sizevartwo}{\sizevarthree}
  $$
  which, after applying a series of equalities and using the fact that
  $\sizevartwo \notin \contextsizedone_1$, coincides with
  $$
  \contextsizedone_1 \contextsep 
  \funcone \typsep \left(\distrelts{\left(\Nat^{\subst{\sizetwo_{\indextwo}}{\sizevartwo}{\sizevarthree}} \typarrow
   \subst{\distrtyptwo}{\sizevartwo}{\subst{\sizetwo_{\indextwo}}{\sizevartwo}{\sizevarthree}}
   \right)^{p_{\indextwo}} \sep \indextwo \in \indexsettwo}\right)
   \proves \valone \typsep 
   \Nat^{\sizesucc{\sizevarthree}} \typarrow 
   \subst{\subst{\distrtyptwo}{\sizevartwo}{\sizesucc{\sizevartwo}}}{\sizevartwo}{\sizevarthree}
  $$
  but also with
  $$
  \contextsizedone_1 \contextsep 
  \funcone \typsep \left(\distrelts{\left(\Nat^{\subst{\sizetwo_{\indextwo}}{\sizevartwo}{\sizevarthree}} \typarrow
   \subst{\subst{\distrtyptwo}{\sizevartwo}{\sizevarthree}}{\sizevarthree}{\subst{\sizetwo_{\indextwo}}{\sizevartwo}{\sizevarthree}}
   \right)^{p_{\indextwo}} \sep \indextwo \in \indexsettwo}\right)
   \proves \valone \typsep 
   \Nat^{\sizesucc{\sizevarthree}} \typarrow 
   \subst{\subst{\distrtyptwo}{\sizevartwo}{\sizevarthree}}{\sizevarthree}{\sizesucc{\sizevarthree}}
  $$
  We can apply the induction hypothesis again, and obtain after rewriting
  $$
  \subst{\contextsizedone_1}{\sizevarone}{\sizeone} \contextsep 
  \funcone \typsep \left(\distrelts{\left(\Nat^{\subst{\sizetwo_{\indextwo}}{\sizevartwo}{\sizevarthree}} \typarrow
   \subst{\subst{\subst{\distrtyptwo}{\sizevartwo}{\sizevarthree}}{\sizevarthree}{\subst{\sizetwo_{\indextwo}}{\sizevartwo}{\sizevarthree}}}{\sizevarone}{\sizeone}
   \right)^{p_{\indextwo}} \sep \indextwo \in \indexsettwo}\right)
   \proves \valone \typsep 
   \Nat^{\sizesucc{\sizevarthree}} \typarrow 
   \subst{\subst{\subst{\distrtyptwo}{\sizevartwo}{\sizevarthree}}{\sizevarthree}{\sizesucc{\sizevarthree}}}{\sizevarone}{\sizeone}
  $$
  where we used the fact that 
  $\forall \indextwo \in \indexsettwo,\ \spine{\sizetwo_\indextwo}\,=\,\sizevartwo \neq \sizevarone$
  so that $\subst{\left(\Nat^{\subst{\sizetwo_{\indextwo}}{\sizevartwo}{\sizevarthree}}\right)}{\sizevarone}{\sizeone}
  \,=\,\Nat^{\subst{\sizetwo_{\indextwo}}{\sizevartwo}{\sizevarthree}}$.
  Since $\sizevarthree \notin \sizeone$, we can exchange
  $\subst{}{\sizevarthree}{\sizesucc{\sizevarthree}}$
  and $\subst{}{\sizevarone}{\sizeone}$.
  For every $\indextwo \in \indexsettwo$, we can also exchange
  $\subst{}{\sizevarone}{\sizeone}$ and 
  $\subst{}{\sizevarthree}{\subst{\sizetwo_{\indextwo}}{\sizevartwo}{\sizevarthree}}$
  since $\spine{\subst{\sizetwo_\indextwo}{\sizevartwo}{\sizevarthree}}\,=\,\sizevarthree \neq \sizevarone$
  and $\sizevarthree \notin \sizeone$. We obtain:
  $$
  \subst{\contextsizedone_1}{\sizevarone}{\sizeone} \contextsep 
  \funcone \typsep \distrelts{\left(\Nat^{\subst{\sizetwo_{\indextwo}}{\sizevartwo}{\sizevarthree}} \typarrow
   \subst{\subst{\subst{\distrtyptwo}{\sizevartwo}{\sizevarthree}}{\sizevarone}{\sizeone}}{\sizevarthree}{\subst{\sizetwo_{\indextwo}}{\sizevartwo}{\sizevarthree}}
   \right)^{p_{\indextwo}} \sep \indextwo \in \indexsettwo}
   \proves \valone \typsep 
   \Nat^{\sizesucc{\sizevarthree}} \typarrow 
   \subst{\subst{\subst{\distrtyptwo}{\sizevartwo}{\sizevarthree}}{\sizevarone}{\sizeone}}{\sizevarthree}{\sizesucc{\sizevarthree}}
  $$
  %
  Additionally, we have:
  \begin{varitemize}
   \item $\underlying{\subst{\contextsizedone_1}{\sizevarone}{\sizeone}}\,=\,\Nat$,
   \item $\sizevarthree \notin \subst{\contextsizedone_1}{\sizevarone}{\sizeone}$,
   \item $\sizevarthree \text{ positive in } \subst{\subst{\nu}{\sizevartwo}{\sizevarthree}}{\sizevarone}{\sizeone}$
   since $\sizevartwo$ was positive in $\distrtyptwo$,
   \item $\forall \indextwo \in \indexsettwo,\ \spine{\subst{\sizetwo_\indextwo}{\sizevartwo}{\sizevarthree}}\,=\,\sizevarthree$
   since $\spine{\sizetwo_\indextwo}\,=\,\sizevartwo$,
   \item and 
   $\distrelts{\left(\Nat^{\subst{\sizetwo_{\indextwo}}{\sizevartwo}{\sizevarthree}} \typarrow
   \subst{\subst{\subst{\distrtyptwo}{\sizevartwo}{\sizevarthree}}{\sizevarone}{\sizeone}}{\sizevarthree}{\subst{\sizetwo_{\indextwo}}{\sizevartwo}{\sizevarthree}}
   \right)^{p_{\indextwo}} \sep \indextwo \in \indexsettwo}$
   induces the same sized walk, which is thus AST, as 
   $\distrelts{\left(\Nat^{\sizetwo_{\indextwo}} \typarrow 
   \subst{\distrtyptwo}{\sizevartwo}{\sizetwo_{\indextwo}}
   \right)^{p_{\indextwo}} \sep \indextwo \in \indexsettwo}$. Indeed, only the spine
   variable changes under the substitution $\subst{}{\sizevartwo}{\sizevarthree}$.
  \end{varitemize}
  Let $\sizethree\,=\,\subst{\sizetwo}{\sizevarone}{\sizeone}$.
  Since all these conditions are met, we can apply the $\letrecname$ rule
  and obtain
  $$
  \subst{\contextsizedone_1}{\sizevarone}{\sizeone},\,\subst{\contextsizedone_2}{\sizevarone}{\sizeone} \contextsep 
  \subst{\contextdistrone}{\sizevarone}{\sizeone}
  \proves \letrec{\funcone}{\valone} \typsep 
   \Nat^{\sizethree} \typarrow 
   \subst{\subst{\subst{\distrtyptwo}{\sizevartwo}{\sizevarthree}}{\sizevarone}{\sizeone}}{\sizevarthree}{\sizethree}
  $$  
  Since $\sizevarone,\,\sizevarthree \notin \sizeone$ and $\sizevarthree \notin \distrtyptwo$,
  we can commute $\subst{}{\sizevarone}{\sizeone}$ and $\subst{}{\sizevarthree}{\sizethree}$
  and compose substitutions to obtain
  $$
  \subst{\contextsizedone}{\sizevarone}{\sizeone} \contextsep 
  \subst{\contextdistrone}{\sizevarone}{\sizeone}
  \proves \letrec{\funcone}{\valone} \typsep 
   \Nat^{\sizethree} \typarrow 
   \subst{\subst{\distrtyptwo}{\sizevartwo}{\sizethree}}{\sizevarone}{\sizeone}
  $$  
  which rewrites to
  $$
  \subst{\contextsizedone}{\sizevarone}{\sizeone} \contextsep 
  \subst{\contextdistrone}{\sizevarone}{\sizeone}
  \proves \letrec{\funcone}{\valone} \typsep 
  \subst{\left(\Nat^{\sizetwo} \typarrow \subst{\distrtyptwo}{\sizevartwo}{\sizetwo}\right)}{\sizevarone}{\sizeone}
  $$
  which allows to conclude.
 \end{varitemize}

\end{proof}
}

\longv{
\subsection{Subject Reduction}}

\noindent
We can now state the main lemma of subject reduction:
\begin{lemma}[Subject Reduction, Fundamental Lemma]
\label{lemma/subject-reduction-one-step}
 Let $\termone \in \setdistrtypedclosedterms{\distrtypone}$ and 
 $\distrone$
 be the unique closed term distribution such that
 $\termone \rcbv \distrone$. Then
 there exists a closed typed distribution $\distrelts{\left(\termthree_{\indextwo} \typsep \distrtyptwo_{\indextwo}\right)^{p_{\indextwo}}
 \sep \indextwo \in \indexsettwo}$
 such that 
 \begin{varitemize}
  \item $\expectype{\left(\termthree_{\indextwo} \typsep \distrtyptwo_{\indextwo}\right)^{p_{\indextwo}}}\ =\ \mu$,
  \item $\pseudorep{\left(\termthree_{\indextwo}\right)^{p_{\indextwo}}\sep \indextwo \in \indexsettwo}$ is a pseudo-representation of
  $\distrone$.
  \end{varitemize}
 Note that the condition on expectations implies that $\bigcup_{\indextwo \in \indexsettwo}\ \supp{\distrtyptwo_{\indextwo}}\ =\ \supp{\distrtypone}$.
\end{lemma}

\longv{
\begin{proof} We proceed by induction on $\termone$.

\begin{itemize}
\item Suppose that $\termone\,=\,\letin{\varone}{\valone}{\termtwo}$,
that $\distrone\,=\,\distrelts{\left(\subst{\termtwo}{\varone}{\valone}\right)^1}$,
and that $\emptyset \contextsep \emptyset \proves \letin{\varone}{\valone}{\termtwo} \typsep \distrtypone$.
By Lemma~\ref{lemma/generation-typing}, there exists $\left(\distrtypthree,\typone\right)$ such that 
$\emptyset \contextsep \emptyset \proves \valone \typsep \typone$ and $\varone \typsep \typone \contextsep \emptyset \proves \termtwo \typsep \distrtypthree$
with $\distrtypthree \subtypeleq \distrtypone$.
By Lemma~\ref{lemma/value-substitution}, $\emptyset \contextsep \emptyset \proves \subst{\termtwo}{\varone}{\valone} \typsep \distrtypthree$, and since
$\distrtypthree \subtypeleq \distrtypone$ we obtain by subtyping that $\emptyset \contextsep \emptyset \proves \subst{\termtwo}{\varone}{\valone} \typsep \distrtypone$.
It follows that $\distrelts{\left(\subst{\termtwo}{\varone}{\valone} \typsep \distrtypone\right)^1}$ is a closed typed distribution
satisfying the requirements of the lemma.

\item Suppose that $\termone \,=\, \left(\abstr{\varone}{\termtwo}\right)\ \valone$,
that $\distrone\,=\,\distrelts{\left(\subst{\termtwo}{\varone}{\valone}\right)^1}$
and that $\emptyset \contextsep \emptyset \proves \left(\abstr{\varone}{\termtwo}\right)\ \valone\typsep \distrtypone$.
Applying Lemma~\ref{lemma/generation-typing} twice, we obtain that
$\varone \typsep \typtwo \contextsep \emptyset \proves \termtwo \typsep \distrtypthree$ and
$\emptyset \contextsep \emptyset \proves \valone \typsep \typone$ with $\typone \subtypeleq \typtwo$ and $\distrtypthree \subtypeleq \distrtypone$.
Applying subtyping to the second judgement gives $\emptyset \contextsep \emptyset \proves \valone \typsep \typtwo$,
and we can apply Lemma~\ref{lemma/value-substitution}
to obtain $\emptyset \contextsep \emptyset \proves \subst{\termtwo}{\varone}{\valone} \typsep \distrtypthree$. Since
$\distrtypthree \subtypeleq \distrtypone$ we obtain by weakening that 
$\emptyset \contextsep \emptyset\proves \subst{\termtwo}{\varone}{\valone} \typsep \distrtypone$.
It follows that $\distrelts{\left(\subst{\termtwo}{\varone}{\valone} \typsep \distrtypone\right)^1}$ is a closed typed distribution
satisfying the requirements of the lemma.

\item Suppose that $\termone \,=\, \termtwo\ \choice_p\ \termthree$,
that $\distrone\,=\,\pseudorep{\termtwo^p,\,\termthree^{1-p}}$
and that $\emptyset \contextsep \emptyset \proves \termtwo\ \choice_p\ \termthree\typsep \distrtypone$.
By Lemma~\ref{lemma/generation-typing}, there exists $\left(\distrtypthree,\distrtypfour\right)$
such that $\emptyset\contextsep \emptyset \proves \termtwo \typsep \distrtypthree$ and 
$\emptyset\contextsep \emptyset \proves \termthree \typsep \distrtypfour$
with $\distrtypthree \choice_p \distrtypfour \subtypeleq \distrtypone$
and $\distrsum{(\distrtypthree \choice_p \distrtypfour)}\,=\,1$.
By Lemma~\ref{lemma/subtyping-prob-sum}, there exists $\left(\distrtypthree',\distrtypfour'\right)$
such that $\distrtypone\,=\,\distrtypthree' \choice_p \distrtypfour'$, $\distrtypthree \subtypeleq \distrtypthree'$
and $\distrtypfour \subtypeleq \distrtypfour'$. By subtyping, $\emptyset \contextsep \emptyset\proves \termtwo \typsep \distrtypthree'$
and $\emptyset \contextsep \emptyset\proves \termthree \typsep \distrtypfour'$. 
We consider the closed typed distribution of pseudo-representation $\pseudorep{\left(\termtwo \typsep \distrtypthree'\right)^p,\,
\left(\termthree \typsep \distrtypfour'\right)^{1-p}}$
which satisfies the requirements of the lemma since its expectation type is 
$p \cdot \distrtypthree' + (1-p) \cdot \distrtypfour' \,=\,\distrtypthree' \choice_p \distrtypfour'\,=\,\distrtypone$.
Note that we use a pseudo-representation to cope with the very specific case in which $\termtwo\,=\,\termthree$ and $\distrtypthree'\,=\,\distrtypfour'$,
in which the representation of the closed typed distribution is $\distrelts{\left(\termtwo \typsep \distrtypthree'\right)^1}$.

\item Suppose that $\termone \,=\,\letin{\varone}{\termtwo}{\termthree}$, that
 $\distrone\,=\,\distrelts{\left(\letin{\varone}{\termfour_{\indextwo}}{\termthree}\right)^{p'_{\indextwo}} \sep \indextwo \in \indexsettwo}$
 and that $\emptyset\contextsep \emptyset \proves \letin{\varone}{\termtwo}{\termthree} \typsep \distrtypone$.
 By Lemma~\ref{lemma/generation-typing}, there exists 
 $\left(\indexsetone,\left(\typone_{\indexone}\right)_{\indexone \in \indexsetone},
   \left(p_{\indexone}\right)_{\indexone \in \indexsetone}, \left(\distrtypthree_{\indexone}\right)_{\indexone \in \indexsetone}\right)$ such that
   \begin{itemize}
    \item $\sum_{\indexone \in \indexsetone}\ p_{\indexone} \cdot \distrtypthree_{\indexone} \subtypeleq \distrtypone$,
    \item $\emptyset\contextsep \emptyset \proves \termtwo \typsep \distrelts{\typone_{\indexone}^{p_{\indexone}} \sep \indexone \in \indexsetone}$,
    \item $\forall \indexone \in \indexsetone,\ \ \varone\typsep \typone_{\indexone}\contextsep \emptyset
    \proves \termthree \typsep \distrtypthree_{\indexone}$.
   \end{itemize}

    This reduction comes, by definition of $\rcbv$, from 
 $\termtwo\, \rcbv\, \distrelts{\termfour_{\indextwo}^{p'_{\indextwo}} \sep \indextwo \in \indexsettwo}$,
 to which we can apply the induction hypothesis: there exists a closed typed distribution 
 $$
 \distrelts{\left(\termfive_{\indexthree} \typsep \distrtypfour_{\indexthree} \right)^{p''_{\indexthree}} \sep \indexthree \in \indexsetthree}
 $$
 which is such that
 $$
 \distrelts{\typone_{\indexone}^{p_{\indexone}} \sep \indexone \in \indexsetone}\ \ =\ \ \sum_{\indexthree \in \indexsetthree}\ 
 p''_{\indexthree} \cdot \distrtypfour_{\indexthree}
 $$
 and that $\pseudorep{\left(\termfive_{\indexthree}  \right)^{p''_{\indexthree}} \sep \indexthree \in \indexsetthree}$
 is a pseudo-representation of $\distrelts{\termfour_{\indextwo}^{p'_{\indextwo}} \sep \indextwo \in \indexsettwo}$.
 It follows that, for every $\indexthree \in \indexsetthree$, we can write
 $\distrtypfour_{\indexthree}$ as the pseudo-representation 
 $\pseudorep{\typone_{\indexone}^{p'''_{\indexthree\indexone}} \sep \indexone \in \indexsetone}$
 where some of the $p'''_{\indexthree\indexone}$ (but not all of them) may be worth $0$.
 This implies that, for all $\indexone \in \indexsetone$,
 $$
 p_{\indexone}\ =\ \sum_{\indexthree \in \indexsetthree}\ p''_{\indexthree} p'''_{\indexthree \indexone}
 $$
 Now, for every $\indexthree \in \indexsetthree$, we can derive 
 $\emptyset\contextsep \emptyset \proves \letin{\varone}{\termfive_{\indexthree}}{\termthree} \typsep 
 \sum_{\indexone \in \indexsetone} p'''_{\indexthree \indexone} \cdot \distrtypthree_{\indexone}$
 from the rule
 $$
 \AxiomC{$\emptyset\contextsep \emptyset \proves
 \termfive_{\indexthree} \typsep \distrelts{\typone_{\indexone}^{p'''_{\indexthree\indexone}} \sep \indexone \in \indexsetone}$}
 \AxiomC{$\varone\typsep \typone_{\indexone}\contextsep \emptyset
 \proves \termthree \typsep \distrtypthree_{\indexone}\ \ \ \left(\forall \indexone \in \indexsetone\right)$}
 \BinaryInfC{$ \emptyset\contextsep \emptyset \proves \letin{\varone}{\termfive_{\indexthree}}{\termthree} \typsep 
 \sum_{\indexone \in \indexsetone} p'''_{\indexthree \indexone} \cdot \distrtypthree_{\indexone}$}
 \DisplayProof
 $$
 so that 
 $\pseudorep{\left(\letin{\varone}{\termfive_{\indexthree}}{\termthree} \typsep 
 \sum_{\indexone \in \indexsetone} p'''_{\indexthree \indexone} \cdot \distrtypthree_{\indexone}\right)^{p''_{\indexthree}}
 \sep \indexthree \in \indexsetthree}$ 
 is a pseudo-representation of a closed typed distribution, whose expectation is
 $$
 \sum_{\indexthree \in \indexsetthree}\ p''_{\indexthree}\ \sum_{\indexone \in \indexsetone} p'''_{\indexthree \indexone} \cdot \distrtypthree_{\indexone}
 \ \ =\ \ 
 \sum_{\indexone \in \indexsetone}\ \left(\sum_{\indexthree \in \indexsetthree}\ p''_{\indexthree}  p'''_{\indexthree \indexone}\right) \cdot \distrtypthree_{\indexone}
  \ \ =\ \ \sum_{\indexone \in \indexsetone}\ p_{\indexone} \cdot \distrtypthree_{\indexone}
 $$
 By Lemma~\ref{lemma/generation-typing}, the sum of $\sum_{\indexone \in \indexsetone}\ p_{\indexone} \cdot \distrtypthree_{\indexone}$ is $1$,
 and it follows that $\distrsum{\distrtypone}\,=\,1$ as well.
 Since $\sum_{\indexone \in \indexsetone}\ p_{\indexone} \cdot \distrtypthree_{\indexone} \subtypeleq \distrtypone$,
 applying Corollary~\ref{corollary/subtyping-probabilistic-sums}
 gives us a family $\left(\distrtyptwo_{\indexone}\right)_{\indexone \in \indexsetone}$
 of distribution types such that, by subtyping, we can derive for every $\indexthree \in \indexsetthree$
 the judgement 
  $\emptyset\contextsep \emptyset \proves \letin{\varone}{\termfive_{\indexthree}}{\termthree} \typsep 
 \sum_{\indexone \in \indexsetone} p'''_{\indexthree \indexone} \cdot \distrtyptwo_{\indexone}$.
 This family $\vec{\distrtyptwo}$ satisfies moreover $\sum_{\indexone \in \indexsetone}\ p_{\indexone} \cdot \distrtyptwo_{\indexone}\ =\ \distrtypone$.
 We therefore consider the closed typed distribution of pseudo-representation
 $$
 \pseudorep{\left(\letin{\varone}{\termfive_{\indexthree}}{\termthree} \typsep 
 \sum_{\indexone \in \indexsetone} p'''_{\indexthree \indexone} \cdot \distrtyptwo_{\indexone}\right)^{p''_{\indexthree}} \sep \indexthree
 \in \indexsetthree}
 $$
 and of expectation type
 $$
 \sum_{\indexthree \in \indexsetthree}\ p''_{\indexthree} \sum_{\indexone \in \indexsetone} p'''_{\indexthree \indexone} \cdot \distrtyptwo_{\indexone}
 \ \ =\ \ \sum_{\indexone \in \indexsetone}\ p_{\indexone} \cdot \distrtyptwo_{\indexone}\ =\ \distrtypone
 $$
 Since $\pseudorep{\left(\termfive_{\indexthree} \typsep \distrtypfour_{\indexthree} \right)^{p''_{\indexthree}} \sep \indexthree \in \indexsetthree}$
 is a pseudo-representation of $\distrelts{\termfour_{\indextwo}^{p'_{\indextwo}} \sep \indextwo \in \indexsettwo}$, we have that
 $\pseudorep{\left(\letin{\varone}{\termfive_{\indexthree}}{\termthree} \typsep \distrtypfour_{\indexthree} \right)^{p''_{\indexthree}} \sep \indexthree \in \indexsetthree}$
 is a pseudo-representation of
 $\distrelts{\left(\letin{\varone}{\termfour_{\indextwo}}{\termthree}\right)^{p'_{\indextwo}} \sep \indextwo \in \indexsettwo}$
 which allows us to conclude.
 
\item Suppose that $\termone \,=\,\caseof{\natsucc \ \valone}{\natsucc \rightarrow \valtwo \smallsep
 \natzero \rightarrow \valthree}$, that
 $\distrone\,=\,\distrelts{\left(\valtwo\ \valone\right)^1}$ and that
 $\emptyset \contextsep \emptyset\proves \caseof{\natsucc \ \valone}{\natsucc \rightarrow \valtwo \smallsep
 \natzero \rightarrow \valthree} \typsep \distrtypone$.
 By Lemma~\ref{lemma/generation-typing}, there exists $\sizeone$ and $\distrtypthree$ such that
    $\emptyset\contextsep \emptyset \proves \natsucc\ \valone \typsep \Nat^{\sizesucc{\sizeone}}$ and 
   $\emptyset\contextsep \emptyset \proves \valtwo \typsep \Nat^{\sizeone} \typarrow \distrtypthree$
   with $\distrtypthree \subtypeleq \distrtypone$.
   Lemma~\ref{lemma:size-constructors-properties} implies that 
   $\emptyset\contextsep \emptyset \proves \valone \typsep \Nat^{\sizeone}$. Using an Application rule, we obtain that
   $\emptyset\contextsep \emptyset \proves \valtwo \ \valone \typsep \distrtypthree$, and subtyping gives 
   $\emptyset\contextsep \emptyset \proves \valtwo \ \valone \typsep \distrtypone$, allowing us to conclude for 
   $\distrelts{\left(\valtwo\ \valone \typsep \distrtypone\right)^1}$.

\item Suppose that $\termone \,=\,\caseof{\natzero}{\natsucc \rightarrow \valtwo \smallsep
 \natzero \rightarrow \valthree}$, that
 $\distrone\,=\,\distrelts{\left(\valthree\right)^1}$ and that
 $\emptyset\contextsep \emptyset \proves \caseof{\natzero}{\natsucc \rightarrow \valtwo \smallsep
 \natzero \rightarrow \valthree} \typsep \distrtypone$.
 By Lemma~\ref{lemma/generation-typing}, there exists $\distrtypthree$ with $\distrtypthree \subtypeleq \distrtypone$ and such that
 $\emptyset\contextsep \emptyset \proves \valthree \typsep \distrtypthree$. By subtyping, 
 $\emptyset\contextsep \emptyset \proves \valthree \typsep \distrtypone$ which allows to conclude for 
 $\distrelts{\left(\valthree \typsep \distrtypone\right)^1}$.

\item Suppose that $\termone\,=\,\left(\letrec{\funcone}{\valone}\right)\ \left(\dataconstrone\ \termvector{\valtwo}\right)$, that
 $\distrone\,=\, \distrelts{\left(\subst{\valone}{\funcone}{\left(\letrec{\funcone}{\valone}\right)}\ \left(\dataconstrone\
 \termvector{\valtwo}\right)\right)^1}$ and that 
 $\emptyset\contextsep \emptyset \proves \left(\letrec{\funcone}{\valone}\right)\ \left(\dataconstrone\ \termvector{\valtwo}\right) \typsep \distrtypone$.
 We apply again Lemma~\ref{lemma/generation-typing}, but this time we rather depict the derivation typing $\termone$ with $\distrtypone$
 it induces, for the sake of clarity. This derivation is of the form (modulo composition of subtyping rules):
 \begin{adjustwidth}{-2cm}{-2cm}
 $$
 \AxiomC{$\pi_1$}
 \noLine
 \UnaryInfC{$\vdots$}
 \noLine
 \UnaryInfC{$\mathit{Hyp}$}
 \noLine
 \UnaryInfC{$\emptyset\contextsep \funcone \typsep \distrelts{\left(\Nat^{\sizefour_{\indextwo}} \typarrow \subst{\distrtypthree}{\sizevarone}{\sizefour_{\indextwo}}
\right)^{p_{\indextwo}} \sep \indextwo \in \indexsettwo} \proves \valone \typsep 
\Nat^{ \sizesucc{\sizevarone}} \typarrow \subst{\distrtypthree}{\sizevarone}{\sizesucc{\sizevarone}}$}
 \UnaryInfC{$\emptyset\contextsep \emptyset \proves \letrec{\funcone}{\valone} \typsep \Nat^{\sizethree} 
 \typarrow \subst{\distrtypthree}{\sizevarone}{\sizethree}$}
 \UnaryInfC{$\emptyset\contextsep \emptyset \proves \letrec{\funcone}{\valone} \typsep \Nat^{\sizesucc{\sizeone}} \typarrow \distrtypone$}
 \AxiomC{$\pi_2$}
 \noLine
 \UnaryInfC{$\vdots$}
 \noLine
 \UnaryInfC{$\emptyset\contextsep \emptyset \proves \dataconstrone\ \termvector{\valtwo} \typsep \Nat^{\sizesucc{\sizetwo}}$}
 \UnaryInfC{$\emptyset\contextsep \emptyset \proves \dataconstrone\ \termvector{\valtwo} \typsep \Nat^{\sizesucc{\sizeone}}$}
 \BinaryInfC{$\emptyset\contextsep \emptyset \proves \left(\letrec{\funcone}{\valone}\right)\
 \left(\dataconstrone\ \termvector{\valtwo}\right) \typsep \distrtypone$}
 \DisplayProof
 $$
 \end{adjustwidth}
 where the two sizes appearing in the types for $\dataconstrone\ \termvector{\valtwo}$ are successors due to Lemma~\ref{lemma:size-constructors-properties},
 and where
 \begin{itemize}
  \item $\mathit{Hyp}$ denotes the additional premises of the $\letrecname$ rule, and contains notably $\positive{\sizevarone}{\distrtypthree}$,
  \item $\sizetwo \sizeleq \sizesucc{\sizetwo} \sizeleq \sizesucc{\sizeone} \sizeleq \sizethree$,
  \item $\subst{\distrtypthree}{\sizevarone}{\sizethree} \subtypeleq \distrtypone$.
 \end{itemize}
 It follows that, for every $\indextwo \in \indexsettwo$, we can deduce that the closed value $\letrec{\funcone}{\valone}$
 has type $\Nat^{\sizefour_{\indextwo}} \typarrow \subst{\distrtypthree}{\sizevarone}{\sizefour_{\indextwo}}$, as proved by the derivation
 $$
 \AxiomC{$\pi_1$}
 \noLine
 \UnaryInfC{$\vdots$}
 \noLine
 \UnaryInfC{$\mathit{Hyp}$}
 \noLine
 \UnaryInfC{$ \emptyset\contextsep\funcone \typsep \distrelts{\left(\Nat^{\sizefour_{\indextwo}} \typarrow \subst{\distrtypthree}{\sizevarone}{\sizefour_{\indextwo}}\right)^{p_{\indextwo}} \sep \indextwo \in \indexsettwo} \proves \valone \typsep 
\Nat^{ \sizesucc{\sizevarone}} \typarrow \subst{\distrtypthree}{\sizevarone}{\sizesucc{\sizevarone}}$}
 \UnaryInfC{$\emptyset\contextsep \emptyset \proves \letrec{\funcone}{\valone} \typsep \Nat^{\sizefour_{\indextwo}} \typarrow 
 \subst{\distrtypthree}{\sizevarone}{\sizefour_{\indextwo}}$}
 \DisplayProof
 $$
 Since 
 $$
 \emptyset\contextsep\funcone \typsep \distrelts{\left(\Nat^{\sizefour_{\indextwo}} \typarrow \subst{\distrtypthree}{\sizevarone}{\sizefour_{\indextwo}}
\right)^{p_{\indextwo}} \sep \indextwo \in \indexsettwo} \proves \valone \typsep 
\Nat^{ \sizesucc{\sizevarone}} \typarrow \subst{\distrtypthree}{\sizevarone}{ \sizesucc{\sizevarone}}
 $$
 we obtain by Lemma~\ref{lemma/distrib-substitution} that 
 $$
 \emptyset\contextsep \emptyset \proves \subst{\valone}{\funcone}{\left(\letrec{\funcone}{\valone}\right)} \typsep 
 \Nat^{ \sizesucc{\sizevarone}} \typarrow \subst{\distrtypthree}{\sizevarone}{ \sizesucc{\sizevarone}}
 $$
 We now apply Lemma~\ref{lemma/size-substitution} to
  $ \emptyset \contextsep \emptyset\proves \subst{\valone}{\funcone}{\left(\letrec{\funcone}{\valone}\right)} \typsep 
 \Nat^{ \sizesucc{\sizevarone}} \typarrow \subst{\distrtypthree}{\sizevarone}{ \sizesucc{\sizevarone}}$
 with the substitution $\subst{}{\sizevarone}{\sizetwo}$
 and we obtain that 
 $ \emptyset\contextsep \emptyset \proves \subst{\valone}{\valone}{\left(\letrec{\funcone}\right)}{\funcone} \typsep 
 \Nat^{ \sizesucc{\sizetwo}} \typarrow \subst{\distrtypthree}{\sizevarone}{ \sizesucc{\sizetwo}}$.
 Using the Application rule, we derive $\emptyset \contextsep \emptyset\proves \subst{\valone}{\funcone}{\left(\letrec{\funcone}{\valone}\right)}
 \ \left(\dataconstrone\ \termvector{\valtwo}\right) \typsep \subst{\distrtypthree}{\sizevarone}{ \sizesucc{\sizetwo}}$.
 Since $\positive{\sizevarone}{\distrtypthree}$ and $\sizesucc{\sizetwo} \sizeleq \sizethree$, by
 Lemma~\ref{lemma:size-substitution-subtyping}, we get that 
  $\subst{\distrtypthree}{\sizevarone}{\sizesucc{\sizetwo}} \subtypeleq \subst{\distrtypthree}{\sizevarone}{\sizethree}$.
 By transitivity of $\subtypeleq$, $\subst{\distrtypthree}{\sizevarone}{\sizesucc{\sizetwo}} \subtypeleq \distrtypone$
 which allows us to conclude by subtyping that 
  $\emptyset\contextsep \emptyset \proves \subst{\valone}{\funcone}{\left(\letrec{\funcone}{\valone}\right)}\ \left(\dataconstrone\
 \termvector{\valtwo}\right) \typsep \distrtypone$.
  The result follows, for $\distrelts{\left(\subst{\valone}{\funcone}{\left(\letrec{\funcone}{\valone}\right)}\ \left(\dataconstrone\
 \termvector{\valtwo}\right) \typsep \distrtypone\right)^1}$.

 %

\end{itemize}
\end{proof}
}

\longv{
\begin{theorem}[Subject Reduction for $\rcbv$]
\label{theorem/subject-reduction}
 Let $n \in \NN$, and $\distrelts{\left(\termone_{\indexone} \typsep \distrtypone_{\indexone}\right)^{p_{\indexone}} \sep \indexone \in \indexsetone}$
 be a closed typed distribution. Suppose that 
 $ \distrelts{\left(\termone_{\indexone} \right)^{p_{\indexone}} \sep \indexone \in \indexsetone}\ \ \rcbv^n\ \ 
 \distrelts{\left(\termtwo_{\indextwo} \right)^{p'_{\indextwo}} \sep \indextwo \in \indexsettwo}
 $
 then there exists a closed typed distribution $\distrelts{\left(\termthree_{\indexthree} \typsep \distrtyptwo_{\indexthree} \right)^{p''_{\indexthree}} \sep \indexthree \in \indexsetthree}$
 such that
 \begin{itemize}
  \item $\expectype{\left(\termone_{\indexone} \typsep \distrtypone_{\indexone}\right)^{p_{\indexone}}}\ \ =\ \ 
  \expectype{\left(\termthree_{\indexthree}\typsep \distrtyptwo_{\indexthree} \right)^{p''_{\indexthree}}}$,
  \item and that $\pseudorep{\left(\termthree_{\indexthree} \right)^{p''_{\indexthree}} \sep \indexthree \in \indexsetthree}$ is a pseudo-representation of
  $\distrelts{\left(\termtwo_{\indextwo} \right)^{p'_{\indextwo}} \sep \indextwo \in \indexsettwo}$.
 \end{itemize}
\end{theorem}
}

\longv{ 
\begin{proof} 
 The proof is by induction on $n$. For $n = 0$, $\rcbv^0$ is the identity relation
 and the result is immediate.
 For $n+1$, we have 
 $$
 \distrelts{\left(\termone_{\indexone} \right)^{p_{\indexone}} \sep \indexone \in \indexsetone}\ \ \rcbv^n\ \ 
 \distrelts{\left(\termfour_{\indexfour} \right)^{p''_{\indexfour}} \sep \indexfour \in \indexsetfour}
 \ \ \rcbv\ \ \distrelts{\left(\termtwo_{\indextwo} \right)^{p'_{\indextwo}} \sep \indextwo \in \indexsettwo}
 $$
 We apply the induction hypothesis and obtain a closed typed distribution
 $\distrelts{\left(\termfive_{\indexfive} \typsep \distrtypthree_{\indexfive} \right)^{p^{(3)}_{\indexfive}} \sep \indexfive \in \indexsetfive}$
 satisfying $\expectype{\left(\termone_{\indexone} \typsep \distrtypone_{\indexone}\right)^{p_{\indexone}}}\ =\ 
 \expectype{\left(\termfive_{\indexfive} \typsep \distrtypthree_{\indexfive} \right)^{p^{(3)}_{\indexfive}}}$ and 
 such that $\pseudorep{\left(\termfive_{\indexfive}\right)^{p^{(3)}_{\indexfive}} \sep \indexfive \in \indexsetfive}$
 is a pseudo-representation of 
 $\distrelts{\left(\termfour_{\indexfour} \right)^{p''_{\indexfour}} \sep \indexfour \in \indexsetfour}$.
 For every $\indexfive \in \indexsetfive$:
 \begin{varitemize}
  \item if $\termfive_\indexfive$ is a value, we set
  $\distrone_\indexfive \,=\,\distrelts{\termfive^1}$ and 
  $\closedtypeddistrone_\indexfive$ to be the closed 
  typed distribution $\closedtypeddistrone_\indexfive\,=\,
  \distrelts{\left(\termsix_\indexsix \typsep \distrtypfour_\indexsix\right)^{p^{(4)}_\indexsix} \sep
  \indexsix \in \indexsetsix_\indexfive}\,=\,
  \left(\termfive_\indexfive\typsep\distrtypthree_\indexfive\right)^1$,
  \item else $\termfive_\indexfive \rcbv \distrone_\indexfive$.
  We apply Lemma~\ref{lemma/subject-reduction-one-step} and obtain a closed 
  typed distribution $\closedtypeddistrone_\indexfive\,=\,
  \distrelts{\left(\termsix_\indexsix \typsep \distrtypfour_\indexsix\right)^{p^{(4)}_\indexsix} \sep
  \indexsix \in \indexsetsix_\indexfive}$ such that
  $\expectype{\left(\termsix_\indexsix \typsep \distrtypfour_\indexsix\right)^{p^{(4)}_\indexsix}}\,=\,
  \distrtypthree_\indexfive$
  and that $\pseudorep{\left(\termsix_\indexsix\right)^{p^{(4)}_\indexsix} \sep
  \indexsix \in \indexsetsix_\indexfive}$ is a pseudo-representation of $\distrone_\indexfive$.
 \end{varitemize}
 We claim that the closed typed distribution defined as
 $$
 \distrelts{\left(\termthree_{\indexthree} \typsep \distrtyptwo_{\indexthree} \right)^{p''_{\indexthree}} \sep \indexthree \in \indexsetthree}
 \ \ =\ \
 \sum_{\indexfive\in\indexsetfive}\ p^{(3)}_\indexfive \cdot \closedtypeddistrone_\indexfive
 $$
 satisfies the required conditions. Indeed, the expectation type is preserved:
 $$
 \begin{array}{rcl}
 \expectype{\left(\termone_{\indexone} \typsep \distrtypone_{\indexone}\right)^{p_{\indexone}}}
 &\ \  =\ \ & 
 \expectype{\left(\termfive_{\indexfive} \typsep \distrtypthree_{\indexfive} \right)^{p^{(3)}_{\indexfive}}}\\
 &\ \  =\ \ & 
 \sum_{\indexfive\in\indexsetfive}\ p^{(3)}_{\indexfive} \cdot \distrtypthree_{\indexfive}\\
 &\ \  =\ \ & 
 \sum_{\indexfive\in\indexsetfive}\ p^{(3)}_{\indexfive} \cdot \expectype{\left(\termsix_\indexsix \typsep \distrtypfour_\indexsix\right)^{p^{(4)}_\indexsix}}\\
 &\ \  =\ \ &
 \expectype{\sum_{\indexfive\in\indexsetfive}\ p^{(3)}_\indexfive \cdot \closedtypeddistrone_\indexfive}\\
 & \ \ =\ \ &
 \expectype{\distrelts{\left(\termthree_{\indexthree} \typsep \distrtyptwo_{\indexthree} \right)^{p''_{\indexthree}} \sep \indexthree \in \indexsetthree}}\\
 \end{array}
 $$
 Moreover, by definition of the family $\left(\distrone_\indexfive\right)_{\indexfive \in \indexsetfive}$,
 $$
 \distrelts{\left(\termfour_{\indexfour} \right)^{p''_{\indexfour}} \sep \indexfour \in \indexsetfour}
 \ \ =\ \ \sum_{\indexfive \in \indexsetfive}\ p^{(3)}_\indexfive \cdot\distrelts{\left(\termfive_\indexfive\right)^1}
 \ \ \rcbv\ \ \distrelts{\left(\termtwo_{\indextwo} \right)^{p'_{\indextwo}} \sep \indextwo \in \indexsettwo}
 \ \ =\ \ \sum_{\indexfive \in \indexsetfive}\ p^{(3)}_\indexfive \cdot\distrone_\indexfive
 $$
 The result follows from the fact that $\pseudorep{\left(\termsix_\indexsix\right)^{p^{(4)}_\indexsix} \sep
  \indexsix \in \indexsetsix_\indexfive}$ is a pseudo-representation of $\distrone_\indexfive$ for every
  $\indexfive \in \indexsetfive$.
 
\end{proof}
}

\longv{
\subsection{Subject Reduction for $\redval$}
}

\longv{
Recall that there is an order $\distrleq$ on distributions, defined pointwise.
}

\longv{
\begin{lemma}
 Suppose that $\termone \redval \distrelts{\valone_{\indexone}^{p_{\indexone}} \sep \indexone \in \indexsetone}$
 and that $\termone \in \setdistrtypedclosedterms{\distrtypone}$.
 Then there exists a closed typed distribution $\distrelts{\left(\valtwo_{\indextwo} \typsep \typone_{\indextwo}\right)^{p'_{\indextwo}}
 \sep \indextwo \in \indexsettwo}$ such that 
 \begin{itemize}
  \item $\expectype{\left(\valtwo_{\indextwo} \typsep \typone_{\indextwo}\right)^{p'_{\indextwo}}}\ \distrleq\ \distrtypone$,
  \item and that $\pseudorep{\left(\valtwo_{\indextwo}\right)^{p'_{\indextwo}}  \sep \indextwo \in \indexsettwo}$
  is a pseudo-representation of $\distrelts{\left(\valone_{\indexone}\right)^{p_{\indexone}} \sep \indexone \in \indexsetone}$.
 \end{itemize}
\end{lemma}
}

\longv{
\begin{proof}
 We have $\termone \rcbv \distrone \valdec \distroneterm + 
 \distrelts{\valone_{\indexone}^{p_{\indexone}} \sep \indexone \in \indexsetone}$.
 By Lemma~\ref{lemma/subject-reduction-one-step}, there exists a closed typed distribution
  $\distrelts{\left(\termthree_{\indexthree} \typsep \distrtyptwo_{\indexthree}\right)^{p''_{\indexthree}}
 \sep \indexthree \in \indexsetthree}$
 such that $\expectype{\left(\termthree_{\indexthree} \typsep \distrtyptwo_{\indexthree}\right)^{p''_{\indexthree}}}
 \ =\ \mu$ and that $\pseudorep{\left(\termthree_{\indexthree}\right)^{p''_{\indexthree}}\sep \indexthree \in \indexsetthree}$ is a pseudo-representation of $\distrone$.
 We consider the pseudo-representation
 $\pseudorep{\left(\valtwo_{\indextwo}\right)^{p'_{\indextwo}}  \sep \indextwo \in \indexsettwo}$
 obtained from $\pseudorep{\left(\termthree_{\indexthree}\right)^{p''_{\indexthree}}\sep \indexthree \in \indexsetthree}$ by removing all the terms which are not values. 
 Note that $\indexsettwo \subseteq \indexsetthree$.
 We obtain in this way a pseudo-representation
 of $\distrelts{\valone_{\indexone}^{p_{\indexone}} \sep \indexone \in \indexsetone}$
 which induces a closed typed representation
 $\distrelts{\left(\valtwo_{\indextwo} \typsep \distrtyptwo_{\indextwo}\right)^{p'_{\indextwo}}
 \sep \indextwo \in \indexsettwo}$ such that $\expectype{\left(\valtwo_{\indextwo} \typsep \distrtyptwo_{\indextwo}\right)^{p'_{\indextwo}}}\ \distrleq\ \distrtypone$. 
\end{proof}
}

\shortv{The proof of this result, and its generalization to the iterated reduction of closed typed distributions, appear in
the long version. They allow us to deduce the following property on the operational semantics of
$\languageproba$-terms:}

\begin{theorem}[Subject Reduction]
\label{th:subject-reduction-semantics}
  Let $\termone \in \setdistrtypedclosedterms{\distrtypone}$. 
  Then there exists a closed typed distribution $\distrelts{\left(\valtwo_{\indextwo} \typsep \typone_{\indextwo}\right)^{p_{\indextwo}}
 \sep \indextwo \in \indexsettwo}$ such that 
 \begin{varitemize}
  \item $\expectype{\left(\valtwo_{\indextwo} \typsep \typone_{\indextwo}\right)^{p_{\indextwo}}}\ \distrleq \ \distrtypone$,
\item and that $\pseudorep{\left(\valtwo_{\indextwo}\right)^{p_{\indextwo}}  \sep \indextwo \in \indexsettwo}$
  is a pseudo-representation of $\semantics{\termone}$.
 \end{varitemize}
\end{theorem}

\noindent
Note that $\expectype{\left(\valtwo_{\indextwo} \typsep \typone_{\indextwo}\right)^{p_{\indextwo}}}\ \distrleq \ \distrtypone$
since the semantics of a term may not be a proper distribution at this stage.
In fact, it will follow from the soundness theorem of Section~\ref{sect:reducibility} that
the typability of $\termone$ implies that $\distrsum{\semantics{\termone}}=1$
and thus that the previous statement is an equality.

\drop{
\begin{proof}
 If $\termone$ is a value, the result is immediate. Consider thus that 
 $\termone$ is not a value.
 We prove that there exists a chain of closed typed distributions
 $\closedtypeddistrone_0\,=\,\emptyset \distrleq \closedtypeddistrone_1 \distrleq \closedtypeddistrone_2 
 \distrleq \cdots$ whose supremum satisfies these requirements.

 REPRENDRE
 First, we write
 $
 \semantics{\termone}\ \ =\ \ \sum_{\indexone \in \NN}\ \distrone_{\indexone}
 $
 with $\termone \redval^n \sum_{\indexone =0}^n\ \distrone_{\indexone}$
 and we introduce two families of distributions of terms (which are not values)
 $\left(\distrtwo_{\indexone}\right)_{\indexone \in \indexsetone}$ and 
 $\left(\distrthree_{\indexone}\right)_{\indexone \in \indexsetone}$
 which are such that $\termone \rcbv^n \sum_{\indexone =0}^n\ \distrone_{\indexone} +
 \distrtwo_n + \distrthree_n$ and $\distrtwo_n \rcbv \distrone_{n+1}$.

 We shall now define by induction three families $\left(\closedtypeddistrone_\indexone\right)_{\indexone \in \NN}$,
 $\left(\closedtypeddistrtwo_\indexone\right)_{\indexone \in \NN}$ and
 $\left(\closedtypeddistrthree_\indexone\right)_{\indexone \in \NN}$ of closed typed distributions
 such that, for every $n \in \NN$:
 \begin{varitemize}
  \item $\expectype{\closedtypeddistrone_{n-1}} \distrleq \expectype{\closedtypeddistrone_{n}} \distrleq \distrtypone$,
  if $n \geq 1$,
  \item the underlying term distribution of $\closedtypeddistrone_n$
  is $\sum_{\indexone = 0}^n \distrone_\indexone$,
  \item the underlying term distribution of $\closedtypeddistrtwo_n$
  is $\distrtwo_n$,
  \item the underlying term distribution of $\closedtypeddistrthree_n$ is $\distrthree_n$.
 \end{varitemize}

\end{proof}
}

\section{Typability Implies Termination: Reducibility Strikes Again}\label{sect:reducibility}
This section is technically the most advanced one of the paper, and proves
that the typing discipline we have introduced indeed enforces almost sure termination.
As already mentioned, the technique we will employ is a substantial generalisation
of Girard-Tait's reducibility. In particular, reducibility must be made quantitative,
in that terms can be said to be reducible \emph{with a certain probability}. This
means that reducibility sets will be defined as sets parametrised by a real number $p$,
called the \emph{degree of reducibility} of the set.
As Lemma~\ref{lemma:degree-of-reducibility-and-degree-of-termination} will emphasize,
this degree of reducibility ensures that terms contained in a reducibility set parametrised
by $p$ terminate with probability at least $p$. These ``intermediate'' degrees of reducibility
are required to handle the fixpoint construction, and show that recursively-defined terms
that are typable are indeed AST --- that is, that they belong to the appropriate reducibility set, parametrised
by $1$.

%
%
%
%
%


\longv{
\subsection{Reducibility Sets for Closed Terms}
}

The first preliminary notion we need is that of a size environment:
\begin{definition}[Size Environment]
  A \emph{size environment} is any function $\seone$ from $\sizevars$
  to $\NN\cup\{\infty\}$.  Given a size environment $\seone$ and a
  size expression $\sizeone$, there is a naturally defined element of
  $\NN\cup\{\infty\}$, which we indicate as
  $\sesem{\sizeone}{\seone}$:
  \begin{varitemize}
   \item $\sesem{\sizesuccit{\sizevarone}{n}}{\seone}\ =\ \seone(\sizevarone) + n$,
   \item $\sesem{\sizeinf}{\seone}\ =\ \infty$.
  \end{varitemize}
\end{definition}
In other words, the purpose of size environments is to give a semantic
meaning to size expressions. Our reducibility sets will be parametrised
not only on a probability, but also on a size environment.

\begin{definition}[Reducibility Sets]\label{def:redsets}~

\begin{varitemize}
  \item For values of simple type $\Nat$, we define the reducibility sets
  $$
  \setredvals{p}{\Nat^\sizeone,\seone}\ \ =\ \ \left\{\natsucc^n\ \natzero\sep p>0\Longrightarrow n < \sesem{\sizeone}{\seone}\right\}.
  $$

  \medbreak
    
  \item Values of higher-order type are in a reducibility set when their applications
  to appropriate values are reducible terms, with an adequate degree of reducibility:
  $$
  \begin{array}{rl}
   \setredvals{p}{\typone\typarrow\distrtypone,\seone}\ \ =\ \ \left\{\valone \in \settypedclosedvalues{\underlying{\typone \typarrow \distrtypone}} 
    \sep\forall q \in (0,1],\ \ \right.& \forall\valtwo\in\setredvals{q}{\typone,\seone},\\ \ \ 
   & \left. \valone\ \valtwo\in\setredtmsfin{pq}{\distrtypone,\seone}\right\}\\
  \end{array}
  $$
  
  \medbreak
  
   \item Distributions of values are reducible with degree $p$ when they consist of values which are themselves
   globally reducible ``enough''. Formally,   
   $\setreddists{p}{\distrtypone,\seone}$ is the set of finite distributions of values -- in the sense that they have a finite support --
   admitting a pseudo-representation
  $\distrone\,=\,\pseudorep{\left(\valone_{\indexone}\right)^{p_{\indexone}} \sep \indexone \in \indexsetone}$
  such that, setting $\distrtypone\,=\,\distrelts{\left(\typone_\indextwo\right)^{p'_\indextwo} \sep \indextwo \in \indexsettwo}$,
  there exists a family $\left(p_{\indexone \indextwo}\right)_{\indexone\in \indexsetone,\indextwo\in\indexsettwo}
  \in [0,1]^{|\indexsetone|\times|\indexsettwo|}$
  of probabilities and a family
  $\left(q_{\indexone \indextwo}\right)_{\indexone\in \indexsetone,\indextwo\in\indexsettwo}\in [0,1]^{|\indexsetone|\times|\indexsettwo|}$
  of degrees of reducibility, satisfying:
  \begin{enumerate}
   \item $\forall \indexone \in \indexsetone,\ \ \forall\indextwo \in \indexsettwo,\ \ \valone_\indexone \in 
   \setredvals{q_{\indexone\indextwo}}{\typone_\indextwo,\seone}$,
   \item $\forall \indexone \in \indexsetone,\ \ \sum_{\indextwo \in \indexsettwo}\ p_{\indexone\indextwo}\ =\ p_\indexone$,
   \item $\forall \indextwo \in \indexsettwo,\ \ \sum_{\indexone \in \indexsetone}\ p_{\indexone\indextwo}\ =\ \distrtypone(\typone_\indextwo)$,
   \item $p \leq \sum_{\indexone \in \indexsetone}\sum_{\indextwo \in \indexsettwo}\ q_{\indexone\indextwo}p_{\indexone\indextwo}$.
  \end{enumerate}
  %
  Note that $(2)$ and $(3)$ imply that $\distrsum{\distrone}\,=\,\distrsum{\distrtypone}$.
  We say that $\pseudorep{\left(\valone_{\indexone}\right)^{p_{\indexone}} \!\sep\! \indexone \in \indexsetone}$
  \emph{witnesses} that $\distrone \in \setreddists{p}{\distrtypone,\seone}$.
  
  \medbreak
  
 \item A term is reducible with degree $p$ when its finite approximations compute distributions of values
 of degree of reducibility arbitrarily close to $p$:
 $$
 \begin{array}{rl}
  \setredtmsfin{p}{\distrtypone,\seone}\ \ =\ \ \left\{\termone \in \settypedclosedterms{\underlying{\distrtypone}}
 \sep \forall 0 \leq r < p, \right. & \ \ \exists \distrtyptwo_r \distrleq \distrtypone,\ \ \exists n_r \in \NN,\\ 
 &\!\!\!\!\!\!\!\!\!\!\!
 \left.\termone \redval^{n_r} \distrone_r \text{ and } \distrone_r \in \setreddists{r}{\distrtyptwo_r,\seone} \right\}\\
 \end{array}$$
 Note that here, unlike to the case of $\setreddists{}{}$, the fact that $\termone \in \settypedclosedterms{\underlying{\distrtypone}}$
 implies that $\distrtypone$ is proper.
\end{varitemize}
\end{definition}
The first thing to observe about reducibility sets as given in
Definition \ref{def:redsets} is that they only deal with closed
terms, and not with arbitrary terms. As such, we cannot rely
\emph{directly} on them when proving AST for typable terms, at
least if we want to prove it by induction on the structure of
type derivations. We will therefore define in the sequel an extension
of these sets to open terms, which will be based on these sets of 
closed terms, and therefore enjoy similar properties.
\longv{Before embarking in the proof that typability implies reducibility, it is
convenient to prove some fundamental properties of reducibility sets, which inform
us about how these sets are structured, and which will be crucial in the sequel.
This is the purpose of the following subsections.}
\longv{
\subsection{Reducibility Sets and Termination}
}
The following lemma, relatively easy to prove, is crucial for the 
understanding of the reducibility sets, for that it shows that the 
degree of reducibility of a term gives information on the sum of
its operational semantics:
\begin{lemma}[Reducibility and Termination]~

\label{lemma:degree-of-reducibility-and-degree-of-termination}
\begin{varitemize}
 \item Let $\distrone \in \setreddists{p}{\distrtypone,\seone}$. Then $\distrsum{\distrone} \geq p$.
 \item Let $\termone \in \setredtmsfin{p}{\distrtypone,\seone}$. Then $\distrsum{\semantics{\termone}} \geq p$.
\end{varitemize}
\end{lemma}

\longv{

\begin{proof}~

\begin{itemize}
 \item Let $\distrone \in \setreddists{p}{\distrtypone,\seone}$, then there exists a pseudo-representation
 $\distrone\,=\,\pseudorep{\left(\valone_{\indexone}\right)^{p_{\indexone}} \sep \indexone \in \indexsetone}$
 and families 
 $\left(p_{\indexone \indextwo}\right)_{\indexone\in \indexsetone,\indextwo\in\indexsettwo}$
  and $\left(q_{\indexone \indextwo}\right)_{\indexone\in \indexsetone,\indextwo\in\indexsettwo}$ of reals of $[0,1]$
  such that 
  $\forall \indexone \in \indexsetone,\ \ \sum_{\indextwo \in \indexsettwo}\ p_{\indexone\indextwo}\ =\ p_\indexone$,
  and that $p \leq \sum_{\indexone \in \indexsetone}\sum_{\indextwo \in \indexsettwo}\ q_{\indexone\indextwo}p_{\indexone\indextwo}$.
 We therefore have:
 $$
 \distrsum{\distrone}\ \ =\ \ \sum_{\indexone \in \indexsetone}\ p_\indexone \ \ =\ \ 
 \sum_{\indexone \in \indexsetone}\sum_{\indextwo \in \indexsettwo}\ p_{\indexone\indextwo}
 \ \ \geq\ \  \sum_{\indexone \in \indexsetone}\sum_{\indextwo \in \indexsettwo}\ q_{\indexone\indextwo}p_{\indexone\indextwo}\ \  \geq\ \  p.	
 $$
 \item Since $\termone \in \setredtmsfin{p}{\distrtypone,\seone}$, for every $0 \leq r < p$, there exists 
 $n_r$ with $\termone \redval^{n_r} \distrone_r$ and $\distrone_r \in \setreddists{r}{\distrtyptwo_r,\seone}$.
 From the previous point, we get that $\distrsum{\distrone_r} \geq r$ for every $0 \leq r < p$.
 It follows from Corollary~\ref{corollary:semantics-is-computed-monotonically} that $\distrsum{\semantics{\termone}} \geq r$ for every $0 \leq r < p$
 and, by taking the supremum, $\distrsum{\semantics{\termone}} \geq p$.
\end{itemize}

\end{proof}

}

\noindent
It follows from this lemma that terms with degree of reducibility $1$ are AST:

\begin{corollary}[Reducibility and AST]
\label{corollary:reducibility-and-ast}
Let $\termone \in \setredtmsfin{1}{\distrtypone,\seone}$. Then $\termone$ is AST.
\end{corollary}

\longv{ 
\subsection{Reducibility Sets and Reducibility Degrees}
We now prove two results related to the reducibility degrees of reducibility
sets. First of all, if the degree of reducibility $p$ is $0$, then no assumption is made
on the probability of termination of terms, distributions or values.
It follows that the three kinds of reducibility sets collapse to the set 
of all affinely simply typable terms, distributions or values:
}
\shortv{
\paragraph{Fundamental Properties.}
Before embarking in the proof that typability implies reducibility, it is
convenient to prove some fundamental properties of reducibility sets, which inform
us about how these sets are structured, and which will be crucial in the sequel.
First of all, if the degree of reducibility $p$ is $0$, then no assumption is made
on the probability of termination of terms, distributions or values.
It follows that the three kinds of reducibility sets collapse to the set 
of all affinely simply typable terms, distributions or values:}
\begin{lemma}[Candidates of Null Reducibility]
\begin{varitemize}
\label{lemma:simple-types-and-candidates-of-proba-zero}
 \item If $\valone \in \settypedclosedvalues{\simpletypone}$, then
  $\valone\in\setredvals{0}{\typone,\seone}$ for every $\typone$ such that $\underlying{\typone} \,=\, \simpletypone$ and every
  size environment $\seone$.
  
 \item Let $\distrone\,=\,\distrelts{\left(\valone_\indexone\right)^{p_\indexone}  \sep \indexone \in \indexsetone}$ be a finite distribution of values.
  If $\forall \indexone \in \indexsetone,\ \ \valone_\indexone \in \settypedclosedvalues{\simpletypone}$, then 
 $\distrone \in \setreddists{0}{\distrtypone,\seone}$
 for every $\distrtypone$ such that $\underlying{\distrtypone}\,=\,\simpletypone$ and $\distrsum{\distrtypone}\,=\,\distrsum{\distrone}$ and every $\seone$.
 
 \item If $\termone \in \settypedclosedterms{\simpletypone}$, then
  $\termone\in\setredtmsfin{0}{\distrtypone,\seone}$ for $\distrtypone$ such that $\underlying{\distrtypone} \,=\,\simpletypone$ and every $\seone$.
\end{varitemize}
\end{lemma}

\longv{

\noindent
\textbf{Structure of the proof.}
In this lemma, as for most lemmas proving properties about $\setredvals{}{}$, $\setreddists{}{}$ and $\setredtmsfin{}{}$,
we use a proof by induction on types. As the property is
defined in a mutual way over $\setredvals{}{}$, $\setreddists{}{}$ and $\setredtmsfin{}{}$, we typically prove it for 
$\setredvals{p}{\Nat^\sizeone,\seone}$ for any size $\sizeone$ refining $\Nat$, and then for 
$\setredvals{p}{\typone \typarrow \distrtypone,\seone}$ by using the associated hypothesis on 
$\setredtmsfin{p}{\distrtypone,\seone}$. Then we prove the property for any distribution type
for $\setreddists{p}{\distrtypone,\seone}$ using induction hypothesis on the 
$\setredvals{p}{\typone,\seone}$ for $\typone \in \supp{\distrtypone}$,
and we prove it for $\setredtmsfin{p}{\distrtypone,\seone}$ using induction hypothesis on $\setredvals{p}{\typone,\seone}$.
The point is that these ingredients allow to give a proof by induction on the simple type underlying the sized type of interest.
In the base case, the sized type is necessarily of the form $\Nat^\sizeone$ for some size $\sizeone$: we prove the statement
on $\setredvals{p}{\Nat^\sizeone,\seone}$ for
all these sized types without using any induction-like hypothesis.
Then we prove the statement for distribution types $\distrtypone\,=\,\distrelts{\left(\Nat^{\sizeone_\indexone}\right)^{p_\indexone}
\sep \indexone \in \indexsetone}$, first on $\setreddists{p}{\distrtypone,\seone}$ by using the results for the 
sets $\setredvals{p}{\Nat^{\sizeone_\indexone},\seone}$. Then we prove the result for 
$\setredtmsfin{p}{\distrtypone,\seone}$ typically using the one for $\setreddists{p}{\distrtypone,\seone}$.

We then switch to higher-order types, and give the proof for $\setredvals{p}{\typone \typarrow \distrtypone,\seone}$, which may use
the results for the other sets on types $\typone$ and $\distrtypone$. Typically, only results on $\setredtmsfin{p}{\distrtypone,\seone}$ are used.
Then the proofs for $\setreddists{p}{\typone \typarrow \distrtypone,\seone}$ and $\setredtmsfin{p}{\typone \typarrow \distrtypone,\seone}$
are typically the same as in the case of distributions over sized types refining $\Nat$: therefore we do not write them again.

This proof scheme will become more clear with the proof of this lemma on candidates of null reducibility:

}

\longv{
\begin{proof}~
\begin{itemize}
 \item Let $\valone \in \settypedclosedvalues{\Nat}$. Every $\typone \refines \Nat$ is of the shape $\typone\,=\,\Nat^\sizeone$ for a size $\sizeone$.
 Let $\seone$ be a size environment. By inspection of the grammar of values and of the simple type system, we see that $\valone$
 must be of the shape $\natsucc^n \ \natzero$ for $n \in \NN$. Note that $\valone$ is closed: it cannot be a variable.
 By definition, $\valone \in \setredvals{0}{\typone,\seone}$.

 \item Let $\simpletypone\,=\,\simpletyptwo \typarrow \simpletypthree$ be a higher-order type, with 
 $\typone \refines \simpletyptwo$ and $\distrtypone \refines \simpletypthree$. Let $\seone$ be a size environment, and 
 $\valone \in \settypedclosedvalues{\simpletypone}$.
 Let $q \in (0,1]$ and $\valtwo\in\setredvals{q}{\typone,\seone}$, we need to prove that $\valone\ \valtwo\in\setredtmsfin{0}{\distrtypone,\seone}$.
 But, by definition of $\setredvals{q}{\typone,\seone}$, $\valtwo \in \settypedclosedvalues{\simpletyptwo}$.
 It follows that $\valone\ \valtwo \in \settypedclosedterms{\simpletypthree}$, and we can apply the induction hypothesis to deduce that
 $\valone\ \valtwo \in \setredtmsfin{0}{\distrtypone,\seone}$, so that by definition $\valone \in \setredvals{0}{\typone,\seone}$.

 \item Let $\distrone\,=\,\distrelts{\left(\valone_\indexone\right)^{p_\indexone}  \sep \indexone \in \indexsetone}$ be a distribution of values
 and $\distrtypone\,=\,\distrelts{\left(\typone_\indextwo\right)^{p'_{\indextwo}} \sep \indextwo \in \indexsettwo}
 \refines \simpletypone$ be a distribution type. Suppose that $\forall \indexone \in \indexsetone$, $\valone_\indexone \in \settypedclosedvalues{\simpletypone}$. Let $\seone$ be a size environment. For every $(\indexone,\indextwo) \in \indexsetone \times \indexsettwo$,
 we set $p_{\indexone \indextwo} \,=\,\frac{p_\indexone p'_\indextwo}{\distrsum{\distrtypone}}$ and $q_{\indexone \indextwo} \,=\,0$.
 We consider the canonical pseudo-representation $\distrone\,=\,\pseudorep{\left(\valone_\indexone\right)^{p_\indexone}  \sep \indexone \in \indexsetone}$
 and check the four conditions to be in $\setreddists{0}{\distrtypone,\seone}$:
 \begin{enumerate}
    \item $\forall \indexone \in \indexsetone,\ \ \forall\indextwo \in \indexsettwo,\ \ \valone_\indexone \in 
   \setredvals{q_{\indexone\indextwo}}{\typone_\indextwo,\seone}$: this is obtained by induction hypothesis,
   \item $\forall \indexone \in \indexsetone,\ \ \sum_{\indextwo \in \indexsettwo}\ p_{\indexone\indextwo}\ =\ p_\indexone$:
   let $\indexone \in \indexsetone$, we have $\sum_{\indextwo \in \indexsettwo}\ p_{\indexone\indextwo}\,=\,
   \frac{p_\indexone}{\distrsum{\distrtypone}} \sum_{\indextwo \in \indexsettwo}\ p'_\indextwo
   \,=\,\frac{p_\indexone}{\distrsum{\distrtypone}}\times \distrsum{\distrtypone}
   \,=\,p_\indexone$.
   \item $\forall \indextwo \in \indexsettwo,\ \ \sum_{\indexone \in \indexsetone}\ p_{\indexone\indextwo}\ =\ \distrtypone(\typone_\indextwo)$:
   let $\indextwo \in \indexsettwo$, we have $\sum_{\indexone \in \indexsetone}\ p_{\indexone\indextwo}
   \,=\,\frac{p'_\indextwo}{\distrsum{\distrtypone}} \sum_{\indexone \in \indexsetone}\ p_\indexone 
   \,=\,\frac{p'_\indextwo}{\distrsum{\distrtypone}} \times \distrsum{\distrone}$.
   But $\distrsum{\distrtypone}\,=\,\distrsum{\distrone}$, so that
   the sum equals $p'_\indextwo$ as requested.
   \item $p \leq \sum_{\indexone \in \indexsetone}\sum_{\indextwo \in \indexsettwo}\ q_{\indexone\indextwo}p_{\indexone\indextwo}$:
   this amounts to $0 \leq 0$, which holds.
 \end{enumerate}

 \item Let $\termone \in \settypedclosedterms{\simpletypone}$ and $\distrtypone \refines \simpletypone$. Let $\seone$ be a size environment.
 Then $\termone \in \setredtmsfin{0}{\distrtypone,\seone}$: the condition on $\termone$ in the definition of $\setredtmsfin{0}{\distrtypone,\seone}$
 is for any $0 \leq r < 0$ so that it's an empty condition in this case.
\end{itemize}
\end{proof}
}

\noindent
As $p$ gives us a lower bound on the sum of the semantics of terms, it is
easily guessed that a term having degree of reducibility $p$ must also have
degree of reducibility $q < p$. The following lemma makes this statement precise:
\begin{lemma}[Downward Closure]
\label{lemma/downward-closure-tred}
 Let $\typone$ be a sized type, $\distrtypone$ be a distribution type and $\seone$ be a size environment. Let $0 \leq q < p \leq 1$. Then:
\begin{varitemize}
  \item For any value $\valone$, $\valone \in \setredvals{p}{\typone,\seone} \ \ \implies\ \ \valone \in \setredvals{q}{\typone,\seone}$, 
  \item For any finite distribution of values $\distrone$, $\distrone \in \setreddists{p}{\distrtypone,\seone}
 \ \ \implies\ \ \distrone \in \setreddists{q}{\distrtypone,\seone}$,
  \item For any term $\termone$, $\termone \in \setredtmsfin{p}{\distrtypone,\seone}
 \ \ \implies\ \ \termone \in \setredtmsfin{q}{\distrtypone,\seone}$.
 \end{varitemize}
\end{lemma}

\longv{
\begin{proof}
  Let $\typone$ be a sized type, $\distrtypone$ be a distribution type and $\seone$ be a size environment.
  If $q = 0$, the result is immediate as a consequence of Lemma~\ref{lemma:simple-types-and-candidates-of-proba-zero}.
  Let $0 < q < p \leq 1$.
 \begin{itemize}
  \item Suppose that $\valone \in \setredvals{p}{\Nat^\sizeone,\seone}$. Since by definition $p,\,q > 0 \ \implies\ 
  \setredvals{p}{\Nat^\sizeone,\seone}\,=\,\setredvals{q}{\Nat^\sizeone,\seone}$, the result holds.
 \item Suppose that $\valone \in \setredvals{p}{\typone \typarrow \distrtypone,\seone}$.
 Then:
 $$
 \begin{array}{rl}
  & \valone \in \setredvals{p}{\typone \typarrow \distrtypone,\seone}\\
  \Longleftrightarrow \ \ & \forall q \in (0,1],\ \ \forall\valtwo\in\setredvals{q}{\typone,\seone},\ \ 
        \valone\valtwo\in\setredtms{pq}{\distrtypone,\seone}\\
  \implies \ \ & \forall q' \in (0,1],\ \ \forall\valtwo\in\setredvals{q'}{\typone,\seone},\ \ 
        \valone\valtwo\in\setredtms{qq'}{\distrtypone,\seone} \qquad \ \ \text{(by IH, since } 0 < qq' < pq \leq 1 \text{)}\\
  \Longleftrightarrow \ \ & \valone \in \setredvals{q}{\typone \typarrow \distrtypone,\seone}\\
 \end{array}
 $$

 \item Suppose that $\distrone \in \setreddists{p}{\distrtypone,\seone}$. Then there exists 
 a pseudo-representation $\distrone\,=\,\pseudorep{\left(\valone_\indexone\right)^{p_\indexone}  \sep \indexone \in \indexsetone}$
 and families of reals
 $\left(p_{\indexone \indextwo}\right)_{\indexone\in \indexsetone,\indextwo\in\indexsettwo}$
  and $\left(q_{\indexone \indextwo}\right)_{\indexone\in \indexsetone,\indextwo\in\indexsettwo}$ 
  satisfying conditions $(1)-(4)$. We have $\distrone \in \setreddists{q}{\distrtypone,\seone}$,
  for the same pseudo-representation,
  since conditions $(1)-(3)$ are the same, and $(4)$ holds as well since $q < p$.
 
 \item Suppose that $\termone \in \setredtmsfin{p}{\distrtypone,\seone}$.
 Then for every $0 \leq r < p$, there exists 
 $\distrtyptwo_r \distrleq \distrtypone$ and $n_r \in \NN$ with $\termone \redval^{n_r} \distrone_r$
 and $\distrone_r \in \setreddists{r}{\distrtyptwo_r,\seone}$.
 So this statement also holds for every $0 \leq r < q$ and $\termone \in \setredtmsfin{q}{\distrtypone,\seone}$.
 \end{itemize}
\end{proof}
}

%
%
%

\longv{
\subsection{Continuity of the Reducibility Sets}
}

\longv{
To prove the lemma of continuity on the reducibility sets, which says that if 
an element is in all the reducibility sets for degrees $q < p$ then it is also in the
set parametrised by the degree $p$, we use the following companion lemma
computing a family of probabilities maximizing the degree of reducibility of a distribution:
}
 
\longv{

\begin{lemma}[Maximizing the Degree of Reducibility of a Distribution]~

 \label{lemma/companion-lemma-continuity-tred}
  Let $\distrone\,=\,\pseudorep{\left(\valone_{\indexone}\right)^{p_{\indexone}} \sep \indexone \in \indexsetone}$ be a finite distribution of values and 
  $\distrtypone\,=\,\distrelts{\left(\typone_\indextwo\right)^{p'_\indextwo} \sep \indextwo \in \indexsettwo}$ be a distribution type.
  Set $q_{\indexone \indextwo}\,=\,\max\left\{q \sep \valone_\indexone \in \setredvals{q}{\typone_\indextwo,\seone} \right\}$
  for every $(\indexone,\indextwo) \in \indexsetone \times \indexsettwo$.
  Then there exists a family $\left(p_{\indexone \indextwo}\right)_{\indexone\in \indexsetone,\indextwo\in\indexsettwo}$
  of reals of $[0,1]$ satisfying:
  \begin{enumerate}
   \item $\forall \indexone \in \indexsetone,\ \ \sum_{\indextwo \in \indexsettwo}\ p_{\indexone\indextwo}\ =\ p_\indexone$,
   \item $\forall \indextwo \in \indexsettwo,\ \ \sum_{\indexone \in \indexsetone}\ p_{\indexone\indextwo}\ =\ \distrtypone(\typone_\indextwo)$,
  \end{enumerate}
  and which maximizes $\sum_{\indexone \in \indexsetone}\sum_{\indextwo \in \indexsettwo}\ q_{\indexone\indextwo}p_{\indexone\indextwo}$.
\end{lemma}
}

\longv{
\begin{proof}
 We use the theory of linear programming in the finite real vectorial space $\RR^n$, taking~\cite{schrijver:theory-linear-programming} as a reference.
 We stick to the notations of this book. The problem then amounts to showing the existence of
 \begin{equation}
  \label{equation/the-linear-problem}
  \max\left\{cx \sep x \geq \vec{0},\ Ax \,=\, b \right\}
 \end{equation}
 where, supposing that we can index vectors and matrices by $\indexone \times \indextwo$ thanks to a bijection 
 $\indexone \times \indextwo \xrightarrow{} \setoneton$ where $n\,=\,\setcard{\indexsetone \times \indexsettwo}$:
 \begin{itemize}
  \item $x$ is the column vector indexed by the finite set $\indexsetone \times \indexsettwo$, where $x_{\indexone\indextwo}$ plays the role of
  $p_{\indexone\indextwo}$,
  \item $c$ is the row vector indexed by $\indexsetone \times \indexsettwo$, with 
  $c_{\indexone\indextwo}\,=\,\max\left\{q \sep \valone_\indexone \in \setredvals{q}{\typone_\indextwo,\seone} \right\}$,
  \item $\vec{0}$ is the null column vector of size $\setcard{\indexsetone \times \indexsettwo}$,
  \item $A$ is the matrix with columns indexed by $\indexsetone \times \indexsettwo$ and rows indexed by
  $\indexsetone + \indexsettwo$, and such that:
  \begin{itemize}
   \item $a_{\indexone',(\indexone,\indextwo)}\,=\,1$ if and only if $\indexone=\indexone'$, and $0$ else,
   \item and $a_{\indextwo',(\indexone,\indextwo)}\,=\,1$ if and only if $\indextwo=\indextwo'$, and $0$ else.
  \end{itemize}

  \item $b$ is the column vector indexed by $\indexsetone+\indexsettwo$ and such that $b_\indexone\,=\,p_\indexone$ and
  $b_\indextwo\,=\,\distrtypone(\typone_\indextwo)$.
 \end{itemize}
 Following \cite[Section 7.4]{schrijver:theory-linear-programming}, the maximum (\ref{equation/the-linear-problem}) exists if and only if:
 \begin{itemize}
  \item the problem is \emph{feasible}: its constraints admit at least a solution,
  \item and if it is \emph{bounded}: there should be an upper bound over (\ref{equation/the-linear-problem}).
 \end{itemize}
 and, also, its existence is equivalent to the one of the maximum of the following problem:
  \begin{equation}
  \label{equation/the-linear-problem-bis}
  \max\left\{cx \sep x \geq \vec{0},\ Ax \leq b \right\}
 \end{equation}
 This reformulation makes the feasibility immediate, for the null vector $x\,=\,\vec{0}$.
 It is as well immediate to see that the problem is bounded: by construction, all the $q_{\indexone\indextwo} \in [0,1]$,
 and $\sum_{\indexone \in \indexsetone}\sum_{\indextwo\in\indexsettwo}\ p_{\indexone\indextwo}\,=\,1$ so that 
 $\sum_{\indexone \in \indexsetone}\sum_{\indextwo\in\indexsettwo}\ q_{\indexone\indextwo}p_{\indexone\indextwo}\, \leq\, 1$.
 The existence of the maximum (\ref{equation/the-linear-problem}) follows, and the Lemma therefore holds.
\end{proof}
}

\longv{It follows that a distribution has a maximal degree of reducibility: the supremum of the degrees
of reducibility is again a degree of reducibility:
\begin{corollary}[Maximizing the Degree of Reducibility of a Distribution II]
\label{corollary/continuity-dred}
 Let $\distrone$ be a finite distribution of values, $\distrtypone$ be a distribution type and $\seone$ be a size environment.
 Suppose that $\distrone \in \setreddists{p}{\distrtypone,\seone}$ for some real $p \in [0,1]$. Then there exists a maximal real
 $p_{\mathit{max}} \in [p,1]$ such that $\distrone \in \setreddists{p_{\mathit{max}}}{\distrtypone,\seone}$
 and $p' > p_{\mathit{max}} \implies \distrone \notin \setreddists{p'}{\distrtypone,\seone}$.
\end{corollary}
}

\longv{
\begin{proof}
 Let $\distrone\,=\,\pseudorep{\left(\valone_{\indexone}\right)^{p_{\indexone}} \sep \indexone \in \indexsetone}$ be a finite distribution of values and 
  $\distrtypone\,=\,\distrelts{\left(\typone_\indextwo\right)^{p'_\indextwo} \sep \indextwo \in \indexsettwo}$ be a distribution type.
 By Lemma \ref{lemma/companion-lemma-continuity-tred}, setting $q_{\indexone \indextwo}\,=\,\max\left\{q \sep \valone_\indexone \in
 \setredvals{q}{\typone_\indextwo,\seone} \right\}$ for every $(\indexone,\indextwo) \in \indexsetone \times \indexsettwo$,
  there exists a family $\left(p_{\indexone \indextwo}\right)_{\indexone\in \indexsetone,\indextwo\in\indexsettwo}$
  of reals of $[0,1]$ which maximizes $w\,=\,
  \sum_{\indexone \in \indexsetone}\sum_{\indextwo \in \indexsettwo}\ q_{\indexone\indextwo}p_{\indexone\indextwo}$.
  It is immediate to see that any increase of a $q_{\indexone\indextwo}$ to $q'$ is contradictory with 
  $\valone_\indexone \in \setredvals{q'}{\typone_\indextwo,\seone}$, and that any decrease of a $q_{\indexone\indextwo}$
  actually decreases $w$. It follows that $p_{\mathit{max}}\,=\,w$.
\end{proof}

}

To analyse the $\letrecname$ construction, we will prove that, for every $\epsilon \in (0,1]$,
performing enough unfoldings of the fixpoint allows to prove that the recursively-defined term
is in a reducibility set parametrised by $1 -\epsilon$. We will be able to conclude on the 
AST nature of recursive constructions using the following continuity lemma, proved using
the theory of linear programming:
%
\begin{lemma}[Continuity]\label{lemma/continuity-lemma-vred} 
Let $\typone$ be a sized type, $\distrtypone$ be a distribution type
and $\seone$ be a size environment.  Let $p \in (0,1]$. Then:
\begin{varitemize}
 \item 
   $\setredvals{p}{\typone,\seone}\,=\,\bigcap_{0<q<p}\ \setredvals{q}{\typone,\seone}$,
 \item 
   $\setreddists{p}{\distrtypone,\seone}\,=\,\bigcap_{0<q<p}\ \setreddists{q}{\distrtypone,\seone}$,
 \item 
   $\setredtmsfin{p}{\distrtypone,\seone}\,=\,\bigcap_{0<q<p}\ \setredtmsfin{q}{\distrtypone,\seone}$.
\end{varitemize}
\end{lemma}

\longv{
\begin{proof}
Let $\typone$ be a sized type, $\distrtypone$ be a distribution type and $\seone$ be a size environment.
Let $p \in (0,1]$.

\begin{itemize}
 \item If $\typone\,=\,\Nat^{\sizeone}$ for some size $\sizeone$, then for every $0 < q < p$ we have 
 $\setredvals{q}{\typone,\seone}\,=\,\setredvals{p}{\typone,\seone}$ so that 
 $\setredvals{p}{\typone,\seone}\,=\,\bigcap_{0 < q < p}\ \setredvals{q}{\typone,\seone}$.
 \item If $\typone\,=\,\typtwo \typarrow \distrtypone$, we proceed by mutual inclusions.
 \begin{itemize}
 \item $\setredvals{p}{\typone,\seone}\,\subseteq\,\bigcap_{0<q<p}\ \setredvals{q}{\typone,\seone}$ is an immediate consequence
 of Lemma~\ref{lemma/downward-closure-tred}.

 \item Let us prove now that $\bigcap_{0<q<p}\ \setredvals{q}{\typone,\seone} \,\subseteq\, \setredvals{p}{\typone,\seone}$.
 Let $\valone \in \bigcap_{0<q<p}\ \setredvals{q}{\typone,\seone}$, it follows that:
 $$
 \begin{array}{rl}
  & \forall q \in (0,p),\ \ \forall q' \in [0,1],\ \ \forall \valtwo \in \setredvals{q'}{\typone,\seone},\ \ 
  \valone \ \valtwo \in \setredtms{qq'}{\distrtypone,\seone}\\
  \implies & \forall q' \in [0,1],\ \ \forall \valtwo \in \setredvals{q'}{\typone,\seone},\ \ \forall q \in (0,p),\ \ 
  \valone \ \valtwo \in \setredtms{qq'}{\distrtypone,\seone}\\
  \implies & \forall q' \in [0,1],\ \ \forall \valtwo \in \setredvals{q'}{\typone,\seone},\ \ 
  \valone \ \valtwo \in \bigcap_{0<q<p}\  \setredtms{qq'}{\distrtypone,\seone}\\
 \end{array}
 $$
 But
 $$
 \bigcap_{0<q<p}\ \setredtms{qq'}{\distrtypone,\seone}\ \ =\ \ \bigcap_{0<r<pq'}\  \setredtms{r}{\distrtypone,\seone}
 \ \ =\ \  \setredtms{pq'}{\distrtypone,\seone} \qquad \ \ \text{ (by IH)}
 $$
 so that
 $$
 \forall q' \in [0,1],\ \ \forall \valtwo \in \setredvals{q'}{\typone,\seone},\ \ 
  \valone \ \valtwo \in \setredtms{pq'}{\distrtypone,\seone}
 $$
 By definition, $\valone \in \setredvals{p}{\typone,\seone}$.
 \end{itemize}
 \item The inclusion $\setreddists{p}{\distrtypone,\seone} \,\subseteq\,\bigcap_{0<q<p}\ \setreddists{q}{\distrtypone,\seone}$ 
 is again an immediate consequence of Lemma~\ref{lemma/downward-closure-tred}.
 Let $\distrone \in \bigcap_{0<q<p}\ \setreddists{q}{\distrtypone,\seone}$.
 Let $(q_n)_{n \in \NN}$ be an increasing sequence of reals of $[0,p)$ converging to $p$.
 For every $n \in \NN$, $\distrone \in \setreddists{q_n}{\distrtypone,\seone}$ so that by
 Corollary~\ref{corollary/continuity-dred} there exists a real
 $p_{\mathit{max},n} \in [q_n,1]$ such that $\distrone \in \setreddists{p_{\mathit{max},n}}{\distrtypone,\seone}$
 and $p' > p_{\mathit{max},n} \implies \distrone \notin \setreddists{p'}{\distrtypone,\seone}$.
 It follows that all the $p_{\mathit{max},n}$ coincide, and that they are greater than $\sup_{n \in \NN} q_n\,=\,p$.
 So $\distrone \in \setreddists{p}{\distrtypone,\seone}$.

  \item The inclusion $\setredtmsfin{p}{\distrtypone,\seone} \,\subseteq\,\bigcap_{0<q<p}\ \setredtmsfin{q}{\distrtypone,\seone}$ 
 is again an immediate consequence of Lemma~\ref{lemma/downward-closure-tred}.
 Let $\termone \in \bigcap_{0<q<p}\ \setredtmsfin{q}{\distrtypone,\seone}$. We need to prove that 
 $\termone \in \setredtmsfin{p}{\distrtypone,\seone}$, that is, that for every 
 $0 \leq r < p$ there exists $\distrtyptwo_r \distrleq \distrtypone$, $n_r \in \NN$, $\distrone_r$ such that
 $\termone \redval^{n_r} \distrone_r$  and that $\distrone_r \in \setreddists{r}{\distrtyptwo_r,\seone}$.
 Let $r \in [0,p)$. Since $\termone \in \setredtmsfin{\frac{p+r}{2}}{\distrtypone,\seone}$ and that $\frac{p+r}{2} > r$,
 we obtain the desired $\distrtyptwo_r \distrleq \distrtypone$, $n_r \in \NN$, $\distrone_r$ having the properties of interest.
 The result follows.
\end{itemize}

\end{proof}
}

\longv{ 
\subsection{Reducibility Sets and Sizes}}

\longv{In this subsection, we show how the sizes appearing in the (sized or distribution) type 
parametrizing a reducibility set relate with the interpretation of size variables contained
in the size environment which also parametrizes it. We prove first the following lemma,
which will be used as a companion for this result:}

\longv{
\begin{lemma}[Commuting Sizes with Environments]
\label{lemma/size-substitution-and-size-environments}
Let $\sizevarone$ be a size variable, $\sizeone,\,\sizetwo$ be two
sizes, and $\seone$ be a size environment.  Suppose that
$\sizeone\,=\,\sizeinf$ or that $\spine{\sizeone}\neq \sizevarone$.
Then
$\sesem{\subst{\sizetwo}{\sizevarone}{\sizeone}}{\seone}\ =\ \sesem{\sizetwo}{\seone\left[\sizevarone
    \mapsto \sesem{\sizeone}{\seone}\right]}$.
\end{lemma}
}

\longv{
\begin{proof}
 By case analysis.
 \begin{itemize}
  \item If $\sizetwo\,=\,\sizesuccit{\sizevartwo}{n}$ for $\sizevartwo \neq \sizevarone$, then
  $\subst{\sizetwo}{\sizevarone}{\sizeone}\,=\,\sizetwo$ and $\sesem{\sizetwo}{\seone}\,=\,\seone(\sizevartwo) + n\,=\,
  \sesem{\sizetwo}  {\seone\left[\sizevarone \mapsto \sesem{\sizeone}{\seone}\right]}$.
  \item If $\sizetwo\,=\,\sizesuccit{\sizevarone}{n}$, then
    \begin{itemize}
     \item if $\sizeone\,=\,\sizesuccit{\sizevartwo}{m}$ for $\sizevartwo \neq \sizevarone$, then
     $\subst{\sizetwo}{\sizevarone}{\sizeone}\,=\,\sizesuccit{\sizevartwo}{n+m}$ and
     $$
     \sesem{\subst{\sizetwo}{\sizevarone}{\sizeone}}{\seone}\,=\,\seone(\sizevartwo) + n + m 
     \,=\,\sesem{\sizesuccit{\sizevartwo}{m}}{\seone} + n
     \,=\,\sesem{\sizeone}{\seone} + n
     \,=\,\sesem{\sizesuccit{\sizevarone}{n}}{\seone\left[\sizevarone \mapsto \sesem{\sizeone}{\seone}\right]}
     \, =\, \sesem{\sizetwo}{\seone\left[\sizevarone \mapsto \sesem{\sizeone}{\seone}\right]},
     $$
     \item if $\sizeone\,=\,\sizeinf$, then $\subst{\sizetwo}{\sizevarone}{\sizeone}\,=\,\sizeinf$ and 
      $\sesem{\subst{\sizetwo}{\sizevarone}{\sizeone}}{\seone}\ =\ \sizeinf 
      \ =\ \sesem{\sizetwo}{\seone\left[\sizevarone \mapsto \sesem{\sizeone}{\seone}\right]}$.
    \end{itemize}
  \item If $\sizetwo\,=\,\sizeinf$, then $\subst{\sizetwo}{\sizevarone}{\sizeone}\,=\,\sizetwo$ and 
  $\sesem{\sizetwo}{\seone}\,=\,\sizeinf\,=\,\sesem{\sizetwo}{\seone\left[\sizevarone \mapsto \sesem{\sizeone}{\seone}\right]}$.
 \end{itemize}

\end{proof}

}

The last fundamental property about reducibility sets which will be crucial
to treat the recursive case is the following, stating that the sizes appearing
in a sized type may be recovered in the reducibility set by using an appropriate
semantics of the size variables, and conversely:
%
\begin{lemma}[Size Commutation]\label{lemma:exchange-size-size-env}
Let $\sizevarone$ be a size variable, $\sizeone$ be a size such that
$\sizeone\,=\,\sizeinf$ or that $\spine{\sizeone}\neq \sizevarone$ and
$\seone$ be a size environment. Then:
\begin{varitemize}
 \item $\setredvals{p}{\subst{\typone}{\sizevarone}{\sizeone}, \seone}\ =\ \setredvals{p}{\typone, 
 \seone\left[\sizevarone \mapsto \sesem{\sizeone}{\seone}\right]}$,
 \item $\setreddists{p}{\subst{\distrtypone}{\sizevarone}{\sizeone}, \seone}\ =\ \setreddists{p}{\distrtypone, 
 \seone\left[\sizevarone \mapsto \sesem{\sizeone}{\seone}\right]}$,
 \item $\setredtms{p}{\subst{\distrtypone}{\sizevarone}{\sizeone}, \seone}\ =\ \setredtms{p}{\distrtypone, 
 \seone\left[\sizevarone \mapsto \sesem{\sizeone}{\seone}\right]}$.
\end{varitemize}
\end{lemma}

\longv{
\begin{proof}~

\begin{itemize}
 \item The first case to consider is $\typone\,=\,\Nat^{\sizetwo}$ for some size $\sizetwo$.
Using Lemma~\ref{lemma/size-substitution-and-size-environments}, we have that 
 $$
 \begin{array}{rcl}
 \setredvals{p}{\subst{\left(\Nat^\sizetwo\right)}{\sizevarone}{\sizeone},\seone} 
 & \ \ =\ \  & \setredvals{p}{\Nat^{\subst{\sizetwo}{\sizevarone}{\sizeone}},\seone}\\
 & \ \ =\ \  &
 \left\{\natsucc^n\natzero\sep p>0\Longrightarrow n < \sesem{\subst{\sizetwo}{\sizevarone}{\sizeone}}{\seone}\right\} \\
 & \ \ =\ \  &
 \left\{\natsucc^n\natzero\sep p>0\Longrightarrow n < \sesem{\sizetwo}{\seone\left[\sizevarone \mapsto \sesem{\sizeone}{\seone}\right]}\right\} \\
 & \ \ =\ \  &
 \setredvals{p}{\Nat^\sizetwo,\seone\left[\sizevarone \mapsto \sesem{\sizeone}{\seone}\right]}\\
 \end{array}
$$

\item We then consider the case of the sized type $\typone \typarrow \distrtypone \refines \simpletyptwo \typarrow \simpletypthree$. We have
 \begin{adjustwidth}{-2cm}{-2cm}
 $$
 \begin{array}{rcl}
 \setredvals{p}{\subst{(\typone\typarrow\distrtypone)}{\sizevarone}{\sizeone},\seone}
 & \ \ =\ \ &
  \setredvals{p}{\subst{\typone}{\sizevarone}{\sizeone}\typarrow\subst{\distrtypone}{\sizevarone}{\sizeone},\seone}\\ 
 & \ \ =\ \ &
 \left\{\valone \in \settypedclosedvalues{\underlying{\subst{\typone}{\sizevarone}{\sizeone}\typarrow\subst{\distrtypone}{\sizevarone}{\sizeone}}} 
    \sep\forall q \in (0,1],\ \ \forall\valtwo\in\setredvals{q}{\subst{\typone}{\sizevarone}{\sizeone},\seone},\ \ 
   \valone\ \valtwo\in\setredtms{pq}{\subst{\distrtypone}{\sizevarone}{\sizeone},\seone}\right\}\\
    & \ \ =\ \ &
 \left\{\valone \in \settypedclosedvalues{\underlying{\typone \typarrow \distrtypone}} 
    \sep\forall q \in (0,1],\ \ \forall\valtwo\in\setredvals{q}{\typone,\seone\left[\sizevarone \mapsto \sesem{\sizeone}{\seone}\right]},\ \ 
   \valone\ \valtwo\in\setredtms{pq}{\distrtypone,\seone\left[\sizevarone \mapsto \sesem{\sizeone}{\seone}\right]}\right\}\\
   & \ \ =\ \ &
    \setredvals{p}{\typone\typarrow\distrtypone,\seone \left[\sizevarone \mapsto \sesem{\sizeone}{\seone}\right]}
 \end{array}
 $$
  \end{adjustwidth}
 where we used the induction hypothesis twice, once on $\simpletyptwo$ and the other time on $\simpletypthree$.

\item Let $\distrone$ be a finite distribution of values
and $\distrtypone\,=\,\distrelts{\left(\typone_\indextwo\right)^{p'_\indextwo} \sep \indextwo \in \indexsettwo}$ be a distribution type.
We have that $\subst{\distrtypone}{\sizevarone}{\sizeone}\,=\,\distrelts{\left(\subst{\typone_\indextwo}{\sizevarone}{\sizeone}\right)^{p'_\indextwo} \sep \indextwo \in \indexsettwo}$.
Suppose that $\distrone \in \setreddists{p}{\subst{\distrtypone}{\sizevarone}{\sizeone}, \seone}$.
Then there exist a pseudo-representation
$\distrone\,=\,\pseudorep{\left(\valone_{\indexone}\right)^{p_{\indexone}} \sep \indexone \in \indexsetone}$
and families $\left(p_{\indexone \indextwo}\right)_{\indexone\in \indexsetone,\indextwo\in\indexsettwo}$
  and $\left(q_{\indexone \indextwo}\right)_{\indexone\in \indexsetone,\indextwo\in\indexsettwo}$ of reals of $[0,1]$
  satisfying:
  \begin{enumerate}
   \item $\forall \indexone \in \indexsetone,\ \ \forall\indextwo \in \indexsettwo,\ \ \valone_\indexone \in 
   \setredvals{q_{\indexone\indextwo}}{\subst{\typone_\indextwo}{\sizevarone}{\sizeone},\seone}$,
   \item $\forall \indexone \in \indexsetone,\ \ \sum_{\indextwo \in \indexsettwo}\ p_{\indexone\indextwo}\ =\ p_\indexone$,
   \item $\forall \indextwo \in \indexsettwo,\ \ \sum_{\indexone \in \indexsetone}\ p_{\indexone\indextwo}\ =\ \distrtypone(\typone_\indextwo)$,
   \item $p \leq \sum_{\indexone \in \indexsetone}\sum_{\indextwo \in \indexsettwo}\ q_{\indexone\indextwo}p_{\indexone\indextwo}$.
  \end{enumerate}
  But $(1)$ is equivalent to $\forall \indexone \in \indexsetone,\ \ \forall\indextwo \in \indexsettwo,\ \ \valone_\indexone \in 
   \setredvals{q_{\indexone\indextwo}}{\typone_\indextwo,\seone\left[\sizevarone \mapsto \sesem{\sizeone}{\seone}\right]}$ by induction hypothesis.
   It follows that $\distrone \in \setreddists{p}{\distrtypone, \seone\left[\sizevarone \mapsto \sesem{\sizeone}{\seone}\right]}$.
   The converse direction proceeds in the exact same way.

   \item Then, $\termone \in \setredtmsfin{p}{\subst{\distrtypone}{\sizevarone}{\sizeone},\seone}$ if and only if
  $$
  \termone \in \settypedclosedterms{\underlying{\distrtypone}} \text{ and }
 \forall 0 \leq r < p,\ \ \exists \distrtyptwo_r \distrleq \distrtypone,\ \ \exists n_r \in \NN,\ \ 
 \termone \redval^{n_r} \distrone_r \text{ and } \distrone_r \in \setreddists{r}{\subst{\distrtyptwo_r}{\sizevarone}{\sizeone},\seone}
 $$
 if and only if, by induction hypothesis,
 $$
 \termone \in \settypedclosedterms{\underlying{\distrtypone}} \text{ and }
 \forall 0 \leq r < p,\ \ \exists \distrtyptwo_r \distrleq \distrtypone,\ \ \exists n_r \in \NN,\ \ 
 \termone \redval^{n_r} \distrone_r \text{ and } \distrone_r \in \setreddists{r}{\distrtyptwo_r,
 \seone\left[\sizevarone \mapsto \sesem{\sizeone}{\seone}\right]}
 $$
 that is, if and only if
 $\termone \in \setredtmsfin{p}{\distrtypone,\seone\left[\sizevarone \mapsto \sesem{\sizeone}{\seone}\right]}$.
  \end{itemize}
\end{proof}
}

\longv{ 
\subsection{Reducibility Sets are Stable by Unfoldings}}
\shortv{\paragraph{Unfoldings.}}
The most difficult step in proving all typable terms 
to be reducible is, unexpectedly, proving that terms 
involving \emph{recursion} are reducible whenever their respective \emph{unfoldings}
are. This very natural concept expresses simply that
any term in the form $\letrec{\funcone}{\valtwo}$ is
assumed to compute the fixpoint of the function
defined by $\valtwo$.
\begin{definition}[$n$-Unfolding]
Suppose that $\valone\,=\,\left(\letrec{\funcone}{\valtwo}\right)$ is
closed, then the $n$-\emph{unfolding} of $\valone$ is:
 \begin{varitemize}
 \item 
   $\valone$ if $n=0$;
 \item 
   $\subst{\valtwo}{\funcone}{\valthree}$ if $n=m+1$ and $\valthree$
   is the $m$-unfolding of $\valone$.
 \end{varitemize}
 We write the set of unfoldings of $\valone$ as
 $\unfoldings{\valone}$. Note that if $\valone$ admits a simple type, then
 all its unfoldings have this same simple type as well. In the sequel,
 we implicitly consider that $\valone$ is simply typed.
\end{definition}
Any unfolding of $\valone=\left(\letrec{\funcone}{\valtwo}\right)$ should behave
like $\valone$ itself: all unfoldings of $\valone$ should be
equivalent. This, however, cannot be proved using simply the operational
semantics. It requires some work, and techniques akin to logical relations,
to prove
this behavioural equivalence between a recursive definition and its unfoldings.
\longv{The first lemma is technical and lists the unfoldings of terms defined recursively
as equal to themselves or to a variable:
\begin{lemma}~
\label{lemma:unfolding-of-a-variable}
\begin{itemize}
 \item Let $\valone\,=\,\funcone$ and $\valtwo \in \unfoldings{\letrec{\funcone}{\valone}}$. Then $\valtwo\,=\,\letrec{\funcone}{\valone}$.
 \item Let $\valone\,=\,\varone\neq\funcone$ and $\valtwo \in \unfoldings{\letrec{\funcone}{\valone}}$. Then $\valtwo\,=\,\letrec{\funcone}{\valone}$
 or $\valtwo\,=\,\varone$. More precisely, the $n$-unfoldings for $n \geq 1$ are all $\varone$.
\end{itemize}
\end{lemma}
}

\longv{The next lemma is the technical core of this section. Think of two terms as related when they are of
the shape $\subst{\termone}{\vec{\varone}}{\vec{\valthree}}$ and
$\subst{\termone}{\vec{\varone}}{\vec{\valthree'}}$, where $\vec{\varone}$
is a sequence of ``holes'' in $\termone$, filled with unfoldings from a \emph{same}
recursively-defined term. Then their rewritings by $\rcbv$ form distributions of pairwise related terms.}
\longv{
\begin{lemma}
\label{lemma:one-step-unfolding}
 Let $\valone\,=\,\left(\letrec{\funcone}{\valtwo}\right)$ be a closed value.
 Let $\vec{\varone},\,\vec{\valthree},\,\vec{\valthree'}$ be a vector of variables and two vectors of terms of $\unfoldings{\valone}$,
 all of the same length.
 Let $\termone$ be a simply-typed term of free variables contained in $\vec{\varone}$, all typed with the simple type of $\valone$.
 Suppose that $\subst{\termone}{\vec{\varone}}{\vec{\valthree}} \rcbv \distrone$.
 Then there exists $\termtwo_1,\ldots,\termtwo_n$, a vector of variables $\vec{\vartwo}$ 
 and $\vec{\valthree_1},\ldots,\vec{\valthree_n},\vec{\valthree'_1},\ldots,\vec{\valthree'_n} \in \unfoldings{\valone}$
 of the same length as $\vec{\vartwo}$ and such that 
 $\distrone\ =\ \distrelts{\left(\subst{\termtwo_{\indexone}}{\vec{\vartwo}}{\vec{\valthree_\indexone}}\right)^{p_\indexone}}$
 and moreover $\subst{\termone}{\vec{\varone}}{\vec{\valthree'}} \rcbv \distrtwo\ =\ \distrelts{\left(\subst{\termtwo_{\indexone}}{\vec{\vartwo}}{\vec{\valthree'_\indexone}}\right)^{p_\indexone}}$.
\end{lemma}
}

\longv{ 
\begin{proof}
 We prove the result by induction on the structure of $\termone$.
 \begin{itemize}
  \item The case where $\termone$ is a variable cannot fit in this setting: either $\termone=\vartwo \notin \vec{\varone}$
  and there is no reduction from $\subst{\termone}{\vec{\varone}}{\vec{\valthree}}$, or $\termone=\varone_\indexone \in \vec{\varone}$
  and there is no reduction either from $\subst{\termone}{\vec{\varone}}{\vec{\valthree}}\,=\,\valthree_\indexone$ since it is a value.
  We can similarly rule out all the cases where $\termone$ is a value.
  \item Suppose that $\termone\,=\,\valone_1\ \valone_2$. We proceed by case exhaustion on $\valone_1$. Three possibilities exist,
  the other ones contradicting the fact that there should be a reduction step from $\termone$:
  \begin{itemize}
   \item If $\valone_1\,=\,\varone_\indexone \in \vec{\varone}$, 
   we distinguish four cases:
   \begin{itemize}
    \item Suppose that $\valthree_\indexone\,=\,\valthree'_\indexone$ are both the $0$-unfolding of $\valone$.
    Then $\subst{\termone}{\vec{\varone}}{\vec{\valthree}}\,=\,\subst{\termone}{\vec{\varone}}{\vec{\valthree'}}$ and the result 
    follows immediately from:
    $$
    \begin{array}{rcl}
    \subst{\termone}{\vec{\varone}}{\vec{\valthree}} & \ \ =\ \ & 
    (\letrec{\funcone}{\valtwo})\ \subst{\valone_2}{\vec{\varone}}{\vec{\valthree}}\\
    & \ \ =\ \ & 
    (\letrec{\funcone}{\valtwo})\ (\natsucc^m\ 0)\\
    & \ \ \rcbv\ \ & 
    \distrelts{\left((\subst{\valtwo}{\funcone}{\letrec{\funcone}{\valtwo}})\ (\natsucc^m\ 0)\right)^1}\\
    \end{array}
    $$
    where in the second line the shape of $\valone_2$ needs to be $\natsucc^m\ \natzero$ by typing constraints.
    Note that $\vec{\vartwo}$ is the empty vector here.
    \item Suppose that $\valthree_\indexone$ is the $n$-unfolding of $\valone$ for $n > 0$, and that 
    $\valthree'_\indexone$ is the $0$-unfolding.
    We have that 
    $$
    \subst{\termone}{\vec{\varone}}{\vec{\valthree}}\ \ =\ \ \subst{\valtwo}{\funcone}{\valthree''}\ \ \subst{\valone_2}{\vec{\varone}}{\vec{\valthree}}
    $$
    where $\valthree''$ is the $(n-1)$-unfolding of $\valone$, and that
    $$
    \begin{array}{rcl}
    \subst{\termone}{\vec{\varone}}{\vec{\valthree'}} & \ \ =\ \ & 
    \left(\letrec{\funcone}{\valtwo}\right)\ \subst{\valone_2}{\vec{\varone}}{\vec{\valthree'}}\\
    & \ \ \rcbv\ \ & \distrelts{\left(\left(\subst{\valtwo}{\funcone}{\letrec{\funcone}{\valtwo}}\right)\
    \subst{\valone_2}{\vec{\varone}}{\vec{\valthree'}}\right)^1}\\
    \end{array}
    $$
    Notice that this reduction is possible since the constraint of simple typing implies that $\valone_2$ is of the shape
    $\natsucc^m\ \natzero$ for some $m \geq 0$. We can therefore rewrite the two terms as
    $$
    \subst{\termone}{\vec{\varone}}{\vec{\valthree}}\ \ =\ \ \subst{\valtwo}{\funcone}{\valthree''}\ \ (\natsucc^m\ \natzero)
    $$
    and%
    $$
    \begin{array}{rcl}
    \subst{\termone}{\vec{\varone}}{\vec{\valthree'}}
    & \ \ \rcbv\ \ & \distrelts{\left(\left(\subst{\valtwo}{\funcone}{\letrec{\funcone}{\valtwo}}\right)\
    (\natsucc^m\ \natzero)\right)^1}\\
    \end{array}
    $$
    We need to distinguish four cases, depending on the structure of $\valtwo$.
    \begin{itemize}
     \item Suppose that $\valtwo$ is a variable different from $\funcone$. Then by Lemma~\ref{lemma:unfolding-of-a-variable}
     there cannot be a step of reduction from $\subst{\termone}{\vec{\varone}}{\vec{\valthree}}$.
     \item Suppose that $\valtwo\,=\,\funcone$. Then by Lemma~\ref{lemma:unfolding-of-a-variable}
     we have $\valthree_\indexone=\valthree'_\indexone=\valthree''$, so that
     $\subst{\termone}{\vec{\varone}}{\vec{\valthree}}\,=\,\subst{\termone}{\vec{\varone}}{\vec{\valthree'}}$
     and the result follows just as for the case where both $\valthree$ and $\valthree'$ were $0$-unfoldings.
     \item Suppose that $\valtwo\,=\,\abstr{\vartwo}{\termthree}$. Then 
    $$
    \begin{array}{rcl}
    \subst{\termone}{\vec{\varone}}{\vec{\valthree}} & \ \ =\ \ & \left(\abstr{\vartwo}{\subst{\termthree}{\funcone}{\valthree''}}\right)\ 
    \left(\natsucc^m\ \natzero\right)\\
    & \ \ \rcbv\ \ & \distrelts{\left(\subst{\subst{\termthree}{\funcone}{\valthree''}}{\vartwo}{\natsucc^m\ \natzero} \right)^1}\\
    & \ \ =\ \ & \distrelts{\left(\subst{\left(\subst{\termthree}{\vartwo}{\natsucc^m\ \natzero}\right)}{\funcone}{\valthree''} \right)^1}\\
    \end{array}
    $$
    Moreover,
    $$
    \begin{array}{rcl}
    \subst{\termone}{\vec{\varone}}{\vec{\valthree'}}
    & \ \ \rcbv\ \ & \distrelts{\left(\left(\subst{\valtwo}{\funcone}{\letrec{\funcone}{\valtwo}}\right)\
    (\natsucc^m\ \natzero)\right)^1}\\
       & \ \ =\ \ & \distrelts{\left(\subst{\left((\abstr{\vartwo}{\termthree}\right)\
    (\natsucc^m\ \natzero))}{\funcone}{\letrec{\funcone}{\valtwo}}\right)^1}\\
    &  \ \ \rcbv\ \ &\distrelts{\left(\subst{\left(\subst{\termthree}{\vartwo}{(\natsucc^m\ \natzero)}\right)}{\funcone}{\letrec{\funcone}{\valtwo}}\right)^1}\\
    \end{array}
    $$
    so that we can conclude with $\vec{y}\,=\,\funcone$ and $\termtwo_1\ =\ \subst{\termthree}{\vartwo}{(\natsucc^m\ \natzero)}$.
    \item Suppose that $\valtwo\ =\ \letrec{\functwo}{\valtwo'}$. Then 
    $$
    \begin{array}{rcl}
    \subst{\termone}{\vec{\varone}}{\vec{\valthree}} & \ \ =\ \ & \left(\subst{(\letrec{\functwo}{\valtwo'})}{\funcone}{\valthree''}\right)\ 
    \left(\natsucc^m\ \natzero\right)\\
    & \ \ \rcbv\ \ & \distrelts{\left(\subst{\subst{\valtwo'}{\functwo}{\letrec{\functwo}{\valtwo'}})}{\funcone}{\valthree''}\ \left(\natsucc^m\ \natzero\right) \right)^1}\\
    & \ \ =\ \ & \distrelts{\left(\subst{\subst{\valtwo'}{\functwo}{\letrec{\functwo}{\valtwo'}} \ \left(\natsucc^m\ \natzero\right))}{\funcone}{\valthree''} \right)^1}\\
    \end{array}
    $$
    Moreover,
    $$
    \begin{array}{rcl}
    \subst{\termone}{\vec{\varone}}{\vec{\valthree'}} & \ \ \rcbv\ \ & \left(\subst{(\letrec{\functwo}{\valtwo'})}{\funcone}{\letrec{\funcone}{\valtwo}}\right)\ 
    \left(\natsucc^m\ \natzero\right)\\
    & \ \ \rcbv\ \ & \distrelts{\left(\subst{\subst{\valtwo'}{\functwo}{\letrec{\functwo}{\valtwo'}})}{\funcone}{\letrec{\funcone}{\valtwo}}\ \left(\natsucc^m\ \natzero\right) \right)^1}\\
    & \ \ =\ \ & \distrelts{\left(\subst{\subst{\valtwo'}{\functwo}{\letrec{\functwo}{\valtwo'}} \ \left(\natsucc^m\ \natzero\right))}{\funcone}{\letrec{\funcone}{\valtwo}} \right)^1}\\
    \end{array}
    $$
    so that we can conclude with $\vec{y}\,=\,\funcone$ and $\termtwo_1\ =\ \subst{\valtwo'}{\functwo}{\letrec{\functwo}{\valtwo'}} \ \left(\natsucc^m\ \natzero\right))$.

    \end{itemize}
     \item Suppose that $\valthree_\indexone$ is the $0$-unfolding of $\valone$, and that $\valthree'_\indexone$ is the $n$-unfolding for $n > 0$.
     Again, the constraint of simple typing implies that $\valone_2$ is of the shape $\natsucc^m\ \natzero$ for some $m \geq 0$.
    We have that 
    $$
    \begin{array}{rcl}
    \subst{\termone}{\vec{\varone}}{\vec{\valthree}}
    & \ \ \rcbv\ \ & \distrelts{\left(\left(\subst{\valtwo}{\funcone}{\letrec{\funcone}{\valtwo}}\right)\
    (\natsucc^m\ \natzero)\right)^1}\\
    & \ \ =\ \ & \distrelts{\left(\subst{(\valtwo\ (\natsucc^m\ \natzero))}{\funcone}{\letrec{\funcone}{\valtwo}})\right)^1}\\
    \end{array}
    $$
    and that
    $$
    \subst{\termone}{\vec{\varone}}{\vec{\valthree'}}\ \ =\ \ \subst{\valtwo}{\funcone}{\valthree''}\ \ (\natsucc^m\ \natzero)
    \ \ =\ \ \subst{(\valtwo\ \ (\natsucc^m\ \natzero))}{\funcone}{\valthree''}
    $$
    where $\valthree''$ is the $(n-1)$-unfolding of $\valone$, so that we can conclude with $\vec{\vartwo} = \funcone$
    and $\termtwo_1\ =\ \valtwo\ \ (\natsucc^m\ \natzero)$.
    \item Suppose that $\valthree_\indexone$ is the $n$-unfolding of $\valone$ for $n > 0$,
    and that $\valthree'_\indexone$ is the $n'$-unfolding for $n' > 0$. We have 
    $$
    \subst{\termone}{\vec{\varone}}{\vec{\valthree}}\ \ =\ \ \subst{\valtwo}{\funcone}{\valthree''}\ \ \subst{\valone_2}{\vec{\varone}}{\vec{\valthree}}
    $$
    where $\valthree''$ is the $(n-1)$-unfolding of $\valone$, and
    $$
    \subst{\termone}{\vec{\varone}}{\vec{\valthree}}\ \ =\ \ \subst{\valtwo}{\funcone}{\valthree'''}\ \ \subst{\valone_2}{\vec{\varone}}{\vec{\valthree}}
    $$
    where $\valthree''$ is the $(n'-1)$-unfolding of $\valone$.
    We proceed by case analysis on $\valtwo$. As we discussed in the case where $\valthree_\indexone$ was a $(n''+1)$-unfolding
    and $\valthree'_\indexone$ a $0$-unfolding, the case where $\valtwo$ is a variable does not lead to a rewriting step.
    It remains to treat two cases:
    \begin{itemize}
     \item Suppose that $\valtwo\,=\,\abstr{\vartwo}{\termthree}$. Then
    $$
    \begin{array}{rcl}
    \subst{\termone}{\vec{\varone}}{\vec{\valthree}}
    & \ \ =\ \ &\subst{\abstr{\vartwo}{\termthree}}{\funcone}{\valthree''}\ \ \subst{\valone_2}{\vec{\varone}}{\vec{\valthree}}\\
     & \ \ \rcbv\ \ & \distrelts{\left(\subst{\subst{\termthree}{\funcone}{\valthree''}}{\vartwo}{\subst{\valone_2}{\vec{\varone}}{\vec{\valthree}}} \right)^1}\\
    & \ \ =\ \ & \distrelts{\left(\subst{\left(\subst{\termthree}{\vartwo}{\subst{\valone_2}{\vec{\varone}}{\vec{\valthree}}}\right)}{\funcone}{\valthree''} \right)^1}\\
    \end{array}
    $$
    and
    $$
    \begin{array}{rcl}
    \subst{\termone}{\vec{\varone}}{\vec{\valthree'}}
    & \ \ =\ \ &\subst{\abstr{\vartwo}{\termthree}}{\funcone}{\valthree'''}\ \ \subst{\valone_2}{\vec{\varone}}{\vec{\valthree}}\\
     & \ \ \rcbv\ \ & \distrelts{\left(\subst{\subst{\termthree}{\funcone}{\valthree'''}}{\vartwo}{\subst{\valone_2}{\vec{\varone}}{\vec{\valthree}}} \right)^1}\\
    & \ \ =\ \ & \distrelts{\left(\subst{\left(\subst{\termthree}{\vartwo}{\subst{\valone_2}{\vec{\varone}}{\vec{\valthree}}}\right)}{\funcone}{\valthree'''} \right)^1}\\
    \end{array}
    $$
    so that we can conclude with $\vec{\vartwo} = \funcone$ and $\termtwo_1\ =\ \subst{\termthree}{\vartwo}{\subst{\valone_2}{\vec{\varone}}{\vec{\valthree}}}$.
    \item Suppose that $\valtwo\,=\,\letrec{\functwo}{\valtwo'}$. Then 
    $$
    \begin{array}{rcl}
    \subst{\termone}{\vec{\varone}}{\vec{\valthree}} & \ \ =\ \ & \left(\subst{(\letrec{\functwo}{\valtwo'})}{\funcone}{\valthree''}\right)\ 
    \subst{\valone_2}{\vec{\varone}}{\vec{\valthree}}\\
    & \ \ \rcbv\ \ & \distrelts{\left(\subst{\subst{\valtwo'}{\functwo}{\letrec{\functwo}{\valtwo'}})}{\funcone}{\valthree''}\ \subst{\valone_2}{\vec{\varone}}{\vec{\valthree}} \right)^1}\\
    & \ \ =\ \ & \distrelts{\left(\subst{\subst{\valtwo'}{\functwo}{\letrec{\functwo}{\valtwo'}} \ \subst{\valone_2}{\vec{\varone}}{\vec{\valthree}})}{\funcone}{\valthree''} \right)^1}\\
    \end{array}
    $$
    where the reduction is possible because the simple typing constraints imply that $\subst{\valone_2}{\vec{\varone}}{\vec{\valthree}}$
    is of the shape $\natsucc^m\ \natzero$ for some $m \in \NN$. Moreover,
    $$
    \begin{array}{rcl}
    \subst{\termone}{\vec{\varone}}{\vec{\valthree'}} & \ \ =\ \ & \left(\subst{(\letrec{\functwo}{\valtwo'})}{\funcone}{\valthree'''}\right)\ 
    \subst{\valone_2}{\vec{\varone}}{\vec{\valthree}}\\
    & \ \ \rcbv\ \ & \distrelts{\left(\subst{(\subst{\valtwo'}{\functwo}{\letrec{\functwo}{\valtwo'}})}{\funcone}{\valthree'''}\ \subst{\valone_2}{\vec{\varone}}{\vec{\valthree}} \right)^1}\\
    & \ \ =\ \ & \distrelts{\left(\subst{(\subst{\valtwo'}{\functwo}{\letrec{\functwo}{\valtwo'}} \ \subst{\valone_2}{\vec{\varone}}{\vec{\valthree}})}{\funcone}{\valthree'''} \right)^1}\\
    \end{array}
    $$
    so that we can conclude with $\vec{\vartwo}=\funcone$ and 
    $\termtwo_1\ =\ \subst{\valtwo'}{\functwo}{\letrec{\functwo}{\valtwo'}} \ \subst{\valone_2}{\vec{\varone}}{\vec{\valthree}}$.
    \end{itemize}
  \end{itemize}

  \item If $\valone_1\ =\ \abstr{\vartwo}{\termthree}$,
  $$
  \begin{array}{rcl}
  \subst{\termone}{\vec{\varone}}{\vec{\valthree}} & \ \ =\ \ &
  \left(\abstr{\vartwo}{\subst{\termthree}{\vec{\varone}}{\vec{\valthree}}\right)}\ \subst{\valone_2}{\vec{\varone}}{\vec{\valthree}}\\
  & \ \ \rcbv\ \ &
  \distrelts{\left(\subst{\subst{\termthree}{\vec{\varone}}{\vec{\valthree}}}{\vartwo}{\subst{\valone_2}{\vec{\varone}}{\vec{\valthree}}}\right)^1}\\
  & \ \ =\ \ &
  \distrelts{\left(\subst{\subst{\termthree}{\vartwo}{\valone_2}}{\vec{\varone}}{\vec{\valthree}}\right)^1}\\
  \end{array}
  $$
  and in the same way
  $$
  \begin{array}{rcl}
  \subst{\termone}{\vec{\varone}}{\vec{\valthree'}} & \ \ \rcbv\ \ &
  \distrelts{\left(\subst{\subst{\termthree}{\vartwo}{\valone_2}}{\vec{\varone}}{\vec{\valthree'}}\right)^1}\\
  \end{array}
  $$
  which allows to conclude with $\termtwo_1\ =\ \subst{\termthree}{\vartwo}{\valone_2}$.
  \item If $\valone_1\ =\ \letrec{\functwo}{\valtwo'}$, by typing constraints $\valone_2\ =\ \natsucc^m\ \natzero$ for some $m \geq 0$.
  It follows that we can reduce $\subst{\termone}{\vec{\varone}}{\vec{\valthree}}$ and
  $\subst{\termone}{\vec{\varone}}{\vec{\valthree'}}$ as follows:
  $$
  \begin{array}{rcl}
  \subst{\termone}{\vec{\varone}}{\vec{\valthree}} & \ \ =\ \ &
  \left(\letrec{\functwo}{\subst{\valtwo'}{\vec{\varone}}{\vec{\valthree}}\right)}\  \natsucc^m\ \natzero\\
  & \ \ \rcbv\ \ &
  \distrelts{\left(\left(\subst{\subst{\valtwo'}{\vec{\varone}}{\vec{\valthree}}}{\functwo}{\letrec{\functwo}{\subst{\valtwo'}{\vec{\varone}}{\vec{\valthree}}}} \right) \  (\natsucc^m\ \natzero)\right)^1}\\
  & \ \ =\ \ &
  \distrelts{\left(\subst{\left(\subst{\valtwo'}{\functwo}{\letrec{\functwo}{\valtwo'}}\right)}{\vec{\varone}}{\vec{\valthree}} \  (\natsucc^m\ \natzero)\right)^1}\\
  & \ \ =\ \ &
  \distrelts{\left(\subst{\left(\subst{\valtwo'}{\functwo}{\letrec{\functwo}{\valtwo'}} \  (\natsucc^m\ \natzero)\right)}{\vec{\varone}}{\vec{\valthree}}\right)^1}\\
  \end{array}
  $$
  and similarly
  $$
  \begin{array}{rcl}
  \subst{\termone}{\vec{\varone}}{\vec{\valthree'}} & \ \ \rcbv\ \ &
  \distrelts{\left(\subst{\left(\subst{\valtwo'}{\functwo}{\letrec{\functwo}{\valtwo'}} \  (\natsucc^m\ \natzero)\right)}{\vec{\varone}}{\vec{\valthree'}}\right)^1}\\
  \end{array}
  $$
  so that we can conclude with $\termtwo_1\ =\ \subst{\valtwo'}{\functwo}{\letrec{\functwo}{\valtwo'}} \  (\natsucc^m\ \natzero)$.
  \end{itemize}
  
  \item Suppose that $\termone\ =\ \letin{\vartwo}{\valfour}{\termfour}$. Then
  $$
  \begin{array}{rcl}
   \subst{\termone}{\vec{\varone}}{\vec{\valthree}} & \ \ =\ \ &
   \letin{\vartwo}{\subst{\valfour}{\vec{\varone}}{\vec{\valthree}}}{\subst{\termfour}{\vec{\varone}}{\vec{\valthree}}} \\
   & \ \ \rcbv\ \ &
   \distrelts{\left(\subst{\subst{\termfour}{\vec{\varone}}{\vec{\valthree}}}{\vartwo}{\subst{\valfour}{\vec{\varone}}{\vec{\valthree}}}\right)^1}\\
   & \ \ =\ \ &
   \distrelts{\left(\subst{\subst{\termfour}{\vartwo}{\valfour}}{\vec{\varone}}{\vec{\valthree}}\right)^1}\\
  \end{array}
  $$
  and similarly $\subst{\termone}{\vec{\varone}}{\vec{\valthree'}} \rcbv \ \distrelts{\left(\subst{\subst{\termfour}{\vartwo}{\valfour}}{\vec{\varone}}{\vec{\valthree'}}\right)^1}$
  from which we can conclude.
  
  \item Suppose that $\termone\ =\ \letin{\vartwo}{\termthree}{\termfour}$ and that 
  $$
  \begin{array}{rcl}
  \subst{\termone}{\vec{\varone}}{\vec{\valthree}} 
  &\ \ =\ \ &
  \letin{\vartwo}{\subst{\termthree}{\vec{\varone}}{\vec{\valthree}}}{\subst{\termfour}{\vec{\varone}}{\vec{\valthree}}}\\
  &\ \ \rcbv \ \ &\distrelts{\left(\letin{\vartwo}{\subst{\termthree'_\indexone}{\vec{\varone}}{\vec{\valthree}}}{\subst{\termfour}{\vec{\varone}}{\vec{\valthree}}}\right)^{p_\indexone}}\\
  &\ \ = \ \ &\distrelts{\left(\letin{\vartwo}{\subst{\termthree''_\indexone}{\vec{\varthree}}{\vec{\valthree}}}{\subst{\termfour}{\vec{\varone}}{\vec{\valthree}}}\right)^{p_\indexone}}\\
    &\ \ = \ \ &\distrelts{\left(\subst{\left(\letin{\vartwo}{\termthree''_\indexone}{\termfour}\right)}{\vec{\varthree},\vec{\varone}}{\vec{\valthree},\vec{\valthree}}\right)^{p_\indexone}}\\
  \end{array}
  $$
  where the third step is obtained by $\alpha$-renaming, and where by definition of $\rcbv$ we have
  $$
  \subst{\termthree}{\vec{\varone}}{\vec{\valthree}} \ \ \rcbv\ \ \distrelts{(\subst{\termthree'_\indexone}{\vec{\varone}}{\vec{\valthree}})^{p_\indexone} \sep \indexone \in \indexsetone}.
  $$
  By induction hypothesis, there exists $\vec{\valthree'_1},\ldots,\vec{\valthree'_n} \in \unfoldings{\valone}$ such that
  $$
  \subst{\termthree}{\vec{\varone}}{\vec{\valthree'}} \ \ \rcbv\ \ \distrelts{(\subst{\termthree'_\indexone}{\vec{\varone}}{\vec{\valthree'_\indexone}})^{p_\indexone} \sep \indexone \in \indexsetone}.
  $$
  Now remark that
  $$
  \begin{array}{rcl}
  \subst{\termone}{\vec{\varone}}{\vec{\valthree'}}& \ \ \rcbv \ \ & 
  \distrelts{\left(\letin{\vartwo}{\subst{\termthree'_\indexone}{\vec{\varone}}{\vec{\valthree'_\indexone}}}{\subst{\termfour}{\vec{\varone}}{\vec{\valthree}}}\right)^{p_\indexone}}\\
  & \ \ = \ \ & 
  \distrelts{\left(\letin{\vartwo}{\subst{\termthree''_\indexone}{\vec{\varthree}}{\vec{\valthree'_\indexone}}}{\subst{\termfour}{\vec{\varone}}{\vec{\valthree}}}\right)^{p_\indexone}}\\
  & \ \ = \ \ & 
  \distrelts{\left(\subst{\left(\letin{\vartwo}{\termthree''_\indexone}{\termfour}\right)}{\vec{\varone},\vec{\varthree}}{\vec{\valthree},\vec{\valthree'_\indexone}} \right)^{p_\indexone}}\\
  \end{array}
  $$
  The result follows for $\vec{\vartwo}\,=\,\vec{\varone},\vec{\varthree}$ and
  $\termtwo_\indexone\ =\ \letin{\vartwo}{\termthree''_\indexone}{\termfour}$.
  
  \item Suppose that $\termone \ =\ \termthree \choice_p \termfour$. Suppose that $\termthree \neq \termfour$. Then 
  $$
  \subst{\termone}{\vec{\varone}}{\vec{\valthree}} \ \ \rcbv\ \ 
  \distrelts{\subst{\termthree}{\vec{\varone}}{\vec{\valthree}}^p,\ \subst{\termfour}{\vec{\varone}}{\vec{\valthree}}^{1-p}}
  $$
  and
  $$
  \subst{\termone}{\vec{\varone}}{\vec{\valthree'}} \ \ \rcbv\ \ 
  \distrelts{\subst{\termthree}{\vec{\varone}}{\vec{\valthree'}}^p,\ \subst{\termfour}{\vec{\varone}}{\vec{\valthree'}}^{1-p}}
  $$
  so that the result holds for $\termtwo_1\,=\,\termthree$ and $\termtwo_2\,=\,\termfour$.
  
  If $\termthree\,=\,\termfour$,%
  $$
  \subst{\termone}{\vec{\varone}}{\vec{\valthree}} \ \ \rcbv\ \ 
  \distrelts{\subst{\termthree}{\vec{\varone}}{\vec{\valthree}}^1}
  $$
  and
  $$
  \subst{\termone}{\vec{\varone}}{\vec{\valthree'}} \ \ \rcbv\ \ 
  \distrelts{\subst{\termthree}{\vec{\varone}}{\vec{\valthree'}}^1}
  $$
  and the result holds as well. Note that the distinction is necessary so as to avoid the use of pseudo-representations in the
  statement of the lemma.
  
  \item Suppose that $\termone\ =\ \casenat{\valone'}{\valfour}{\valfive}$. By typing constraints,
  $\valone'\ =\ \natsucc^m \ \natzero$ or $\valone'\,=\,\vartwo$ is a variable.
  \begin{itemize}
   \item If $\valone'\,=\,\natzero$, 
   $
   \subst{\termone}{\vec{\varone}}{\vec{\valthree}} \ \ \rcbv\ \ \distrelts{\left(\subst{\termfive}{\vec{\varone}}{\vec{\valthree}}\right)^1}
   $
   and
   $
   \subst{\termone}{\vec{\varone}}{\vec{\valthree'}} \ \ \rcbv\ \ \distrelts{\left(\subst{\termfive}{\vec{\varone}}{\vec{\valthree'}}\right)^1}
   $
   so that we can conclude.
   \item If $\valone'\,=\,\natsucc^m\ \natzero$ with $m > 0$, we can conclude in the same way.
   \item In the latter case, there is no reduction from
  $\subst{\termone}{\vec{\varone}}{\vec{\valthree}}$ unless
  $\subst{\valone'}{\vec{\varone}}{\vec{\valthree}}$ is of the shape $\valone'\ =\ \natsucc^m \ \natzero$.
  But this is of type $\Nat$ and cannot therefore be an unfolding of $\valone$, so that this case is impossible.
  \end{itemize}

 \end{itemize}

\end{proof}
}

\longv{This result can be extended to a $n$-step rewriting process; however pseudo-representations are required
to keep the statement true, as we explain in the proof.}

\longv{
\begin{lemma}
\label{lemma:relating-unfoldings-iterated-steps}
 Let $\valone\,=\,\left(\letrec{\funcone}{\valtwo}\right)$ be a closed value.
 Let $\termone$ be a simply-typed term of free variables contained in $\vec{\varone}$, all typed with the simple type of $\valone$.
 Let $\vec{\valthree},\,\vec{\valthree'} \in \unfoldings{\valone}$ and $n \in \NN$.
 Then there exists a distribution of values of pseudo-representation
 $\pseudorep{\valfour_\indexone^{p_\indexone}\sep \indexone \in \indexsetone}$,
 a vector of variables $\vec{\vartwo}$ and families of vectors $\left(\vec{\valthree_\indexone}\right)_{\indexone \in \indexsetone},
 \left(\vec{\valthree'_\indexone}\right)_{\indexone \in \indexsetone}$ of the same length as $\vec{\vartwo}$,
 all such that $\subst{\termone}{\vec{\varone}}{\vec{\valthree}} \redval^n 
 \pseudorep{\left(\subst{\valfour_\indexone}{\vec{\vartwo}}{\vec{\valthree_\indexone}}\right)^{p_\indexone}  \sep \indexone \in \indexsetone}$
 and that $\subst{\termone}{\vec{\varone}}{\vec{\valthree'}} \redval^n \pseudorep{\left(\subst{\valfour_\indexone}{\vec{\vartwo}}{\vec{\valthree'_\indexone}}\right)^{p_\indexone} \sep \indexone \in \indexsetone}$.
\end{lemma}
}

\longv{ 
\begin{proof}
By iteration of Lemma~\ref{lemma:one-step-unfolding}. The pseudo-representations come from the fact that some terms
in different reduction branches may converge to the same value, say, in the reduction from
$\subst{\termone}{\vec{\varone}}{\vec{\valthree}}$ but not in the one from $\subst{\termone}{\vec{\varone}}{\vec{\valthree'}}$.
\end{proof}
}


\longv{The following lemma is of technical interest. It states that, given two pseudo-representations of a distribution
-- one of the shape exhibited in the previous lemmas and used for relating terms with unfoldings, the other one
being a pseudo-representation witnessing the belonging to a set $\setreddists{}{}$ -- there exists a third one
which ``combines'' both:}
\longv{
\begin{lemma}
 \label{lemma:cutting-pseudo-representations}
 Suppose that 
 $\distrone_r\ =\ \pseudorep{\left(\subst{\valfour_\indexone}{\vec{\vartwo}}{\vec{\valthree_\indexone}}\right)^{p_\indexone}  \sep \indexone \in \indexsetone} \ =\ \pseudorep{\left(\valfour'_\indextwo\right)^{p'_\indextwo}\sep \indextwo \in \indexsettwo}$.
 Then there exists a set $\indexsetthree$, two applications $\pi_1\,:\,\indexsetthree \rightarrow \indexsetone$ and 
 $\pi_2\,:\,\indexsetthree \rightarrow \indexsettwo$
 and a pseudo-representation
 $\distrone_r \ =\ \pseudorep{\left(\subst{\valfour''_\indexthree}{\vec{\vartwo}}{\vec{\valthree_{\pi_1(\indexthree)}}}\right)^{p''_\indexthree}
  \sep \indexthree \in \indexsetthree}$
  such that
 \begin{itemize}
  \item $\forall \indexthree \in \indexsetthree,\ \ \valfour''_\indexthree\ =\ \valfour_{\pi_1(\indexthree)}$,
  \item  $\forall \indexone \in \indexsetone,\ \ \sum_{\indexthree \in \pi_1^{-1}(\indexone)}\ p''_\indexthree\ =\ p_\indexone$,
  \item $\forall \indexthree \in \indexsetthree,\ \ \subst{\valfour''_\indexthree}{\vec{\vartwo}}{\vec{\valthree''_{\indexthree}}}
  \ =\ \valfour'_{\pi_2(\indexthree)}$,
  \item $\forall \indextwo \in \indexsettwo,\ \ \sum_{\indexthree \in \pi_2^{-1}(\indextwo)}\ p''_\indexthree\ =\ p'_\indextwo$.
 \end{itemize}

\end{lemma}
}

\longv{
\begin{proof}
 Let $\distrone\ =\ \distrelts{\left(\valfive_\indexfour\right)^{p''_\indexfour}\sep \indexfour \in \indexsetfour}$ be the \emph{representation} of $\distrone$. We build $\indexsetthree$, $\pi_1$ and $\pi_2$ as follows. The construction starts from the empty set and the empty maps,
 and is iterated on every $\indexfour \in \indexsetfour$.
 First, we set $\indexsetone_\indexfour \ =\ \left\{\indexone \in \indexsetone \sep \valfive_\indexfour\ =\ \subst{\valfour_\indexone}{\vec{\vartwo}}{\vec{\valthree_\indexone}}\right\}$ and $\indexsettwo_{\indexfour}\ =\ \left\{\indextwo \in \indexsettwo \sep 
 \valfive_\indexfour \ =\ \valfour'_\indextwo\right\}$. We suppose that both these sets are enumerated, and will write
 them $\indexsetone_\indexfour\ =\ \left\{\indexone_0,\ldots,\indexone_{n_\indexfour}\right\}$
 and $\indexsettwo_\indexfour\ =\ \left\{\indextwo_0,\ldots,\indextwo_{m_\indexfour}\right\}$.
 We consider the set of reals 
 $$
 R\ =\ 
 \left\{0,\,p_{\indexone_1},\,p_{\indexone_1} + p_{\indexone_2},\,\ldots,\,\sum_{r=0}^{n_\indexfour}\ p_{\indexone_r} \right\}
 \ \cup\ \left\{0,\,p'_{\indextwo_1},\,p'_{\indextwo_1} + p'_{\indextwo_2},\,\ldots,\,\sum_{r'=0}^{m_\indexfour}\ p'_{\indexone_{r'}} \right\}
 \ \subset \ [0,p''_\indexfour]
 $$
 This set is ordered, as a set of reals, so that we have a maximal enumeration
 $$
 0\ =\ \alpha_0 < \alpha_1 < \cdots < \alpha_s \ =\ p
 $$
 where maximality means that $\beta \in R \implies \exists t,\ \ \beta \,=\,\alpha_t$.
 We add $s$ elements to the set $\indexsetthree$ produced during the examination of previous elements of $\indexsetfour$:
 $\indexsetthree := \indexsetthree \uplus \{0,\ldots,s-1\}$. For every $t \in \{0,\ldots,s-1\}$, we define:
 \begin{itemize}
  \item $p''_t\,=\,\alpha_{t+1} - \alpha_t$,
  \item $\pi_1(t)\ =\ \indexone_k \in \indexsetone_\indexfour$ where $\sum_{r=0}^{k-1}\ p_{\indexone_r} \leq \alpha_t$
  and $\sum_{r=0}^{k}\ p_{\indexone_r} \geq \alpha_t$,
    \item $\pi_2(t)\ =\ \indextwo_k \in \indexsettwo_\indexfour$ where $\sum_{r=0}^{k-1}\ p_{\indextwo_r} \leq \alpha_t$
  and $\sum_{r=0}^{k}\ p_{\indextwo_r} \geq \alpha_t$.
 \end{itemize}
 We claim that the set $\indexsetthree$ resulting of this constructive process satisfies the equalities of the lemma.
\end{proof}
}

\longv{The series of previous lemmas allows to deduce that a term is reducible if and only if
the terms to which it is related are:}
\longv{
\begin{lemma}
\label{lemma:unfolding-stability-tred}
 Let $\valone\,=\,\left(\letrec{\funcone}{\valtwo}\right)$ be a closed value.
 Let $\termone$ be a simply-typed term of free variables contained in $\vec{\varone}$, all typed with the simple type of $\valone$.
 Let $\vec{\valthree},\,\vec{\valthree'} \in \unfoldings{\valone}$.
 Then $\subst{\termone}{\vec{\varone}}{\vec{\valthree}} \in \setredtmsfin{p}{\distrtypone,\seone}$ if and only if 
 $\subst{\termone}{\vec{\varone}}{\vec{\valthree'}} \in \setredtmsfin{p}{\distrtypone,\seone}$.
\end{lemma}
}

\longv{
\begin{proof}
 We prove that $\subst{\termone}{\vec{\varone}}{\vec{\valthree}} \in \setredtmsfin{p}{\distrtypone,\seone}$ implies that
 $\subst{\termone}{\vec{\varone}}{\vec{\valthree'}} \in \setredtmsfin{p}{\distrtypone,\seone}$, the converse direction being exactly symmetrical.
 The proof proceeds by induction on the simple type refining $\distrtypone$.\\
 
 Suppose that $\distrtypone \refines \Nat$. 
 Let $r \in [0,p)$.
 Since $\subst{\termone}{\vec{\varone}}{\vec{\valthree}} \in \setredtmsfin{p}{\distrtypone,\seone}$, there exists $n_r$
 and $\distrtyptwo_r \distrleq \distrtypone$ such that 
 $\subst{\termone}{\vec{\varone}}{\vec{\valthree}} \redval^{n_r} \distrone_r$ and that 
 $\distrone_r \in \setreddists{r}{\distrtyptwo_r,\seone}$.
 Lemma~\ref{lemma:relating-unfoldings-iterated-steps} implies that
 there exists a distribution of values of pseudo-representation
 $\pseudorep{\valfour_\indexone^{p_\indexone}\sep \indexone \in \indexsetone}$,
 a vector of variables $\vec{\vartwo}$ and families of vectors $\left(\vec{\valthree_\indexone}\right)_{\indexone \in \indexsetone},
 \left(\vec{\valthree'_\indexone}\right)_{\indexone \in \indexsetone}$ of the same length as $\vec{\vartwo}$ all such that
 $\distrone_r\ =\ \pseudorep{\left(\subst{\valfour_\indexone}{\vec{\vartwo}}{\vec{\valthree_\indexone}}\right)^{p_\indexone}  \sep \indexone \in \indexsetone}$
 and that $\subst{\termone}{\vec{\varone}}{\vec{\valthree'}} \redval^n \distrtwo_r\ =\  \pseudorep{\left(\subst{\valfour_\indexone}{\vec{\vartwo}}{\vec{\valthree'_\indexone}}\right)^{p_\indexone} \sep \indexone \in \indexsetone}$.
 By typing constraints coming from the subject reduction property, all the 
 $\subst{\valfour_\indexone}{\vec{\vartwo}}{\vec{\valthree_\indexone}}$ and
 $\subst{\valfour_\indexone}{\vec{\vartwo}}{\vec{\valthree'_\indexone}}$ have the simple type $\Nat$.
 This implies that all these terms are of the shape $\natsucc^m\ \natzero$ for $m \geq 0$,
 and thus that the $\valfour_\indexone$ cannot contain a variable from $\vec{\vartwo}$, as their simple type is of the shape $\Nat \typarrow \simpletypone$.
 It follows that, for every index $\indexone \in \indexsetone$, 
 $\subst{\valfour_\indexone}{\vec{\vartwo}}{\vec{\valthree_\indexone}}\ =\ 
 \subst{\valfour_\indexone}{\vec{\vartwo}}{\vec{\valthree'_\indexone}}$.
 This implies that $\distrtwo_r\ =\ \distrone_r \in \setreddists{r}{\distrtyptwo_r,\seone}$, and thus that 
 $\subst{\termone}{\vec{\varone}}{\vec{\valthree'}} \in \setredtmsfin{p}{\distrtypone,\seone}$.\\
 
 Suppose that $\distrtypone \refines \simpletypone \typarrow \simpletyptwo$.
 Let $r \in [0,p)$.
 Since $\subst{\termone}{\vec{\varone}}{\vec{\valthree}} \in \setredtmsfin{p}{\distrtypone,\seone}$, there exists $n_r$
 and $\distrtyptwo_r \distrleq \distrtypone$ such that 
 $\subst{\termone}{\vec{\varone}}{\vec{\valthree}} \redval^{n_r} \distrone_r$ and that 
 $\distrone_r \in \setreddists{r}{\distrtyptwo_r,\seone}$.
 Lemma~\ref{lemma:relating-unfoldings-iterated-steps} implies that
 there exists a distribution of values of pseudo-representation
 $\pseudorep{\valfour_\indexone^{p_\indexone}\sep \indexone \in \indexsetone}$,
 a vector of variables $\vec{\vartwo}$ and families of vectors $\left(\vec{\valthree_\indexone}\right)_{\indexone \in \indexsetone},
 \left(\vec{\valthree'_\indexone}\right)_{\indexone \in \indexsetone}$ of the same length as $\vec{\vartwo}$ all such that
 $\distrone_r\ =\ \pseudorep{\left(\subst{\valfour_\indexone}{\vec{\vartwo}}{\vec{\valthree_\indexone}}\right)^{p_\indexone}  \sep \indexone \in \indexsetone}$
 and that $\subst{\termone}{\vec{\varone}}{\vec{\valthree'}} \redval^n \distrtwo_r\ =\  \pseudorep{\left(\subst{\valfour_\indexone}{\vec{\vartwo}}{\vec{\valthree'_\indexone}}\right)^{p_\indexone} \sep \indexone \in \indexsetone}$.
 Since $\distrone_r \in \setreddists{r}{\distrtyptwo_r,\seone}$, 
 there is a pseudo-representation $\distrone_r\ =\ \pseudorep{(\valthree'_{\indextwo})^{p'_\indextwo} \sep \indextwo \in \indexsettwo}$
 witnessing this fact.
 By Lemma~\ref{lemma:cutting-pseudo-representations},
 there exists a pseudo-representation
 $\distrone_r \ =\ \pseudorep{\left(\subst{\valfour''_\indexthree}{\vec{\vartwo}}{\vec{\valthree_{\pi_1(\indexthree)}}}\right)^{p''_\indexthree}
  \sep \indexthree \in \indexsetthree}$ satisfying a series of additional properties.
 These properties ensure two crucial facts for our purpose:
 \begin{itemize}
  \item $\subst{\termone}{\vec{\varone}}{\vec{\valthree}} \redval^n \pseudorep{\left(\subst{\valfour''_\indexthree}{\vec{\vartwo}}{\vec{\valthree_{\pi_1(\indexthree)}}}\right)^{p''_\indexthree} \sep \indexthree \in \indexsetthree}$
  and $\subst{\termone}{\vec{\varone}}{\vec{\valthree'}} \redval^n \distrtwo_r\ =\  \pseudorep{\left(\subst{\valfour''_\indexthree}{\vec{\vartwo}}{\vec{\valthree'_{\pi_1(\indexthree)}}}\right)^{p''_\indexthree} \sep \indexthree \in \indexsetthree}$,
  \item and $\pseudorep{\left(\subst{\valfour''_\indexthree}{\vec{\vartwo}}{\vec{\valthree_{\pi_1(\indexthree)}}}\right)^{p''_\indexthree} \sep \indexthree \in \indexsetthree}$ is a pseudo-distribution witnessing that $\distrone_r \in \setreddists{r}{\distrtyptwo_r,\seone}$.
  Setting $\distrtypone\,=\,\distrelts{(\typone_\indexfour)^{p'''_\indexfour} \sep \indexfour \in \indexsetfour}$,
  there exists therefore families $\left(p''_{\indexthree \indexfour}\right)_{\indexthree\in \indexsetthree,\indexfour\in\indexsetfour}$
  and $\left(q_{\indexthree \indexfour}\right)_{\indexthree\in \indexsetthree,\indexfour\in\indexsetfour}$ of reals of $[0,1]$
  satisfying:
  \begin{enumerate}
   \item $\forall \indexthree \in \indexsetthree,\ \ \forall\indexfour \in \indexsetfour,\ \ \subst{\valfour''_\indexthree}{\vec{\vartwo}}{\vec{\valthree_{\pi_1(\indexthree)}}} \in 
   \setredvals{q_{\indexthree\indexfour}}{\typone_\indexfour,\seone}$,
   \item $\forall \indexthree \in \indexsetthree,\ \ \sum_{\indexfour \in \indexsetfour}\ p''_{\indexthree\indexfour}\ =\ p''_\indexthree$,
   \item $\forall \indexfour \in \indexsetfour,\ \ \sum_{\indexthree \in \indexsetthree}\ p''_{\indexthree\indexfour}\ =\ \distrtypone(\typone_\indexfour)$,
   \item $p \leq \sum_{\indexthree \in \indexsetthree}\sum_{\indexfour \in \indexsetfour}\ q_{\indexthree\indexfour}p''_{\indexthree\indexfour}$.
  \end{enumerate}
 \end{itemize}
 We now prove that 
 $\forall \indexthree \in \indexsetthree,\ \ \forall\indexfour \in \indexsetfour,\ \ \subst{\valfour''_\indexthree}{\vec{\vartwo}}{\vec{\valthree'_{\pi_1(\indexthree)}}} \in 
   \setredvals{q_{\indexthree\indexfour}}{\typone_\indexfour,\seone}$.
 Let $\indexthree \in \indexsetthree$ and $\indexfour \in \indexsetfour$.
 Let $\typone_\indexfour\ =\ \typthree \typarrow \distrtyptwo$.
 $$
 \begin{array}{ll}
  & \subst{\valfour''_\indexthree}{\vec{\vartwo}}{\vec{\valthree_{\pi_1(\indexthree)}}} \in 
   \setredvals{q_{\indexthree\indexfour}}{\typone_\indexfour,\seone}\\
  \Longleftrightarrow \quad & 
  \forall q \in (0,1],\ \ \forall\valfive\in\setredvals{q}{\typthree,\seone},\ \ \subst{\valfour''_\indexthree}{\vec{\vartwo}}{\vec{\valthree_{\pi_1(\indexthree)}}}\ \valfive \in\setredtmsfin{qq_{\indexthree\indexfour}}{\distrtyptwo,\seone}\\
  \Longleftrightarrow \quad & 
  \forall q \in (0,1],\ \ \forall\valfive\in\setredvals{q}{\typthree,\seone},\ \ \subst{\left(\valfour''_\indexthree\ \valfive\right)}{\vec{\vartwo}}{\vec{\valthree_{\pi_1(\indexthree)}}}\  \in\setredtmsfin{qq_{\indexthree\indexfour}}{\distrtyptwo,\seone}\\
  \Longleftrightarrow \quad & 
  \forall q \in (0,1],\ \ \forall\valfive\in\setredvals{q}{\typthree,\seone},\ \ \subst{\left(\valfour''_\indexthree\ \valfive\right)}{\vec{\vartwo}}{\vec{\valthree'_{\pi_1(\indexthree)}}}\  \in\setredtmsfin{qq_{\indexthree\indexfour}}{\distrtyptwo,\seone}
  \quad \ \ \text{by IH}\\
  \Longleftrightarrow \quad & 
  \forall q \in (0,1],\ \ \forall\valfive\in\setredvals{q}{\typthree,\seone},\ \ \subst{\valfour''_\indexthree}{\vec{\vartwo}}{\vec{\valthree'_{\pi_1(\indexthree)}}}\ \valfive  \in\setredtmsfin{qq_{\indexthree\indexfour}}{\distrtyptwo,\seone}
  \\ 
  \Longleftrightarrow & 
  \subst{\valfour''_\indexthree}{\vec{\vartwo}}{\vec{\valthree'_{\pi_1(\indexthree)}}} \in 
   \setredvals{q_{\indexthree\indexfour}}{\typone_\indexfour,\seone}\\
 \end{array}
 $$
  This will implies that $\pseudorep{\left(\subst{\valfour''_\indexthree}{\vec{\vartwo}}{\vec{\valthree'_{\pi_1(\indexthree)}}}\right)^{p''_\indexthree} \sep \indexthree \in \indexsetthree}$ witnesses that $\distrtwo_r \in \setreddists{r}{\distrtyptwo_r,\seone}$,
 for the same families of reals $p''_{\indexthree\indexfour}$ and $q_{\indexthree\indexfour}$.
 Now for every $r \in [0,p)$, there exists $n_r$ and $\distrtyptwo_r \distrleq \distrtypone$ such that 
 $\subst{\termone}{\vec{\varone}}{\vec{\valthree'}} \redval^{n_r} \distrtwo_r$ and that 
 $\distrtwo_r \in \setreddists{r}{\distrtyptwo_r,\seone}$: we have that 
 $\subst{\termone}{\vec{\varone}}{\vec{\valthree'}} \in \setredtmsfin{p}{\distrtypone,\seone}$.

\end{proof}
}

\longv{The following lemma shows that reducible values are reducible terms:}
\longv{
\begin{lemma}[Reducible Values are Reducible Terms]
\label{lemma:values-are-in-tred-iff-in-vred}
 Let $\valone$ be a value. Then $\valone \in \setredtmsfin{p}{\distrelts{\typone^1},\seone}$
 if and only if $\valone \in \setredvals{p}{\typone,\seone}$.
\end{lemma}
}

\longv{
Note that, conversely, we may have $\valone \in \setredtmsfin{p}{\distrtypone,\seone}$ where $\distrtypone$ is not Dirac.
For instance, $\natzero \in \setredtmsfin{1}{\distrtypone,\seone}$ for
$\distrtypone\,=\,\distrelts{(\Nat^{\sizevarone})^{\frac{1}{2}},(\Nat^{\sizesucc{\sizevarone}})^{\frac{1}{2}}}$.
}

\longv{
\begin{proof}~

\begin{itemize}
 \item Suppose that $\valone \in \setredvals{p}{\typone,\seone}$. Let $r \in [0,p)$. We must prove that
 there exists $n_r$ and $\distrtyptwo_r$ such that $\valone \rcbv^{n_r} \distrelts{\valone^1}$ and that
 $\distrelts{\valone^1} \in \setreddists{r}{\distrtyptwo_r,\seone}$.
 Necessarily $n_r = 0$ and $\distrtyptwo_r\,=\,\distrelts{\typone^1}$.
 Since $\valone \in \setredvals{p}{\typone,\seone}$, $\distrelts{\valone^1} \in \setreddists{r}{\distrtyptwo_r,\seone}$:
 take the canonical pseudo-representation $\pseudorep{\valone^1}$ and $p_{11}=1$, $q_{11}=r$.

 \item Suppose that $\valone \in \setredtmsfin{p}{\distrelts{\typone^1},\seone}$.
 It follows that, for every $r \in [0,p)$, there exists $n_r$ and $\distrtyptwo_r$ such that $\valone \rcbv^{n_r} \distrelts{\valone^1}$ and that
 $\distrelts{\valone^1} \in \setreddists{r}{\distrtyptwo_r,\seone}$. Again, since, $\valone$ is a value, we necessarily have
 $n_r = 0$ and $\distrtyptwo_r\,=\,\distrelts{\typone^1}$.
 Since $\distrelts{\valone^1} \in \setreddists{r}{\distrtyptwo_r,\seone}$, there is a pseudo-representation
 $\pseudorep{\valone^{p_1},\ldots,\valone^{p_n}}$ such that $\sum_{\indexone = 1}^n\ p_\indexone\ =\ 1$,
 and a family $\left(q_{\indexone 1}\right)_{\indexone \in \indexsetone}$ which is such that
 $r \leq \sum_{\indexone \in \indexsetone}\ p_{\indexone 1} q_{\indexone 1}$,
 where $p_{\indexone 1}\,=\,p_{\indexone}$.
 
 Suppose that there is no $q_{\indexone 1}$ greater or equal to $r$. Then
 $\forall \indexone \in \indexsetone,\ \ q_{\indexone 1} < r$ and
 $$
 \sum_{\indexone \in \indexsetone}\ p_{\indexone 1}q_{\indexone 1} \ <\ 
 \sum_{\indexone \in \indexsetone}\ p_{\indexone 1} r \ =\ r \sum_{\indexone \in \indexsetone}\ p_{\indexone 1} = r
 $$
 which is a contradiction. So there exists $q_{\indexone 1} \geq r$ and therefore
 $\valone \in \setredvals{q_{\indexone 1}}{\typone,\seone}$. By Lemma~\ref{lemma/downward-closure-tred},
 $\valone \in \setredvals{r}{\typone,\seone}$.
 Since the result is true for all $r \in [0,p)$, we obtain by Lemma~\ref{lemma/continuity-lemma-vred}
 that $\valone \in \setredvals{p}{\typone,\seone}$.
\end{itemize}

\end{proof}
}

\longv{We finally deduce from the two previous lemmas the proposition of interest, relating the reducibility of a 
recursively-defined term with the one of its unfoldings:}
\begin{proposition}[Reducibility is Stable by Unfolding]
\label{corollary:unfolding-stability-vred}
 Let $n \in \NN$ and $\valone\,=\,\left(\letrec{\funcone}{\valtwo}\right)$ be a closed value. Suppose that $\valthree$ is the $n$-unfolding of $\valone$.
 Then $\valone \in \setredvals{p}{\Nat^\sizeone \typarrow \distrtypone,\seone}$ if and only if 
 $\valthree \in \setredvals{p}{\Nat^\sizeone \typarrow \distrtypone,\seone}$.
\end{proposition}

\longv{
\begin{proof}
A direct consequence of Lemma~\ref{lemma:unfolding-stability-tred} and Lemma~\ref{lemma:values-are-in-tred-iff-in-vred}.
\end{proof}
}

\longv{\subsection{Reducibility Sets vs. Reductions and Probabilistic Combinations}}

\longv{If a distribution, obtained as partial approximation of the semantics $\semantics{\termone}$
of a term $\termone$, is reducible for a type $\distrtypone_n$, then all the partial approximations
of $\semantics{\termone}$ obtained by iterating at least as many times the reduction relation
$\redval$ have the same degree of reducibility, for a greater type:
\begin{lemma}
\label{lemma:pumping-dred}
 Suppose that $\termone \redval^n \distrone_n \in \setreddists{p}{\distrtypone_n,\seone}$
 for $\distrtypone_n \distrleq \distrtypone$, with $\distrsum{\distrtypone}=1$.
 Suppose that, for $m \geq n$, $\termone \redval^m \distrone_m$.
 Then there exists $\distrtypone_n \distrleq \distrtypone_m \distrleq \distrtypone$ such that
 $\distrone_m \in \setreddists{p}{\distrtypone_m,\seone}$.
\end{lemma}

}

\longv{

\begin{proof}
 Let $\distrtypone_n\,=\,\distrelts{\left(\typone_\indextwo\right)^{p'_\indextwo} \sep \indextwo \in \indexsettwo}$.
 Since $\distrone_n \in \setreddists{p}{\distrtypone_n,\seone}$,
 there exists a pseudo-representation $\distrone_n\,=\,\pseudorep{\valone_\indexone^{p_\indexone} \sep \indexone \in \indexsetone}$
 and two families of reals $\left(p_{\indexone \indextwo}\right)_{\indexone\in \indexsetone,\indextwo\in\indexsettwo}$
  and $\left(q_{\indexone \indextwo}\right)_{\indexone\in \indexsetone,\indextwo\in\indexsettwo}$
 such that
 \begin{enumerate}
   \item $\forall \indexone \in \indexsetone,\ \ \forall\indextwo \in \indexsettwo,\ \ \valone_\indexone \in 
   \setredvals{q_{\indexone\indextwo}}{\typone_\indextwo,\seone}$,
   \item $\forall \indexone \in \indexsetone,\ \ \sum_{\indextwo \in \indexsettwo}\ p_{\indexone\indextwo}\ =\ p_\indexone$,
   \item $\forall \indextwo \in \indexsettwo,\ \ \sum_{\indexone \in \indexsetone}\ p_{\indexone\indextwo}\ =\ p'_\indextwo$,
   \item $p \leq \sum_{\indexone \in \indexsetone}\sum_{\indextwo \in \indexsettwo}\ q_{\indexone\indextwo}p_{\indexone\indextwo}$.
  \end{enumerate}
 By Lemma~\ref{lemma/semantics-is-computed-monotonically2}, we have $\distrone_n \distrleq \distrone_m$ so that the distribution $\distrone_m$
 admits a pseudo-representation $\distrone_m\,=\,\pseudorep{\valone_\indexone^{p_\indexone} \sep \indexone \in \indexsetone \uplus \indexsetthree}$
 extending the one of $\distrone_n$.
 We now need to define appropriate families of reals $\left(p'_{\indexone \indextwo}\right)_{\indexone\in \indexsetone \uplus \indexsetthree,\indextwo\in\indexsettwo}$ and $\left(q'_{\indexone \indextwo}\right)_{\indexone\in \indexsetone \uplus \indexsetthree,\indextwo\in\indexsettwo}$.
 We set:
 \begin{itemize}
  \item $\forall \indexone \in \indexsetone,\ \forall \indextwo\in \indexsettwo,\ \ 
  p'_{\indexone \indextwo}\ =\ p_{\indexone \indextwo}$,
  \item $\forall \indexone \in \indexsetone,\ \forall \indextwo\in \indexsettwo,\ \ 
  q'_{\indexone \indextwo}\ =\ q_{\indexone \indextwo}$,
  \item $\forall \indexone \in \indexsetthree,\ \forall \indextwo\in \indexsettwo,\ \ 
  q'_{\indexone \indextwo}\ =\ 0$
 \end{itemize}
 and we choose the $\left(p'_{\indexone \indextwo}\right)_{\indexone \in \indexsetthree,\indextwo\in \indexsettwo}$ arbitrarily in $[0,1]$ 
 under the constraints that $\forall \indexone \in \indexsetthree,\ \ \sum_{\indextwo \in \indexsettwo}\ p'_{\indexone\indextwo}\ =\ p_\indexone$
 and that $\forall \indextwo \in \indexsettwo,\ \sum_{\indexone \in \indexsetone \uplus \indexsetthree}\ p'_{\indexone\indextwo} \leq \distrtypone(\typone_\indextwo)$.
 These constraints are feasible since 
 $\sum_{\indexone \in \indexsetone \uplus \indexsetthree}\ \sum_{\indextwo \in \indexsettwo}\ p'_{\indexone\indextwo}\ =\
 \sum_{\indexone \in \indexsetone \uplus \indexsetthree}\ p_\indexone \leq 1 = \distrsum{\distrtypone}$.
 We then set $\distrtypone_m\,=\,\distrelts{\left(\typone_\indextwo\right)^{\sum_{\indexone \in \indexsetone \uplus \indexsetthree}\ p'_{\indexone\indextwo}} \sep \indextwo \in \indexsettwo} \distrleq \distrtypone$.
 Let us check that $\distrone_m \in \setreddists{p}{\distrtypone_m,\seone}$.
 \begin{enumerate}
  \item $\forall \indexone \in \indexsetone,\ \ \forall\indextwo \in \indexsettwo,\ \ \valone_\indexone \in 
   \setredvals{q_{\indexone\indextwo}}{\typone_\indextwo,\seone}$ and $\forall \indexone \in \indexsetthree,\ \ \forall\indextwo \in \indexsettwo,\ \ \valone_\indexone \in \setredvals{0}{\typone_\indextwo,\seone}$ as this set contains all terms of simple type $\underlying{\typone_\indextwo}$
   by Lemma~\ref{lemma:simple-types-and-candidates-of-proba-zero},
   \item $\forall \indexone \in \indexsetone,\ \ \sum_{\indextwo \in \indexsettwo}\ p'_{\indexone\indextwo}\ =\ p_\indexone$ by definition
   and $\forall \indexone \in \indexsetthree,\ \ \sum_{\indextwo \in \indexsettwo}\ p'_{\indexone\indextwo}\ =\ p_\indexone$ by construction,
   \item $\forall \indextwo \in \indexsettwo,\ \ \sum_{\indexone \in \indexsetone \uplus \indexsetthree}\ p'_{\indexone\indextwo}\ =\ \distrtypone_m(\typone_\indextwo)$ by definition of $\distrtypone_m$,
   \item %
   $$
   \begin{array}{rcl}
   p & \ \leq\ &\sum_{\indexone \in \indexsetone}\sum_{\indextwo \in \indexsettwo}\ q_{\indexone\indextwo}p_{\indexone\indextwo}\\
   & \ =\ & \sum_{\indexone \in \indexsetone}\sum_{\indextwo \in \indexsettwo}\ q'_{\indexone\indextwo}p'_{\indexone\indextwo} + 0\\
   & \ =\ & \sum_{\indexone \in \indexsetone}\sum_{\indextwo \in \indexsettwo}\ q'_{\indexone\indextwo}p'_{\indexone\indextwo} + 
   \sum_{\indexone \in \indexsetthree}\sum_{\indextwo \in \indexsettwo}\ q'_{\indexone\indextwo}p'_{\indexone\indextwo}\\
   & \ =\ & \sum_{\indexone \in \indexsetone \uplus \indexsetthree}\sum_{\indextwo \in \indexsettwo}\ q'_{\indexone\indextwo}p'_{\indexone\indextwo} \\
   \end{array}
   $$
 \end{enumerate}
 So $\distrone_m \in \setreddists{p}{\distrtypone_m,\seone}$.

\end{proof}

}

\longv{When two distributions $\distrone$ and $\distrtwo$ are reducible, 
with respective degrees of reducibility $p'$ and $p''$, 
their probabilistic combination $\distrone \choice_p \distrtwo$
is reducible as well with degree of reducibility $pp' + (1-p)p''$,
for the distribution type computed by $\choice_p$:

\begin{lemma}
\label{lemma:backtracking-dred}
 Suppose that $\underlying{\distrtypone}\,=\,\underlying{\distrtyptwo}$, that 
 $\distrone \in \setreddists{p'}{\distrtypone,\seone}$
 and that $\distrtwo \in \setreddists{p''}{\distrtyptwo,\seone}$.
 Then $p\distrone + (1-p) \distrtwo \in \setreddists{pp' + (1-p)p''}{\distrtypone \choice_p \distrtyptwo,\seone}$.
\end{lemma}

}

\longv{

\begin{proof}
Let $\distrtypone\,=\,\distrelts{\left(\typone_\indextwo\right)^{p'_\indextwo} \sep \indextwo \in \indexsettwo}$.
 Since $\distrone \in \setreddists{p'}{\distrtypone,\seone}$,
 there exists a pseudo-representation $\distrone\,=\,\pseudorep{\valone_\indexone^{p_\indexone} \sep \indexone \in \indexsetone}$
 and two families of reals $\left(p_{\indexone \indextwo}\right)_{\indexone\in \indexsetone,\indextwo\in\indexsettwo}$
  and $\left(q_{\indexone \indextwo}\right)_{\indexone\in \indexsetone,\indextwo\in\indexsettwo}$
 such that
 \begin{enumerate}
   \item $\forall \indexone \in \indexsetone,\ \ \forall\indextwo \in \indexsettwo,\ \ \valone_\indexone \in 
   \setredvals{q_{\indexone\indextwo}}{\typone_\indextwo,\seone}$,
   \item $\forall \indexone \in \indexsetone,\ \ \sum_{\indextwo \in \indexsettwo}\ p_{\indexone\indextwo}\ =\ p_\indexone$,
   \item $\forall \indextwo \in \indexsettwo,\ \ \sum_{\indexone \in \indexsetone}\ p_{\indexone\indextwo}\ =\ p'_\indextwo$,
   \item $p' \leq \sum_{\indexone \in \indexsetone}\sum_{\indextwo \in \indexsettwo}\ q_{\indexone\indextwo}p_{\indexone\indextwo}$.
  \end{enumerate}
 Let $\distrtyptwo\,=\,\distrelts{\left(\typtwo_\indexfour\right)^{p'''_\indexfour} \sep \indexfour \in \indexsetfour}$. 
 Since $\distrtwo \in \setreddists{p''}{\distrtyptwo,\seone}$,
 there exists a pseudo-representation $\distrtwo\,=\,\pseudorep{\valtwo_\indexthree^{p''_\indexthree} \sep \indexthree \in \indexsetthree}$
 and two families of reals $\left(p'_{\indexthree \indexfour}\right)_{\indexthree\in \indexsetthree,\indexfour\in\indexsetfour}$
  and $\left(q'_{\indexthree \indexfour}\right)_{\indexthree\in \indexsetthree,\indexfour\in\indexsetfour}$
 such that
 \begin{enumerate}
   \item $\forall \indexthree \in \indexsetthree,\ \ \forall\indexfour \in \indexsetfour,\ \ \valtwo_\indexthree \in 
   \setredvals{q'_{\indexthree\indexfour}}{\typtwo_\indexfour,\seone}$,
   \item $\forall \indexthree \in \indexsetthree,\ \ \sum_{\indexfour \in \indexsetfour}\ p'_{\indexthree\indexfour}\ =\ p''_\indexthree$,
   \item $\forall \indexfour \in \indexsetfour,\ \ \sum_{\indexthree \in \indexsetthree}\ p'_{\indexthree\indexfour}\ =\ p'''_\indexfour$,
   \item $p'' \leq \sum_{\indexthree \in \indexsetthree}\sum_{\indexfour \in \indexsetfour}\ q'_{\indexthree\indexfour}p'_{\indexthree\indexfour}$.
  \end{enumerate}
  We suppose that $\indexsetone$ and $\indexsetthree$ are disjoint, and that $\indextwo \in \indexsettwo \cap \indexsetfour \Leftrightarrow 
  \typone_\indextwo = \typtwo_\indextwo$.
  To prove that $p\distrone + (1-p) \distrtwo \in \setreddists{pp' + (1-p)p''}{\distrtypone \choice_p \distrtyptwo,\seone}$, we consider the pseudo-representation
  \begin{equation}
  \label{eq:lemma-probabilistic-dred-combination1}
  p\distrone + (1-p) \distrtwo\ \ =\ \ \pseudorep{\valone_\indexone^{pp_\indexone} \sep \indexone \in \indexsetone}
  + \pseudorep{\valtwo_\indexthree^{(1-p)p''_\indexthree} \sep \indexthree \in \indexsetthree}
  \end{equation}
  and we write the distribution type $\distrtypone \choice_p \distrtyptwo$ as
  $$
  \distrelts{\left(\typone_\indextwo\right)^{pp'_\indextwo} \sep \indextwo \in \indexsettwo\setminus(\indexsettwo \cap \indexsetfour)}
  \ +\ \distrelts{\left(\typone_\indextwo\right)^{pp'_\indextwo+(1-p)p'''_\indextwo} \sep \indextwo \in \indexsettwo \cap \indexsetfour}
  \ +\ \distrelts{\left(\typtwo_\indexfour\right)^{(1-p)p'''_\indexfour} \sep \indexfour \in \indexsetfour\setminus(\indexsettwo \cap \indexsetfour)}
  $$
 We set $\indexsetfive\,=\,\indexsetone + \indexsetthree$ and $\indexsetsix\,=\,\indexsettwo+ \indexsetfour$.
  We now need to define appropriate families of reals $\left(p''_{\indexfive\indexsix}\right)_{\indexfive \in \indexsetfive,\indexsix \in \indexsetsix}$
  and $\left(q''_{\indexfive\indexsix}\right)_{\indexfive \in \indexsetfive,\indexsix \in \indexsetsix}$. We proceed as follows:
  \begin{itemize}
   \item if $\indexfive \in \indexsetone$ and $\indexsix \in \indexsettwo$,
   $p''_{\indexfive\indexsix}\,=\,pp_{\indexfive\indexsix}$ and $q''_{\indexfive\indexsix}\,=\,q_{\indexfive\indexsix}$,
   \item if $\indexfive \in \indexsetone$ and $\indexsix \in \indexsetfour$,
   $p''_{\indexfive\indexsix}\,=\,0$ and $q''_{\indexfive\indexsix}\,=\,0$,
   \item if $\indexfive \in \indexsetthree$ and $\indexsix \in \indexsettwo$,
   $p''_{\indexfive\indexsix}\,=\,0$ and $q''_{\indexfive\indexsix}\,=\,0$,
   \item if $\indexfive \in \indexsetthree$ and $\indexsix \in \indexsetfour$,
   $p''_{\indexfive\indexsix}\,=\,(1-p)p'_{\indexfive\indexsix}$ and $q''_{\indexfive\indexsix}\,=\,q'_{\indexfive\indexsix}$.
  \end{itemize}
  Let us prove that (\ref{eq:lemma-probabilistic-dred-combination1}) together with these two families provide a witness that 
  $p\distrone + (1-p) \distrtwo \in \setreddists{pp' + (1-p)p''}{\distrtypone \choice_p \distrtyptwo,\seone}$ by checking the four usual conditions.
  We write $\valthree_\indexfive$ either for $\valone_\indexone$ or $\valtwo_\indexthree$, depending on the context.
  We write similarly $\typthree_\indexsix$ for $\typone_\indextwo$ or $\typtwo_\indexfour$.
  \begin{enumerate}
   \item $\forall \indexfive \in \indexsetfive,\ \ \forall\indexsix \in \indexsetsix,\ \ \valthree_\indexfive \in 
   \setredvals{q''_{\indexfive\indexsix}}{\typthree_\indexsix,\seone}$ is proved by case exhaustion:
   \begin{itemize}
    \item $\forall \indexfive \in \indexsetone,\ \ \forall\indexsix \in \indexsettwo,\ \ \valone_\indexfive \in 
   \setredvals{q_{\indexfive\indexsix}}{\typone_\indexsix,\seone}$ since $\distrone \in \setreddists{p'}{\distrtypone,\seone}$,
   \item $\forall \indexfive \in \indexsetthree,\ \ \forall\indexsix \in \indexsetfour,\ \ \valtwo_\indexfive \in 
   \setredvals{q'_{\indexfive\indexsix}}{\typtwo_\indexsix,\seone}$ since $\distrtwo \in \setreddists{p''}{\distrtyptwo,\seone}$,
   \item in the two remaining cases, $q''_{\indexfive\indexsix}\,=\,0$ and by Lemma~\ref{lemma:simple-types-and-candidates-of-proba-zero}
   the result holds.
   \end{itemize}
   \item We proceed again by case exhaustion.
   \begin{itemize}
    \item If $\indexfive \in \indexsetone$, $\sum_{\indexsix \in \indexsetsix}\ p''_{\indexfive\indexsix}\ =\ 
    \sum_{\indexsix \in \indexsettwo}\ p''_{\indexfive\indexsix} + \sum_{\indexsix \in \indexsetfour}\ p''_{\indexfive\indexsix}\ =\ 
    \sum_{\indexsix \in \indexsettwo}\ pp_{\indexfive\indexsix}\ =\ pp_\indexfive$.
    \item If $\indexfive \in \indexsetthree$, $\sum_{\indexsix \in \indexsetsix}\ p''_{\indexfive\indexsix}\ =\ 
    \sum_{\indexsix \in \indexsetfour}\ (1-p)p'_{\indexfive\indexsix}\ =\ (1-p)p''_\indexfive$.
   \end{itemize}

   \item  We proceed again by case exhaustion.
   \begin{itemize}
    \item Suppose that $\indexsix \in \indexsettwo \setminus (\indexsettwo \cap\indexsetfour)$.
    Then $\sum_{\indexfive \in \indexsetfive}\ p''_{\indexfive\indexsix}\ =\ 
    \sum_{\indexfive \in \indexsetone}\ p''_{\indexfive\indexsix}\ =\ \sum_{\indexfive \in \indexsetone}\ pp_{\indexfive\indexsix} 
    \ =\ pp'_\indexfive$.
    \item Suppose that $\indexsix \in \indexsetfour \setminus (\indexsettwo \cap\indexsetfour)$.
    Then $\sum_{\indexfive \in \indexsetfive}\ p''_{\indexfive\indexsix}\ =\ 
    \sum_{\indexfive \in \indexsetthree}\ p''_{\indexfive\indexsix}\ =\ \sum_{\indexfive \in \indexsetthree}\ (1-p)p'_{\indexfive\indexsix} 
    \ =\ (1-p)p'''_\indexfive$.
    \item Suppose that $\indexsix \in \indexsettwo \cap\indexsetfour$.
    Then $\sum_{\indexfive \in \indexsetfive}\ p''_{\indexfive\indexsix}\ =\ 
    \sum_{\indexfive \in \indexsetone}\ p''_{\indexfive\indexsix}\ +\  \sum_{\indexfive \in \indexsetthree}\ p''_{\indexfive\indexsix}
    \ =\ pp'_\indexfive + (1-p)p'''_\indexfive$.
   \end{itemize}
   \item 
   $$
   \begin{array}{rl}
   & \sum_{\indexfive \in \indexsetfive}\sum_{\indexsix \in \indexsetsix}\ q''_{\indexfive\indexsix}p''_{\indexfive\indexsix} \\
    \ =\ &
   \sum_{\indexfive \in \indexsetone}\sum_{\indexsix \in \indexsettwo}\ q''_{\indexfive\indexsix}p''_{\indexfive\indexsix} 
   + \sum_{\indexfive \in \indexsetthree}\sum_{\indexsix \in \indexsetfour}\ q''_{\indexfive\indexsix}p''_{\indexfive\indexsix}\\
    \ =\ &
   \sum_{\indexfive \in \indexsetone}\sum_{\indexsix \in \indexsettwo}\ q_{\indexfive\indexsix}pp_{\indexfive\indexsix} 
   + \sum_{\indexfive \in \indexsetthree}\sum_{\indexsix \in \indexsetfour}\ q'_{\indexfive\indexsix}(1-p)p'_{\indexfive\indexsix}\\ 
   \ =\ &
   p \sum_{\indexfive \in \indexsetone}\sum_{\indexsix \in \indexsettwo}\ q_{\indexfive\indexsix}p_{\indexfive\indexsix} 
   + (1-p) \sum_{\indexfive \in \indexsetthree}\sum_{\indexsix \in \indexsetfour}\ q'_{\indexfive\indexsix}p'_{\indexfive\indexsix}\\
    \ \geq\ & pp' + (1-p)p''\\
   \end{array}
   $$
   It follows that $p\distrone + (1-p) \distrtwo \in \setreddists{pp' + (1-p)p''}{\distrtypone \choice_p \distrtyptwo,\seone}$.

  \end{enumerate}

\end{proof}

}

\longv{This lemma generalizes to the $n$-ary case of a weighted sum of distributions:
\begin{lemma}
\label{lemma:backtracking-dred-multiple-instances}
 Let $\left(\distrtypone_\indexone\right)_{\indexone \in \indexsetone}$ be a family of distribution types of same underlying type.
 For every $\indexone \in \indexsetone$, let $\distrone_\indexone \in \setreddists{q_\indexone}{\distrtypone_\indexone,\seone}$.
 Let $(p_\indexone)_{\indexone \in \indexsetone}$ be a family of reals of $[0,1]$ such that 
 $\sum_{\indexone \in \indexsetone}\ p_\indexone \leq 1$.
 Then $\sum_{\indexone \in \indexsetone}\ p_\indexone \distrone_\indexone \in \setreddists{\sum_{\indexone \in \indexsetone}\ p_\indexone q_\indexone}{\sum_{\indexone \in \indexsetone}\ p_\indexone \distrtypone_\indexone,\seone}$.
\end{lemma}

}

\longv{

\begin{proof}
Similar to the proof of Lemma~\ref{lemma:backtracking-dred}. 
\end{proof}

}

\longv{$\setredtms{}{}$ is closed by anti-reduction for Dirac distributions, but also in the case corresponding 
to the reduction of a choice operator:}

\longv{
\begin{lemma}[Reductions and Sets of Candidates]~
 \label{lemma:reduction-and-candidates}
\begin{itemize}
 \item Suppose that $\termone \rcbv \distrelts{\termtwo^1}$ and that $\termtwo \in \setredtmsfin{p}{\distrtypone,\seone}$.
 Then $\termone \in \setredtmsfin{p}{\distrtypone,\seone}$.
 \item Suppose that $\termone \rcbv \distrelts{\termtwo^p,\termthree^{1-p}}$, that $\termtwo \in \setredtmsfin{p'}{\distrtypone,\seone}$
 and that $\termthree \in \setredtmsfin{p''}{\distrtyptwo,\seone}$. Then 
 $\termone \in \setredtmsfin{pp'+(1-p)p''}{\distrtypone \choice_p \distrtyptwo,\seone}$.
\end{itemize}
\end{lemma}
}

\longv{
\begin{proof}~

 \begin{itemize}
  \item Since $\termtwo \in \setredtmsfin{p}{\distrtypone,\seone}$, for every 
  $0 \leq r < p$ there exists $\distrtyptwo_r \distrleq \distrtypone$ and $n_r \in \NN$
  such that $\termtwo \redval^{n_r} \distrone_r \in \setreddists{r}{\distrtyptwo_r,\seone}$.
  Recall that $\redval^{n_r +1}\,=\,\rcbv \circ \redval^{n_r}$. It follows that 
  $\termone \redval^{n_r +1} \distrone_r$ which has the required properties, so that 
  $\termone \in \setredtmsfin{p}{\distrtypone,\seone}$.
  \item
  Let $0 \leq r < pp'+(1-p)p''$. Let $(r',r'')$ be such that $r = pr' + (1-p)r''$,
  $0 \leq r' < p'$ and $0 \leq r'' < p''$.
  Since $\termtwo \in \setredtmsfin{p'}{\distrtypone,\seone}$, there exists $n_{r'}$
  and $\distrtypone_{r'} \distrleq \distrtypone$ such that $\termtwo \redval^{n_{r'}} \distrone_{r'} \in \setreddists{r'}{\distrtypone_{r'},\seone}$.
  Since $\termthree \in \setredtmsfin{p''}{\distrtyptwo,\seone}$, there exists $m_{r''}$
  and $\distrtyptwo_{r''} \distrleq \distrtyptwo$ such that $\termthree \redval^{m_{r''}} \distrtwo_{r''} \in \setreddists{r''}{\distrtyptwo_{r''},\seone}$.
  Suppose that $n_{r'} \leq m_{r''}$, the dual case being exactly symmetrical.
  By Lemma~\ref{lemma:pumping-dred}, by denoting $\distrone_{r''}$ the distribution such that $\termtwo \redval^{m_{r''}} \distrone_{r''}$,
  there exists $\distrtypone_{r'} \distrleq\distrtypone_{r''} \distrleq \distrtypone$ such that
  $\distrone_{r''} \in \setreddists{r'}{\distrtypone_{r''},\seone}$.
  Now $\termone \redval^{m_{r''}+1} p\distrone_{r''} + (1-p) \distrtwo_{r''}$,
  and by Lemma~\ref{lemma:backtracking-dred} we have 
  $p\distrone_{r''} + (1-p) \distrtwo_{r''} \in \setreddists{pr'+(1-p)r''}{\distrtypone_{r''} \choice_p \distrtyptwo_{r''},\seone}$.
  Since by construction $\distrtypone_{r''} \choice_p \distrtyptwo_{r''} \distrleq \distrtypone \choice_p \distrtyptwo$,
  we can conclude that $\termone \in \setredtmsfin{pp'+(1-p)p''}{\distrtypone \choice_p \distrtyptwo,\seone}$.
  %
  %
  %
 \end{itemize}

\end{proof}
}

\longv{\subsection{Subtyping Soundness}}

\longv{Reducibility sets are monotonic with respect to the subtyping order $\subtypeleq$:}

\longv{ 
\begin{lemma}[Subtyping Soundness]~
 \label{lemma:subtyping-and-candidates}
 
 \begin{itemize}
  \item Suppose that $\typone \subtypeleq \typtwo$. Then, for every $p \in [0,1]$ and $\seone$, $\setredvals{p}{\typone,\seone} \subseteq \setredvals{p}{\typtwo,\seone}$.
  \item Suppose that $\distrtypone \subtypeleq \distrtyptwo$ and that $\distrsum{\distrtypone}\,=\,\distrsum{\distrtyptwo}$. Then, for every $p \in [0,1]$ and $\seone$, $\setreddists{p}{\distrtypone,\seone} \subseteq
  \setreddists{p}{\distrtyptwo,\seone}$.
  \item Suppose that $\distrtypone \subtypeleq \distrtyptwo$.
  Then, for every $p \in [0,1]$ and $\seone$, $\setredtmsfin{p}{\distrtypone,\seone} \subseteq \setredtmsfin{p}{\distrtyptwo,\seone}$.
 \end{itemize}

\end{lemma}

}

\longv{ 
\begin{proof} The proof is by mutual induction on the statements following the shape of the simple type refined 
by $\typone$ and $\distrtypone$, as earlier.
 \begin{itemize}
  \item Suppose that $\typone \refines \Nat$. Then $\typone \,=\,\Nat^{\sizeone}$ and $\typtwo \,=\,\Nat^{\sizetwo}$
  with $\sizeone \sizeleq \sizetwo$. Let $\valone \in \setredvals{p}{\typone,\seone}$. There are three possibilities:
  \begin{itemize}
  \item Either $\sizeone\,=\,\sizesuccit{\sizevarone}{k}$
  and $\sizetwo\,=\,\sizesuccit{\sizevarone}{k'}$ with $k \leq k'$.
  Then $\valone$ is of the shape $\natsucc^n\ \natzero$. If $p=0$ the result is immediate. Else we have
  $n < \sesem{\sizeone}{\seone}\,=\,\seone(\sizevarone)+k \leq 
  \seone(\sizevarone)+k' \,=\,\sesem{\sizetwo}{\seone}$ so that $\valone \in \setredvals{p}{\typtwo,\seone}$.
  \item Or $\sizeone\,=\,\sizesuccit{\sizevarone}{k}$ and $\sizetwo\,=\,\infty$. In this case $\valone$ is of the shape $\natsucc^n\ \natzero$
  and therefore $\valone \in \setredvals{p}{\typtwo,\seone}$.
  \item Or $\sizeone\,=\,\sizetwo\,=\,\sizeinf$. In this case $\typone\,=\,\typtwo$ and the result is immediate.
  \end{itemize}
  \item Suppose that $\typone\,=\,\typone' \typarrow \distrtypone$ and that $\typtwo\,=\,\typtwo' \typarrow \distrtyptwo$.
  Let $p \in [0,1]$, $\seone$ be a size environment, and $\valone \in \setredvals{p}{\typone,\seone}$.
  We have that $\typtwo' \subtypeleq \typone'$ and $\distrtypone \subtypeleq \distrtyptwo$.
  It follows, by induction hypothesis, that $\setredvals{p'}{\typtwo',\seone} \subseteq \setredvals{p'}{\typone',\seone}$
  and that $\setredtmsfin{p'}{\distrtypone,\seone} \subseteq \setredtmsfin{p'}{\distrtyptwo,\seone}$
  for every $p' \in [0,1]$.
  Since $\valone \in \setredvals{p}{\typone,\seone}$, for every $q \in (0,1]$ and $\valtwo \in \setredvals{q}{\typone',\seone}$,
  $\valone\ \valtwo \in \setredtmsfin{pq}{\distrtypone,\seone} \subseteq \setredtmsfin{pq}{\distrtyptwo,\seone}$.
  As $\setredvals{q}{\typtwo',\seone} \subseteq \setredvals{q}{\typone',\seone}$, $\valone \in \setredvals{p}{\typtwo,\seone}$.
  \item Suppose that $\distrtypone\,=\,\distrelts{\typone_{\indextwo}^{p'_{\indextwo}} \sep \indextwo \in \indexsettwo}$
  and that $\distrtyptwo\,=\,\distrelts{\typtwo_{\indexthree}^{p''_{\indexthree}} \sep \indexthree \in \indexsetthree}$.
  By definition of subtyping, there exists $\funcone \,:\, \indexsettwo \to \indexsetthree$ such that for all $\indextwo \in \indexsettwo$, $\typone_{\indextwo}\,\subtypeleq\,\typtwo_{\funcone(\indextwo)}$ and that 
  for all $\indexthree \in \indexsetthree,$ $\sum_{\indextwo \in \funcone^{-1}(\indexthree)}\ p'_{\indextwo} \leq p''_{\indexthree}$.
  Note that since $\distrsum{\distrtypone}\,=\,\distrsum{\distrtyptwo}$, this is in fact an equality.
  Let $\distrone \in \setreddists{p}{\distrtypone,\seone}$, then there exists a pseudo-representation
  $\distrone\,=\,\pseudorep{\left(\valone_{\indexone}\right)^{p_{\indexone}} \sep \indexone \in \indexsetone}$
  and families $\left(p_{\indexone \indextwo}\right)_{\indexone\in \indexsetone,\indextwo\in\indexsettwo}$
  and $\left(q_{\indexone \indextwo}\right)_{\indexone\in \indexsetone,\indextwo\in\indexsettwo}$ of reals of $[0,1]$
  satisfying:
  \begin{enumerate}
   \item $\forall \indexone \in \indexsetone,\ \ \forall\indextwo \in \indexsettwo,\ \ \valone_\indexone \in 
   \setredvals{q_{\indexone\indextwo}}{\typone_\indextwo,\seone}$,
   \item $\forall \indexone \in \indexsetone,\ \ \sum_{\indextwo \in \indexsettwo}\ p_{\indexone\indextwo}\ =\ p_\indexone$,
   \item $\forall \indextwo \in \indexsettwo,\ \ \sum_{\indexone \in \indexsetone}\ p_{\indexone\indextwo}\ =\ p'_\indextwo$,
   \item $p \leq \sum_{\indexone \in \indexsetone}\sum_{\indextwo \in \indexsettwo}\ q_{\indexone\indextwo}p_{\indexone\indextwo}$.
  \end{enumerate}
  By induction hypothesis, for every $\indextwo \in \indexsettwo$, $\setredvals{q_{\indexone\indextwo}}{\typone_\indextwo,\seone}
  \subseteq \setredvals{q_{\indexone\indextwo}}{\typtwo_{f(\indextwo)},\seone}$.
  We now prove that $\pseudorep{\left(\valone_{\indexone}\right)^{p_{\indexone}} \sep \indexone \in \indexsetone}$
  witnesses that $\distrone \in \setreddists{p}{\distrtyptwo,\seone}$. We need to define families
  of reals $\left(p'_{\indexone \indexthree}\right)_{\indexone\in \indexsetone,\indexthree\in\indexsetthree}$
  and $\left(q'_{\indexone \indexthree}\right)_{\indexone\in \indexsetone,\indexthree\in\indexsetthree}$ satisfying the four usual conditions.
  To this end, for every $\indexone\in \indexsetone,\ \indexthree\in\indexsetthree$,  we set
  $$
  p'_{\indexone\indexthree}\ \ =\ \ \sum_{\indextwo \in f^{-1}(\indexthree)}\ p_{\indexone\indextwo}
  $$
  and
  $$
  q'_{\indexone\indexthree}\ \ =\ \ \max_{\indextwo \in f^{-1}(\indexthree)}\ q_{\indexone \indextwo}
  $$
  Let us check that the four conditions hold.
  \begin{enumerate}
   \item $\forall \indexone \in \indexsetone,\ \ \forall\indexthree \in \indexsetthree,\ \ \valone_\indexone \in 
   \setredvals{q'_{\indexone\indexthree}}{\typtwo_{f(\indextwo)},\seone}$ by induction hypothesis and by definition of 
   $q'_{\indexone\indexthree}$,

   \item $\forall \indexone \in \indexsetone,\ \ \sum_{\indexthree \in \indexsetthree}\ p'_{\indexone\indexthree}\ =\ 
   \sum_{\indexthree \in \indexsetthree}\ \sum_{\indextwo \in f^{-1}(\indexthree)}\ p_{\indexone\indextwo}\ =\
   \sum_{\indextwo \in \indexsettwo}\ p_{\indexone\indextwo} \ =\ 
   p_\indexone$,

   \item $\forall \indexthree \in \indexsetthree,\ \ \sum_{\indexone \in \indexsetone}\ p'_{\indexone\indexthree} \ =\ 
   \sum_{\indexone \in \indexsetone}\ \sum_{\indextwo \in f^{-1}(\indexthree)}\ p_{\indexone\indextwo}
   \ =\  \sum_{\indextwo \in f^{-1}(\indexthree)}\ \sum_{\indexone \in \indexsetone}\ p_{\indexone\indextwo}
   \ =\  \sum_{\indextwo \in f^{-1}(\indexthree)}\ p'_\indextwo
   \ =\ p''_\indexthree$,
   
   \item 
   $$
   \begin{array}{rcl}
   p & \ \leq \  &\sum_{\indexone \in \indexsetone}\ \sum_{\indextwo \in \indexsettwo}\ q_{\indexone\indextwo}p_{\indexone\indextwo}\\
   & \ = \  &\sum_{\indexone \in \indexsetone}\ \sum_{\indexthree \in \indexsetthree}\ \sum_{\indextwo \in f^{-1}(\indexthree)}\ q_{\indexone\indextwo}p_{\indexone\indextwo}\\
   & \ \leq \  &\sum_{\indexone \in \indexsetone}\ \sum_{\indexthree \in \indexsetthree}\ \sum_{\indextwo \in f^{-1}(\indexthree)}\ q'_{\indexone f(\indextwo)}p_{\indexone\indextwo}\\
   & \ = \  &\sum_{\indexone \in \indexsetone}\ \sum_{\indexthree \in \indexsetthree}\ q'_{\indexone\indexthree} \sum_{\indextwo \in f^{-1}(\indexthree)}\ p_{\indexone\indextwo}\\
   & \ =\ & \sum_{\indexone \in \indexsetone}\sum_{\indexthree \in \indexsetthree}\ q'_{\indexone\indexthree}p'_{\indexone\indexthree}\\
   \end{array}
   $$
  \end{enumerate}
  It follows that $\distrone \in \setreddists{p}{\distrtyptwo,\seone}$.

 \item Suppose that $\distrtypone\,=\,\distrelts{\typone_{\indextwo}^{p'_{\indextwo}} \sep \indextwo \in \indexsettwo}$
  and that $\distrtyptwo\,=\,\distrelts{\typtwo_{\indexthree}^{p''_{\indexthree}} \sep \indexthree \in \indexsetthree}$.
  By definition of subtyping, there exists $\funcone \,:\, \indexsettwo \to \indexsetthree$ such that for all $\indextwo \in \indexsettwo$, $\typone_{\indextwo}\,\subtypeleq\,\typtwo_{\funcone(\indextwo)}$ and that 
  for all $\indexthree \in \indexsetthree,$ $\sum_{\indextwo \in \funcone^{-1}(\indexthree)}\ p'_{\indextwo} \leq p''_{\indexthree}$.
  Let $\termone \in \setredtmsfin{p}{\distrtypone,\seone}$. Then, for every $0 \leq r < p$, there exists
  $\distrtypone'_r \distrleq \distrtypone$ and $n_r$ such that $\termone \redval^{n_r} \distrone_r \in \setreddists{r}{\distrtypone'_r,\seone}$.
  By definition of $\distrtypone'_r \distrleq \distrtypone$,
  $\distrtypone'_r\ =\ \pseudorep{\typone_{\indextwo}^{q'_{\indextwo}} \sep \indextwo \in \indexsettwo}$
  with $q'_\indextwo \leq p'_\indextwo$ for every $\indextwo \in \indexsettwo$.
  We set $\distrtyptwo'_r\ =\ \pseudorep{\typtwo_{f(\indextwo)}^{q'_{\indextwo}} \sep \indextwo \in \indexsettwo}$
  which is such that $\distrsum{\distrtypone'_r} = \distrsum{\distrtyptwo'_r}$ and, by construction,
  $\distrtypone'_r \subtypeleq \distrtyptwo'_r$ so that we can apply the induction hypothesis and obtain
  that $\termone \redval^{n_r} \distrone_r \in \setreddists{r}{\distrtyptwo'_r,\seone}$. The result follows,
  since by construction $\distrtyptwo'_r \distrleq \distrtyptwo$.
 \end{itemize}

\end{proof}

}

\longv{\subsection{Reducibility Sets for Open Terms}}
\shortv{\paragraph{Extension to Open Terms.}}
We are now ready to extend the notion of reducibility
set from the realm of \emph{closed} terms to the one of \emph{open}
terms. This turns out to be subtle. The guiding intuition is that one would
like to define a term $\termone$ with free variables in $\vec{\varone}$
to be reducible iff any closure $\subst{\termone}{\vec{\varone}}{\vec{\valone}}$
is itself reducible in the sense of Definition \ref{def:redsets}. What happens, however,
to the underlying degree of reducibility $p$? How do we relate the degrees of reducibility
of $\vec{\valone}$ with the one of $\subst{\termone}{\vec{\varone}}{\vec{\valone}}$?
The answer is contained in the following definition:
\begin{definition}[Reducibility Sets for Open Terms]
  Suppose that $\contextsizedone$ is a sized context in the form
  $\varone_1\typsep\typone_1,\ldots,\varone_n\typsep\typone_n$,
  and that $\vartwo$ is a variable distinct from $\varone_1,\ldots,\varone_n$.
  Then we define the following sets of terms and values:
  $$
  \begin{array}{rcll}
    \setredopentms{\contextone \contextsep \emptyset}{\distrtypone,\seone}
    & \ \ =\ \ &
    \left\{\termone \sep \right.
    & \forall (q_\indexone)_\indexone \in [0,1]^n,\ \ \forall \left(\valone_1,\,\ldots,\,\valone_n\right) \in \prod_{\indexone=1}^n \ 
    \setredvals{q_\indexone}{\typone_\indexone,\seone},\ \ \\[0.15cm]
    & & &
    \qquad \qquad \qquad \qquad \qquad \quad
    \left. \subst{\termone}{\vec{\varone}}{\vec{\valone}}\in\setredtmsfin{\prod_{\indexone=1}^{n}q_\indexone}{\distrtypone,\seone}\right\}\\[0.5cm]
    \setredopenvals{\contextone \contextsep \emptyset}{\distrtypone,\seone}
    & \ \ =\ \ &
    \left\{\valtwo \sep \right.
    & \forall (q_\indexone)_\indexone \in [0,1]^n,\ \ \forall \left(\valone_1,\,\ldots,\,\valone_n\right) \in \prod_{\indexone=1}^n \ 
    \setredvals{q_\indexone}{\typone_\indexone,\seone},\ \ \\[0.15cm]
    & & &
    \qquad \qquad \qquad \qquad \qquad \quad 
    \left. \subst{\valtwo}{\vec{\varone}}{\vec{\valone}}\in\setredvals{\prod_{\indexone=1}^{n}q_\indexone}{\distrtypone,\seone}\right\}\\[0.5cm]
    \end{array}
    $$
    $$
    \begin{array}{rcll}
    
    \setredopentms{\contextone \contextsep \vartwo\typsep\{\typtwo_\indextwo^{p_\indextwo}\}_{\indextwo\in \indexsettwo}}{\distrtypone,\seone}
    & \ \ =\ \ &
    \left\{\termone \sep \right.
    & \forall (q_\indexone)_\indexone \in [0,1]^n,\ \ \forall \vec{\valone} \in \prod_{\indexone=1}^n \ 
    \setredvals{q_\indexone}{\typone_\indexone,\seone},\ \ \\[0.15cm]
    & & & \forall \left(q'_\indextwo\right)_{\indextwo} \in [0,1]^\indexsettwo,\ \ \forall\valtwo\in \bigcap_{\indextwo \in \indexsettwo}\ 
    \setredvals{q'_\indextwo}{\typone_\indextwo,\seone},\ \ \\[0.15cm]
    & & & \qquad \qquad\qquad \ \ 
    \left. \subst{\termone}{\vec{\varone},\vartwo}{\vec{\valone},\valtwo}\in\setredtmsfin{\alpha}{\distrtypone,\seone}\right\}\\[0.5cm]
    \setredopenvals{\contextone \contextsep \vartwo\typsep\{\typtwo_\indextwo^{p_\indextwo}\}_{\indextwo\in \indexsettwo}}{\distrtypone,\seone}
    & \ \ =\ \ &
    \left\{\valthree \sep \right.
    & \forall (q_\indexone)_\indexone \in [0,1]^n,\ \ \forall \vec{\valone} \in \prod_{\indexone=1}^n \ 
    \setredvals{q_\indexone}{\typone_\indexone,\seone},\ \ \\[0.15cm]
   & & & \forall \left(q'_\indextwo\right)_{\indextwo} \in [0,1]^\indexsettwo,\ \ \forall\valtwo\in \bigcap_{\indextwo \in \indexsettwo}\ 
    \setredvals{q'_\indextwo}{\typone_\indextwo,\seone},\ \ \\[0.15cm]
    & & &  \qquad \qquad\qquad \ \ 
    \left. \subst{\valthree}{\vec{\varone},\vartwo}{\vec{\valone},\valtwo}\in\setredvals{\alpha}{\distrtypone,\seone}\right\}\\[0.5cm]
  \end{array}
  $$
  where
  \shortv{$
  \alpha\ \ =\ \ \left(\prod_{\indexone=1}^n q_\indexone \right)
      \left(\left(\sum_{\indextwo\in \indexsettwo} p_\indextwo q'_\indextwo\right) + 1 - \left(\sum_{\indextwo\in\indexsettwo}p_\indextwo\right)\right)
  $}
  is called the degree of reducibility.
  \longv{the degree of reducibility $\alpha$ is defined as
  $$
  \alpha\ \ =\ \ \left(\prod_{\indexone=1}^n q_\indexone \right)
      \left(\left(\sum_{\indextwo\in \indexsettwo} p_\indextwo q'_\indextwo\right) + 1 - \left(\sum_{\indextwo\in\indexsettwo}p_\indextwo\right)\right).
  $$}
  \longv{
  Note that this contains:
    $$
  \begin{array}{rcll}
    \setredopentms{\emptyset\contextsep\emptyset}{\distrtypone,\seone}
    & \ \ =\ \ & \setredtmsfin{1}{\distrtypone,\seone} \!\!\!\!\!\!\!\!\!\!\!\!\!\!\!\!\!\!\!\!\!\!\!
    \!\!\!\!\!\!\!& \\[0.2cm]
    \setredopenvals{\emptyset\contextsep\emptyset}{\typone,\seone}
    & \ \ =\ \ & \setredvals{1}{\typone,\seone} \!\!\!\!\!\!\!\!\!\!\!\!\!\!\!\!\! & \\[0.2cm]
    \setredopentms{\emptyset \contextsep \vartwo\typsep\{\typtwo_i^{p_i}\}_{i\in I}}{\distrtypone,\seone}
    & \ \ =\ \ &
    \left\{\termone \sep \right.&
    \forall \left(q_\indexone\right)_{\indexone} \in [0,1]^\indexsetone,\ \ \forall\valone\in \bigcap_{\indexone \in \indexsetone}\ 
    \setredvals{q_i}{\typone_i,\seone},\ \ \\[0.15cm]
    & & &
    \qquad \qquad \qquad \qquad 
    \left. \subst{\termone}{\vartwo}{\valone}\in\setredtmsfin{\sum_{i\in I}p_iq_i + 1 - \left(\sum_{\indextwo\in\indexsettwo}p_\indextwo\right)}{\distrtypone,\seone}\right\}\\[0.5cm]
    \setredopenvals{\emptyset \contextsep \vartwo\typsep\{\typtwo_i^{p_i}\}_{i\in I}}{\distrtypone,\seone}
    & \ \ =\ \ &
    \left\{\valtwo \sep \right.&
    \forall \left(q_\indexone\right)_{\indexone} \in [0,1]^\indexsetone,\ \ \forall\valone\in \bigcap_{\indexone \in \indexsetone}\ 
    \setredvals{q_i}{\typone_i,\seone},\ \ \\[0.15cm]
    & & &
    \qquad \qquad \qquad \qquad 
    \left. \subst{\valtwo}{\vartwo}{\valone}\in\setredvals{\sum_{i\in I}p_iq_i + 1 - \left(\sum_{\indextwo\in\indexsettwo}p_\indextwo\right)}{\distrtypone,\seone}\right\}\\[0.5cm]
  \end{array}
  $$}
  \shortv{$\!\!$}Note \longv{also}that these sets extend the ones for closed terms: in particular,
  $\setredopentms{\emptyset \contextsep \emptyset}{\distrtypone,\seone}\,=\,\setredtms{1}{\distrtypone,\seone}$.

\end{definition}


\longv{As for closed terms (Lemma~\ref{lemma:values-are-in-tred-iff-in-vred}), reducible values are reducible terms:}

\begin{lemma}[Reducible Values are Reducible Terms]
\label{lemma/inclusion-redvals-redterms}
 For every $\contextsizedone,\, \contextdistrone,\,\typone$ and $\seone$,
 $\valone \in \setredopenvals{\contextsizedone\contextsep \contextdistrone}{\typone,\seone}$
 if and only if $\valone \in \setredopentms{\contextsizedone \contextsep \contextdistrone}{\distrelts{\typone^1},\seone}$.
 An immediate consequence is that
  $\setredopenvals{\contextsizedone\contextsep \contextdistrone}{\typone,\seone}\ 
 \subseteq\ \setredopentms{\contextsizedone \contextsep \contextdistrone}{\distrelts{\typone^1},\seone}$.
 
\end{lemma}

\longv{
\begin{proof}
 Corollary of Lemma~\ref{lemma:values-are-in-tred-iff-in-vred} and of the definitions of the candidates for open sets.
\end{proof}
}

\longv{The following easy lemma relates the reducibility of natural numbers, and will be used to 
treat the case of the rules Succ and Zero in the proof of typing soundness:}

\longv{
\begin{lemma}
\label{lemma:successor-on-redval}
\begin{itemize}
 \item $\valone \in \setredvals{p}{\Nat^{\sizeone},\seone} \ \ \implies\ \ \natsucc\ \valone \in \setredvals{p}{\Nat^{\sizesucc{\sizeone}},\seone}$
 \item For every size $\sizeone$, $\natzero \in \setredvals{p}{\Nat^{\sizesucc{\sizeone}},\seone}$.
\end{itemize}
\end{lemma}
}

\longv{ 
\begin{proof}
First point:
 \begin{itemize}
  \item Suppose that $\valone \in \setredvals{p}{\Nat^{\sizesuccit{\sizevarone}{k}},\seone}$ and that $p > 0$.
  Then $\valone\,=\,\natsucc^n\ \natzero$ for some $ n < \sesem{\sizeone}{\seone}$.
  Then $\natsucc\ \valone\,=\,\natsucc^{n+1}\ \natzero$ satisfies $n+1 < \sesem{\sizesucc{\sizeone}}{\seone}= \sesem{\sizeone}{\seone}+1$,
  so that $\natsucc\ \valone \in \setredvals{p}{\Nat^{\sizesuccit{\sizevarone}{k+1}},\seone}$.
  \item Suppose that $\valone \in \setredvals{p}{\Nat^{\sizeinf},\seone}$ or that $p=0$. 
  By definition, $\valone \,=\,\natsucc^n\ \natzero$ for $n \in \NN$. It follows that 
  $\natsucc\ \valone \in \setredvals{p}{\Nat^{\sizesucc{\sizeone}},\seone}$.
 \end{itemize}
 
 Second point:
 \begin{itemize}
  \item Suppose that $p=0$. Then $\natzero \in \setredvals{p}{\Nat^{\sizesucc{\sizeone}},\seone}$, by definition.
  \item Else we need to prove that $\sesem{\sizeone}{\seone} > 0$. But $\sizesucc{\sizeone}$ is either
  $\sizeinf$, in which case $\sesem{\sizeone}{\seone}\,=\,\infty$, or it is of the shape 
  $\sizesuccit{\sizevarone}{k}$ for $k > 0$, and $\sesem{\sizesuccit{\sizevarone}{k}}{\seone}\ =\ \seone(\sizevarone) + k > 0$.
 \end{itemize}

\end{proof}
}

\longv{\subsection{Reducibility and Sized Walks}}
\shortv{\paragraph*{Reducibility and Sized Walks.}}
To handle the fixpoint rule, we need to relate the notion of sized walk
which guards it with the reducibility sets, and in particular with the degrees
of reducibility we can attribute to recursively-defined terms.
\begin{definition}[Probabilities of Convergence in Finite Time]
 Let us consider a sized walk. We define the associated
\emph{probabilities of convergence in finite time}
$\left(\probaconv{n}{m}\right)_{n \in \NN,m\in\NN}$ as follows: $
\forall n \in \NN,\ \ \forall m \in \NN, $
the real number $\probaconv{n}{m}$ is defined as the probability
that, starting from $m$, the sized walk reaches $0$ in
\emph{at most} $n$ steps.
\end{definition}
The point is that, for an AST sized walk, the more we iterate, the closer we get
to reaching 0 in finite time $n$ with probability $1$.
\begin{lemma}[Finite Approximations of AST]
\label{lemma:sized-walks-get-arbitrarily-close-to-proba-one-in-finite-time}
Let $m \in \NN$ and $\epsilon \in (0,1]$. Consider a sized walk, and its associated probabilities of convergence in finite time
$\left(\probaconv{n}{m}\right)_{n \in \NN,m\in\NN}$. If the sized walk is AST, there exists $n \in \NN$ such that
$\probaconv{n}{m} \geq 1 - \epsilon$. 
\end{lemma}

\longv{ 
\begin{proof}
 Suppose, by contradiction, that there exists $\epsilon \in (0,1]$ such that there is no $n \in \NN$
 with $\probaconv{n}{m} \geq 1 - \epsilon$. Then $\lim_{n \in \NN} \probaconv{n}{m} \leq 1 - \epsilon$. But this limit 
 should be worth 1 as we supposed the sized walk to be AST.
 
\end{proof}
}

\longv{The following lemma allows to treat the base case of Lemma~\ref{lemma:sized-walk-argument-for-letrec}:
\begin{lemma}
\label{lemma/base-case-letrec}
 Suppose that $\valone$ is a closed value of simple type $\Nat \typarrow \simpletypone$.
 Then, for every $\Nat^\sizevarone \typarrow \distrtypone \refines \Nat \typarrow \simpletypone$,
 and for every size environment $\seone$ such that $\seone(\sizevarone) = 0$,
 we have $\valone \in \setredvals{1}{\Nat^\sizevarone \typarrow \distrtypone,\seone}$.
\end{lemma}
}

\longv{
\begin{proof}
 To prove that $\valone \in \setredvals{1}{\Nat^\sizevarone \typarrow \distrtypone,\seone}$, we need to show that
 for every $q \in (0,1]$ and every $\valtwo \in \setredvals{q}{\Nat^{\sizevarone},\seone}$
 we have that $\valone\ \valtwo \in \setredtmsfin{q}{\distrtypone,\seone}$.
 This is always the case, as $\setredvals{q}{\Nat^{\sizevarone},\seone}$ is the empty set by definition:
 there is no term of the shape $\natsucc^n\ \natzero$ with $n < \seone(\sizevarone) = 0$.
\end{proof}

}

The following lemma is the crucial result relating sized walks with
the reducibility sets. It proves that, when the sized walk is AST, and
after substitution of the variables of the context by reducible values
in the recursively-defined term, we can prove the degree of
reducibility to be any probability $\probaconv{n}{m}$ of convergence
in finite time.
\begin{lemma}[Convergence in Finite Time and $\letrecname$]
\label{lemma:sized-walk-argument-for-letrec}
Consider the distribution type
$\distrtypone\ = \distrelts{\left(\Nat^{\sizeone_{\indextwo}} \typarrow \subst{\distrtyptwo}{\sizevarone}{\sizeone_{\indextwo}}
\right)^{p_{\indextwo}} \sep \indextwo \in \indexsettwo}$.
Let $\contextsizedone$ be the sized context $\varone_1\typsep\Nat^{\sizetwo_1},\ldots,\,\varone_l\typsep\Nat^{\sizetwo_l}$.
Suppose that
$
\contextsizedone \contextsep
\funcone \typsep \distrtypone \proves \valone \typsep \Nat^{\sizesucc{\sizevarone}} \typarrow \subst{\distrtyptwo}{\sizevarone}{\sizesucc{\sizevarone}}
$
and that $\distrtypone$ induces an AST sized walk. 
Denote $\left(\probaconv{n}{m}\right)_{n \in \NN,m\in\NN}$ its associated probabilities of convergence in finite time.
Suppose that $\valone\in\setredopenvals{\contextsizedone \contextsep \funcone \typsep \distrtypone}{\Nat^{ \sizesucc{\sizevarone}} \typarrow \subst{\distrtyptwo}{\sizevarone}{\sizesucc{\sizevarone}},\seone}$ for every $\seone$.
%
Let $\vec{\valtwo} \in \prod_{\indexone=1}^l\ \setredvals{1}{\Nat^{\sizetwo_i},\seone}$, then for every $(n,m) \in \NN^2$, we have that
$$
\letrec{\funcone}{\subst{\valone}{\vec{\varone}}{\vec{\valtwo}}}\ \ \in\ \ \setredvals{\probaconv{n}{m}}{\Nat^\sizevarone \typarrow \distrtyptwo,\seone\left[\sizevarone \mapsto m\right]}
$$
 
\end{lemma}
%
%
\shortv{\begin{proof} We give a sketch of the proof, to be found in the long version~\cite{dal-lago-grellois:monadic-affine-sized-types-full}.
The proof is by recurrence on $n$. The main case relies on the decomposition
$
  \probaconv{n+1}{m'+1}\ \ =\ \ \sum_{\indextwo \in \indexsettwo}\ p_\indextwo \probaconv{n}{m' + k_\indextwo}\ 
  \ +\ \ 1 - \left(\sum_{\indextwo \in \indexsettwo}\ p_{\indextwo}\right)
$.
The induction hypothesis allows then to state that for every $\indextwo \in \indexsettwo$
we have $\letrec{\funcone}{\subst{\valone}{\vec{\varone}}{\vec{\valtwo}}}\ \ \in\ \ \setredvals{\probaconv{n}{m' + k_\indextwo}}{\Nat^\sizevarone \typarrow \distrtyptwo,\seone\left[\sizevarone \mapsto m' + k_\indextwo\right]}$.
We use the Size Commutation lemma (Lemma~\ref{lemma:exchange-size-size-env})
to obtain that $\letrec{\funcone}{\subst{\valone}{\vec{\varone}}{\vec{\valtwo}}}$ is in an appropriate intersection
of reducibility sets, and the hypothesis that $\valone\in\setredopenvals{\contextsizedone \contextsep \funcone \typsep \distrtypone}{\Nat^{ \sizesucc{\sizevarone}} \typarrow \subst{\distrtyptwo}{\sizevarone}{\sizesucc{\sizevarone}},\seone\left[\sizevarone \mapsto m'\right]}$ then implies
that 
$\subst{\valone}{\vec{\varone},\funcone}{\vec{\valtwo},\letrec{\funcone}{\subst{\valone}{\vec{\varone}}{\vec{\valtwo}}}}
  \ \in\ \setredvals{\probaconv{n+1}{m'+1}}{\Nat^{\sizevarone} \typarrow \distrtyptwo,\seone\left[\sizevarone \mapsto m'+1\right]}$,
using the Size Commutation lemma once again.
As this term is an unfolding of $\letrec{\funcone}{\subst{\valone}{\vec{\varone}}{\vec{\valtwo}}}$, we conclude using 
Proposition~\ref{corollary:unfolding-stability-vred}.\qed
\end{proof}
}

\longv{ 
\begin{proof}
We prove the statement by induction on $n$.
\begin{itemize}
 \item If $n=0$, we have two cases.
 \begin{itemize}
  \item If $m = 0$, then Lemma~\ref{lemma/base-case-letrec} implies that
  $\letrec{\funcone}{\subst{\valone}{\vec{\varone}}{\vec{\valtwo}}}\ \ \in\ \ \setredvals{1}{\Nat^\sizevarone \typarrow \distrtyptwo,\seone\left[\sizevarone \mapsto 0\right]}$
  so that by downward closure (Lemma~\ref{lemma/downward-closure-tred}) we obtain
  $\letrec{\funcone}{\subst{\valone}{\vec{\varone}}{\vec{\valtwo}}}\ \ \in\ \ \setredvals{\probaconv{n}{0}}{\Nat^\sizevarone \typarrow \distrtyptwo,\seone\left[\sizevarone \mapsto 0\right]}$.
  \item If $m \neq 0$, then $\probaconv{n}{m} = 0$. The hypothesis of the lemma imply that
  $\letrec{\funcone}{\subst{\valone}{\vec{\varone}}{\vec{\valtwo}}} \refines \Nat \typarrow \underlying{\distrtyptwo}$, and we conclude using
  Lemma~\ref{lemma:simple-types-and-candidates-of-proba-zero}.
 \end{itemize}
 \item Suppose that $n \geq 1$.
 \begin{itemize}
  \item If $m = 0$, the result is immediate as in the previous case.

  \item Suppose that $m >0$. Then $m = m' +1$.
  By definition, $\sizeone_\indextwo$ must be of the shape
$\sizesuccit{\sizevarone}{k_\indextwo}$ with $k_{\indextwo} \geq 0$
for every $\indextwo \in \indexsettwo$.  We set
$\indexsetone\ =\ \left\{k_\indextwo \sep \indextwo \in
\indexsettwo\right\}$ and $q_{k_\indextwo}\,=\,p_\indextwo$ for
every $\indextwo \in \indexsettwo$.  The sized walk induced by the
distribution type $\distrtypone$ is then the sized walk associated to
$\left(\indexsetone,\left(q_\indexone)_{\indexone \in
  \indexsetone}\right)\right)$, which from the state $m' +1 \in\NN\setminus\{0\}$ moves:
  \begin{varitemize}
  \item 
    to the state $m' + k_\indextwo$ with probability $p_\indextwo$, 
    for every $\indextwo \in \indexsettwo$;
  \item 
    to $0$ with probability $1 -
    \left(\sum_{\indextwo \in \indexsettwo}\ p_{\indextwo}\right)$.
  \end{varitemize}
  It follows that
  \begin{equation}
  \label{eq:decomposition-proba}
  \probaconv{n+1}{m'+1}\ \ =\ \ \sum_{\indextwo \in \indexsettwo}\ p_\indextwo \probaconv{n}{m' + k_\indextwo}\ 
  \ +\ \ 1 - \left(\sum_{\indextwo \in \indexsettwo}\ p_{\indextwo}\right)
  \end{equation}
  For every $\indextwo \in \indexsettwo$, let us apply the induction hypothesis and obtain
  $$
  \letrec{\funcone}{\subst{\valone}{\vec{\varone}}{\vec{\valtwo}}}\ \ \in\ \ \setredvals{\probaconv{n}{m'+k_\indextwo}}{\Nat^\sizevarone \typarrow \distrtyptwo,\seone\left[\sizevarone \mapsto m'+k_\indextwo\right]}
  $$
  By Lemma~\ref{lemma:exchange-size-size-env},
  $$
  \letrec{\funcone}{\subst{\valone}{\vec{\varone}}{\vec{\valtwo}}}\ \ \in\ \ \setredvals{\probaconv{n}{m'+k_\indextwo}}{\Nat^{\sizesuccit{\sizevarone}{k_\indextwo}} \typarrow \subst{\distrtyptwo}{\sizevarone}{\sizesuccit{\sizevarone}{k_\indextwo}},\seone\left[\sizevarone \mapsto m'\right]}
  \ =\ 
  \setredvals{\probaconv{n}{m'+k_\indextwo}}{\Nat^{\sizeone_{\indextwo}} \typarrow \subst{\distrtyptwo}{\sizevarone}{\sizeone_{\indextwo}},\seone\left[\sizevarone \mapsto m'\right]}
  $$
  Since this is valid for every $\indextwo \in \indexsettwo$, we have that
  $$
  \letrec{\funcone}{\subst{\valone}{\vec{\varone}}{\vec{\valtwo}}}\ \ \in\ \ 
  \bigcap_{\indextwo \in \indexsettwo}
  \setredvals{\probaconv{n}{m'+k_\indextwo}}{\Nat^{\sizeone_{\indextwo}} \typarrow \subst{\distrtyptwo}{\sizevarone}{\sizeone_{\indextwo}},\seone\left[\sizevarone \mapsto m'\right]}
  $$
  and since $\valone\in\setredopenvals{\contextsizedone \contextsep \funcone \typsep \distrtypone}{\Nat^{ \sizesucc{\sizevarone}} \typarrow \subst{\distrtyptwo}{\sizevarone}{\sizesucc{\sizevarone}},\seone\left[\sizevarone \mapsto m'\right]}$
  we obtain
  $$
  \subst{\valone}{\vec{\varone},\funcone}{\vec{\valtwo},\letrec{\funcone}{\subst{\valone}{\vec{\varone}}{\vec{\valtwo}}}}
  \ \in\ \setredvals{\sum_{\indextwo \in \indexsettwo}\ p_\indextwo \probaconv{n}{m + k_\indextwo}\ 
  \ +\ \ 1 - \left(\sum_{\indextwo \in \indexsettwo}\ p_{\indextwo}\right)}{\Nat^{ \sizesucc{\sizevarone}} \typarrow \subst{\distrtyptwo}{\sizevarone}{\sizesucc{\sizevarone}},\seone\left[\sizevarone \mapsto m'\right]}
  $$
  which, by (\ref{eq:decomposition-proba}) and by Lemma~\ref{lemma:exchange-size-size-env}, gives
  $$
  \subst{\valone}{\vec{\varone},\funcone}{\vec{\valtwo},\letrec{\funcone}{\subst{\valone}{\vec{\varone}}{\vec{\valtwo}}}}
  \ \in\ \setredvals{\probaconv{n+1}{m'+1}}{\Nat^{\sizevarone} \typarrow \distrtyptwo,\seone\left[\sizevarone \mapsto m'+1\right]}
  $$
  But this term is an unfolding of $\letrec{\funcone}{\subst{\valone}{\vec{\varone}}{\vec{\valtwo}}}$, so that by Corollary~\ref{corollary:unfolding-stability-vred} we obtain 
  $$
 \letrec{\funcone}{\subst{\valone}{\vec{\varone}}{\vec{\valtwo}}}\ \ \in\ \ \setredvals{\probaconv{n}{m}}{\Nat^\sizevarone \typarrow \distrtyptwo,\seone\left[\sizevarone \mapsto m\right]}
 $$
 \end{itemize}

\end{itemize}

\end{proof}
}

\longv{\subsection{Size Environments Mapping Sizes to Infinity}}
\shortv{When $m=\infty$, the previous lemma does not allow to conclude, and an additional argument is required. 
Indeed, it does not make sense to consider
a sized walk beginning from $\infty$: the meaning of this size is in fact \emph{any integer}, not 
the ordinal $\omega$. The following lemma justifies this vision by proving that, if a term is in a reducibility set for any finite interpretation
of a size, then it is also in the set where the size is interpreted as $\sizeinf$.}
\longv{When $m=\infty$, the previous lemma does not allow to conclude, and an additional argument is required. 
Indeed, it does not make sense to consider
a sized walk beginning from $\infty$: the meaning of this size is in fact \emph{any integer}, not 
the ordinal $\omega$. Before we justify this understanding, we need the following companion lemma.}

\longv{ 
\begin{lemma}
\label{lemma:positivity-and-inclusion-of-candidates}
 If $\positive{\sizevarone}{\typone}$, then 
 \begin{itemize}
  \item $\setredvals{p}{\typone,\seone\left[\sizevarone \mapsto n\right]}
 \subseteq \setredvals{p}{\typone,\seone\left[\sizevarone \mapsto \sizeinf\right]}$,
  \item $\setreddists{p}{\distrtypone,\seone\left[\sizevarone \mapsto n\right]}
 \subseteq \setreddists{p}{\distrtypone,\seone\left[\sizevarone \mapsto \sizeinf\right]}$,
 \item $\setredtms{p}{\distrtypone,\seone\left[\sizevarone \mapsto n\right]}
 \subseteq \setredtms{p}{\distrtypone,\seone\left[\sizevarone \mapsto \sizeinf\right]}$.
 \end{itemize}
\end{lemma}

\begin{proof}~
\begin{itemize}
 \item  Let $\sizeone = \sizesuccit{\sizevarone}{n}$. We have $\sesem{\sizeone}{\seone\left[\sizevarone \mapsto 0\right]}\,=\,n$.
 Using Lemma~\ref{lemma:exchange-size-size-env}, we obtain
 $$
  \setredvals{p}{\subst{\typone}{\sizevarone}{\sizeone},\seone\left[\sizevarone \mapsto 0\right]}
  \ \ =\ \ \setredvals{p}{\typone,\seone\left[\sizevarone \mapsto \sesem{\sizeone}{\seone\left[\sizevarone \mapsto 0\right]}\right]}
  \ \ =\ \ \setredvals{p}{\typone,\seone\left[\sizevarone \mapsto n\right]}
 $$
 By the same lemma,
 $$
  \setredvals{p}{\subst{\typone}{\sizevarone}{\sizeinf},\seone\left[\sizevarone \mapsto 0\right]}
  \ \ =\ \ \setredvals{p}{\typone,\seone\left[\sizevarone \mapsto \sesem{\sizeinf}{\seone\left[\sizevarone \mapsto 0\right]}\right]}
  \ \ =\ \ \setredvals{p}{\typone,\seone\left[\sizevarone \mapsto \sizeinf\right]}
 $$
 Since $\positive{\sizevarone}{\typone}$ and $\sizeone \sizeleq \sizeinf$, Lemma~\ref{lemma:size-substitution-subtyping}
 implies that $\subst{\typone}{\sizevarone}{\sizeone} \subtypeleq \subst{\typone}{\sizevarone}{\sizeinf}$.
 By Lemma~\ref{lemma:subtyping-and-candidates}, we obtain
 $
 \setredvals{p}{\subst{\typone}{\sizevarone}{\sizeone},\seone\left[\sizevarone \mapsto 0\right]} \ \subseteq
 \ \setredvals{p}{\subst{\typone}{\sizevarone}{\sizeinf},\seone\left[\sizevarone \mapsto 0\right]}
 $
 and thus $\setredvals{p}{\typone,\seone\left[\sizevarone \mapsto n\right]}
 \subseteq \setredvals{p}{\typone,\seone\left[\sizevarone \mapsto \sizeinf\right]}$.
 \item Let $\distrone \in \setreddists{p}{\distrtypone,\seone\left[\sizevarone \mapsto n\right]}$.
 It follows that $\distrone\,=\,\pseudorep{\left(\valone_{\indexone}\right)^{p_{\indexone}} \sep \indexone \in \indexsetone}$ and that,
 setting $\distrtypone\,=\,\distrelts{\left(\typone_\indextwo\right)^{p'_\indextwo} \sep \indextwo \in \indexsettwo}$,
  there exists families $\left(p_{\indexone \indextwo}\right)_{\indexone\in \indexsetone,\indextwo\in\indexsettwo}$
  and $\left(q_{\indexone \indextwo}\right)_{\indexone\in \indexsetone,\indextwo\in\indexsettwo}$ of reals of $[0,1]$
  satisfying:
  \begin{enumerate}
   \item $\forall \indexone \in \indexsetone,\ \ \forall\indextwo \in \indexsettwo,\ \ \valone_\indexone \in 
   \setredvals{q_{\indexone\indextwo}}{\typone_\indextwo,\seone\left[\sizevarone \mapsto n\right]}$,
   \item $\forall \indexone \in \indexsetone,\ \ \sum_{\indextwo \in \indexsettwo}\ p_{\indexone\indextwo}\ =\ p_\indexone$,
   \item $\forall \indextwo \in \indexsettwo,\ \ \sum_{\indexone \in \indexsetone}\ p_{\indexone\indextwo}\ =\ \distrtypone(\typone_\indextwo)$,
   \item $p \leq \sum_{\indexone \in \indexsetone}\sum_{\indextwo \in \indexsettwo}\ q_{\indexone\indextwo}p_{\indexone\indextwo}$.
  \end{enumerate}
   Since 
   $\forall \indexone \in \indexsetone,\ \ \forall\indextwo \in \indexsettwo,\ \ \valone_\indexone \in 
   \setredvals{q_{\indexone\indextwo}}{\typone_\indextwo,\seone\left[\sizevarone \mapsto n\right]}
   \subseteq \setredvals{q_{\indexone\indextwo}}{\typone_\indextwo,\seone\left[\sizevarone \mapsto \sizeinf\right]}$,
   we obtain that $\distrone \in \setreddists{p}{\distrtypone,\seone\left[\sizevarone \mapsto \sizeinf\right]}$
   using the same witnesses.
 
 \item Let $\termone \in \setredtms{p}{\distrtypone,\seone\left[\sizevarone \mapsto n\right]}$.
 It follows that for every $0 \leq r < p$, there exists $\distrtyptwo_r \distrleq \distrtypone$ and $n_r \in \NN$
 such that $\termone \redval^{n_r} \distrone_r$ and $\distrone_r \in \setreddists{r}{\distrtyptwo_r,\seone\left[\sizevarone \mapsto n\right]}$.
 But $\setreddists{r}{\distrtyptwo_r,\seone\left[\sizevarone \mapsto n\right]} \subseteq
 \setreddists{r}{\distrtyptwo_r,\seone\left[\sizevarone \mapsto \sizeinf\right]}$,
 so that $\termone \in \setredtms{p}{\distrtypone,\seone\left[\sizevarone \mapsto \sizeinf\right]}$.
 \end{itemize}
\end{proof}
}

\longv{The following lemma proves that $\sizeinf$ stands for ``every integer''.
It proves indeed that, if a term is in a reducibility set for any finite interpretation
of a size, then it is also in the set where the size is interpreted as $\sizeinf$:}
\begin{lemma}[Reducibility for Infinite Sizes]
\label{lemma:letrec-infinite-size-case}
 Suppose that $\positive{\sizevarone}{\distrtyptwo}$ and that
 $\valtwo$ is the value $\letrec{\funcone}{\valone}$. 
 If $\valtwo \in\ \setredvals{p}{\Nat^\sizevarone \typarrow
     \distrtyptwo,\seone\left[\sizevarone \mapsto n\right]}$  
     for every $n \in \NN$, then
   $\valtwo \in\ \setredvals{p}{\Nat^\sizevarone \typarrow
     \distrtyptwo,\seone\left[\sizevarone \mapsto \sizeinf\right]}$.
\end{lemma}

\longv{ 
\begin{proof}
Suppose that $\positive{\sizevarone}{\distrtyptwo}$ and that, for every $n \in \NN$, 
$\letrec{\funcone}{\valone}\ \ \in\ \ \setredvals{p}{\Nat^\sizevarone \typarrow \distrtyptwo,\seone\left[\sizevarone \mapsto n\right]}$.
Let $\valtwo \in \setredvals{p}{\Nat^\sizevarone,\seone\left[\sizevarone \mapsto \sizeinf\right]}$.
Then $\valtwo = \natsucc^m\ \natzero$ for some $m \in \NN$. It follows that 
$\valtwo \in \setredvals{p}{\Nat^\sizevarone,\seone\left[\sizevarone \mapsto m+1\right]}$. But 
$\letrec{\funcone}{\valone}\ \ \in\ \ \setredvals{p}{\Nat^\sizevarone \typarrow \distrtyptwo,\seone\left[\sizevarone \mapsto m+1\right]}$,
so that 
$$
\left(\letrec{\funcone}{\valone}\right)\ \valtwo \in \setredtms{p}{\distrtyptwo,\seone\left[\sizevarone \mapsto m+1\right]}
$$
By Lemma~\ref{lemma:positivity-and-inclusion-of-candidates}, since $\positive{\sizevarone}{\distrtyptwo}$, we obtain that 
$$
\left(\letrec{\funcone}{\valone}\right)\ \valtwo \in \setredtms{p}{\distrtyptwo,\seone\left[\sizevarone \mapsto \sizeinf\right]}
$$
It follows that 
$$
\letrec{\funcone}{\valone}\ \ \in\ \ \setredvals{p}{\Nat^\sizevarone \typarrow \distrtyptwo,\seone\left[\sizevarone \mapsto \sizeinf\right]}
$$

\end{proof}
}

\longv{\subsection{A Last Technical Lemma}}

\longv{The following technical lemma will allow us to deal with the Let rule in the proof of typing soundness.
\begin{lemma}
 \label{lemma:arithmetic-pumping}
 Let $(q_\indexone)_\indexone \in (0,1]^n$, $(q'_\indextwo)_\indextwo \in (0,1]^m$ and $(q''_\indexthree)_\indexthree \in (0,1]^l$.
 Let $\indexsetfour$ and $\indexsetfive$ be two sets of indexes.
 Let $0 \leq r'' < \left(\prod_{\indexone=1}^n\ q_\indexone\right)\left(\prod_{\indextwo=1}^m\ q'_\indextwo\right)\left(\prod_{\indexthree=1}^l\ q''_\indexthree\right)$.
 Suppose that, for every $0 \leq r < \prod_{\indextwo=1}^m\ q'_\indextwo$, there exists two families
 $(p^r_{\indexfour\indexfive})_{\indexfour \in \indexsetfour,\indexfive \in \indexsetfive}$ and
 $(q^r_{\indexfour\indexfive})_{\indexfour \in \indexsetfour,\indexfive \in \indexsetfive}$ of reals of $[0,1]$
 satisfying
 \begin{equation}
  \label{eq:lemma-arith1}
  r \leq \sum_{\indexfour\in\indexsetfour}\ \sum_{\indexfive \in \indexsetfive}\ p^r_{\indexfour\indexfive} q^r_{\indexfour\indexfive}
 \end{equation}
 and 
 \begin{equation}
  \label{eq:lemma-arith2}
  \sum_{\indexfour\in\indexsetfour}\ \sum_{\indexfive \in \indexsetfive}\ p^r_{\indexfour\indexfive} \leq 1
 \end{equation}
 Then there exists $0 \leq r < \prod_{\indextwo=1}^m\ q'_\indextwo$ and
 a family $(r'_{\indexfour\indexfive})_{\indexfour \in \indexsetfour,\indexfive \in \indexsetfive}$
 satisfying
 \begin{equation}
  \label{eq:lemma-arith3}
  \forall \indexfour \in \indexsetfour,\ \ \forall \indexfive \in \indexsetfive,\ \ 
  0 \leq r'_{\indexfour\indexfive} < \left(\prod_{\indexone=1}^n\ q_\indexone\right)\left(\prod_{\indexthree=1}^l\ q''_\indexthree\right)q^r_{\indexfour\indexfive}
 \end{equation}
 and
 \begin{equation}
  \label{eq:lemma-arith4}
  r'' \leq \sum_{\indexfour\in\indexsetfour}\ \sum_{\indexfive \in \indexsetfive}\ p^r_{\indexfour\indexfive} r'_{\indexfour\indexfive}
 \end{equation}
\end{lemma}
}

\longv{
\begin{proof}
Since  $r'' < \left(\prod_{\indexone=1}^n\ q_\indexone\right)\left(\prod_{\indextwo=1}^m\ q'_\indextwo\right)\left(\prod_{\indexthree=1}^l\
q''_\indexthree\right)$, there exists $\epsilon > 0$ and $\forall \indexfour\in\indexsetfour,\ \forall \indexfive\in\indexsetfive,\ 
\epsilon_{\indexfour\indexfive} > 0$ satisfying
\begin{equation}
 \label{eq:lemma-arith5}
0 < \epsilon + \sum_{\indexfour\in\indexsetfour}\ \sum_{\indexfive \in \indexsetfive}\ \epsilon_{\indexfour\indexfive} < 
\left(\prod_{\indexone=1}^n\ q_\indexone\right)\left(\prod_{\indextwo=1}^m\ q'_\indextwo\right)\left(\prod_{\indexthree=1}^l\
q''_\indexthree\right) - r''
\end{equation}
We pick $r$ such that 
\begin{equation}
 \label{eq:lemma-arith6}
 \prod_{\indextwo=1}^m\ q'_\indextwo - \epsilon < r < \prod_{\indextwo=1}^m\ q'_\indextwo 
\end{equation}
and this induces families $(p^r_{\indexfour\indexfive})_{\indexfour \in \indexsetfour,\indexfive \in \indexsetfive}$ and
 $(q^r_{\indexfour\indexfive})_{\indexfour \in \indexsetfour,\indexfive \in \indexsetfive}$ of reals of $[0,1]$
 satisfying (\ref{eq:lemma-arith1}) and (\ref{eq:lemma-arith2}).
 We choose a family $(r'_{\indexfour\indexfive})_{\indexfour \in \indexsetfour,\indexfive \in \indexsetfive}$ such that 
 $$
   \forall \indexfour \in \indexsetfour,\ \ \forall \indexfive \in \indexsetfive,\ \ 
  \left(\prod_{\indexone=1}^n\ q_\indexone\right)\left(\prod_{\indexthree=1}^l\ q''_\indexthree\right)q^r_{\indexfour\indexfive} - \epsilon_{\indexfour\indexfive} < r'_{\indexfour\indexfive} < \left(\prod_{\indexone=1}^n\ q_\indexone\right)\left(\prod_{\indexthree=1}^l\ q''_\indexthree\right)q^r_{\indexfour\indexfive}
  $$
 By (\ref{eq:lemma-arith2}) and since $\left(\prod_{\indexone=1}^n\ q_\indexone\right)\left(\prod_{\indexthree=1}^l\ q''_\indexthree\right)$,
 we obtain from (\ref{eq:lemma-arith5}) that
 $$
 \left(\prod_{\indexone=1}^n\ q_\indexone\right)\left(\prod_{\indexthree=1}^l\ q''_\indexthree\right) \epsilon + \sum_{\indexfour\in\indexsetfour}\ \sum_{\indexfive \in \indexsetfive}\ p^r_{\indexfour\indexfive} \epsilon_{\indexfour\indexfive} < 
 \left(\prod_{\indexone=1}^n\ q_\indexone\right)\left(\prod_{\indextwo=1}^m\ q'_\indextwo\right)\left(\prod_{\indexthree=1}^l\
 q''_\indexthree\right) - r''
 $$
 Thus
 $$
  r''  < 
   \left(\prod_{\indexone=1}^n\ q_\indexone\right)\left(\prod_{\indexthree=1}^l\ q''_\indexthree\right)
 \left(\left(\prod_{\indextwo=1}^m\ q'_\indextwo\right) - \epsilon\right)
 - \sum_{\indexfour\in\indexsetfour}\ \sum_{\indexfive \in \indexsetfive}\ p^r_{\indexfour\indexfive} \epsilon_{\indexfour\indexfive}
 $$
 By (\ref{eq:lemma-arith6}) and then (\ref{eq:lemma-arith1}):
  $$
  r''  < 
   \left(\prod_{\indexone=1}^n\ q_\indexone\right)\left(\prod_{\indexthree=1}^l\ q''_\indexthree\right)
   \left(\sum_{\indexfour\in\indexsetfour}\ \sum_{\indexfive \in \indexsetfive}\ p^r_{\indexfour\indexfive} q^r_{\indexfour\indexfive}\right)
 - \sum_{\indexfour\in\indexsetfour}\ \sum_{\indexfive \in \indexsetfive}\ p^r_{\indexfour\indexfive} \epsilon_{\indexfour\indexfive}
 $$
 which rewrites to
  $$
  r''  < 
   \sum_{\indexfour\in\indexsetfour}\ \sum_{\indexfive \in \indexsetfive}\ p^r_{\indexfour\indexfive}\left(\left(\prod_{\indexone=1}^n\ q_\indexone\right)\left(\prod_{\indexthree=1}^l\ q''_\indexthree\right)
    q^r_{\indexfour\indexfive} - \epsilon_{\indexfour\indexfive}\right)
 $$
 and by definition of $(r'_{\indexfour\indexfive})_{\indexfour \in \indexsetfour,\indexfive \in \indexsetfive}$ we obtain
  $$
  r''  < 
   \sum_{\indexfour\in\indexsetfour}\ \sum_{\indexfive \in \indexsetfive}\ p^r_{\indexfour\indexfive}r'_{\indexfour\indexfive}
 $$
 as requested.
\end{proof}
}


\longv{\subsection{Typing Soundness}}

All these fundamental lemmas allow us to prove the following proposition, which expresses
that all typable terms are reducible and is the key step towards the fact that typability implies AST:
\begin{proposition}[Typing Soundness]
\label{prop:typing-soundness}
  If $\contextsizedone \contextsep \contextdistrone \proves\termone\typsep\distrtypone$, then
  $\termone\in\setredopentms{\contextsizedone \contextsep \contextdistrone}{\distrtypone,\seone}$ for every $\seone$.
  Similarly, if $\contextsizedone \contextsep \contextdistrone\proves\valone\typsep\typone$, then
  $\valone\in\setredopenvals{\contextsizedone \contextsep \contextdistrone}{\typone,\seone}$ for every $\seone$.
\end{proposition}

\begin{proof}
We proceed by induction on the derivation of the sequent $\contextsizedone \contextsep \contextdistrone \proves\termone\typsep\distrtypone$.
When $\termone\,=\,\valone$ is a value, we know by Lemma~\ref{lemma/dirac-types-values} that $\distrtypone\,=\,\distrelts{\typone^1}$; and we prove that
$\valone\in\setredopenvals{\contextsizedone \contextsep \contextdistrone}{\typone,\seone}$ for every $\seone$.
By Lemma~\ref{lemma/inclusion-redvals-redterms} we obtain that
$\valone\in\setredopentms{\contextsizedone \contextsep \contextdistrone}{\distrtypone,\seone}$ for every $\seone$.
We proceed by case analysis on the last rule of the derivation:

\longv{
We suppose in the following that $\contextsizedone$ is a sized context which can be enumerated in the form
  $\varone_1\typsep\typone_1,\ldots,\varone_n\typsep\typone_n$,
  and that $\vartwo$ is a variable distinct from $\varone_1,\ldots,\varone_n$.
We proceed accordingly to the last rule of the derivation:
}
\begin{varitemize}
 \longv{
 \item \textbf{Var:} Suppose that $\contextsizedone,\,\vartwo \typsep \typtwo \contextsep \contextdistrone \ \proves\ \vartwo\typsep \typtwo$.
 Let $(q_\indexone)_\indexone \in [0,1]^{n+1}$ and 
 $\left(\valone_1,\,\ldots,\,\valone_n,\,\valtwo\right) \in \left(\prod_{\indexone=1}^n \ 
    \setredvals{q_\indexone}{\typone_\indexone,\seone}\right) \times \setredvals{q_{n+1}}{\typtwo,\seone}$.
 \begin{itemize}
  \item If $\contextdistrone\,=\, \emptyset$, we need to prove that $\subst{\vartwo}{\vec{\varone},\vartwo}{\vec{\valone},\valtwo}
  \ =\ \valtwo \in 
  \setredvals{\prod_{\indexone=1}^{n+1} q_{\indexone}}{\typtwo,\seone}$. This is immediate since
  $\prod_{\indexone=1}^{n+1} q_{\indexone} \leq q_{n+1}$, using Lemma~\ref{lemma/downward-closure-tred}.

\item If $\contextdistrone\,=\, \varthree\typsep\distrelts{\typthree_\indextwo^{p_\indextwo} \sep \indextwo\in \indexsettwo}$, let 
 $\left(q'_\indextwo\right)_{\indextwo \in \indexsettwo} \in [0,1]^\indexsettwo$ and
 $\valthree\in \bigcap_{\indextwo \in \indexsettwo}\ \setredvals{q'_\indextwo}{\typone_\indextwo,\seone}$.
 We need to prove that $\subst{\vartwo}{\vec{\varone},\vartwo,\varthree}{\vec{\valone},\valtwo,\valthree} \ =\ \valtwo \in 
  \setredvals{\left(\prod_{\indexone=1}^{n+1} q_\indexone \right)
    \left(\sum_{\indextwo\in \indexsettwo} p_\indextwo q'_\indextwo\right)}{\typtwo,\seone}$.
   But, again, $\left(\prod_{\indexone=1}^{n+1} q_\indexone \right)
    \left(\sum_{\indextwo\in \indexsettwo} p_\indextwo q'_\indextwo\right) \leq q_{n+1}$
    since $q_\indexone \leq 1$ for every $\indexone$, $q'_\indextwo \leq 1$ for every $\indextwo$ and
    $\sum_{\indextwo \in \indexsettwo}\ p_\indextwo \,=\,1$.
    We conclude using Lemma~\ref{lemma/downward-closure-tred}.
   \end{itemize}
   
   \bigbreak
  }

 \longv{
 \item \textbf{Var':} Suppose that $\contextsizedone \contextsep \vartwo \typsep \distrelts{\typtwo^1} \ \proves\ \vartwo\typsep \typtwo$.
 Let $(q_\indexone)_\indexone \in [0,1]^{n+1}$ and 
 $\left(\valone_1,\,\ldots,\,\valone_n,\,\valtwo\right) \in \left(\prod_{\indexone=1}^n \ \setredvals{q_\indexone}{\typone_\indexone,\seone}\right)
 \times \setredvals{q_{n+1}}{\typtwo,\seone}$.
 We need to prove that $\subst{\vartwo}{\vec{\varone},\vartwo}{\vec{\valone},\valtwo}
  \ =\ \valtwo \in 
  \setredvals{\prod_{\indexone=1}^{n+1} q_{\indexone}}{\typtwo,\seone}$.
  This is immediate since
  $\prod_{\indexone=1}^{n+1} q_{\indexone} \leq q_{n+1}$, using Lemma~\ref{lemma/downward-closure-tred}.
 
 \bigbreak
 }
 
 \longv{ 
 \item \textbf{Succ:} Suppose that $\contextsizedone \contextsep \contextdistrone\ \proves\ \natsucc\ \valone \typsep \Nat^{\sizesucc{\sizeone}}$.
 Suppose moreover that $\contextdistrone\,=\,\emptyset$. 
 Let $(q_\indexone)_\indexone \in [0,1]^{n}$ and 
 $\left(\valtwo_1,\,\ldots,\,\valtwo_n\right) \in \prod_{\indexone=1}^n \ \setredvals{q_\indexone}{\typone_\indexone,\seone}$.
 We need to prove that $\subst{\left(\natsucc\ \valone\right)}{\vec{\varone}}{\vec{\valtwo}} \in \setredvals{\prod_{\indexone=1}^n q_\indexone}{\Nat^{\sizesucc{\sizeone}},\seone}$. But 
 $\subst{\left(\natsucc\ \valone\right)}{\vec{\varone}}{\vec{\valtwo}}\ =\ \natsucc\ \left(\subst{\valone}{\vec{\varone}}{\vec{\valtwo}}\right)$
 and, by induction hypothesis, $\subst{\valone}{\vec{\varone}}{\vec{\valtwo}} \in \setredvals{\prod_{\indexone=1}^n q_\indexone}{\Nat^{\sizeone},\seone}$.
 By Lemma~\ref{lemma:successor-on-redval}, 
 $\subst{\left(\natsucc\ \valone\right)}{\vec{\varone}}{\vec{\valtwo}} \in \setredvals{\prod_{\indexone=1}^n q_\indexone}{\Nat^{\sizesucc{\sizeone}},\seone}$ and we can conclude.
 The case where $\contextdistrone\neq \emptyset$ is similar.

 \bigbreak
 }
 
 \longv{ 
 \item \textbf{Zero:} Suppose that $\contextsizedone \contextsep \contextdistrone\ \proves\ \natzero \typsep \Nat^{\sizesucc{\sizeone}}$.
  Suppose moreover that $\contextdistrone\,=\,\emptyset$. 
 Let $(q_\indexone)_\indexone \in [0,1]^{n}$ and 
 $\left(\valone_1,\,\ldots,\,\valone_n\right) \in \prod_{\indexone=1}^n \ \setredvals{q_\indexone}{\typone_\indexone,\seone}$.
 By Lemma~\ref{lemma:successor-on-redval}, $\subst{\natzero}{\vec{\varone}}{\vec{\valone}}\ =\ 0 \in \setredvals{\prod_{\indexone=1}^n q_\indexone}{\Nat^{\sizesucc{\sizeone}},\seone}$.
 The case where $\contextdistrone\neq \emptyset$ is similar.
 
 \bigbreak
 }
 
 \longv{
 \item \textbf{$\lambda$:} Suppose that 
 $\contextsizedone \contextsep \contextdistrone\ \proves\ \abstr{\vartwo}{\termone} \typsep \typone \typarrow\distrtypone$,
 with $\contextdistrone\ =\ \varthree\typsep \distrelts{\left(\typtwo_\indextwo\right)^{p_\indextwo} \sep \indextwo \in \indexsettwo}$.
 Let $(q_\indexone)_\indexone \in [0,1]^{n}$ and 
 $\left(\valone_1,\,\ldots,\,\valone_n\right) \in \prod_{\indexone=1}^n \ \setredvals{q_\indexone}{\typone_\indexone,\seone}$.
 Let  $\left(q'_\indextwo\right)_{\indextwo \in \indexsettwo} \in [0,1]^\indexsettwo$ and
 $\valtwo\in \bigcap_{\indextwo \in \indexsettwo}\ \setredvals{q'_\indextwo}{\typone_\indextwo,\seone}$.
 We need to prove that
 $$
 \subst{\left(\abstr{\vartwo}{\termone}\right)}{\vec{\varone},\varthree}{\vec{\valone},\valtwo}
 \ =\ \abstr{\vartwo}{\subst{\termone}{\vec{\varone},\varthree}{\vec{\valone},\valtwo}}
 \in \setredvals{\left(\prod_{\indexone=1}^n q_\indexone \right)
    \left(\sum_{\indextwo\in \indexsettwo} p_\indextwo q'_\indextwo\right)}{\typone \typarrow \distrtypone,\seone}
 $$
 Therefore, let $q'' \in (0,1]$ and $\valthree \in \setredvals{q''}{\typone,\seone}$.
 We now have to prove that 
 \begin{equation}
  \label{eq:reducibility-lambda1}
  \left(\abstr{\vartwo}{\subst{\termone}{\vec{\varone},\varthree}{\vec{\valone},\valtwo}}\right)\ \valthree \ \in\ 
 \setredtmsfin{q''\left(\prod_{\indexone=1}^n q_\indexone \right)
    \left(\sum_{\indextwo\in \indexsettwo} p_\indextwo q'_\indextwo\right)}{\distrtypone,\seone}
 \end{equation}
 But
 $$
 \left(\abstr{\vartwo}{\subst{\termone}{\vec{\varone},\varthree}{\vec{\valone},\valtwo}}\right)\ \valthree
 \ \ \rcbv\ \ 
 \subst{\termone}{\vec{\varone},\varthree,\vartwo}{\vec{\valone},\valtwo,\valthree}
 $$
 Since $\contextsizedone,\,\varone\typsep\typone \contextsep \contextdistrone \ \proves\ \termone \typsep\distrtypone$ by typing,
 the induction hypothesis ensures that 
 $\subst{\termone}{\vec{\varone},\varthree,\vartwo}{\vec{\valone},\valtwo,\valthree} \in 
 \setredtmsfin{q''\left(\prod_{\indexone=1}^n q_\indexone \right)
    \left(\sum_{\indextwo\in \indexsettwo} p_\indextwo q'_\indextwo\right)}{\distrtypone,\seone}$
 and by Lemma~\ref{lemma:reduction-and-candidates} we obtain that (\ref{eq:reducibility-lambda1}) holds, which allows to conclude.

 The case where $\contextdistrone \,=\, \emptyset$ is similar.
 
 \bigbreak
 }

 \longv{
 \item \textbf{Sub:} Suppose that $\contextsizedone \contextsep \contextdistrone\ \proves\ \termone\typsep\distrtyptwo$
 is derived from  $\contextsizedone \contextsep \contextdistrone\ \proves\ \termone\typsep \distrtypone$
 where $\distrtypone\ \subtypeleq\ \distrtyptwo$. Suppose that $\contextdistrone \,=\, \emptyset$.
  Let $(q_\indexone)_\indexone \in [0,1]^{n}$ and 
 $\left(\valone_1,\,\ldots,\,\valone_n\right) \in \prod_{\indexone=1}^n \ \setredvals{q_\indexone}{\typone_\indexone,\seone}$.
 By induction hypothesis, $\subst{\termone}{\vec{\varone}}{\valone} \in \setredtmsfin{\prod_{\indexone=1}^n q_\indexone}{\distrtypone,\seone}$
 so that by Lemma~\ref{lemma:subtyping-and-candidates} we have
 $\subst{\termone}{\vec{\varone}}{\valone} \in \setredtmsfin{\prod_{\indexone=1}^n q_\indexone}{\distrtyptwo,\seone}$
 which allows to conclude.
 
 The case where $\contextdistrone \neq \emptyset$ is similar.
  
 \bigbreak
 }
 
 \longv{
 \item \textbf{App:} Suppose that $\contextsizedone ,\, \contextsizedtwo ,\, \contextsizedthree \contextsep
 \contextdistrone ,\, \contextdistrtwo\ \proves\ \valone\ \valtwo \typsep \distrtypone$.
 Suppose that $\contextdistrone ,\, \contextdistrtwo\, =\, \emptyset$. 
 We set $\contextsizedone\ =\ \varone_1\typsep\typone_1,\ldots,\varone_n\typsep\typone_n$,
 $\contextsizedtwo\ =\ \vartwo_1\typsep\typtwo_1,\ldots,\vartwo_m\typsep\typtwo_m$
 and $\contextsizedthree\ =\ \varthree_1\typsep\typthree_1,\ldots,\varthree_l\typsep\typthree_l$.
 Let $(q_\indexone)_\indexone \in [0,1]^{n}$, $(q'_\indextwo)_\indextwo \in [0,1]^{m}$,
 $(q''_\indexthree)_\indexthree \in [0,1]^{l}$,
 $\left(\valone_1,\,\ldots,\,\valone_n\right) \in \prod_{\indexone=1}^n \ \setredvals{q_\indexone}{\typone_\indexone,\seone}$,
 $\left(\valtwo_1,\,\ldots,\,\valtwo_m\right) \in \prod_{\indextwo=1}^m \ \setredvals{q'_\indextwo}{\typtwo_\indextwo,\seone}$,
 and $\left(\valthree_1,\,\ldots,\,\valthree_l\right) \in \prod_{\indexthree=1}^l \ \setredvals{q_\indexthree}{\typthree_\indexthree,\seone}$.
 We need to prove that
 \begin{equation}
 \label{eq:reducibility-app1}
 \subst{\left(\valone\ \valtwo\right)}{\vec{\varone},\vec{\vartwo},\vec{\varthree}}{\vec{\valone},\vec{\valtwo},\vec{\valthree}}
 \ \ =\ \ \subst{\valone}{\vec{\varone},\vec{\vartwo},\vec{\varthree}}{\vec{\valone},\vec{\valtwo},\vec{\valthree}}\ \ 
 \subst{\valtwo}{\vec{\varone},\vec{\vartwo},\vec{\varthree}}{\vec{\valone},\vec{\valtwo},\vec{\valthree}}
 \end{equation}
 is in $\setredtms{\left(\prod_{\indexone=1}^n q_\indexone\right)\left(\prod_{\indextwo=1}^m q'_\indextwo\right)\left(\prod_{\indexthree=1}^l q''_\indexthree\right)}{\distrtypone,\seone}$.
 \begin{itemize}
  \item Suppose that $\prod_{\indexone=1}^n q_\indexone =0$. Then we need to prove that (\ref{eq:reducibility-app1}) is in 
  $\setredtms{0}{\distrtypone,\seone}$, which is immediate by Lemma~\ref{lemma:simple-types-and-candidates-of-proba-zero}
  as it is of simple type $\underlying{\distrtypone}$.
  \item Suppose that $\prod_{\indexone=1}^n q_\indexone \neq 0$. It follows that $\forall \indexone \in \indexsetone,\ q_\indexone \neq 0$.
  We have that  $\contextsizedone,\,\contextsizedtwo \contextsep \emptyset \ \proves\ \valone\typsep \typone \typarrow\distrtypone$
 which, by induction hypothesis, gives that $\valone\in\setredopenvals{\contextsizedone, \contextsizedtwo \contextsep \emptyset}{\typone \typarrow \distrtypone,\seone}$. Note that for every $\indexone \in \indexsetone$ we have $\typone_\indexone \refines \Nat$;
 since $q_\indexone \neq 0$, we have by definition of the sets of candidates that 
 $\setredvals{q_\indexone}{\typone_\indexone,\seone}\ =\ \setredvals{1}{\typone_\indexone,\seone}$.
 It follows that $\subst{\valone}{\vec{\varone},\vec{\vartwo}}{\vec{\valone},\vec{\valtwo}}\ =\
 \subst{\valone}{\vec{\varone},\vec{\vartwo},\vec{\varthree}}{\vec{\valone},\vec{\valtwo}\vec{\valthree}}
 \,\in\,\setredvals{\left(\prod_{\indexone=1}^n 1\right)\left(\prod_{\indextwo=1}^m q'_\indextwo\right)}{\typone \typarrow \distrtypone,\seone}
 \ =\ \setredvals{\prod_{\indextwo=1}^m q'_\indextwo}{\typone \typarrow \distrtypone,\seone}$.
 Since $\contextsizedone,\,\contextsizedthree \contextsep \contextdistrtwo \ \proves\ \valtwo \typsep \typone$,
 we obtain similarly from the induction hypothesis that 
 $\subst{\valtwo}{\vec{\varone},\vec{\vartwo},\vec{\varthree}}{\vec{\valone},\vec{\valtwo}\vec{\valthree}}
 \,\in\,\setredvals{\prod_{\indexthree=1}^l q''_\indexthree}{\typone,\seone}$.
 By definition of $\setredvals{\prod_{\indextwo=1}^m q'_\indextwo}{\typone \typarrow \distrtypone,\seone}$, we obtain that
 $$
 \subst{\valone}{\vec{\varone},\vec{\vartwo},\vec{\varthree}}{\vec{\valone},\vec{\valtwo},\vec{\valthree}}\ \ 
 \subst{\valtwo}{\vec{\varone},\vec{\vartwo},\vec{\varthree}}{\vec{\valone},\vec{\valtwo},\vec{\valthree}}
 \ \in\ \setredtmsfin{\left(\prod_{\indextwo=1}^m q'_\indextwo\right)\left(\prod_{\indexthree=1}^l q''_\indexthree\right)}{\distrtypone,\seone}
 $$
 and by downwards closure (Lemma~\ref{lemma/downward-closure-tred}) we obtain that (\ref{eq:reducibility-app1})
 is in $\setredtms{\left(\prod_{\indexone=1}^n q_\indexone\right)\left(\prod_{\indextwo=1}^m q'_\indextwo\right)\left(\prod_{\indexthree=1}^l q''_\indexthree\right)}{\distrtypone,\seone}$ so that we can conclude. 
 \end{itemize}

 The case where $\contextdistrone ,\, \contextdistrtwo \neq \emptyset$ is similar.
 
 \bigbreak
 }

 \longv{
 \item \textbf{Choice:} Suppose that $\contextsizedone \contextsep \contextdistrone \choice_p \contextdistrtwo\ 
 \proves\ \termone \choice_p \termtwo \typsep \distrtypone \choice_p \distrtyptwo$.
 Suppose that $\contextdistrone \neq \emptyset$ and that $\contextdistrtwo \neq \emptyset$.
 We set $\contextdistrone\,=\,\vartwo \typsep \distrelts{\typtwo_\indextwo^{p_\indextwo} \sep \indextwo \in \indexsettwo}$
 and $\contextdistrtwo\,=\,\vartwo \typsep \distrelts{\left(\typtwo'_\indexthree\right)^{p'_\indexthree} \sep \indexthree \in \indexsetthree}$
 where we suppose that $\indextwo \in \indexsettwo \cap \indexsetfour \Leftrightarrow 
  \typone_\indextwo = \typtwo_\indextwo$. We obtain that
 $$
 \begin{array}{rl}
 \contextdistrone \choice_p \contextdistrtwo \ =\ \vartwo \typsep 
 & \distrelts{\typtwo_\indextwo^{pp_\indextwo} \sep \indextwo \in \indexsettwo\setminus(\indexsettwo \cap \indexsetthree)} + 
 \distrelts{\left(\typtwo_\indexfour\right)^{pp_\indexfour+(1-p)p'_\indexfour} \sep \indexfour \in \indexsettwo \cap \indexsetthree}\\
 & + \distrelts{\left(\typtwo'_\indexthree\right)^{(1-p)p'_\indexthree} \sep \indexthree \in \indexsetthree\setminus(\indexsettwo \cap \indexsetthree)}\\
 \end{array}
 $$
 Let $(q_\indexone)_\indexone \in [0,1]^{n}$, $(q'_\indextwo)_\indextwo \in [0,1]^{|\indexsettwo\setminus(\indexsettwo \cap \indexsetthree)|}$,
 $(q''_\indexfour)_\indexfour \in [0,1]^{|\indexsettwo \cap \indexsetthree|}$,
 $(q'''_\indexthree)_\indexthree \in [0,1]^{|\indexsetthree\setminus(\indexsettwo \cap \indexsetthree)|}$,
 $\left(\valone_1,\,\ldots,\,\valone_n\right) \in \prod_{\indexone=1}^n \ \setredvals{q_\indexone}{\typone_\indexone,\seone}$,
 and 
 $$
 \valtwo \in \bigcap_{\indextwo \in \indexsettwo\setminus(\indexsettwo \cap \indexsetthree)} \setredvals{q'_\indextwo}{\typtwo_\indextwo,\seone}
 \cap \bigcap_{\indexfour \in \indexsettwo \cap \indexsetthree} \setredvals{q''_\indexfour}{\typtwo_\indexfour,\seone}
 \cap \bigcap_{\indexthree \in \indexsetthree\setminus(\indexsettwo \cap \indexsetthree)} \setredvals{q'''_\indexthree}{\typtwo'_\indexthree,\seone}
 $$
 We need to prove that $\subst{\left(\termone \choice_p \termtwo \right)}{\vec{\varone},\vartwo}{\vec{\valone},\valtwo}$ is in
 $$
 \begin{array}{l}
 \setredtmsfin{\left(\prod_{\indexone=1}^n q_\indexone\right)\left(
 \sum_{\indextwo \in \indexsettwo\setminus(\indexsettwo \cap \indexsetthree)} pp_\indextwo q'_\indextwo
 + \sum_{\indexfour \in \indexsettwo \cap \indexsetthree} (pp_\indexfour + (1-p)p'_\indexfour) q''_\indexfour 
 + \sum_{\indexthree \in \indexsetthree\setminus(\indexsettwo \cap \indexsetthree)} (1-p)p'_\indexthree q'''_\indexthree
 \right)}{\distrtypone \choice_p \distrtyptwo,\seone}\\
 =\ \ 
 \setredtmsfin{
 p\left(\prod_{\indexone=1}^n q_\indexone\right)\left(\sum_{\indextwo \in \indexsettwo\setminus(\indexsettwo \cap \indexsetthree)} p_\indextwo q'_\indextwo
 + \sum_{\indexfour \in \indexsettwo \cap \indexsetthree} p_\indexfour q''_\indexfour \right) + 
 (1-p)\left(\prod_{\indexone=1}^n q_\indexone\right)\left(\sum_{\indexfour \in \indexsettwo \cap \indexsetthree} p'_\indexfour q''_\indexfour
 + \sum_{\indexthree \in \indexsetthree\setminus(\indexsettwo \cap \indexsetthree)} p'_\indexthree q'''_\indexthree
 \right)}{\distrtypone \choice_p \distrtyptwo,\seone}
 \\
 \end{array}
 $$
 Typing gives us that $\contextsizedone \contextsep \contextdistrone\ \proves\ \termone \typsep \distrtypone$, which by induction hypothesis
 implies that
 $$
 \subst{\termone}{\vec{\varone},\vartwo}{\vec{\valone},\valtwo} \ \in\ \setredtmsfin{
 \left(\prod_{\indexone=1}^n q_\indexone\right)\left(\sum_{\indextwo \in \indexsettwo\setminus(\indexsettwo \cap \indexsetthree)} p_\indextwo q'_\indextwo
 + \sum_{\indexfour \in \indexsettwo \cap \indexsetthree} p_\indexfour q''_\indexfour \right) }{\distrtypone,\seone}
 $$
 Typing also implies that $\contextsizedone \contextsep \contextdistrtwo\ \proves\ \termtwo \typsep \distrtyptwo$, and provides by induction
 hypothesis
 $$
 \subst{\termtwo}{\vec{\varone},\vartwo}{\vec{\valone},\valtwo} \ \in\ 
 \setredtmsfin{
 \left(\prod_{\indexone=1}^n q_\indexone\right)\left(\sum_{\indexfour \in \indexsettwo \cap \indexsetthree} p'_\indexfour q''_\indexfour
 + \sum_{\indexthree \in \indexsetthree\setminus(\indexsettwo \cap \indexsetthree)} p'_\indexthree q'''_\indexthree
 \right)}{\distrtypone \choice_p \distrtyptwo,\seone}
 $$
 Since 
 $$
 \subst{\left(\termone \choice_p \termtwo \right)}{\vec{\varone},\vartwo}{\vec{\valone},\valtwo}
 \ \ \rcbv\ \ \distrelts{\left(\subst{\termone}{\vec{\varone},\vartwo}{\vec{\valone},\valtwo}\right)^p,\left(\subst{\termtwo}{\vec{\varone},\vartwo}{\vec{\valone},\valtwo}\right)^{1-p}}
 $$
 Lemma~\ref{lemma:reduction-and-candidates} allows to conclude.
 
 The cases where $\contextdistrone = \emptyset$ or $\contextdistrtwo = \emptyset$ are treated similarly.
 
 \bigbreak
 }
 
 \longv{
 \item \textbf{Let:} Suppose that
 $\contextsizedone,\,\contextsizedtwo,\,\contextsizedthree \contextsep \contextdistrone ,\,
 \left(\sum_{\indexone \in \indexsetone}\ p_{\indexone} \cdot \contextdistrtwo_{\indexone}\right)
 \proves \letin{\varone}{\termone}{\termtwo} \typsep \sum_{\indexone \in \indexsetone}\ p_{\indexone} \cdot \distrtypone_{\indexone}$.
 Let $\contextsizedone\ =\ \varone_1\typsep\typone_1,\ldots,\varone_n\typsep\typone_n$,
 $\contextsizedtwo\ =\ \vartwo_1\typsep\typtwo_1,\ldots,\vartwo_m\typsep\typtwo_m$,
 and $\contextsizedthree\ =\ \varthree_1\typsep\typthree_1,\ldots,\varthree_m\typsep\typthree_l$.
 Let $(q_\indexone)_\indexone \in [0,1]^n$, $(q'_\indextwo)_\indextwo \in [0,1]^m$ and $(q''_\indexthree)_\indexthree \in [0,1]^l$.
 Let $(\valone_1,\ldots,\valone_n) \in \prod_{\indexone = 1}^n \setredvals{q_\indexone}{\typone_\indexone,\seone}$,
 $(\valtwo_1,\ldots,\valtwo_m) \in \prod_{\indextwo = 1}^m \setredvals{q'_\indextwo}{\typtwo_\indextwo,\seone}$,
 and $(\valthree_1,\ldots,\valthree_l) \in \prod_{\indexthree = 1}^l \setredvals{q''_\indexthree}{\typthree_\indexthree,\seone}$.
 There are two subcases here.
 \begin{itemize}
  \item Suppose that $\termone$ is a value. Then the last typing rule is
  $$
    \AxiomC{$\contextsizedone,\,\contextsizedtwo \contextsep \contextdistrone
 \proves \termone \typsep \typone$}
 \AxiomC{$\contextsizedone,\,\contextsizedthree,\,\varone\typsep \typone \contextsep \contextdistrtwo
 \proves \termtwo \typsep \distrtypone$}
   \AxiomC{$\underlying{\contextsizedone}\,=\,\Nat$}
 \TrinaryInfC{$\contextsizedone,\,\contextsizedtwo,\,\contextsizedthree \contextsep \contextdistrone ,\,
 \contextdistrtwo \proves \letin{\varone}{\termone}{\termtwo} \typsep \distrtypone$}
  \DisplayProof
  $$
  We treat the case where $\contextdistrone = \contextdistrtwo = \emptyset$, the two other ones being similar.
  We need to prove that   
  \begin{equation}
  \label{eq:reducibility-let1}
   \subst{\left(\letin{\varone}{\termone}{\termtwo}\right)}{\vec{\varone},\vec{\vartwo},\vec{\varthree}}{\vec{\valone},\vec{\valtwo},\vec{\valthree}}
  \in \setredtmsfin{\left(\prod_{\indexone\in\indexsetone} q_\indexone\right)\left(\prod_{\indextwo\in\indexsettwo} q'_\indextwo\right)\left(\prod_{\indexthree\in\indexsetthree} q''_\indexthree\right)}{\distrtypone,\seone}
  \end{equation}
  We now distinguish two cases.
  \begin{itemize}
   \item Suppose that $\prod_{\indexone\in\indexsetone} q_\indexone\,=\,0$. Then (\ref{eq:reducibility-let1}) holds immediately
   since by Lemma~\ref{lemma:simple-types-and-candidates-of-proba-zero} all the terms of simple type $\underlying{\distrtypone}$
   are in $\setredtmsfin{0}{\distrtypone,\seone}$.
   \item Else for every $\indexone \in \indexsetone$ we have $\setredvals{q_\indexone}{\typone_\indexone,\seone}
   \,=\,\setredvals{1}{\typone_\indexone,\seone}$.
    Since $\contextsizedone,\,\contextsizedtwo \contextsep \contextdistrone \proves \termone \typsep \typone$,
  we obtain by induction hypothesis that $\subst{\termone}{\vec{\varone},\vec{\vartwo}}{\vec{\valone},\vec{\valtwo}}
  \in \setredtmsfin{\left(\prod_{\indexone\in\indexsetone} 1\right)\left(\prod_{\indextwo\in\indexsettwo} q'_\indextwo\right)}{\typone,\seone}$.
  None of the $\vec{\varthree}$ occur in $\termone$, so 
  $\subst{\termone}{\vec{\varone},\vec{\vartwo},\vec{\valthree}}{\vec{\valone},\vec{\valtwo},\vec{\valthree}}
  \in \setredtmsfin{\prod_{\indextwo\in\indexsettwo} q'_\indextwo}{\typone,\seone}$.
  Since $\contextsizedone,\,\contextsizedthree,\,\varone\typsep \typone \contextsep \contextdistrtwo
 \proves \termtwo \typsep \distrtypone$, we obtain by induction hypothesis that
   $\subst{\termtwo}{\vec{\varone},\vec{\varthree},\varone}{\vec{\valone},\vec{\valtwo},\vec{\valthree},\subst{\termone}{\vec{\varone},\vec{\vartwo},\vec{\varthree}}{\vec{\valone},\vec{\valthree}}}$
   is in $\setredtmsfin{\left(\prod_{\indexone\in\indexsetone} q_\indexone\right)\left(\prod_{\indextwo\in\indexsettwo} q'_\indextwo\right)\left(\prod_{\indexthree\in\indexsetthree} q''_\indexthree\right)}{\distrtypone,\seone}$.
   Since none of the variables of $\vec{\vartwo}$ occur in this term, we obtain
   $$
   \subst{\termtwo}{\vec{\varone},\vec{\vartwo},\vec{\varthree},\varone}{\vec{\valone},\vec{\valtwo},\vec{\valthree},\subst{\termone}{\vec{\varone},\vec{\vartwo},\vec{\varthree}}{\vec{\valone},\vec{\valtwo},\vec{\valthree}}}
   \ \in\ \setredtmsfin{\left(\prod_{\indexone\in\indexsetone} q_\indexone\right)\left(\prod_{\indextwo\in\indexsettwo} q'_\indextwo\right)\left(\prod_{\indexthree\in\indexsetthree} q''_\indexthree\right)}{\distrtypone,\seone}
   $$
  Now
  $$
  \begin{array}{rl}
   &  \subst{\left(\letin{\varone}{\termone}{\termtwo}\right)}{\vec{\varone},\vec{\vartwo},\vec{\varthree}}{\vec{\valone},\vec{\valtwo},\vec{\valthree}}\\
   =\ \ & \letin{\varone}{\subst{\termone}{\vec{\varone},\vec{\vartwo},\vec{\varthree}}{\vec{\valone},\vec{\valtwo},\vec{\valthree}}}{\subst{\termtwo}}{\vec{\varone},\vec{\vartwo},\vec{\varthree}}{\vec{\valone},\vec{\valtwo},\vec{\valthree}}\\
   \rcbv \ \ &
   \distrelts{\left(\subst{\subst{\termtwo}{\vec{\varone},\vec{\vartwo},\vec{\varthree}}{\vec{\valone},\vec{\valtwo},\vec{\valthree}}}{\varone}{
   \subst{\termone}{\vec{\varone},\vec{\vartwo},\vec{\varthree}}{\vec{\valone},\vec{\valtwo},\vec{\valthree}}}\right)^1}
   \\
   =\ \ & \distrelts{\left(\subst{\termtwo}{\vec{\varone},\vec{\vartwo},\vec{\varthree},\varone}{\vec{\valone},\vec{\valtwo},\vec{\valthree},\subst{\termone}{\vec{\varone},\vec{\vartwo},\vec{\varthree}}{\vec{\valone},\vec{\valtwo},\vec{\valthree}}}\right)^1}
  \end{array}
  $$
  and it follows from Lemma~\ref{lemma:reduction-and-candidates} that (\ref{eq:reducibility-let1}) holds, allowing to conclude.
  \end{itemize}

  \item Suppose that $\termone$ is not a value.
  We treat in a first time the case where $\contextdistrone = \contextdistrtwo = \emptyset$.
  The case where $\contextdistrone \neq \emptyset$ is exactly similar, while the case where
  $\contextdistrtwo \neq \emptyset$ reveals the reason why a sum 
  $\sum_{\indextwo\in\indexsettwo} p_\indextwo q'_\indextwo$ appears in the definitions of $\setredopentms{}{}$
  and $\setredopenvals{}{}$. The last typing rule is
 $$
 \AxiomC{$\contextsizedone,\,\contextsizedtwo \contextsep \emptyset
 \proves \termone \typsep \distrelts{\typone_{\indexsix}^{p_{\indexsix}} \sep \indexsix \in \indexsetsix}$}
 \AxiomC{$\contextsizedone,\,\contextsizedthree,\,\varone\typsep \typone_{\indexsix} \contextsep \emptyset
 \proves \termtwo \typsep \distrtypone_{\indexsix}$}
   \AxiomC{$\underlying{\contextsizedone}\,=\,\Nat$}
 \TrinaryInfC{$\contextsizedone,\,\contextsizedtwo,\,\contextsizedthree \contextsep \emptyset
 \proves \letin{\varone}{\termone}{\termtwo} \typsep \sum_{\indexsix \in \indexsetsix}\ p_{\indexsix} \cdot \distrtypone_{\indexsix}$}
 \DisplayProof
 $$

  We need to prove that   
  \begin{equation}
  \label{eq:reducibility-let1}
  \begin{array}{rl}
  &\subst{\left(\letin{\varone}{\termone}{\termtwo}\right)}{\vec{\varone},\vec{\vartwo},\vec{\varthree}}{\vec{\valone},\vec{\valtwo},\vec{\valthree}}\\
  =\ \ &
  \letin{\varone}{\subst{\termone}{\vec{\varone},\vec{\vartwo},\vec{\varthree}}{\vec{\valone},\vec{\valtwo},\vec{\valthree}}}{\subst{\termtwo}{\vec{\varone},\vec{\vartwo},\vec{\varthree}}{\vec{\valone},\vec{\valtwo},\vec{\valthree}}}
  \\
   & \quad \in \setredtmsfin{\left(\prod_{\indexone\in\indexsetone} q_\indexone\right)\left(\prod_{\indextwo\in\indexsettwo} q'_\indextwo\right)\left(\prod_{\indexthree\in\indexsetthree} q''_\indexthree\right)}{\sum_{\indexsix \in \indexsetsix}\ p_{\indexsix} \cdot \distrtypone_{\indexsix},\seone} 
  \end{array}
  \end{equation}
  We now distinguish two cases.
   \begin{itemize}
   \item Suppose that $\left(\prod_{\indexone\in\indexsetone} q_\indexone\right)\left(\prod_{\indextwo\in\indexsettwo} q'_\indextwo\right)\left(\prod_{\indexthree\in\indexsetthree} q''_\indexthree\right)\,=\,0$. 
   Then (\ref{eq:reducibility-let1}) holds immediately
   since by Lemma~\ref{lemma:simple-types-and-candidates-of-proba-zero} all the terms of simple type $\underlying{\sum_{\indexsix \in \indexsetsix}\ p_{\indexsix} \cdot \distrtypone_{\indexsix}}$
   are in $\setredtmsfin{0}{\sum_{\indexsix \in \indexsetsix}\ p_{\indexsix} \cdot \distrtypone_{\indexsix},\seone}$.
   \item Else, we use the induction hypothesis on 
   $\contextsizedone,\,\contextsizedtwo \contextsep \emptyset
 \proves \termone \typsep \distrelts{\typone_{\indexsix}^{p_{\indexsix}} \sep \indexsix \in \indexsetsix}$.
 Since $\underlying{\typone_\indexone}=\Nat$, for every $\indexone \in \indexsetone$ we have $\setredvals{q_\indexone}{\typone_\indexone,\seone}
   \,=\,\setredvals{1}{\typone_\indexone,\seone}$. Together with the fact that $\vec{\varthree}$ does not appear in $\termone$,
   we obtain that
   $$
   \subst{\termone}{\vec{\varone},\vec{\vartwo},\vec{\varthree}}{\vec{\valone},\vec{\valtwo},\vec{\valthree}}
   \in \setredtmsfin{\prod_{\indextwo\in\indexsettwo} q'_\indextwo}{\distrelts{\typone_{\indexsix}^{p_{\indexsix}} \sep \indexsix \in \indexsetsix},\seone}
   $$
   By definition, for every $0 \leq r < \prod_{\indextwo=1}^m\ q'_\indextwo$,
   there exists $n_r$ and $\distrtyptwo_r = \distrelts{\typone_{\indexfive}^{p_{r,\indexfive}} \sep \indexfive \in \indexsetfive_r} \distrleq
   \distrelts{\typone_{\indexsix}^{p_{\indexsix}} \sep \indexsix \in \indexsetsix}$ with $\indexsetfive_r \subseteq \indexsetsix$
   such that 
   $$
   \subst{\termone}{\vec{\varone},\vec{\vartwo},\vec{\varthree}}{\vec{\valone},\vec{\valtwo},\vec{\valthree}}
   \ \ \redval^{n_r}\ \ \distrone_r\ =\ \pseudorep{\valfour_{\indexfour}^{p''_{r,\indexfour}}\sep \indexfour \in \indexsetfour_r}
   \ \in\ \setreddists{r}{\distrtyptwo_r,\seone}
   $$
   This implies the existence of two families
 $(p^r_{\indexfour\indexfive})_{\indexfour \in \indexsetfour_r,\indexfive_r \in \indexsetfive}$ and
 $(q^r_{\indexfour\indexfive})_{\indexfour \in \indexsetfour_r,\indexfive_r \in \indexsetfive}$ of reals of $[0,1]$
 satisfying in particular 
 \begin{equation}
 \label{eq:let-arith1}
  r \leq \sum_{\indexfour\in\indexsetfour_r}\ \sum_{\indexfive \in \indexsetfive_r}\ p^r_{\indexfour\indexfive} q^r_{\indexfour\indexfive}
 \end{equation}
  \begin{equation}
 \label{eq:let-arith2}
 \sum_{\indexfour\in\indexsetfour_r}\ \sum_{\indexfive \in \indexsetfive_r}\ p^r_{\indexfour\indexfive} \leq 1
 \end{equation}
  \begin{equation}
 \label{eq:let-arith3}
 \forall \indexfour \in \indexsetfour ,\ \ \sum_{\indexfive \in \indexsetfive_r}\ p^r_{\indexfour\indexfive}\ =\ p''_{r,\indexfour}
 \end{equation}
  \begin{equation}
 \label{eq:let-arith3bis}
 \forall \indexfive \in \indexsetfive ,\ \ \sum_{\indexfour \in \indexsetfour_r}\ p^r_{\indexfour\indexfive}\ =\ p_{r,\indexfive}
 \end{equation}
 and
  \begin{equation}
 \label{eq:let-arith4}
 \forall \indexfour \in \indexsetfour_r,\ \ \forall \indexfive \in \indexsetfive_r,\ \ \valfour_\indexfour \in \setredvals{q^r_{\indexfour\indexfive}}{\typone_{\indexfive},\seone}
 \end{equation}
  By (\ref{eq:let-arith1}) and (\ref{eq:let-arith2}), we can apply Lemma~\ref{lemma:arithmetic-pumping}
  and we obtain $0 \leq r < \prod_{\indextwo=1}^m\ q'_\indextwo$ and
 a family $(r'_{\indexfour\indexfive})_{\indexfour \in \indexsetfour_r,\indexfive \in \indexsetfive_r}$
 satisfying
 \begin{equation}
  \label{eq:let-arith5}
  \forall \indexfour \in \indexsetfour_r,\ \ \forall \indexfive \in \indexsetfive_r,\ \ 
  0 \leq r'_{\indexfour\indexfive} < \left(\prod_{\indexone=1}^n\ q_\indexone\right)\left(\prod_{\indexthree=1}^l\ q''_\indexthree\right)q^r_{\indexfour\indexfive}
 \end{equation}
 and
 \begin{equation}
  \label{eq:let-arith6}
  r'' \leq \sum_{\indexfour\in\indexsetfour}\ \sum_{\indexfive \in \indexsetfive}\ p^r_{\indexfour\indexfive} r'_{\indexfour\indexfive}
 \end{equation}
   We now consider $r$ to be fixed to this value given by the lemma, this providing $\distrone_r$, $\distrtyptwo_r$ and so on.

   Since $\contextsizedone,\,\contextsizedthree,\,\varone\typsep \typone_{\indexsix} \contextsep \emptyset
 \proves \termtwo \typsep \distrtypone_{\indexsix}$, we obtain by induction hypothesis using (\ref{eq:let-arith4}) that for every 
 $\indexfour \in \indexsetfour$ and $\indexfive\in\indexsetfive$ we have 
 \begin{equation}
  \label{eq:let-arith7}
  \subst{\termtwo}{\vec{\varone},\vec{\vartwo},\vec{\varthree},\varone}{\vec{\valone},\vec{\valtwo},\vec{\valthree},\valfour_{\indexfour}}
  \ \in \ \setredtmsfin{\left(\prod_{\indexone=1}^n\ q_\indexone\right)\left(\prod_{\indexthree=1}^l\ q''_\indexthree\right)q^r_{\indexfour\indexfive}}{\distrtypone_\indexfive,\seone}
 \end{equation}
 By $\ref{eq:let-arith5}$, there exists for every $\indexfour \in \indexsetfour$ and $\indexfive\in\indexsetfive$
 and index $m_{\indexfour\indexfive}$ and a type $\distrtypone'_{\indexfour \indexfive} \subtypeleq \distrtypone_\indexfive$
 such that 
 \begin{equation}
  \label{eq:let-arith8}
  \subst{\termtwo}{\vec{\varone},\vec{\vartwo},\vec{\varthree},\varone}{\vec{\valone},\vec{\valtwo},\vec{\valthree},\valfour_{\indexfour}}
  \ \redval^{m_{\indexfour \indexfive}}
  \distrtwo_{\indexfour \indexfive} \ \in
  \ \setreddists{r'_{\indexfour\indexfive}}{\distrtypone'_{\indexfour\indexfive},\seone}
 \end{equation}
   Now set
   $$
   m \ \ =\ \ \max_{\indexfour \in \indexsetfour,\indexfive \in \indexsetfive}\ m_{\indexfour\indexfive}
   $$
   By Lemma~\ref{lemma:pumping-dred}, we obtain types 
   $\distrtypone'_{\indexfour\indexfive} \distrleq \distrtypone''_{\indexfour\indexfive} \distrleq \distrtypone_{\indexfive}$
   and distributions $\distrtwo'_{\indexfour \indexfive}$
   such that all the reduction lengths are the same:
   \begin{equation}
  \label{eq:let-arith9}
  \subst{\termtwo}{\vec{\varone},\vec{\vartwo},\vec{\varthree},\varone}{\vec{\valone},\vec{\valtwo},\vec{\valthree},\valfour_{\indexfour}}
  \ \redval^{m}
  \distrtwo'_{\indexfour \indexfive} \ \in
  \ \setreddists{r'_{\indexfour\indexfive}}{\distrtypone''_{\indexfour\indexfive},\seone}
 \end{equation}
   Now it follows of (\ref{eq:let-arith3}) that
   $$
   \distrone_r\ \ =\ \ \pseudorep{\valfour_{\indexfour}^{p^r_{\indexfour,\indexfive}}\sep \indexfour \in \indexsetfour_r,\ \indexfive \in \indexsetfive_r}
   $$
   which allows us to use Lemma~\ref{lemma:operational-decomposition-let}, obtaining that
   $$
  \subst{\left(\letin{\varone}{\termone}{\termtwo}\right)}{\vec{\varone},\vec{\vartwo},\vec{\varthree}}{\vec{\valone},\vec{\valtwo},\vec{\valthree}}\ \ 
  \redval^{n_r+m+1}\ \ 
  \sum_{\indexfour \in \indexsetfour}\ \sum_{\indexfive \in \indexsetfive}\ p^r_{\indexfour,\indexfive} \cdot \distrtwo'_{\indexfour \indexfive}
   $$
  By (\ref{eq:let-arith9}) and Lemma~\ref{lemma:backtracking-dred-multiple-instances}, we obtain that
  $$
  \sum_{\indexfour \in \indexsetfour}\ \sum_{\indexfive \in \indexsetfive}\ p^r_{\indexfour,\indexfive} \cdot \distrtwo'_{\indexfour \indexfive}
  \ \in\ 
  \setreddists{\sum_{\indexfour \in \indexsetfour}\ \sum_{\indexfive \in \indexsetfive}\ p^r_{\indexfour,\indexfive}r'_{\indexfour,\indexfive}
  }{\sum_{\indexfour \in \indexsetfour}\ \sum_{\indexfive \in \indexsetfive}\ p^r_{\indexfour,\indexfive}\distrtypone''_{\indexfour,\indexfive},\seone} 
  $$
  By (\ref{eq:let-arith6}) and downward closure (Lemma~\ref{lemma/downward-closure-tred}) we obtain 
  $$
  \sum_{\indexfour \in \indexsetfour}\ \sum_{\indexfive \in \indexsetfive}\ p^r_{\indexfour,\indexfive} \cdot \distrtwo'_{\indexfour \indexfive}
  \ \in\ 
  \setreddists{r''}{\sum_{\indexfour \in \indexsetfour}\ \sum_{\indexfive \in \indexsetfive}\ p^r_{\indexfour,\indexfive}\distrtypone''_{\indexfour,\indexfive},\seone} 
  $$
   and since by (\ref{eq:let-arith3bis}) we have $\sum_{\indexfour \in \indexsetfour}\ \sum_{\indexfive \in \indexsetfive}\ p^r_{\indexfour,\indexfive}\distrtypone''_{\indexfour,\indexfive} \distrleq \sum_{\indexsix \in \indexsetsix}\ p_\indexsix \distrtypone_\indexsix$ we can conclude that
   $$
  \subst{\left(\letin{\varone}{\termone}{\termtwo}\right)}{\vec{\varone},\vec{\vartwo},\vec{\varthree}}{\vec{\valone},\vec{\valtwo},\vec{\valthree}}
  \ \ \in\ \ 
   \setredtmsfin{\left(\prod_{\indexone\in\indexsetone} q_\indexone\right)\left(\prod_{\indextwo\in\indexsettwo} q'_\indextwo\right)\left(\prod_{\indexthree\in\indexsetthree} q''_\indexthree\right)}{\sum_{\indexsix \in \indexsetsix}\ p_{\indexsix} \cdot \distrtypone_{\indexsix},\seone} 
   $$
   \end{itemize}
  \end{itemize}
 \bigbreak
 }
 
 \longv{
 \item \textbf{Case:} Suppose that $\contextsizedone,\,\contextsizedtwo \contextsep \contextdistrone \proves \caseof{\valone}{\natsucc \rightarrow \valtwo
 \smallsep  \natzero \rightarrow \valthree} \typsep \distrtypone$.
 Suppose that $\contextdistrone \,=\, \emptyset$.
 We set $\contextsizedone\ =\ \varone_1\typsep\typone_1,\ldots,\varone_n\typsep\typone_n$
 and $\contextsizedtwo\ =\ \vartwo_1\typsep\typtwo_1,\ldots,\vartwo_m\typsep\typtwo_m$.

 Let $(q_\indexone)_\indexone \in [0,1]^{n}$, $(q'_\indextwo)_\indextwo \in [0,1]^{m}$, 
 $\left(\valone_1,\,\ldots,\,\valone_n\right) \in \prod_{\indexone=1}^n \ \setredvals{q_\indexone}{\typone_\indexone,\seone}$ and 
 $\left(\valone'_1,\,\ldots,\,\valone'_m\right) \in \prod_{\indextwo=1}^m \ \setredvals{q'_\indextwo}{\typtwo_\indextwo,\seone}$.
 We need to prove that
 $$
 \subst{\left(\caseof{\valone}{\natsucc \rightarrow \valtwo\smallsep  \natzero \rightarrow \valthree}\right)}{\vec{\varone},\vec{\vartwo}}{\vec{\valone},\vec{\valone'}} \ \in \ \setredtmsfin{\left(\prod_{\indexone=1}^n q_\indexone\right)\left(\prod_{\indextwo=1}^m q'_\indextwo\right)}{\distrtypone,\seone}
 $$
 i.e. that
 $$
 \caseof{\subst{\valone}{\vec{\varone},\vec{\vartwo}}{\vec{\valone},\vec{\valone'}}}{\natsucc \rightarrow \subst{\valtwo}{\vec{\varone},\vec{\vartwo}}{\vec{\valone},\vec{\valone'}}\smallsep  \natzero \rightarrow \subst{\valthree}{\vec{\varone},\vec{\vartwo}}{\vec{\valone},\vec{\valone'}}}
 $$
 is in $\setredtmsfin{\left(\prod_{\indexone=1}^n q_\indexone\right)\left(\prod_{\indextwo=1}^m q'_\indextwo\right)}{\distrtypone,\seone}$.
 Since $\contextsizedone \contextsep \emptyset \proves \valone \typsep \Nat^{\sizesucc{\sizeone}}$, we have by induction hypothesis that
 $\subst{\valone}{\vec{\varone}}{\vec{\valone}} \in \setredtmsfin{\prod_{\indexone=1}^n q_\indexone}{\distrelts{\left(\Nat^{\sizesucc{\sizeone}}\right)^1},\seone}$.
 Since it is a value, we have by Lemma~\ref{lemma:values-are-in-tred-iff-in-vred} the stronger statement that
 $\subst{\valone}{\vec{\varone}}{\vec{\valone}} \in \setredvals{\prod_{\indexone=1}^n q_\indexone}{\Nat^{\sizesucc{\sizeone}},\seone}$
 which implies that $\subst{\valone}{\vec{\varone}}{\vec{\valone}}$ is of the shape $\natsucc^k\ \natzero$ for $k \in \NN$
 satisfying $\prod_{\indexone=1}^n q_\indexone\, \neq \,0 \implies k < \sesem{\sizesucc{\sizeone}}{\seone}$.
 The typing also ensures that none of the variables of $\vec{\vartwo}$ occurs in $\valone$,
 so that $\subst{\valone}{\vec{\varone}}{\vec{\valone}}\ =\ \subst{\valone}{\vec{\varone},\vec{\vartwo}}{\vec{\valone},\vec{\valone'}}$.
 \begin{itemize}
  \item If $k=0$, then 
  $$
  \begin{array}{l}
  \caseof{\natzero}{\natsucc \rightarrow \subst{\valtwo}{\vec{\varone},\vec{\vartwo}}{\vec{\valone},\vec{\valone'}}\smallsep  \natzero \rightarrow \subst{\valthree}{\vec{\varone},\vec{\vartwo}}{\vec{\valone},\vec{\valone'}}}\\
  \ \ \rcbv\ \ \distrelts{\left(\subst{\valthree}{\vec{\varone},\vec{\vartwo}}{\vec{\valone},\vec{\valone'}}\right)^1}
  \\
  \end{array}
  $$
  Since $\contextsizedtwo \contextsep \emptyset \proves \valthree \typsep \distrtypone$, by induction hypothesis,
  we have that 
  $$
  \subst{\valthree}{\vec{\vartwo}}{\vec{\valone'}} \in 
  \setredtmsfin{\prod_{\indextwo=1}^m q'_\indextwo}{\distrtypone,\seone}
  $$
  and also, by the typing hypothesis, that none of the variables of $\vec{\varone}$ is free in $\subst{\valthree}{\vec{\vartwo}}{\vec{\valone'}}$
  so that $\subst{\valthree}{\vec{\vartwo}}{\vec{\valone'}} \ =\ \subst{\valthree}{\vec{\varone},\vec{\vartwo}}{\vec{\valone},\vec{\valone'}}$.
  But $\prod_{\indexone=1}^n q_\indexone \leq 1$, so that the downward-closure property of Lemma~\ref{lemma/downward-closure-tred}
  induces that
  $$
  \subst{\valthree}{\vec{\varone},\vec{\vartwo}}{\vec{\valone},\vec{\valone'}} \in
  \setredtmsfin{\left(\prod_{\indexone=1}^n q_\indexone\right)\left(\prod_{\indextwo=1}^m q'_\indextwo\right)}{\distrtypone,\seone}
  $$
  Now the closure by anti-reduction of Lemma~\ref{lemma:reduction-and-candidates} ensures that 
  $$
  \caseof{\subst{\valone}{\vec{\varone},\vec{\vartwo}}{\vec{\valone},\vec{\valone'}}}{\natsucc \rightarrow \subst{\valtwo}{\vec{\varone},\vec{\vartwo}}{\vec{\valone},\vec{\valone'}}\smallsep  \natzero \rightarrow \subst{\valthree}{\vec{\varone},\vec{\vartwo}}{\vec{\valone},\vec{\valone'}}}
 $$
 is in $\setredtmsfin{\left(\prod_{\indexone=1}^n q_\indexone\right)\left(\prod_{\indextwo=1}^m q'_\indextwo\right)}{\distrtypone,\seone}$.
  \item If $k>0$, then 
  $$
  \begin{array}{l}
  \caseof{\natsucc^{k}\ \natzero}{\natsucc \rightarrow \subst{\valtwo}{\vec{\varone},\vec{\vartwo}}{\vec{\valone},\vec{\valone'}}\smallsep  \natzero \rightarrow \subst{\valthree}{\vec{\varone},\vec{\vartwo}}{\vec{\valone},\vec{\valone'}}}\\
  \ \ \rcbv\ \ \distrelts{\left(\left(\subst{\valtwo}{\vec{\varone},\vec{\vartwo}}{\vec{\valone},\vec{\valone'}}\right) \ \left(\natsucc^{k-1}\ \natzero\right) \right)^1}\\
  \end{array}
  $$
  By typing hypothesis, we have  
  $\contextsizedtwo \contextsep \emptyset \proves \valtwo \typsep \Nat^{\sizeone} \typarrow \distrtypone$
  and the induction hypothesis provides
  $\subst{\valtwo}{\vec{\vartwo}}{\vec{\valone'}} \in \setredtmsfin{\prod_{\indextwo=1}^m q'_\indextwo}{\distrelts{\left(\Nat^{\sizeone} \typarrow \distrtypone\right)^1},\seone}$
  which, since none of the $\vec{\varone}$ appears freely in $\valtwo$, and by Lemma~\ref{lemma:values-are-in-tred-iff-in-vred},
  implies that $\subst{\valtwo}{\vec{\varone},\vec{\vartwo}}{\vec{\valone},\vec{\valone'}} \in \setredvals{\prod_{\indextwo=1}^m q'_\indextwo}{\Nat^\sizeone \typarrow \distrtypone,\seone}$.
  \begin{itemize}
   \item Suppose that $\prod_{\indexone=1}^n q_\indexone\, \neq \,0$. It follows that $k < \sesem{\sizesucc{\sizeone}}{\seone}$
   and therefore $k - 1 < \sesem{\sizeone}{\seone}$ which implies that 
   $\natsucc^{k-1}\ \natzero \in \setredvals{1}{\Nat^{\sizeone},\seone}$. Since $\subst{\valtwo}{\vec{\varone},\vec{\vartwo}}{\vec{\valone},\vec{\valone'}} \in \setredvals{\prod_{\indextwo=1}^m q'_\indextwo}{\distrtypone,\seone}$, we obtain that 
   $\left(\subst{\valtwo}{\vec{\varone},\vec{\vartwo}}{\vec{\valone},\vec{\valone'}}\right) \ \left(\natsucc^{k-1}\ \natzero\right)$
   is in $\setredtmsfin{\prod_{\indextwo=1}^m q'_\indextwo}{\distrtypone,\seone}$. By closure by anti-reduction
   (Lemma~\ref{lemma:reduction-and-candidates}), we have that
   $$
    \caseof{\subst{\valone}{\vec{\varone},\vec{\vartwo}}{\vec{\valone},\vec{\valone'}}}{\natsucc \rightarrow \subst{\valtwo}{\vec{\varone},\vec{\vartwo}}{\vec{\valone},\vec{\valone'}}\smallsep  \natzero \rightarrow \subst{\valthree}{\vec{\varone},\vec{\vartwo}}{\vec{\valone},\vec{\valone'}}}
   $$
   is in $\setredtmsfin{\prod_{\indextwo=1}^m q'_\indextwo}{\distrtypone,\seone}$
   and by downward closure (Lemma~\ref{lemma/downward-closure-tred}) we obtain that it is also in
   $\setredtmsfin{\left(\prod_{\indexone=1}^n q_\indexone\right)\left(\prod_{\indextwo=1}^m q'_\indextwo\right)}{\distrtypone,\seone}$,
   from which we conclude.
   \item Suppose that $\prod_{\indexone=1}^n q_\indexone\, = \,0$. Then all we need to prove is that
   $$
    \caseof{\subst{\valone}{\vec{\varone},\vec{\vartwo}}{\vec{\valone},\vec{\valone'}}}{\natsucc \rightarrow \subst{\valtwo}{\vec{\varone},\vec{\vartwo}}{\vec{\valone},\vec{\valone'}}\smallsep  \natzero \rightarrow \subst{\valthree}{\vec{\varone},\vec{\vartwo}}{\vec{\valone},\vec{\valone'}}}
   $$
   is in $\setredtmsfin{0}{\distrtypone,\seone}$. But this term has simple type $\underlying{\distrtypone}$
   and by Lemma~\ref{lemma:simple-types-and-candidates-of-proba-zero} the result holds.
  \end{itemize}

 \end{itemize}

 The case where $\contextdistrone \neq \emptyset$ is similar.
 
 \bigbreak
 }

 \item \textbf{$\letrecname$:} Suppose that $\contextsizedone,\,\contextsizedtwo
 \contextsep \contextdistrone \proves \letrec{\funcone}{\valone} \typsep  
\Nat^{\sizetwo} \typarrow \subst{\distrtyptwo}{\sizevarone}{\sizetwo}$.
 We treat the case where $\contextsizedtwo\,=\,\contextdistrone\,=\,\emptyset$.
 The general case is easily deduced using the downward-closure of the reducibility
 sets (Lemma~\ref{lemma/downward-closure-tred}).
 Let $\contextsizedone\ =\ \varone_1\typsep\Nat^{\sizetwo_1},\ldots,\,\varone_n\typsep\Nat^{\sizetwo_n}$. 
 We need to prove that, for every family $\left(q_\indexone\right)_\indexone \in [0,1]^n$
 and every $\left(\valtwo_1,\ldots,\valtwo_n\right) \in \prod_{\indexone=1}^n \setredvals{q_\indexone}{\Nat^{\sizetwo_\indexone},\seone}$,
 we have
 $$
 \subst{\left(\letrec{\funcone}{\valone}\right)}{\vec{\varone}}{\vec{\valtwo}} \ 
 = \ \left(\letrec{\funcone}{\subst{\valone}{\vec{\varone}}{\vec{\valtwo}}}\right)
 \ \in\ \setredvals{\prod_{\indexone = 1}^n q_\indexone}{\Nat^{\sizetwo} \typarrow \subst{\distrtyptwo}{\sizevarone}{\sizetwo},\seone}
 $$
 If there exists $\indexone \in \indexsetone$ such that $q_\indexone = 0$, the result is immediate as the term is simply-typed and
 Lemma~\ref{lemma:simple-types-and-candidates-of-proba-zero} applies.
 Else, for every $\indexone \in \indexsetone$, we have by definition that
 $\setredvals{q_\indexone}{\Nat^{\sizetwo_\indexone},\seone}\ =\ \setredvals{1}{\Nat^{\sizetwo_\indexone},\seone}$.
 Since the sets $\setredvals{}{}$ are downward-closed (Lemma~\ref{lemma/downward-closure-tred}), it is in fact enough to prove that
 for every $\left(\valtwo_1,\ldots,\valtwo_n\right) \in \prod_{\indexone=1}^n \setredvals{1}{\Nat^{\sizetwo_\indexone},\seone}$,
 we have
 $$
\letrec{\funcone}{\subst{\valone}{\vec{\varone}}{\vec{\valtwo}}}
 \ \in\ \setredvals{1}{\Nat^{\sizetwo} \typarrow \subst{\distrtyptwo}{\sizevarone}{\sizetwo},\seone}
 $$
 Moreover, by size commutation (Lemma~\ref{lemma:exchange-size-size-env}),
 $$
 \setredvals{1}{\Nat^{\sizetwo} \typarrow \subst{\distrtyptwo}{\sizevarone}{\sizetwo},\seone}
 \ =\ \setredvals{1}{\Nat^{\sizevarone} \typarrow \distrtyptwo,\seone\left[\sizevarone \mapsto \sesem{\sizetwo}{\seone}\right]}
 $$
 Let us therefore prove the stronger fact that, for \emph{every} integer $m \in \NN \cup \{\sizeinf\}$,
 $$
 \letrec{\funcone}{\subst{\valone}{\vec{\varone}}{\vec{\valtwo}}}
 \ \in\ \setredvals{1}{\Nat^{\sizevarone} \typarrow \distrtyptwo,\seone\left[\sizevarone \mapsto m\right]}
 $$
 Now, the typing derivation gives us that
 $
 \contextsizedone \contextsep
 \funcone \typsep \distrtypone \proves \valone \typsep \Nat^{ \sizesucc{\sizevarone}} \typarrow \subst{\distrtyptwo}{\sizevarone}{\sizesucc{\sizevarone}}
 $
 and that $\distrtypone$ induces an AST sized walk. 
 Denote $\left(\probaconv{n}{m}\right)_{n \in \NN,m\in\NN}$ its associated probabilities of convergence in finite time.
 By induction hypothesis, $\valone\in\setredopenvals{\contextsizedone \contextsep \funcone \typsep \distrtypone}{\Nat^{ \sizesucc{\sizevarone}} \typarrow \subst{\distrtyptwo}{\sizevarone}{\sizesucc{\sizevarone}},\seone}$ for every $\seone$ and we can apply Lemma~\ref{lemma:sized-walk-argument-for-letrec}.
 It follows that, for every $(n,m) \in \NN$,
 $$
 \letrec{\funcone}{\subst{\valone}{\vec{\varone}}{\vec{\valtwo}}}\ \ \in\ \ \setredvals{\probaconv{n}{m}}{\Nat^\sizevarone \typarrow \distrtyptwo,\seone\left[\sizevarone \mapsto m\right]}
 $$
 Let $\epsilon \in (0,1)$. By Lemma~\ref{lemma:sized-walks-get-arbitrarily-close-to-proba-one-in-finite-time},
 there exists $n \in \NN$ such that $\probaconv{n}{m} \geq 1 - \epsilon$. Using downward closure (Lemma~\ref{lemma/downward-closure-tred})
 and quantifying over all the $\epsilon$, we obtain
 $$
 \letrec{\funcone}{\subst{\valone}{\vec{\varone}}{\vec{\valtwo}}}\ \ \in\ \ \bigcap_{0< \epsilon < 1} \setredvals{1-\epsilon}{\Nat^\sizevarone \typarrow \distrtyptwo,\seone\left[\sizevarone \mapsto m\right]}
 $$
 so that, by continuity of $\setredvals{}{}$ (Lemma~\ref{lemma/continuity-lemma-vred}), we obtain
 \begin{equation}
 \label{eq:letrec1}
 \letrec{\funcone}{\subst{\valone}{\vec{\varone}}{\vec{\valtwo}}}\ \ \in\ \ \setredvals{1}{\Nat^\sizevarone \typarrow \distrtyptwo,\seone\left[\sizevarone \mapsto m\right]} 
 \end{equation}
 for every $m \in \NN$, allowing us to conclude. It remains however to treat the case where $m = \sizeinf$.
 Since $\positive{\sizevarone}{\distrtyptwo}$ and that (\ref{eq:letrec1}) holds for every $m \in \NN$,
 Lemma~\ref{lemma:letrec-infinite-size-case} applies and we obtain the result.

 \shortv{ 
 \bigbreak	
 
 \item \textbf{Other cases:} the other cases are treated in the long version~\cite{dal-lago-grellois:monadic-affine-sized-types-full}.\qed
 }
\end{varitemize}

\end{proof}

\noindent
This proposition, together with the definition of $\setredopentms{}{}$,
implies the main result of the paper, namely that typability implies almost-sure termination:

\begin{theorem}
 Suppose that $\termone \in \setdistrtypedclosedterms{\distrtypone}$. Then $\termone$ is AST.
\end{theorem}

\begin{proof}
 Suppose that $\termone \in \setdistrtypedclosedterms{\distrtypone}$, then by Proposition~\ref{prop:typing-soundness}
 we have $\termone \in \setredopentms{\emptyset\contextsep\emptyset}{\distrtypone,\seone}$ for every $\seone$.
 By definition, $\setredopentms{\emptyset\contextsep\emptyset}{\distrtypone,\seone} = \setredtmsfin{1}{\distrtypone,\seone}$.
 Corollary~\ref{corollary:reducibility-and-ast} then implies that $\termone$ is AST.
\end{proof}

\section{Conclusions and Perspectives}

We presented a type system for an affine, simply-typed
$\lambda$-calculus enriched with a probabilistic choice operator,
constructors for the natural numbers, and recursion. This affinity
constraint implies that a given higher-order variable may occur
(freely) at most once in any probabilistic branch of a program. The
type system we designed decorates the affine simple types with
\emph{size information}, allowing to incorporate in the types relevant
information about the recursive behaviour of the functions contained
in the program.  A guard condition on the typing rule for $\letrecname$,
formulated with reference to an appropriate Markov chain, ensures that
typable terms are AST. The proof of soundness of this type system for
AST relies on a quantitative extension of the reducibility method, to
accommodate sets of candidates to the infinitary and probabilistic
nature of the computations we consider.

A first natural question is the one of the decidability of type
inference for our system. In the deterministic case, this question was
only addressed by Barthe and colleagues in an unpublished
tutorial~\cite{barthe-gregoire-riba:type-based-termination-tutorial},
and their solution is technically involved, especially when it comes
to dealing with the fixpoint rule.  We believe that their approach
could be extended to our system of monadic sized types, and hope that
it could provide a decidable type inference procedure for it. However,
this extension will certainly be challenging, as we need to
appropriately infer distribution types associated with AST sized walks
in the $\letrecname$ rule.

Another perspective would be to study the general, \emph{non-affine}
case. This is challenging, for two reasons. First, the system of size
annotations needs to be more expressive in order to distinguish
between various occurrences of a same function symbol in a same
probabilistic branch. A solution would be to use the combined power of
dependent types -- which already allowed Xi to formulate an
interesting type system for
termination in the deterministic case~\cite{xi:dependent-types-program-termination} -- and of
linearity: we could use \emph{linear dependent
  types}~\cite{dal-lago-gaboardi:linear-dependent-types} to formulate
an extension of the monadic sized type system keeping track of
\emph{how many} recursive calls are performed, and of the size of
\emph{each} recursive argument. The second challenge would then be to
associate, in the typing rule for $\letrecname$, this information contained
in linear dependent types with an appropriate random process. This
random process should be kept decidable to guarantee that at least
\emph{derivation} checking can be automated, and there will probably be a
trade-off between the duplication power we allow in programs and the
complexity of deciding AST for the guard in the $\letrecname$ rule.

The extension of our type system to deal with general inductive datatypes is
essentially straightforward.  Other perspectives would be to enrich
the type system so as to be able to treat coinductive data,
polymorphic types, or ordinal sizes, three features present in most
system of sized types dealing with the traditional deterministic case,
but which we chose not to address in this paper to focus on the
already complex task of accommodating sized types to a probabilistic
and higher-order framework.

\newpage
\bibliographystyle{splncs}
\bibliography{main}

\begin{thebibliography}{10}

\bibitem{manning-schutze:foundations-statistical-language-processing}
Manning, C.D., Sch{\"{u}}tze, H.:
\newblock Foundations of statistical natural language processing.
\newblock {MIT} Press (2001)

\bibitem{pearl:machine-learning}
Pearl, J.:
\newblock Probabilistic reasoning in intelligent systems - networks of
  plausible inference.
\newblock Morgan Kaufmann series in representation and reasoning. (1989)

\bibitem{thrun:robotic-mapping}
Thrun, S.:
\newblock Robotic mapping: A survey.
\newblock In: Exploring Artificial Intelligence in the New Millenium, Morgan
  Kaufmann (2002)

\bibitem{shannon-schapiro:computability-proba-machines}
de~Leeuw, K., Moore, E.F., Shannon, C.E., Shapiro, N.:
\newblock Computability by probabilistic machines.
\newblock Automata Studies \textbf{34} (1956)  183--212

\bibitem{motwani-raghavan:randomized-algorithms}
Motwani, R., Raghavan, P.:
\newblock Randomized Algorithms.
\newblock Cambridge University Press (1995)

\bibitem{barthe-gregoire-beguelin:certicrypt}
Barthe, G., Gr{\'{e}}goire, B., B{\'{e}}guelin, S.Z.:
\newblock Formal certification of code-based cryptographic proofs.
\newblock In Shao, Z., Pierce, B.C., eds.: {POPL} 2009, {ACM} (2009)  90--101

\bibitem{barthe-gregoire-heraud-beguelin:easycrypt}
Barthe, G., Gr{\'{e}}goire, B., Heraud, S., B{\'{e}}guelin, S.Z.:
\newblock Computer-aided security proofs for the working cryptographer.
\newblock In Rogaway, P., ed.: {CRYPTO} 2011. Volume 6841 of LNCS., Springer
  (2011)  71--90

\bibitem{goodman-et-al:church-language-generative-models}
Goodman, N.D., Mansinghka, V.K., Roy, D.M., Bonawitz, K., Tenenbaum, J.B.:
\newblock Church: a language for generative models.
\newblock In McAllester, D.A., Myllym{\"{a}}ki, P., eds.: {UAI} 2008, {AUAI}
  Press (2008)  220--229

\bibitem{bournez-kirchner:proba-rewrite-strategies}
Bournez, O., Kirchner, C.:
\newblock Probabilistic rewrite strategies. applications to {ELAN}.
\newblock In Tison, S., ed.: {RTA} 2002. Volume 2378 of LNCS., Springer (2002)
  252--266

\bibitem{esparza-gaiser-kiefer:probabilistic-termination-using-patterns}
Esparza, J., Gaiser, A., Kiefer, S.:
\newblock Proving termination of probabilistic programs using patterns.
\newblock In Madhusudan, P., Seshia, S.A., eds.: {CAV} 2012. Volume 7358 of
  LNCS., Springer (2012)  123--138

\bibitem{fioriti-hermans:probabilistic-termination}
Fioriti, L.M.F., Hermanns, H.:
\newblock Probabilistic termination: Soundness, completeness, and
  compositionality.
\newblock In Rajamani, S.K., Walker, D., eds.: {POPL} 2015, {ACM} (2015)
  489--501

\bibitem{chatterjee-et-al:analysis-proba-termination}
Chatterjee, K., Fu, H., Novotn{\'{y}}, P., Hasheminezhad, R.:
\newblock Algorithmic analysis of qualitative and quantitative termination
  problems for affine probabilistic programs.
\newblock In Bod{\'{\i}}k, R., Majumdar, R., eds.: {POPL} 2016, {ACM} (2016)
  327--342

\bibitem{chatterjee-et-al:termination-analysis-proba-positivstellensatz}
Chatterjee, K., Fu, H., Goharshady, A.K.:
\newblock Termination analysis of probabilistic programs through
  positivstellensatz's.
\newblock In Chaudhuri, S., Farzan, A., eds.: {CAV} 2016. Volume 9779 of LNCS.,
  Springer (2016)  3--22

\bibitem{hughes-pareto-sabry:sized-types}
Hughes, J., Pareto, L., Sabry, A.:
\newblock Proving the correctness of reactive systems using sized types.
\newblock In Boehm, H., Jr., G.L.S., eds.: {POPL}'96, {ACM} Press (1996)
  410--423

\bibitem{hofmann:mixed-modal-linear-lambda-calc}
Hofmann, M.:
\newblock A mixed modal/linear lambda calculus with applications to
  {Bellantoni}-{Cook} safe recursion.
\newblock In Nielsen, M., Thomas, W., eds.: {CSL} '97. Volume 1414 of LNCS.,
  Springer (1997)  275--294

\bibitem{girard-taylor-lafont:proofs-and-types}
Girard, J.Y., Taylor, P., Lafont, Y.:
\newblock Proofs and Types.
\newblock Cambridge University Press, New York, NY, USA (1989)

\bibitem{dal-lago-hofmann:realizability-models-icc}
{Dal Lago}, U., Hofmann, M.:
\newblock Realizability models and implicit complexity.
\newblock Theor. Comput. Sci. \textbf{412}(20) (2011)  2029--2047

\bibitem{barthe-et-al:type-based-termination}
Barthe, G., Frade, M.J., Gim{\'{e}}nez, E., Pinto, L., Uustalu, T.:
\newblock Type-based termination of recursive definitions.
\newblock MSCS \textbf{14}(1) (2004)  97--141

\bibitem{barthe-gregoire-riba:type-based-termination-tutorial}
Barthe, G., Gr{\'{e}}goire, B., Riba, C.:
\newblock A tutorial on type-based termination.
\newblock In Bove, A., Barbosa, L.S., Pardo, A., Pinto, J.S., eds.: {ALFA}
  Summer School 2008. Volume 5520 of LNCS., Springer (2008)  100--152

\bibitem{barthe-gregoire-riba:type-based-termination-sized-products}
Barthe, G., Gr{\'{e}}goire, B., Riba, C.:
\newblock Type-based termination with sized products.
\newblock In Kaminski, M., Martini, S., eds.: {CSL} 2008. Volume 5213 of LNCS.,
  Springer (2008)  493--507

\bibitem{abel:termination-checking-with-types}
Abel, A.:
\newblock Termination checking with types.
\newblock {ITA} \textbf{38}(4) (2004)  277--319

\bibitem{xi:dependent-types-program-termination}
Xi, H.:
\newblock Dependent types for program termination verification.
\newblock Higher-Order and Symbolic Computation \textbf{15}(1) (2002)  91--131

\bibitem{amadio-coupet-grimal:guard-condition-type-theory}
Amadio, R.M., Coupet{-}Grimal, S.:
\newblock Analysis of a guard condition in type theory.
\newblock In Nivat, M., ed.: {FoSSaCS}'98. Volume 1378 of LNCS., Springer
  (1998)  48--62

\bibitem{bournez-garnier:proving-past}
Bournez, O., Garnier, F.:
\newblock Proving positive almost-sure termination.
\newblock In Giesl, J., ed.: {RTA} 2005. Volume 3467 of LNCS., Springer (2005)
  323--337

\bibitem{chakarov-sankaranarayanan:proba-programs-martingales}
Chakarov, A., Sankaranarayanan, S.:
\newblock Probabilistic program analysis with martingales.
\newblock In Sharygina, N., Veith, H., eds.: {CAV} 2013. Volume 8044 of LNCS.,
  Springer (2013)  511--526

\bibitem{cappai-dal-lago:equivalences-metrics-ptime}
Cappai, A., {Dal Lago}, U.:
\newblock On equivalences, metrics, and polynomial time.
\newblock In Kosowski, A., Walukiewicz, I., eds.: {FCT} 2015. Volume 9210 of
  LNCS., Springer (2015)  311--323

\bibitem{dal-lago-parisen-tolding:ho-carac-probaptime}
{Dal Lago}, U., {Parisen Toldin}, P.:
\newblock A higher-order characterization of probabilistic polynomial time.
\newblock Inf. Comput. \textbf{241} (2015)  114--141

\bibitem{dal-lago-zorzi:probabilistic-operational-semantics-lambda-calculus}
{Dal Lago}, U., Zorzi, M.:
\newblock Probabilistic operational semantics for the lambda calculus.
\newblock {RAIRO} - Theor. Inf. and Applic. \textbf{46}(3) (2012)  413--450

\bibitem{bradzdil-et-al:one-counter-markov-decision-processes}
Brázdil, T., Brožek, V., Etessami, K., Kučera, A., Wojtczak, D.:
\newblock One-counter markov decision processes.
\newblock {21st} ACM-SIAM Symposium on Discrete Algorithms (2010)

\bibitem{dal-lago:geometry-linear-ho-recursion}
{Dal Lago}, U.:
\newblock The geometry of linear higher-order recursion.
\newblock In: {LICS} 2005, {IEEE} Computer Society (2005)  366--375

\bibitem{sabry-felleisen:reasoning-about-programs-in-cps}
Sabry, A., Felleisen, M.:
\newblock Reasoning about programs in continuation-passing style.
\newblock Lisp and Symbolic Computation \textbf{6}(3-4) (1993)  289--360

\bibitem{schrijver:theory-linear-programming}
Schrijver, A.:
\newblock Theory of Linear and Integer Programming.
\newblock John Wiley \& Sons, Inc., New York, NY, USA (1986)

\bibitem{dal-lago-grellois:monadic-affine-sized-types-full}
{Dal Lago}, U., Grellois, C.:
\newblock Probabilistic termination by monadic affine sized typing (extended
  version).
\newblock Available at \url{http://eternal.cs.unibo.it/ptmast.pdf} (2016)

\bibitem{dal-lago-gaboardi:linear-dependent-types}
{Dal Lago}, U., Gaboardi, M.:
\newblock Linear dependent types and relative completeness.
\newblock In: {LICS} 2011, {IEEE} Computer Society (2011)  133--142

\end{thebibliography}

\end{document}